\documentclass[12pt]{article}
\usepackage[T1]{fontenc}
\usepackage[utf8]{inputenc}
\usepackage[english]{babel}
\usepackage{geometry, booktabs}
\usepackage{setspace}
\usepackage{latexsym}
\usepackage{amsfonts}
\usepackage{mathrsfs,amsthm,amsmath}
\usepackage[dvipsnames]{xcolor} 
\usepackage{dsfont}
\usepackage{bbm}
\usepackage{bm}
\usepackage{graphicx}
\usepackage{bigints}
\usepackage{mathdots}
\usepackage{subfigure}
\RequirePackage[authoryear]{natbib}
\usepackage{multirow}
\usepackage{hhline}
\usepackage{amssymb}
\usepackage[normalem]{ulem} 
\usepackage{hyperref}
\usepackage[nameinlink]{cleveref}
\usepackage{lipsum}

\usepackage{authblk}

\hypersetup{
	colorlinks=true,
	urlcolor=blue,
	citecolor=blue,
	linktoc=all,
	linkcolor=blue
 } 

\theoremstyle{plain}

\newtheorem{prop}{Proposition}

\newtheorem{theorem}{Theorem}[section]
\newtheorem{lemma}[theorem]{Lemma}
\theoremstyle{definition}

\newcommand{\e}{{\rm e}}
\newcommand{\D}{{\rm d}}
\newcommand{\iid}{\stackrel{\mbox{\scriptsize iid}}{\sim}}
\newcommand{\ind}{\stackrel{\mbox{\scriptsize ind}}{\sim}}
\newcommand{\indic}{\mathbb{I}}
\newcommand{\indep}{\perp \!\!\! \perp}
\newcommand{\virgolette}[1]{`#1'}
\renewcommand{\P}{ \mathbb{P} } 

\newcommand{\E}{E}

\newcommand{\X}{\mathbb{X}}
\newcommand{\N}{\mathbb{N}}
\newcommand{\R}{\mathbb{R}}

\newcommand{\bmn}{ \bm{n} } 
\newcommand{\bmm}{ \bm{m} } 
\newcommand{\bmw}{ \bm{w} } 
\newcommand{\bmgamma}{ \bm{\gamma} } 

\newcommand{\bmX}{ \bm{X} } 

\newcommand{\Lcr}{\mathscr{L}}
\newcommand{\Pcr}{\mathscr{P}}
\newcommand{\Qcr}{\mathscr{Q}}

\newcommand{\Ccr}{\mathscr{C}}

\newcommand{\mbar}{\overline{m}}

\newcommand{\mstar}{m^{\star}}
\newcommand{\nstar}{n^{\star}}

\newcommand{\App}[1]{{\color{black}{#1}}}

\graphicspath{{img/}}

\begin{document}

\title{\textbf{Bayesian discovery of species in multiple areas}}

\author[1]{Alessandro Colombi}
\author[2]{Raffaele Argiento}
\author[3]{Federico Camerlenghi}
\author[4]{Lucia Paci}

\affil[1]{Bocconi Institute for Data Science and Analytics, Bocconi University}
\affil[2]{Department of Economics, Università degli studi di Bergamo}
\affil[3]{Department of Economics, Management and Statistics, University of Milano-Bicocca}
\affil[4]{Department of Statistical Sciences, Università Cattolica del Sacro Cuore}

\date{}

\maketitle
	
\begin{abstract}
In ecology, the description of species composition and biodiversity calls for statistical methods that involve estimating features of interest in unobserved samples based on an observed one. In the last decade, the Bayesian nonparametrics literature has thoroughly investigated the case where data arise from a homogeneous population. In this work, we propose a novel framework to address heterogeneous populations, specifically dealing with scenarios where data arise from two areas. This setting significantly increases the mathematical complexity of the problem and, as a consequence, it has received limited attention in the literature. While early approaches leverage computational methods, we provide a distributional theory for the in-sample analysis of any observed sample and enable out-of-sample prediction for the number of unseen distinct and shared species in additional samples of arbitrary sizes. The latter also extends the frequentist estimators, which solely deal with one-step-ahead prediction. Furthermore, our results can be applied to address sample size determination in sampling problems aimed at detecting distinct and shared species. Our results are illustrated in a real-world dataset concerning a population of ants in the city of Trieste.
\end{abstract}


\section{Introduction}
\label{section:intro}
In ecology, there is often a strong interest in describing species composition and biodiversity and, when the study involves more than a single area, in comparing them. Given that the concept of biodiversity lacks a single universally accepted mathematical definition \citep{colwell2004}, many similarity indexes have been proposed in the literature. For a single assemblage of data, the most widely used indices are Shannon's and Simpson's diversity, while Jaccard's and Sørensen's indices are classic choices when comparisons across multiple areas are investigated. A detailed overview of these and many other alternatives can be found in \cite{Chao2006indici}.
Regardless of the number of areas considered, such investigations generally follow a two-step process. The first step is to choose a sampling strategy to detect as many species as possible. Since both time and financial constraints often limit sampling efforts, it is crucial to develop statistical tools that, based on the current sample, can estimate how many additional observations are needed to discover new species. These predictions enable researchers to assess whether additional sampling efforts are feasible and justified in terms of time and cost.
Due to these practical limitations, collected samples rarely capture the full diversity of the population, leaving some species undetected. Consequently, biodiversity indices estimated from sample data are often negatively biased. Therefore, the second step is usually to estimate, based on the observed data,
the species richness, i.e., the total number of species, both observed and unobserved.

This work extends beyond ecology and is relevant to many other disciplines. Thus, the term “species” can be interpreted in a broader sense. For instance, it can refer to fields such as topic modelling \citep{efron76}, typos in texts \citep{nayak89}, bugs in computer code \citep{chao93}, or genomics \citep{Mao2004}.

\subsection{Estimators for a single area}
\label{section:intro_1area}
The case in which the investigation focuses on a single area is a long-established and extensively discussed problem in statistics, dating back to \cite{Fisher43}. Here, the $n\geq 1$ observed individuals represent a random sample $(X_1,\ldots,X_n)$ drawn from an unknown discrete distribution $P$. In this setting, the celebrated Good-Turing estimator \citep{goodturing53} provides an estimate of the probability that the $(n+1)$th observation coincides with a species whose frequency in the original sample is exactly equal to $f$, with $f\geq0$. In particular, for $f=0$, this estimates the probability of observing a new, previously undetected species as the relative frequency of the species observed exactly once in the sample, commonly known as singletons.
The Good-Toulmin estimator \citep{GoodToulmin} represents a $m$-steps ahead generalisation for the probability of discovering a new species.
These estimators are nonparametric since they do not rely on any assumption for the unknown generating distribution $P$, making them flexible and easy to use both for deriving stopping rules in sampling strategies \citep{Rasmussen79} and for estimating species richness \citep{Chao1est}. See \cite{Orlitsky16} for a review of generalisations and improvements of the Good-type estimators.

An alternative Bayesian nonparametric approach has been introduced by \cite{Lij(07)}. In this framework, the authors only require $(X_1,\ldots, X_n)$ to be an exchangeable sequence. By de Finetti's representation theorem, this hypothesis is equivalent to assuming that the unknown distribution $P$ is a random probability measure governed by a prior distribution $\Pcr$.
Given $P$, the $n$ observations are assumed to be independent and identically distributed according to $P$.
In particular, \cite{Lij(07)} studied priors $\Pcr$ belonging to the broad class of Gibbs-type priors \citep{deblasi2015}, which encompasses notable examples such as the Dirichlet process \citep{ferguson73} and the two parameters Poisson-Dirichlet model \citep{pitman95}.
In this setting, \cite{favaro2012} presents a Bayesian generalisation of the Good-Turing and Good-Toulimin estimators. Based on an observed sample of size $n$ and an additionally unobserved sample of size $m\geq0$, they provide an estimator for the probability that the $(n+m+1)$th observation coincides with a species whose frequency, within the sample of size $n+m$, is exactly $f$. Thus, for $m=0$ and $f=0$, this is a Bayesian analogue of the Good-Turing and Good-Toulmin estimators, respectively. As for their frequentist counterparts, these estimators are crucial in sampling problems.

Regarding the estimation of species richness, \cite{favaro2009} approached the problem as an extrapolation of the unseen species problem. Specifically, for each $m\geq0$, they derive an estimator for the number of new distinct species that would be observed in an additional unobserved sample $(X_{n+1},\ldots, X_{n+m})$. As $m$ varies, this estimate forms the so-called extrapolation curve \citep{Gotelli2001}, and species richness can be considered the limit for $m\rightarrow\infty$. However, whenever $\Pcr$ is chosen in the family of Gibbs-type priors, this limit may either be infinite or finite, dividing the family into two subgroups. The choice of $\Pcr$ depends on which hypothesis is more realistic for the specific application.
Notable examples of the first group, characterised by unlimited growth, include the Dirichlet process and the two-parameter Poisson–Dirichlet model. In contrast, the second group, namely the Gibbs-type priors with negative parameters, encompasses Gnedin's model \citep{gnedin2010} and the finite Dirichlet process \citep{argiento2022annals}.
Moving beyond the Gibbs-type priors framework, \cite{zito2023} proposed a tunable Bayesian nonparametric approach between these two regimes.
We refer to \cite{BalocchiFavaro24} for an exhaustive review on this topic.

\subsection{Estimators for multiple areas}
\label{section:intro_2area}
The study of estimators for multiple areas is much less extensive than for a single population. Even with just $d=2$ different areas,
the mathematical complexity increases significantly because we need to estimate the number of shared species, i.e., those occurring in both areas. Indeed, this statistic is crucial for evaluating the similarity or dissimilarity between the two areas.
To get an idea of the increased difficulty, we point out that
the total number of distinct species can be expressed as the sum of two quantities: the number of observed distinct species (a known quantity) and the number of unobserved species (an unknown quantity). Rather, for shared species, there are three distinct unknowns: species that are shared but undetected in both areas and species that are shared but observed in only one of the two areas.
The exponential increase in the number of unknowns (which is $2^d-1$) explains why the case of two areas is the most studied in the literature. For scenarios with $d>2$ areas, comparisons are made between all possible pairs. In fact, commonly used multi-area indices are designed for the two-area case, since summarising comparisons across three or more areas into a single index is challenging, see \cite{ChaoLB2009}. 

In the frequentist framework, the model assumes two discrete and unknown probability distributions, $P_1$ and $P_2$, with potentially different supports, whose intersection is not empty and is ordered in such a way that the first $S$ species are the shared ones. 
Then, two random samples $\bmX_j=\left(X_{j,1},\ldots,X_{j,n_j}\right)$, $j=1,2$, of sizes $n_1$ and $n_2$ are taken from $P_1$ and $P_2$, respectively.
In this setting, \cite{Yue2012}, \cite{Chuang2015} and \cite{chao2017} proposed Turing-type estimators for the probability of discovering a new shared species in the next pair of observations. However, the three studies addressed only the case of one-step-ahead prediction, despite \cite{Yue2012} highlighting the need for sampling strategies where $m>1$ pairs of observations are considered. Consequently, a generalisation of the Good-Toulmin estimator for the shared species problem remains open.
On the other hand, various methods have been proposed to estimate shared species richness.

The seminal work by \cite{Chao2000} introduced a method based on sample coverage, which has been improved in two ways. The first improvement relies on Laplace's approximations \citep{Chao2006Laplace}, while the second version presents a lower bound of the initial method \citep{ChaoLB2009}, which proves to be more useful in practice. Finally, \cite{Chuang2015} presents a jackknife estimator for the richness of shared species. All these contributions are inferential in nature: they aim to estimate an unknown population-level quantity and rely on asymptotic assumptions, such as sufficiently large sample sizes. In other words, they answer the question, ``How many species that we have not yet observed actually exist?'', even though, under their assumptions, some of these species may never be encountered.

In the Bayesian nonparametric framework, a natural extension of the exchangeability assumption to the case of two areas is to assume that observations are partially exchangeable, meaning they are exchangeable within each group but not across the two groups.
As in the single-group case, by virtue of de Finetti's representation theorem, this condition implies that the unknown distributions $\left(P_1, P_2\right)$ form a vector of dependent random probability measures
governed by a prior distribution $\Qcr_2$. Given $\left(P_1, P_2\right)$, the observations are independent between groups and independent and identically distributed according to $P_1$ and $P_2$, respectively.
Although Gibbs-type priors are an established choice for exchangeable data, the range of choices for $\Qcr_2$ is much broader in the two-group case.
A key difficulty in this setting is that it is not enough to specify solely the marginal priors for $P_1$ and $P_2$, but modelling the dependence between these two distributions is crucial. Numerous constructions have been proposed in the literature, 
see e.g., \cite{Muller2004}, \cite{teh2006hierarchical}, \cite{LijoiNipotiPrunster2014}, \cite{camerlenghi2019distribution}, and \cite{casarin2020}. For a comprehensive overview, see \cite{Quintana2022}, while we refer to the recent work by \cite{franzolini25} for a unified framework of multivariate species sampling models, offering new directions for this area of research.
These approaches have been employed as Bayesian nonparametric priors in model-based clustering via mixture models. However, their application in species sampling problems has been limited by the lack of closed-form estimators, particularly for predicting shared species across areas.
To the best of our knowledge, the only exceptions are \cite{bacallado2015}  and \cite{cam2017_multisample}, where posterior inference relies mainly on computational methods.

\subsection{Summary of the contribution}
\label{section:intro_summary}
In this paper, we take a step forward and develop a Bayesian nonparametric methodology to study the problem of unseen distinct and shared species when observations are collected in two different areas. Our approach is predictive in nature and aims to answer the question: ``How many species not yet observed in the two areas will be discovered in a future sample?''
The proposed method is based on the Vector of Finite Dirichlet Process ($\operatorname{Vec-FDP}$), recently introduced by \cite{colombi2023mixture}, and provides closed-form expressions for estimating the quantity of interest. The $\operatorname{Vec-FDP}$ prior assumes 
a common species composition in the two areas, and that the number of species is finite, yet random, while allowing for variation in their proportions.
Working with finite samples, not all species will necessarily be shared between areas, allowing for the existence of area-specific species.


Specifically, the main goals of this work are to derive a distributional theory for: (i) in-sample analysis for any observed samples of sizes $n_1$ and $n_2$, and (ii) out-of-sample predictions of the number of unseen distinct and shared species in additional unobserved samples of sizes $m_1$ and $m_2$. Hence, our results can be used to address sample size determination in sampling problems designed to detect distinct and shared species.
Moreover, we are able to evaluate the sample coverage for distinct and shared species, i.e., the proportion of species observed in a sample \citep{goodturing53}, both in-sample and out-of-sample.
A remarkable finding is that our result holds for any finite future sample sizes $m_1$ and $m_2$, thereby extending the results in the frequentist literature, which only provides solutions for one-step ahead predictions. 
Unlike previous Bayesian contributions, our methodology provides a joint description of all the relevant quantities and yields an efficient computational strategy for evaluating the out-of-sample quantities of interest without relying on expensive and approximate computational methods. 
The proposed estimators are illustrated using both synthetic data, highlighting similarities and differences with frequentist approaches, and a real-world dataset. Specifically, we analyse an ant population in Trieste, sampled from two parks: one just outside the city and the other in its centre.

The rest of the paper is organized as follows. Section \ref{section:VecFDP} describes the Vec-FDP prior model and its properties. The main theoretical results about the in-sample analysis and the out-of-sample prediction are presented in Section \ref{section:prior} and Section \ref{section:posterior}, respectively. Section \ref{subsection:indici_intro} provides an interpretation of the model parameters in terms of diversity indices and outlines two strategies for their estimation. A simulation study is presented in Section \ref{section:simulation_summary} while the results of the analysis of a real dataset are given in Section \ref{section:application}. We conclude with a discussion in Section \ref{section:discussion}. 

\section{Vector of Finite Dirichlet Processes}
\label{section:VecFDP}
Consider a sample $\bmX = (\bmX_1,\bmX_2)$ from $d=2$ partially exchangeable sequences, that is, $\bmX_j~=~\left(X_{j,1},\dots,X_{j,n_j}\right)$, for $j=1,2$, each taking values in a Polish space $\X$.
This is equivalent to assuming that the distribution of $\bmX$ is invariant under permutations occurring within elements of vectors $\bmX_1$ and $\bmX_2$, but it is not invariant under permutations among them; see \cite{camerlenghi2019distribution} for a recent contribution.
Hereafter, we refer to observations coming from two areas or groups, keeping in mind that the definition of a \virgolette{group} is problem-specific and can be interpreted in a broader sense.

By virtue of de Finetti's representation theorem, assuming partial exchangeability is equivalent to saying that there exists a vector $\left(P_1,P_2\right)$ of random probability measures such that
\begin{equation}
	\label{eqn:partial_ex}
	\begin{split}
		(X_{1,i_1}, X_{2, i_2}) \mid \left(P_1,P_2\right)
		&\ind
		P_1 \otimes P_2,
		\quad
		\left(i_1, i_2\right) \in \{1,\dots,n_1\}\times\{1,\dots,n_2\} \\
		\left(P_1,P_2\right) &\sim \Qcr_2,
	\end{split}
\end{equation}
where $\Qcr_2$ denotes the prior distribution of the vector in a Bayesian setting.
The definition of $\Qcr_2$ is a crucial issue as it governs the 
properties of the statistical model that generates the data $\bmX = \left(\bmX_1,\bmX_2\right)$. In this work, we study the model in Equation \eqref{eqn:partial_ex} under the Vector of Finite Dirichlet Process (Vec-FDP) prior, as introduced in \cite{colombi2023mixture}.
In this framework, the random probability measures $P_1$ and $P_2$ are constructively defined as follows,
\begin{equation}
	\label{eqn:Pj_def}
	P_j(A)
	\, \stackrel{a.s.}{=} \,
	\sum_{m=1}^{M} w_{j,m}\delta_{\tau_m}
	(A),
	\quad
	j = 1,2,
\end{equation}
where $A$ is a measurable set, $\delta_{\tau_m}$ stands for the delta-Dirac mass at $\tau_m$,
and the nonnegative weights
$\bmw_j~=~\left(w_{1,1},\dots,w_{j,M}\right)$ are the species' proportions specific to population $j$, which sum up to one almost surely.
Then, defining a prior for $\left(P_1,P_2\right)$ is equivalent to place a prior over $\left(M,\tau_1,\dots,\tau_M,\bmw_1,\bmw_2\right)$.
Here, conditionally on $M$, $\left(\tau_1,\dots,\tau_M\right)$ are common random atoms across the two random probability measures, which are assumed to be independent and identically distributed with a common distribution $P_0$, that is a diffuse probability measure on $\X$.
Moreover, we assume $\bmw_1$ independent of $\bmw_2$ and such that
$\bmw_j \mid M \sim \operatorname{Dir}_M\left(\gamma_j,\dots,\gamma_j\right)$, for $j=1,2$, where $\operatorname{Dir}_M\left(\gamma_j,\dots,\gamma_j\right)$ denotes the $M$-dimensional symmetric Dirichlet distribution with group-specific parameter $\gamma_j$.

Finally, $M$ is supposed to be a positive integer-valued random variable whose probability mass function is denoted as $q_M$.
For the sake of simplicity, we follow \cite{colombi2023mixture} and choose a $1$-shifted Poisson distribution, namely,
$ q_M(m) =  e^{-\Lambda}\Lambda^{m-1}/(m-1)!$, for all integers $m\geq1$, and denote it as $M\sim\operatorname{Pois}_1\left(\Lambda\right)$.
Nevertheless, we point out that all our results hold for any distribution on positive integers.
This completes the prior specification, and we write
\begin{equation}
	\label{eqn:VecFDP_prior}
	\begin{split}
		\left(P_1,P_2\right) \sim \operatorname{Vec-FDP}\left(\Lambda,\bmgamma,P_0\right),
	\end{split}
\end{equation}
where $\bmgamma = \left(\gamma_1,\gamma_2\right)$. 
By placing a $\operatorname{Vec-FDP}$ prior on $\left(P_1,P_2\right)$, we assume that the two groups share the same finite, yet random, number of species $M$ that appear with different frequencies $w_{j,m}$ in the two groups, i.e., in the two areas. Note that the assumption of full sharing of atoms is a popular strategy in Bayesian nonparametrics, see e.g., \cite{teh2006hierarchical}, \cite{CAM2021} and many examples in \cite{franzolini25}. Since a Dirichlet prior is chosen for the weights, the same $M$ species would eventually be observed in the two groups if we were able to obtain an infinite sample from $P_1$ and $P_2$. However, in practice, we always work with finite samples of sizes $n_1\geq1$ and $n_2\geq1$. 

Let $n~=~n_1+n_2$ be the total number of observations. Let the integers label the observations according to their order of arrival by group, that is,
observations are indexed first by group $j = 1, 2$, and then by within-group order of arrival.
Within each group $j$, a random number of $K_{j,n_j}$ distinct species will be observed, which we label $\bmX^*_j~=~\left\{X^*_{j,1},\dots,X^*_{j,K_{j,n_j}}\right\}$. Let $K_{j,n_j}=r_j$ be a realisation of this random variable. 
Since $P_1$ and $P_2$ share the same support, ties between the two groups are expected, i.e., $\P\left(X^*_{1,k} = X^*_{2,k^\prime}\right)>0$.
Thus, the set of labels for the distinct species in the entire sample is obtained as
$\bmX^{**}=\bmX^*_1\cup\bmX^*_2=\left\{X^{**}_1,\dots,X^{**}_{\mathcal K_{n_1,n_2}}\right\}$, where $\mathcal K_{n_1,n_2}$ denotes the overall number of distinct species, which is, in general, smaller than the sum of $K_{1,n_1}$ and $K_{2,n_2}$. Such a difference is due to the random number of species shared between the two groups, namely
\begin{equation}
	\label{eqn:linear_rel_prior}
	\mathcal S_{n_1,n_2}=K_{1,n_1}+K_{2,n_2}-\mathcal K_{n_1,n_2}.
\end{equation}
In the following, we let $\mathcal K_{n_1,n_2}=r$ and $\mathcal S_{n_1,n_2}=t$ denote the realisations of the distinct and shared number of species, respectively. Moreover, we refer to group-specific quantities as \textit{local} quantities, while we call \textit{global} quantities those related to the joint sequence $\bmX$. Therefore, $K_{j,n_j}$ is also named the local number of distinct species in group $j$ while $\mathcal K_{n_1,n_2}$ is the global number of distinct species. \App{We refer to Section \ref{app:prior_quantities} for a more detailed description of these quantities.}

\subsection{pEPPF and predictive distribution}
\label{section:predictive}
The marginal law of the sample $\bmX = (\bmX_1,\bmX_2)$ from model \eqref{eqn:partial_ex} with $\Qcr_2$ chosen as in Equation \eqref{eqn:VecFDP_prior} is uniquely determined by the species labels $\bmX^{**}$ and their abundances.
The latter are defined by the vectors of frequency counts $\bmn_j = (n_{j,1},\ldots,n_{j,r})$, where $n_{j,l}$ represents the number of observations in the $j$ th group
that coincide with the $l$ distinct value, indexed according to the order of arrival by groups. 
These counts must satisfy the following constraints:
\begin{equation}
	\label{eqn:cluster_vincoli}
	n_{j,l}\geq0,
	\quad
	n_{1,l}+n_{2,l} >0,
	\quad
	\sum_{l=1}^{r} n_{j,l} = n_j
	\quad
	l = 1,\dots,r;\ j = 1,2.
\end{equation}
Consistent with the description in Section \ref{section:VecFDP}, some of the counts $n_{j,l}$ may also be zero.

Extending the results of \cite{pitman1996}, one obtains that the marginal likelihood $\Lcr\left(\bmX\right)$, which is obtained by integrating $\left(P_1,P_2\right)$ out of the model \eqref{eqn:partial_ex}, admits the following factorisation:
$ \Lcr\left(\bmX\right)=\Lcr\left(\bmn_1,\bmn_2,\bmX^{**}\right)=\Lcr\left(\bmn_1,\bmn_2\right)\prod_{l=1}^{r} P_0(\D X^{**}_{l}),
$ see \cite{franzolini25}.
In what follows, we drop the value of the labels $\bmX^{**}$ as it is not relevant for our purposes. Consequently, the main object of interest is the law of abundances $\Lcr\left(\bmn_1,\bmn_2\right)$. 
The law of the random partition induced by a partially exchangeable sequence $\left(\bmX_1,\bmX_2\right)$ is described through a probabilistic object called the partially Exchangeable Partition Probability Function (pEPPF); see, e.g., \cite{LijoiNipotiPrunster2014} and \cite{camerlenghi2019distribution}.
Under a $\operatorname{Vec-FDP}$ prior, the pEPPF equals
\begin{equation}
	\label{eqn:peppf}
	\Pi_r^{(n)}\left(\bmn_1,\bmn_2\right)  \ = \
	V_{n_1,n_2}^r \
	\prod_{j=1}^d \prod_{l=1}^r  (\gamma_j)_{n_{j,l}} \, ,
\end{equation}
where $d=2$, $\left(\bmn_1,\bmn_2\right)$ satisfy the constraints given in Equation \eqref{eqn:cluster_vincoli} and
\begin{equation}
	\label{eqn:Vprior}
	\begin{split}
		V_{n_1,n_2}^r \ = \
		\sum_{m=1}^\infty
		\
		(m)_{r\downarrow}
		\
		q_M(m)
		\
		\prod_{j=1}^d  \frac{1}{(\gamma_jm)_{n_j}}.
	\end{split}
\end{equation}
\App{
	See Section \ref{app:proof_peppf} and \cite{colombi2023mixture} for an alternative expression.
}
In this work, we let $(m)_{r\downarrow}~=~m(m-1)\dots(m-r+1)$ denote the falling factorial of order $r$ and $(x)_{n}~=~\Gamma(x+n)/\Gamma(x)$ is the Pochhammer symbol, also known as the rising factorial when $n$ is a natural number. The sum in Equation \eqref{eqn:Vprior} can start from $m=r$ since $(m)_{r\downarrow} = 0$ for all $m<r$.
See Section \ref{section:Vcoeff} for further analysis and properties of the $V$ coefficients, such as convergence, asymptotics, and a recurrence relationship.

The pEPPF in Equation \eqref{eqn:peppf} represents the sampling model we assume generates the data. However, since its form is rather involved and hard to interpret, it is common to describe the data-generating mechanism by inspecting the predictive distributions, which are a generalisation of the well-known Chinese restaurant franchise process introduced in \cite{teh2006hierarchical}.
Confining our attention to the first group, \cite{colombi2023mixture} showed that, taken $(\bmn_1,\bmn_2)$ observations as in Equation \eqref{eqn:cluster_vincoli}, the $(n_1+1)$th observation can either be equal to one of the \virgolette{old} (already observed) species, with probability proportional to $V_{n_1+1,n_2}^r \left(n_{1,l}+\gamma_1\right)$, or to a \virgolette{new} (never observed before) species, whose label is drawn from $P_0$, with probability proportional to $V_{n_1+1,n_2}^{r+1}\gamma_1$. Namely,
\begin{equation*}
	\label{eqn:predictive_j}
	\P\left( X_{1,n_1+1}\in\,A\,\mid \bmX \right)
	=
	\frac{V_{n_1+1,n_2}^r}{V_{n_1,n_2}^r}
	\sum_{l=1}^r \left(n_{1,l}+\gamma_1\right)\delta_{X^{**}_{l}}(\,A\,)
	+
	\frac{V_{n_1+1,n_2}^{r+1}}{V_{n_1,n_2}^r}\gamma_1
	P_0(\,A\,),
\end{equation*}
The case of a new observation in group $2$ trivially follows.
	In Section \ref{app:joint_CRFP}, we generalise this result as we report the predictive distribution for a new pair of observations, one in each group.
For the sake of brevity, we present here a concise version of the probability of a new pair of observations that highlights
the distinction between old and new species, thereby neglecting which of the $r$ species is selected in the case of an old species.
The unnormalised probabilities are
summarised in Table \ref{tab:joint_CRFP}, where 
for each group $j$, $q_j^{\text{old}} = \sum_{l=1}^r(n_{j,l} + \gamma_j)$ the weight 
of generating
an observed species, and $q_j^{\text{new}} = \gamma_j$ as the weight associated with 
a new species. The normalising constant is $V^{r}_{n_1+1,n_2+1}$.
\begin{table}
	\centering
	\caption{Unnormalised probabilities of observing an old species and a new species in each group when a new pair of observations is considered.}
	\label{tab:joint_CRFP}
	\renewcommand{\arraystretch}{1.2}
	\begin{tabular}{lcccl}
		\hline
		& \multicolumn{1}{l}{} & \multicolumn{1}{l}{} & \multicolumn{1}{l}{Group 1} &  \\
		&  & Old & \multicolumn{2}{c}{New} \\ \hline
		\multicolumn{1}{l}{} & Old & $V^{r}_{n_1+1,n_2+1}q_1^{\text{old}}q_2^{\text{old}}$ & \multicolumn{2}{c}{$V^{r+1}_{n_1+1,n_2+1}q_1^{\text{new}}q_2^{\text{old}}$} \\
		\multicolumn{1}{l}{Group 2} & \multicolumn{1}{l}{} & \multicolumn{1}{l}{} & \multicolumn{2}{l}{} \\
		\multicolumn{1}{l}{} & New & $V^{r+1}_{n_1+1,n_2+1}q_1^{\text{old}}q_2^{\text{new}}$ & \multicolumn{2}{c}{$\left(V^{r+1}_{n_1+1,n_2+1} + V^{r+2}_{n_1+1,n_2+1}\right)q_1^{\text{new}}q_2^{\text{new}}$} \\ \hline
	\end{tabular}
\end{table}

Table \ref{tab:joint_CRFP} includes all four possible cases, which involve either generating a new observation or not in each of the two groups. The apex of each coefficient $V$ indicates the total number of distinct species in the larger sample of size $(n_1+1,n_2+1)$.
    Specifically, we highlight that $r$ can be increased by a single unit even in the scenario where a new species is observed in both groups. This is because the new species could be the same in both groups, i.e., a previously unobserved shared species.
	\section{In-sample analysis}
	\label{section:prior}
	\subsection{Correlation}
	\label{subsection:correlation}
	A natural question that arises when moving from analysing a single group to multiple groups is quantifying the interaction between these two.
	Here, we aim to provide a quantitative answer by examining the statistical dependence between the data generating models for the two groups, namely between $P_1$ and $P_2$.
	Expanding the result of \cite{colombi2023mixture}, we obtain a closed-form expression for the correlation between $P_1$ and $P_2$ when evaluated on the same measurable set $A$, namely
	\begin{equation}
		\label{eqn:correlation}
		\text{cor}\left(P_1(A),P_2(A)\right)
		\ = \
		\frac{\E\left(1/M\right)}
		{
			\sqrt{\left(1+\gamma_1\right)\left(1+\gamma_2\right)}
			\sqrt{\E\left(\frac{1}{1+\gamma_1M}\right)
				\E\left(\frac{1}{1+\gamma_2M}\right)}
		} \, .
	\end{equation}
	The expression \eqref{eqn:correlation} does not depend on the choice of the set $A$.
	Thus, it may be considered an overall measure of dependence between the two random probability measures.
	Furthermore, if $M\sim\operatorname{Pois}_1\left(\Lambda\right)$, then the numerator in Equation \eqref{eqn:correlation} is available in closed-form and it equals
	\begin{equation}
		\label{eqn:Exp1overM_Pois}
		\E\left(1/M\right)= \Lambda^{-1}\left(1-\e^{-\Lambda}\right) \, .
	\end{equation}
	Additionally, the following limiting values hold:
	\begin{equation*}
		\label{eqn:corr_limits}
		\begin{split}
			\lim_{\gamma_1,\gamma_2\rightarrow0}
			\text{cor}\left(P_1(A),P_2(A)\right)
			\ = \
			\E\left(1/M\right)
			,\quad
			\lim_{\gamma_1,\gamma_2\rightarrow+\infty}
			\text{cor}\left(P_1(A),P_2(A)\right) \ = \ 1 \, .
		\end{split}
	\end{equation*}
	The above limits suggest an interpretation of the $\gamma_j$'s as homogeneity parameters:
	large values of $\gamma_j$'s indicate similar groups that share most of the distinct species. Conversely, small values of $\gamma_j$'s result in the minimum value of Equation \eqref{eqn:correlation}. Interestingly, this value is not zero but depends on the choice of the prior distribution of $M$. Intuitively, larger expected values of $M$ drive the correlation between the two populations toward zero.
	
	\subsection{In-sample statistics}
	\label{subsection:formule_priori}
	In this section, we investigate the distributions of the most relevant in-sample statistics for a sample $\bmX = (\bmX_1,\bmX_2)$ of sizes $n_1$ and $n_2$ from model \eqref{eqn:partial_ex} under the $\operatorname{Vec-FDP}$ prior given in Equation \eqref{eqn:VecFDP_prior}.
	Mathematically, this is equivalent to studying the properties of the Bayesian nonparametric prior. We recall the quantities of interest introduced in Section \ref{section:VecFDP} \App{and further described in \ref{app:prior_quantities}}:
	the local number of distinct species, $\mathcal K_{j,n_j}$, one for each group, the global number of distinct species $\mathcal K_{n_1,n_2}$, and the number of shared species $\mathcal S_{n_1,n_2}$.
	Previous works derived the marginal distribution of both $K_{j,n_j}$ \citep{Lij(07),argiento2022annals} and $\mathcal K_{n_1,n_2}$ \citep{colombi2023mixture}, \App{ as
		reported in Equations \eqref{eqn:K12_prior} and \eqref{eqn:Kj_prior}}, respectively.
	Conversely, the distribution of $\mathcal S_{n_1,n_2}$ has not yet been derived, although it is linearly related to the number of local and global distinct species by Equation \eqref{eqn:linear_rel_prior}.
	This is also because the joint distribution of $K_{1,n_1}$, $K_{2,n_2}$ and $\mathcal K_{n_1,n_2}$ is required to derive the distribution of $\mathcal S_{n_1,n_2}$.
	Theorem \ref{thm:jointprior_d2} overcomes this limitation.
	\begin{theorem}
		\label{thm:jointprior_d2}
		Let $\bmX=\left(\bmX_1,\bmX_2\right)$ be a sample of sizes $n_1$ and $n_2$ from model \eqref{eqn:partial_ex} under the $\operatorname{Vec-FDP}$ prior in Equation \eqref{eqn:VecFDP_prior}.
		Then, the joint distribution of the local number of distinct species $K_{1,n_1}$ and $K_{2,n_2}$ and the global number of distinct species $\mathcal K_{n_1,n_2}$ equals
		\begin{equation}
			\label{eqn:K1K2K12_prior_d2}
			\begin{split}
				\P\left(\,\mathcal K_{n_1, n_2} = r,\,
				K_{1, n_1} = r_1,\,
				K_{2, n_2} = r_2 \,\right)
				= V^r_{n_1,n_2}\,
				\frac{r_1!r_2!}{r^*_1!r^*_2!t!}
				\,
				\prod_{j=1}^2
				\lvert C(n_j, r_j; -\gamma_j )\rvert,
			\end{split}
		\end{equation}
		for $r\in \left\{1,\ldots,r_1+r_2\right\}$ and
		$ r_j \in\left\{\{1, \ldots, \min\{r,n_j\} \right\}\,(j=1,2)$ and where we defined $r^*_1~=~r~-~r_2$, $r^*_2~=~r~-~r_1$ and $t~=~r_1~+~r_2-r$.
		The coefficient $C (\cdot,\cdot;\cdot)$ in Equation \eqref{eqn:K1K2K12_prior_d2} denotes the generalised factorial coefficient, as defined in \cite{chara2002}.
	\end{theorem}
	
	\App{The proof of Theorem \ref{thm:jointprior_d2} is provided in Section \ref{app:proof_joint_prior}.}
	The generalised factorial coefficients in Equation \eqref{eqn:K1K2K12_prior_d2} can be computed via the triangular recurrence relationships described in \cite{chara2002}.
	The quantities $r^*_1$ and $r^*_2$ represent the number of species observed only in groups $1$ and $2$, respectively, while $t$ is the number of shared species between the two groups.
	\App{See Section \ref{app:Cnumbers} for further details on generalized factorial coefficients. }
	Equation \eqref{eqn:K1K2K12_prior_d2} is a joint distribution and enables the evaluation of all other linearly dependent in-sample statistics. For instance, the distribution of the shared species, $\P\left(\mathcal S_{n_1,n_2} = t\right)$, is obtained by summing the expression in Equation \eqref{eqn:K1K2K12_prior_d2} for all $r,r_1,r_2$ such that $t=r_1+r_2-r$, $t = 0,\ldots,\min\{n_1,n_2\}$. Alternatively, one may draw Monte Carlo samples from Equation \eqref{eqn:K1K2K12_prior_d2} to obtain a Monte Carlo estimation of the distribution of interest.
	Furthermore, since Equation \eqref{eqn:linear_rel_prior} is a linear relationship, the choice of $\mathcal S_{n_1,n_2}$ as a dependent variable is arbitrary. If one is interested in only the global quantities, it is possible to derive the joint prior $\P\left(\mathcal K_{n_1,n_2} = r, \mathcal S_{n_1,n_2} = t\right)$, integrating out one of the local quantities. \App{Such expression is reported in Section \ref{app:insample_marginal}}. Then, the marginal distribution of $\mathcal S_{n_1,n_2}$ is computed exactly as
	$
	\P\left(\mathcal S_{n_1,n_2} = t\right) =
	\sum_{r=t}^{n}
	\P\left(\mathcal K_{n_1,n_2} = r, \mathcal S_{n_1,n_2} = t\right)
	$.

	\section{Out-of-sample prediction}
	\label{section:posterior}
	The present section addresses the task of out-of-sample prediction of new distinct and shared species.
	Given $\bmX=\left(\bmX_1,\bmX_2\right)$, an additional sample comprising $m_1$ and $m_2$ individuals is considered, resulting in enlarged samples of sizes $n_1+m_1$ and $n_2+m_2$, namely $\left(X_{j,1},\dots,X_{j,n_j+m_j}\right)\,$, for $j=1,2$.
	Specifically, our interest lies in assessing the probability of discovering new local and global distinct species, as well as new shared species. Regarding the new local species, we follow the definition of \cite{Lij(07)}, i.e., $K^{(n_j)}_{j,m_j} = K_{j,n_j+m_j} - K_{j,n_j}$, for $j=1,2$. This definition can be extended to the global number of new distinct species, here defined as
	$\mathcal K^{(n_1,n_2)}_{m_1,m_2} = \mathcal K_{n_1+m_1,n_2+m_2} - \mathcal K_{n_1,n_2}$. We point out that the latter counts only the species that are unobserved in both groups. For example, observing a species that has already been spotted in the first group but was missing in the second would increase $K^{(n_2)}_{2,m_2}$ by one and leave $\mathcal K^{(n_1,n_2)}_{m_1,m_2}$ unchanged.
	Further caution is required when interpreting the number of new shared species, defined as $\mathcal S^{(n_1,n_2)}_{m_1,m_2} = \mathcal S_{n_1+m_1,n_2+m_2} - \mathcal S_{n_1,n_2}$. This not only takes into account the species that are new in both groups, but also the species that were first observed in one group only and belong to the additional sample of the other group.
	The relationship between such posterior quantities mimics that of their prior counterparts 
	namely,
	\begin{equation}
		\label{eqn:linear_rel_post}
		\mathcal S^{(n_1,n_2)}_{m_1,m_2} = K^{(n_1)}_{1,m_1} + K^{(n_2)}_{2,m_2} - \mathcal K^{(n_1,n_2)}_{m_1,m_2}.
	\end{equation}
	\App{We refer to Section \ref{app:post_quantities} for a more detailed description of these predictive quantities.}
	
	\subsection{Posterior of the total number of species}
	\label{subsection:qM_post}
	The model introduced in Section \ref{section:VecFDP} assumes a finite number of species $M$, which is usually unknown and is assumed to be random. As data are observed, at least $\mathcal K_{n_1,n_2} = r$ species must exist, hence it is convenient to reparametrise the total number of species as
    $M \ = \ r + M^\star$, where $M^\star$ is interpreted as the random number of yet not discovered species. Following the posterior representation theorem for the $\operatorname{Vec-FDP}$ prior in \cite{colombi2023mixture}, the posterior distribution of $M^\star$ has a probability mass function $q^{\star}_{M\mid \bmX}$ defined as
	\begin{equation}
		\label{eqn:qMpost}
		\begin{split}
			q^{\star}_{M\mid \bmX}(\mstar) \ = \
			\frac{1}{V_{n_1,n_2}^r}
			(\mstar+r)_{r\downarrow}
			q_M(\mstar+r)
			\prod_{j=1}^d \frac{1}{(\gamma_j (\mstar+r))_{n_j}},
		\end{split}
	\end{equation}
	where  $\mstar= 0, 1 , 2\ldots$. Specifically, $M^{\star}$ may equal zero since $q^{\star}_{M\mid \bmX}(0)>0$. See Section \ref{app:proof_qMpost} for the equivalence between Equation \eqref{eqn:qMpost} and the corresponding formulation presented in \cite{colombi2023mixture}. Furthermore, \App{in Section \ref{app:proof_ExpqMpost}, we show that posterior expected value of $M^\star$ equals}
	\begin{equation}
		\label{eqn:ExpqMpost}
		\E\left(M^\star\mid\bmX\right) \ = \ \frac{V_{n_1,n_2}^{r+1}}{V_{n_1,n_2}^{r}}
		\, .
	\end{equation}
	In particular, it admits the following asymptotic approximation for large sample sizes:
	\begin{equation}
		\label{eqn:ExpqMpost_approx}
		\E\left(M^\star\mid\bmX\right)
		\ = \
		(r+1)\,\frac{q_M(r+1)}{q_M(r)}\,(\gamma_1 r)_{\gamma_1}(\gamma_2 r)_{\gamma_2}\,n_1^{-\gamma_1}n_2^{-\gamma_2}\,\left(1+o(1)\right) \, .
	\end{equation}
	Equation \eqref{eqn:ExpqMpost_approx} shows that $\E\left(M^\star\mid\bmX\right)$ goes to zero when $n_1,n_2\rightarrow\infty$.  This aligns with our modelling assumption, i.e., with an infinite amount of data, all possible species would have already been observed, leaving no room for further discoveries. Moreover, this also highlights the crucial role of the parameters $\gamma_1$ and $\gamma_2$ in governing the discovery rate of new species. If the values are greater than one, the expression quickly approaches zero. This means that after only a few observations, the expectation of discovering new species rapidly decreases. Conversely, values much smaller than one enable the discovery of new species even when a large number of observations is considered.
	
	\subsection{Joint predictive distribution}
	\label{subsection:joint_predictive}
	
	Following the same approach as in Section \ref{subsection:formule_priori}, it is of primary interest to derive both the marginal and joint distributions of the random variables on the right side of Equation \eqref{eqn:linear_rel_post}. These are stated in the following theoretical results.
	
	\begin{theorem}
		\label{thm:jointpost_d2}
		Let $\bmX=\left(\bmX_1,\bmX_2\right)$ be a sample of sizes $n_1$ and $n_2$ from model \eqref{eqn:partial_ex} under the $\operatorname{Vec-FDP}$ prior given in Equation \eqref{eqn:VecFDP_prior}.
		Let $\mathcal K_{n_1,n_2} = r$ and $\mathcal S_{n_1,n_2} = t$ be the observed number of global distinct and shared species, and let $K_{j,n_j} = r_j$, for $j=1,2$ be the observed local distinct species. Then,
		\begin{equation}
			\label{eqn:joint_post_d2}
			\begin{split}
				& \P\left(\mathcal K^{(n_1,n_2)}_{m_1,m_2} = k,
				\ K^{(n_1)}_{1, m_1} = k_1,
				\ K^{(n_2)}_{2,m_2} = k_2
				\mid \bmX \right) = \\
				& 
				\frac{V^{r+k}_{n_1+m_1,n_2+m_2}}{V^{r}_{n_1,n_2}} 
				\prod_{j=1}^2
				\lvert C(m_j,k_j; -\gamma_j,-(\gamma_j r_j + n_j) )\rvert
				\sum_{s^*=0}^{k}
				\sum_{k_1^*=0}^{k-s^*}
				\frac{k_1!k_2!}{s^*!k^*_1!k^*_2!}
				\binom{r^*_1}{s_{12}}
				\binom{r^*_2}{s_{21}},
			\end{split}
		\end{equation}
		for non-negative integers $k$, $k_1$, $k_2$ such that
		$0\leq k \leq k_1+k_2$ and $0\leq k_j \leq m_j$ for $j=1,2$. Specifically, in Equation \eqref{eqn:joint_post_d2} we set
		$k^*_2~=~k~-~s^*~-~k^*_1$, $s_{12}~=~k_2~+~k^*_1~-~k$ and $s_{21}~=~k_1~-~k^*_1~-~s^*$.  Finally, the coefficient $C (\, \cdot \, , \, \cdot\, ; \, \cdot \, , \, \cdot \, )$ in Equation \eqref{eqn:joint_post_d2} denotes the non-central generalised factorial coefficient, as defined in \cite{chara2002}.
	\end{theorem}
	The proof of Theorem \ref{thm:jointpost_d2} is provided in Section \ref{app:proof_thm_post}.
	The non-central generalised factorial coefficient satisfies a specific recursive relation that helps with its evaluation. See \cite{chara2002} for details.
	As discussed in Section \ref{app:post_quantities}, all auxiliary quantities involved in Equation \eqref{eqn:joint_post_d2} are interpretable. In fact, $s^*$ is the number of new shared species among the $k$ new distinct species, then $k^*_1$ is the number of new distinct species in group $1$ but missing in group $2$ while $s_{12}$ is the number of species that were first only observed in group $1$ and then also observed in group $2$. Finally, $k^*_2$ and $s_{21}$ are defined accordingly.
	The marginal distributions for the local and global number of distinct species are reported in Equations \eqref{eqn:Kpost_d2} and \eqref{eqn:Kj_post}. These are the posterior counterparts of Equations \eqref{eqn:K12_prior} and \eqref{eqn:Kj_prior}.
	\begin{prop}
		\label{thm:marginal_post_global}
		Under the same hypothesis of Theorem \ref{thm:jointpost_d2}, the marginal distribution of the global number of new distinct species $\mathcal K_{m_1,m_2}^{(n_1,n_2)}$ is
		\begin{equation}
			\label{eqn:Kpost_d2}
			\begin{split}
				&\P\left(\mathcal K_{m_1,m_2}^{(n_1,n_2)} =k \mid \bmX\right)  = 	\frac{V^{r+k}_{n_1+m_1,n_2+m_2}}{V^r_{n_1,n_2}}
				\\
				& 
				\qquad  \times \sum_{k^*_1=0}^k  \,
				\sum_{k^*_2=k-k^*_1}^{k}
				\frac{(k^*_1+s^*)!(k^*_2+s^*)!}{k^*_1!k^*_2!s^*!}
				\prod_{j=1}^2
				\lvert C(m_j, k^*_j+s^*; -\gamma_j, -(\gamma_j r + n_j))\rvert ,
			\end{split}
		\end{equation}
		for $k\in\{0,\dots,m_1+m_2\}$ and where $s^* = k - k^*_1 - k^*_2$.
	\end{prop}
	The proof of Proposition \ref{thm:marginal_post_global} is provided in Section \ref{app:proof_K12_post}. The marginal distribution of the local number of new distinct species in group $j$, $K_{j,m_j}^{(n_j)}$, follows from Equation \eqref{eqn:Kpost_d2} and it equals
	\begin{equation}
		\label{eqn:Kj_post}
		\begin{split}
			\P\left(K^{(n_j)}_{j,m_j} = k_j \mid \bmX\right) \ = \
			\frac{V^{r_j+k_j}_{n_j+m_j}}{V^{r_j}_{n_j}} \, \lvert C(m_j, k_j; -\gamma_j, -(\gamma_j r_j + n_j))\rvert,
		\end{split}
	\end{equation}
	for $k_j\in\{0,\dots,m_j\}$.
	The latter coincides with the findings in \cite{deblasi2015} about Gibbs-type priors with negative parameters.
	Furthermore, we highlight that Theorem \ref{thm:jointprior_d2} requires 
	$K_{j,n_j} \leq \mathcal K_{n_1,n_2}$. This condition, however, is not necessary in Theorem \ref{thm:jointpost_d2}, as it is possible for the number of global discoveries to exceed the number of local discoveries, i.e., $K^{(n_j)}_{j, m_j} > \mathcal K^{(n_1,n_2)}_{m_1,m_2}$.
	An example illustrating this scenario is presented in Section \ref{app:post_quantities}.

	\subsection{Discovering shared species}
	\label{subsection:PrShared_post}
	Similarly to Section \ref{subsection:formule_priori}, Equations \eqref{eqn:Kj_post} and \eqref{eqn:Kpost_d2} allow us to compute the posterior expected values of the number of new distinct species and, by means of Equation \eqref{eqn:linear_rel_post}, to derive the Bayesian estimator of the number of new shared species, that is
	\begin{equation}
		\label{eqn:shared_estimator}
		\E\left[\mathcal S^{(n_1,n_2)}_{m_1,m_2}\mid\bmX\right]  =  
		\E\left[K^{(n_1)}_{1, m_1}\mid\bmX\right]  +  
		\E\left[K^{(n_2)}_{2,m_2}\mid\bmX\right]  -  
		\E\left[\mathcal K^{(n_1,n_2)}_{m_1,m_2}\mid\bmX\right] \, .
	\end{equation}
	The associated uncertainty is quantified through the posterior distribution, namely, $\P\left(\mathcal S^{(n_1,n_2)}_{m_1,m_2} = s\mid\bmX\right) $ that is obtained by summing the expression in Equation \eqref{eqn:joint_post_d2} for all $k,k_1,k_2$ such that $s=k_1+k_2-k$ and for all $s=\{0,\ldots,m_1+m_2\}$. Alternatively, it can be estimated via Monte Carlo sampling, drawing samples from Equation \eqref{eqn:joint_post_d2}.
	
	The shared species sample coverage is defined as the proportion of shared species that are observed in the sample. Being able to estimate it allows us to assess the number of shared species that have been observed, and therefore, to decide whether it is worth continuing the experiment and sampling more observations. Moreover, the shared species sample coverage on $m$-steps ahead facilitates determining the size of the additional sample that ensures the coverage exceeds a specified threshold.
	In our setting, the shared coverage probability for $m$-steps ahead, i.e., the probability of not discovering new shared species in the additional sample, follows from Equation \eqref{eqn:joint_post_d2} and equals
	
	\begin{equation}
		\label{eqn:Shzero_m_steps}
		\P\left(\mathcal S^{(n_1,n_2)}_{m_1,m_2} = 0\mid\bmX\right) = 
		\sum_{k_1=0}^{m_1}
		\sum_{k_2=0}^{m_2}
		\frac{V_{n_1+m_1,n_2+m_2}^{r+k_1+k_2}}{V_{n_1,n_2}^{r}}
		\prod_{j=1}^d 
		\lvert C(m_j,k_j; -\gamma_j,-(\gamma_j r_j + n_j) )\rvert.
	\end{equation}
	In particular, the one-step ahead coverage probability 
is obtained from Equation \eqref{eqn:Shzero_m_steps} by setting $m_1 = 1$ and $m_2 = 1$.
	In this case, we can explicitly write the whole distribution $\P\left(\mathcal S^{(n_1,n_2)}_{1,1} = s\mid\bmX\right)$ for each $s\in\{0,1,2\}$ and not just for $s=0$ as it happens for $m_j>1$.
	For the sake of brevity, such distribution is reported in Section \ref{app:species_discovery}, alongside the probability of discovering new local and global distinct species.
	
	Here, we report the probability of discovering at least one new shared species; that is,
	\begin{equation}
		\label{eqn:Sh0_1_step}
		\begin{split}
			& \P\left(\mathcal S^{(n_1,n_2)}_{1,1} > 0\mid\bmX\right) \,= \,
			1 -
			\frac{V_{n_1+1,n_2+1}^{r}}{V_{n_1,n_2}^r}
			(\gamma_1r_1+n_1)(\gamma_2r_2+n_2)
			\\
			& \qquad -
			\frac{V_{n_1+1,n_2+1}^{r+1}}{V_{n_1,n_2}^r}
			\left\{
			\gamma_1(\gamma_2r_2+n_2) + \gamma_2(\gamma_1r_1+n_1)
			\right\}
			-
			\frac{V_{n_1+1,n_2+1}^{r+2}}{V_{n_1,n_2}^r}
			\gamma_1\gamma_2.
		\end{split}
	\end{equation}
	The ratios of $V$ coefficients in Equation \eqref{eqn:Sh0_1_step} represent the three different scenarios where the new pair of observations yields none, one, or two new distinct global species. It is worth noting that increasing values of the probability in Equation \eqref{eqn:Sh0_1_step} indicate that the observed sample is sufficiently exhaustive, suggesting that further data collection may not be necessary.
    Equation \eqref{eqn:Sh0_1_step} can be compared with the Good-Turing estimators proposed in the frequentist literature, as presented in Section \ref{subsection:SS_exp2} where a simulation study is carried out.
    Nevertheless, to the best of our knowledge, neither the frequentist nor the Bayesian literature provides an analogue of the $m$-step-ahead discovery probability in Equation \eqref{eqn:Shzero_m_steps}; we illustrate its practical use in Section \ref{section:application}.

	\section{Diversity indices}
	\label{subsection:indici_intro}
	In ecology, the concept of diversity is tied not only to the number of distinct species present in an area but also to their heterogeneity. For instance, having ten equally represented species is intuitively very different from having one highly abundant species and the remaining nine extremely rare \citep{colwell2004}. In the literature, there is no unique quantitative definition of diversity; instead, a variety of indices have been proposed to measure it. Among these, we focus on Simpson's diversity index \citep{Simpson49}, which captures both richness and evenness in species distributions.
	Assuming that the unknown discrete distribution $P_j$ that generates the population is made up of $M$ distinct species with species proportions $w_{j,m}$, Simpson's diversity index is
	\begin{equation}
		\label{eq:simpson}
		\rho_j = \sum_{m=1}^M w^2_{j,m}.
	\end{equation}
	Useful alternatives are suitable transformations of $\rho_j$, such as the Gini-Simpson index, defined as $1-\rho_j$, or the inverse-Simpson index, $1/\rho_j$ \citep{colwell2004}.
	The index $\rho_j$ ranges between $1/M$ (when all species are uniformly distributed) and one (when one species is abundant and the remaining are negligible).
	These extreme cases correspond to the cases of maximum and minimum heterogeneity, respectively.
	

    When we
	assume the population is generated from a vector $\left(P_1,P_2\right)$ of unknown discrete distribution, each having $M$ different species whose proportions are $w_{j,m}$, for 
	$j=1,2$ and $m=1,\dots,M$, the Morisita index \citep{Morisita59} can be used to quantify the similarity between the two areas, namely, $2\rho_{12}/(\rho_1+\rho_2)$, where 
	\begin{equation}
		\label{eq:rho12}
		\rho_{12}=\sum_{m=1}^M w_{1,m}w_{2,m}.
	\end{equation}
	%
	Increasing values of the Morisita index show evidence of identical communities, i.e., with the same species proportions. \cite{chao2017} highlighted an important probabilistic interpretation of the Morisita index. The numerator in Equation \eqref{eq:rho12} represents the probability of selecting the same shared species when two observations, one from each group, are randomly drawn from this population. The denominator of the Morisita index is instead the sum of the Simpson's diversity indices in the two groups. 
    In Section \ref{section:posterior_diversity}, we present the Bayesian estimators of $\rho_j$, $j=1,2$, and $\rho_{12}$ given an observed sample $\bmX=(\bmX_1,\bmX_2)$ of sizes $n_1$ and $n_2$, with $r$ distinct species and $t$ shared species.
    
    \subsection{Parameter estimation}
	\label{subsection:param_estimation}
    The Bayesian model in Equation \eqref{eqn:partial_ex} under the $\operatorname{Vec-FDP}$ prior in Equation \eqref{eqn:VecFDP_prior} is governed by the parameters $\bmgamma = (\gamma_1,\gamma_2)$ and $\Lambda$. 
    These parameters are unknown and must be estimated from the data. To this end, we propose two different estimation strategies.
    The first strategy, following \cite{MasoeroCamScaled}, is a plug-in approach based on the maximum marginal likelihood estimator. Namely, we find $\hat\bmgamma$ and $\hat \Lambda$ that maximise the pEPPF in Equation \eqref{eqn:peppf}. In the following, we refer to \textit{Bayes I} estimators as the estimators proposed in the previous sections for the global and local quantities  
    when the parameters $\bm{\gamma}$ and $\Lambda$ are estimated via the maximum marginal likelihood estimator.
    
       In addition, we provide an alternative estimation strategy that relies on a diversity-based interpretation of $\left(\bmgamma,\Lambda\right)$. 
       To this end, note that Equation \eqref{eqn:partial_ex} does not assume the existence of a unique, unknown vector of probability distributions $\left(P_1,P_2\right)$ but rather assumes that $\left(P_1,P_2\right)$ is random. Therefore, we can integrate out this source of randomness by taking the expected values, obtaining
	\begin{align}
		\label{eqn:ssj}
		&\E\left(\rho_j\right) \, = \, 
		\E\left(\sum_{m=1}^M w_{j,m}^2\right)
		\, = \,
		\left(1 + \gamma_j\right)\,\E\left(\frac{1}{1 + \gamma_jM}\right) \, , \\
		&\E\left(\rho_{12}\right) \, = \, 
		\E\left(\sum_{m=1}^M w_{1,m}w_{2,m}\right)
		\, = \,
		\E\left(1/M\right) \, .
		\label{eqn:sp1p2}
	\end{align}
	\App{Proof of Equations \eqref{eqn:ssj} and \eqref{eqn:sp1p2} are deferred to Section \ref{app:proof_stima_param}.}
	The expected value of $\rho_{12}$ in Equation \eqref{eqn:sp1p2} solely depends on $q_M$. This is due to the fact that we model the dependence in the vector of random probability measures $\left(P_1,P_2\right)$ through the random number of species $M$ and by imposing common atoms. 
	Note that our choice of $q_M$ as $\operatorname{Pois}_1(\Lambda)$ yields explicit expressions for Equation \eqref{eqn:sp1p2} depending on $\Lambda$;  see Equation \eqref{eqn:Exp1overM_Pois}. Then, $\Lambda$ regulates the amount of similarity between the two areas.
	Furthermore, the limits of the expected Simpson index in Equation \eqref{eqn:ssj} further illuminate the interpretation of $\gamma_1$ and $\gamma_2$ as homogeneity parameters, as described in Section \ref{subsection:correlation}.
	In fact, the limits of Equation \eqref{eqn:ssj} for $\gamma_j\rightarrow0$ and $\gamma_j\rightarrow\infty$ are equal to one and $\E\left(1/M\right)$, respectively, which represent the cases of minimum and maximum heterogeneity.
	Therefore, $\gamma_1$ and $\gamma_2$ are homogeneity parameters since heterogeneity decreases as $\gamma_j$ increases.

	Our diversity-based strategy for estimating $\gamma_1$, $\gamma_2$ and $\Lambda$ consists of two steps. 
	Firstly, we use the observed data to get estimates of $\rho_j$, $j=1,2$, in Equation \eqref{eq:simpson} and $\rho_{12}$ in Equation \eqref{eq:rho12}. 
	This step can be achieved with standard routines, such as using the estimator in Equation \eqref{eqn:ssj_freq}. 
	Then, we plug such estimates into the left-hand sides of Equations \eqref{eqn:ssj} and \eqref{eqn:sp1p2} and solve with respect to $\gamma_1$, $\gamma_2$, and $\Lambda$.
    In the following, we refer to \textit{Bayes II} estimators as the estimators for the local and global quantities, when using the strategy just described to estimate $\bm\gamma$ and $\Lambda$. In Section \ref{section:simulation_summary}, we show that the computational effort required by the diversity-based strategy is negligible compared to the alternative based on the maximum marginal likelihood.

    We conclude by noting that other estimation procedures can also be considered.
    For instance, in the case of a single population, \cite{Lij(07)} and \cite{favaro2009} suggested maximising the prior distribution of the number of distinct species evaluated at their observed values. Alternatively, \cite{BalocchiCamFavaro24} employed a fully Bayesian approach by specifying suitable hyperpriors and estimating the parameters via Markov chain Monte Carlo (MCMC) methods.

	\section{Simulation study}
	\label{section:simulation_summary}
    \subsection{Data generation and competitors}
    \label{subsection:data_generation}
    We conduct a simulation study to evaluate our methodology and compare it with existing estimators. The data-generating mechanism, adapted from \cite{Yue2012}, assumes that $\left(P_1,P_2\right)$ are two discrete probability distributions, each consisting of $M_{\text{true}}$ species. The group-specific species proportions are denoted by $p_{j,m}$, with $j=1,2$ and $m=1,\dots,M_{\operatorname{true}}$, which are specified according to several alternatives to cover a wide range of benchmark distributions and misspecification scenarios. 
    In particular, we consider the following alternatives, referred to here as settings:\\
        (i) Dirichlet weights -- we set $M_{\text{true}} = 60$ and $p_{j,m}$ to be randomly generated from symmetric Dirichlet distributions with parameters $\gamma_j \in \{0.1, 0.5\}$; all three possible combinations are denoted $D_1,D_2,D_3$ as explained in Table \ref{tab:SS_D};\\
        (ii) Geometric weights -- we set $M_{\text{true}} = 60$ and $p_{j,m}$ to be deterministically assigned according to a geometric decay, i.e., $p_{j,m} \propto \alpha_j^m$, with $\alpha_j \in \{0.8,\,0.85,\,0.9\}$, as in \cite{Yue2012}; all possible combinations of these three values are considered, yielding six cases denoted as $G_1,\ldots,G_6$ and explained in 
        Table \ref{tab:SS_G};\\
        (iii) Zipf's weights -- we set $M_{\text{true}} = 60$ and $p_{j,m}$ to be deterministically assigned as $p_{j,m} \propto m^{-s_j}$, with $s_j \in \{1.3,\,2\}$, similar to \cite{franzolini25}; all three possible combinations are denoted $Z_1,Z_2,Z_3$ and explained in 
        Table \ref{tab:SS_Z}.
    
    In the Geometric and Zipf settings, the probabilities $p_{j,m}$ are deterministic and monotonically decreasing in $m$, implying that species with a high probability of being observed in one group are likely to appear in the other group as well. Consequently, only these few species are likely to be observed as shared, leaving a negligible probability for all others. To mitigate this effect, before assigning the probability mass function, the $M_{\text{true}}$ species in each population are randomly permuted. This procedure prevents the Morisita index described in Section \ref{subsection:indici_intro} from being exactly equal to one whenever the distribution parameters of the two groups coincide, e.g., $\alpha_1 = \alpha_2$ or $s_1 = s_2$. The empirical values of the Morisita index after shuffling are reported in Figure \ref{fig:SS_Morosita}.
    Finally, Figures \ref{fig:AccCrv_D3}-\ref{fig:AccCrv_Z3} display the accumulation curves for the four statistics of interest for some selected cases $(D_3,G_6,Z_3)$.
    
    All three settings described above satisfy the assumption that the same set of $M_{\text{true}}$ species is shared. However, the accumulation curves of shared species show that, for a finite number of observations, the number of observed shared species is smaller than $M_{\text{true}}$.
    Nevertheless, we also consider a fourth setting in which this assumption no longer holds. Specifically, inspired by \citep{Muller2004}, we define $(P_1,P_2)$ as
    \begin{equation}
    \label{eq:gen_muller}
        P_j = c Q_0 + (1-c)Q_j\,,
  \end{equation}
    for any scalar $c \in [0,1]$. The shared component 
    $Q_0$ 
    is assumed to follow a symmetric Dirichlet distribution of size $M_{\operatorname{com}}$ with parameter $\delta_0$, while the two idiosyncratic components $Q_j$, $j=1,2$, follow symmetric Dirichlet distributions of sizes $M_{\operatorname{id}}^1$ and $M_{\operatorname{id}}^2$ with parameters $\delta_1$ and $\delta_2$, respectively.
    In particular, 
    we set $\delta_1 = \delta_2$, denoting this common value as $\delta$ in the following.
    Overall, the total number of species is $M_{\operatorname{tot}} = M_{\operatorname{com}} + M_{\operatorname{id}}^1 + M_{\operatorname{id}}^2$. In our experiments, we fix $M_{\operatorname{tot}} = 80$ and set $M_{\operatorname{com}} = c M_{\operatorname{tot}}$, rounded down to the nearest integer. The remaining species are equally divided between the two idiosyncratic components, i.e., $M_{\operatorname{id}}^j = (M_{\operatorname{tot}} - M_{\operatorname{com}})/2$ for $j=1,2$, with rounding applied to ensure that the total number of species matches the desired value.
    Regarding the choice of parameters, we consider all configurations with $c \in \{0,0.5,1\}$ and $(\delta_0,\delta) \in \{0.1,0.5\}^2$, resulting in a total of 12 configurations, denoted as $A_{1,1},\ldots,A_{3,4}$. The corresponding names are reported in Table \ref{tab:SS_A}.
  Figures \ref{fig:AccCrv_A14}-\ref{fig:AccCrv_A34} display the accumulation curves for the four statistics of interest for some selected cases $(A_{1,4},A_{2,4},A_{3,4})$.
    Finally, in each setting, we independently sample $n_j$ species with replacement from $P_j$, for $j=1,2$. Here, we consider the balanced case where $n_1 = n_2$, while we defer the unbalanced case with $n_1 \ll n_2$ to Sections \ref{app:SS_Exp1_additional} and \ref{app:SS_Exp2_additional} of the Supplementary material.
    In both cases, we denote the total number of observations by $n = n_1 + n_2$.

    Whenever possible, we compare the performance of our estimators with that of two simpler benchmark models. The first benchmark models the two populations independently, assuming that each area-specific model coincides with the marginals induced by our proposed framework. Specifically, we consider $P_j = \sum_{m=1}^{M_j} w_{j,m} \delta_{\tau_{j,m}}$, for $j=1,2$ where the collections 
  $\{M_j, (w_{j,1},\dots  , w_{j,M_j}), (\tau_{j,1}, \dots, \tau_{j,M_j})\}$ are distributed independently as a Finite Dirichlet Process (FDP; \citealt{argiento2022annals}), i.e., $P_j\sim \operatorname{FDP}\left(\Lambda_j,\gamma_j\right)$. 
    In the following, we refer to this model and the corresponding estimators as \textit{Independent}. Since any dependence between the two areas is completely ignored, this approach is only suitable for the estimation of local quantities such as $K_{j,n_j}$ and $K_{j,m_j}^{(n_j)}$.
    The second benchmark model, instead, discards any information about group membership and pools all observations as if they originated from a single area. For this reason, we refer to this approach as \textit{Pooled}. In this case, the merged sample $\bm{X}_1 \cup \bm{X}_2$ is assumed to be exchangeable from a single homogeneous population, so that $P_1 = P_2 = P$, with 
    $P \sim \operatorname{FDP}\left(\Lambda_0,\gamma_0\right)$.
    Under this assumption, information on local quantities is irretrievably lost. However, we can still compare the global number of distinct species $\mathcal K^{(n_1,n_2)}_{m_1,m_2}$ with the corresponding quantity computed under the \textit{Pooled} model, namely
    $\mathcal K^{(m_1+m_2)}_{n_1+n_2} = K_{m_1+m_2} - K_{n_1+n_2}$.
    The \textit{Independent} and \textit{Pooled} estimators of the quantities of interest are obtained by estimating the corresponding parameters, $(\Lambda_j,\gamma_j)$ for $j=1,2$ and $(\Lambda_0,\gamma_0)$ via maximum marginal likelihood. An explicit expression for the associated EPPF can be found in \cite{argiento2022annals}.
    The performance of our proposed model is evaluated through two separate experiments, described in Sections \ref{subsection:SS_exp1} and \ref{subsection:SS_exp2}, respectively.

    In the case of one-step-ahead discovery probability of shared species, we compare our approach against two nonparametric frequentist competitors that represent two alternative generalisations of the Good-Turing estimator for the problem of shared species.
	Specifically, let $f_{\nu_1,\nu_2}$ be the number of species that appeared exactly $\nu_1$ and $\nu_2$ times in the first and second groups, respectively. Moreover, let $f_{\nu_1,+}$ ($f_{+,\nu_2}$) be the number of species that appeared exactly $\nu_1$ ($\nu_2$) times in the first (second) group and at least once in the second (first) group. Then, the estimator proposed by \cite{Yue2012} (\textit{Yue}) 
    is $\P_{\text{Yue}}\left(\mathcal S^{(n_1,n_2)}_{1,1} > 0\right) = \left(f_{1+} + f_{+1} + f_{11} \right)/n_1$, and it is only defined for $n_1 = n_2$. On the other hand, \cite{chao2017} (\textit{Chao}) proposed 
    $\P_{\text{Chao}}\left(\mathcal S^{(n_1,n_2)}_{1,1} > 0\right) = f_{1+}/n_1 + f_{+1}/n_2 + f_{11}/n_1n_2 $. 

    \subsection{Experiment 1}
    \label{subsection:SS_exp1}
    The first experiment assesses the ability of the estimators introduced in Section \ref{section:posterior} to predict the number of additional local and global distinct species, as well as the number of shared species, in a future unobserved test set.
    The experiment proceeds as follows. First, we generate a training dataset according to the data-generating mechanism described in Section \ref{subsection:data_generation}, with balanced sample sizes $n_1 = n_2$ taking values in the set $\{100,200,\ldots,800\}$. 
    We then generate $m_1 = m_2 = 200$ additional observations, which constitute the unobserved test set. This procedure is repeated over $100$ independently generated pairs of training and test sets.
    Each training set is used to estimate the model parameters as described in Section \ref{subsection:param_estimation}.
    We recall that \textit{Bayes I} and \textit{Bayes II} estimators refer to maximum marginal likelihood and diversity-based estimation strategies, respectively. 
    These estimates are subsequently employed to predict the expected number of new species in the test set, following the methodology outlined in Section \ref{section:posterior}. Specifically, we use Equation \eqref{eqn:joint_post_d2} to compute $\E\left[K^{(n_j)}_{j, m_j}\mid\bmX\right]$ for $j=1,2$ as well as $\E\left[\mathcal K^{(n_1,n_2)}_{m_1,m_2}\mid\bmX\right]$. Finally, $\E\left[\mathcal S^{(n_1,n_2)}_{m_1,m_2}\mid\bmX\right]$ is computed as in Equation \eqref{eqn:shared_estimator}. 
    We compare the estimates of the new local and global distinct species with the \textit{Independent} and \textit{Pooled} estimators introduced in Section \ref{subsection:data_generation}. However, we are not aware of any $m$-steps ahead estimator for the shared species. 
    
    Out-of-sample performance is evaluated by computing the Root Mean Squared Error (RMSE) of each estimated quantity with respect to its true value.
    For the sake of space, Figure \ref{fig:SS_Pred_1} reports results only for selected scenarios $(D_3,G_6,Z_3,A_{1,4},A_{2,4},A_{3,4})$ when the training sample size is fixed at $n_1 = n_2 = 400$.
    Furthermore, Table \ref{tab:SS1_paramest} reports the mean and standard deviation of the estimated total number of species, $K_n + M^\star$, as well as the model parameters $\left(\Lambda,\gamma_1,\gamma_2\right)$. 
    Results for all remaining configurations are reported on the Github repository \url{https://github.com/alessandrocolombi/HSSM}.

    \begin{figure}[ht!]
        \centering
        \includegraphics[width=0.485\linewidth]{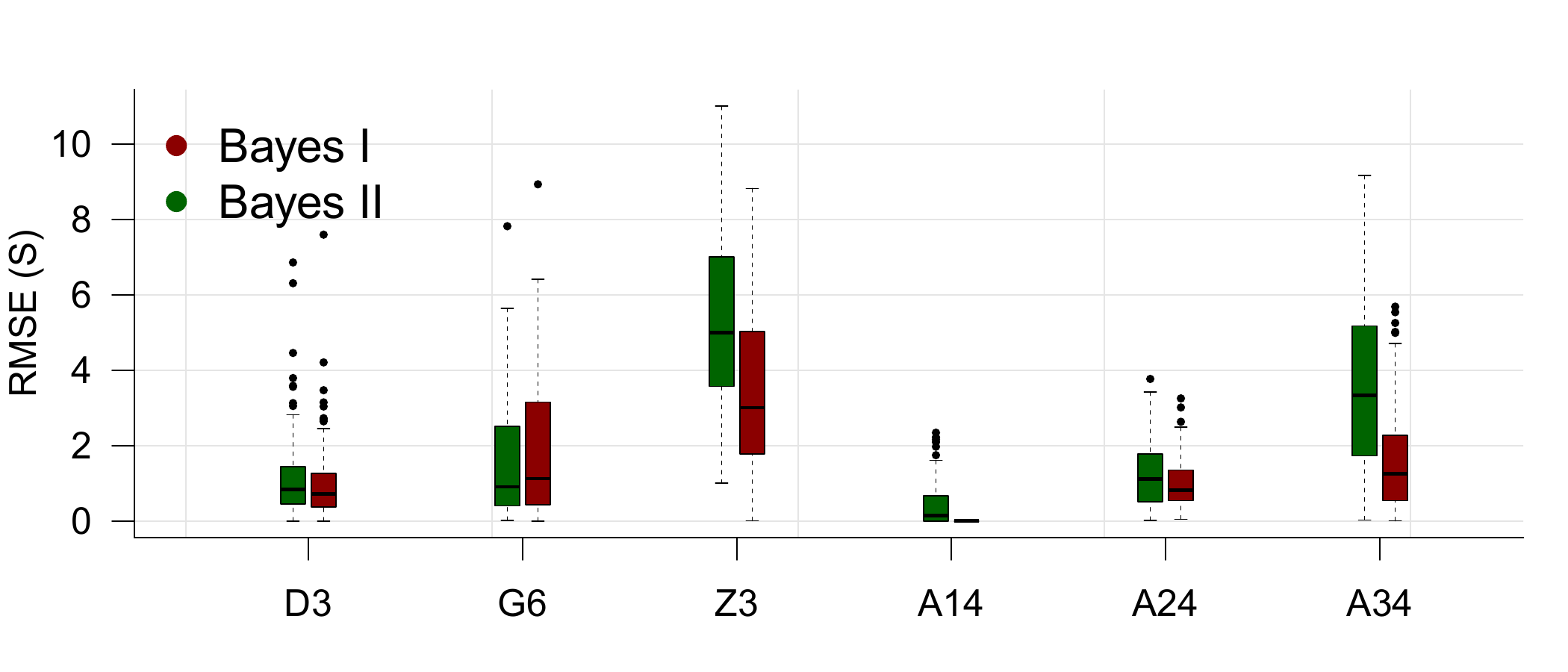}
        \hfill
        \includegraphics[width=0.485\linewidth]{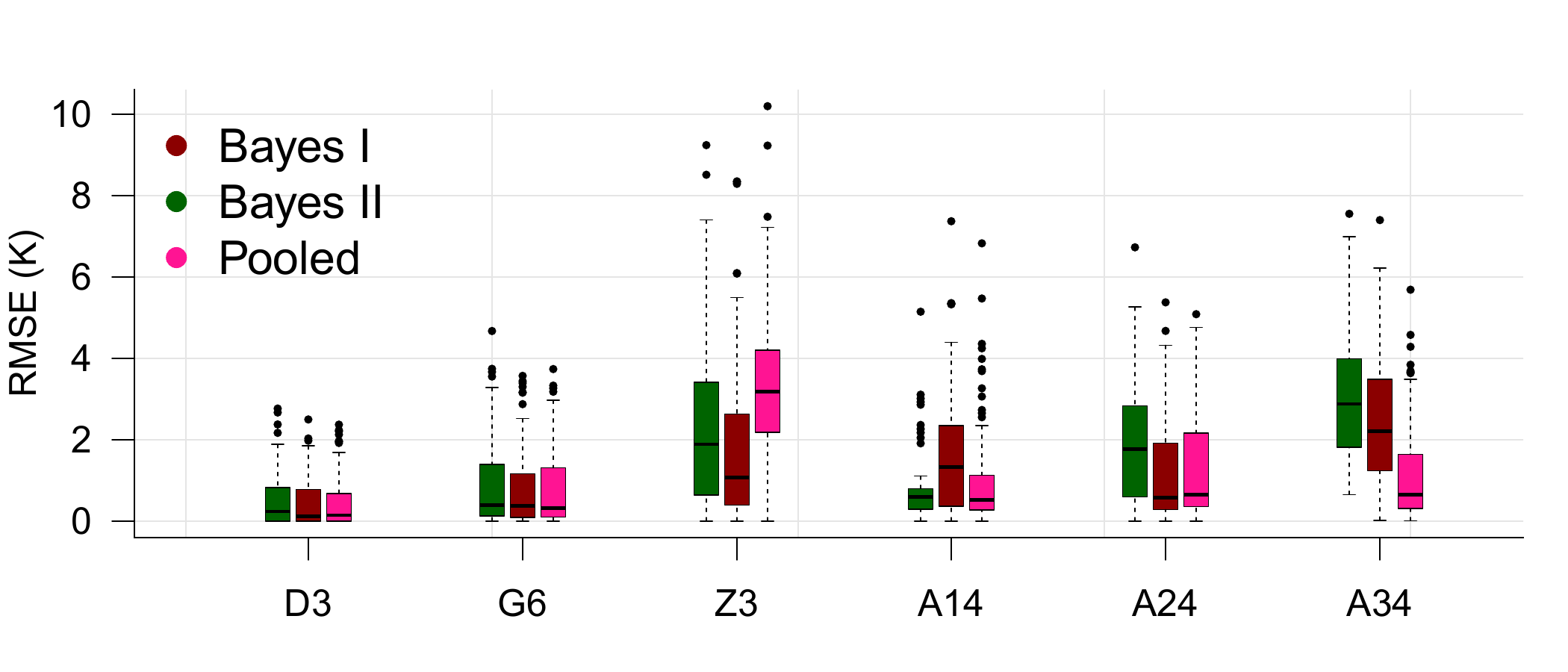}
        \\
        \includegraphics[width=0.485\linewidth]{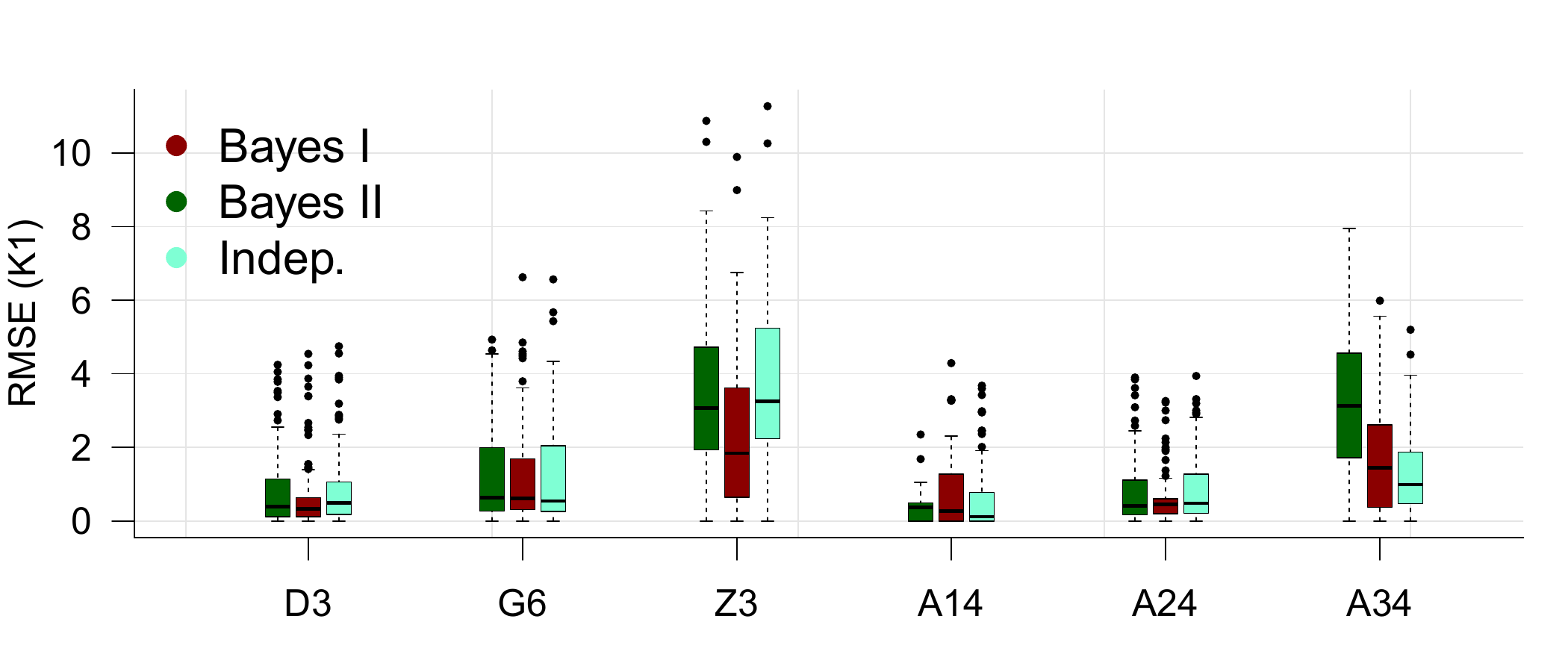}
        \hfill
        \includegraphics[width=0.485\linewidth]{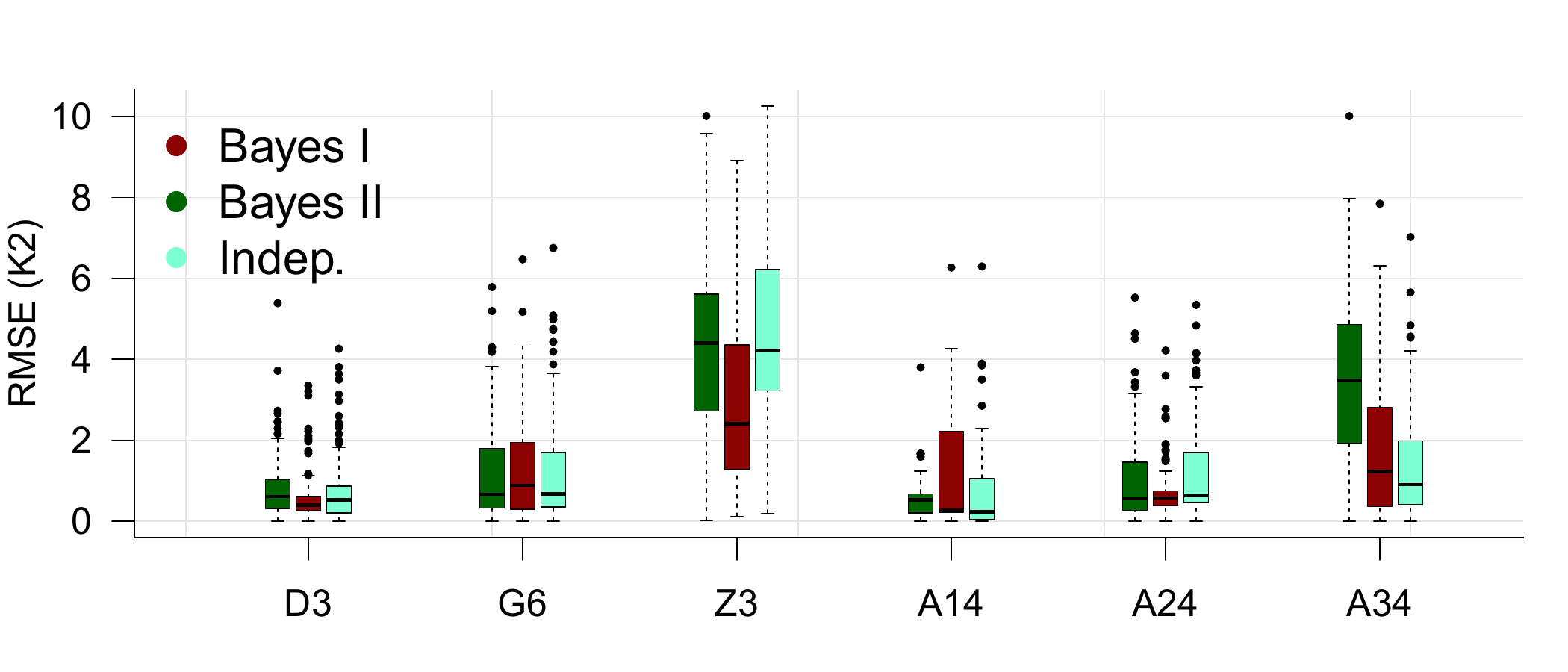}
        \hfill
        \caption{Experiment 1: RMSE of out-of-sample predictions for new shared species (top-left panel), new global distinct species (top-right panel), and new local distinct species (bottom-left panel) and (bottom-right panel) across selected scenarios.}
        \label{fig:SS_Pred_1}
    \end{figure}
    Figure \ref{fig:SS_Pred_1} highlights that \textit{Bayes I} and \textit{Bayes II} are the only procedures that provide simultaneous out-of-sample predictions for all four target quantities. In contrast, \textit{Pooled} and \textit{Independent} focus on a single quantity and therefore cannot deliver a coherent joint assessment. Notably, our method is the only one that provides predictions for the number of new shared species. Overall, our estimators remain competitive with the benchmark methods while offering the advantage of relying on a unified model at the cost of only one additional parameter. 
    The main exception is scenario $A_{3,4}$, as expected: this configuration is generated assuming $c=1$ in Equation \eqref{eq:gen_muller}, i.e., under a single common population. Hence, 
    \textit{Pooled} estimator achieves more accurate predictions for the number of new global distinct species.
    
    The scenario $Z_3$ is the most challenging, as evidenced by uniformly larger RMSEs in all quantities. setting \textit{Bayes I} improves over the competitors, while \textit{Bayes II} outperforms \textit{Pooled} for global distinct discoveries and remains comparable to \textit{Independent} for local discoveries. This is particularly relevant since, in $Z_3$, the estimate of the total number of species $K_n+M^\star$ is out of scale (see Table \ref{tab:SS1_paramest}), indicating that our predictive performance is robust to misspecification (Zipf's weights in $Z_3$ differ markedly from the Dirichlet-type decay). 
    A similar phenomenon occurs in scenario $A_{1,4}$. Here, the likelihood-based fit underlying \textit{Bayes I} is strongly misspecified because, by construction, there are no shared species. Nevertheless, the predicted number of new shared species is correctly concentrated at zero, and the remaining quantities still exhibit small median RMSEs (typically below one), although with slightly higher variability.
    In the same scenario, \textit{Bayes II} behaves differently: since its fitting is driven by a diversity-based strategy, it relies more directly on observed summary quantities and less on the full parametric specification of the data-generating model. Hence, it yields substantially more accurate predictions than \textit{Bayes I}, up to a small residual error on shared-species discovery, which is not fully shrinking to zero. Finally, regarding the estimation of $K_n+M^\star$, both approaches recover the true value under the model in $D_3$ and under geometric weights in $G_6$; the mild underestimation observed in $A_{2,4}$ and $A_{3,4}$ does not appear to compromise the predictive performance of the methods.

    Finally, Figure \ref{fig:SS1_ExecTime} shows the computational time required by \textit{Bayes I} and \textit{Bayes II} to fit the model and perform 
    out-of-sample predictions.
    A key advantage of \textit{Bayes II} is that the cost of parameter estimation is independent of the observed sample size, while \textit{Bayes I} directly involves the computation of the marginal likelihood and therefore becomes more expensive as the sample size increases. More broadly, these results highlight the substantial computational benefit of our approach relative to previous solutions available in the literature. Thanks to closed-form expressions, prediction can be performed in less than half a second, which is comparable to the time required for parameter estimation via \textit{Bayes II} and negligible relative to the cost of fitting the model via \textit{Bayes I}. Such computational times are typically unattainable for more complex Bayesian nonparametric models that require MCMC methods.


    \subsection{Experiment 2}
    \label{subsection:SS_exp2}
    The second experiment assesses the ability of our model to estimate the one-step ahead probability of discovering a new local or global distinct species, as well as a new shared species. 
    The estimator for the shared discovery probability is given in Equation \eqref{eqn:Sh0_1_step} while those related to local and global quantities are reported in Equations \eqref{eqn:PrK_1step_full} and \eqref{eqn:PrKj_1step_full}. We compare our estimators with those discussed in Section \ref{subsection:data_generation}.
	
	Given $\left(P_1,P_2\right)$ generated as described in Section \ref{subsection:data_generation}, we evaluate all competing methods on a grid of sample sizes $n_1 = n_2 = \{50,100,\ldots,400\}$. 
	The experiment is then repeated for $100$ independently generated datasets.
	Results are compared with the true probability of discovering a new species, which is given in Equation \ref{eqn:Pr1step_true}.
    Figure \ref{fig:SS_Pr1step_1} reports the results for selected scenarios $(D_3,G_6,Z_3,A_{1,4},A_{2,4},A_{3,4})$ and only for sample size $n_1=n_2=400$. 
    Results for all remaining configurations are reported on the Github repository \url{https://github.com/alessandrocolombi/HSSM}.

    \begin{figure}[ht!]
        \centering
        \includegraphics[width=0.485\linewidth]{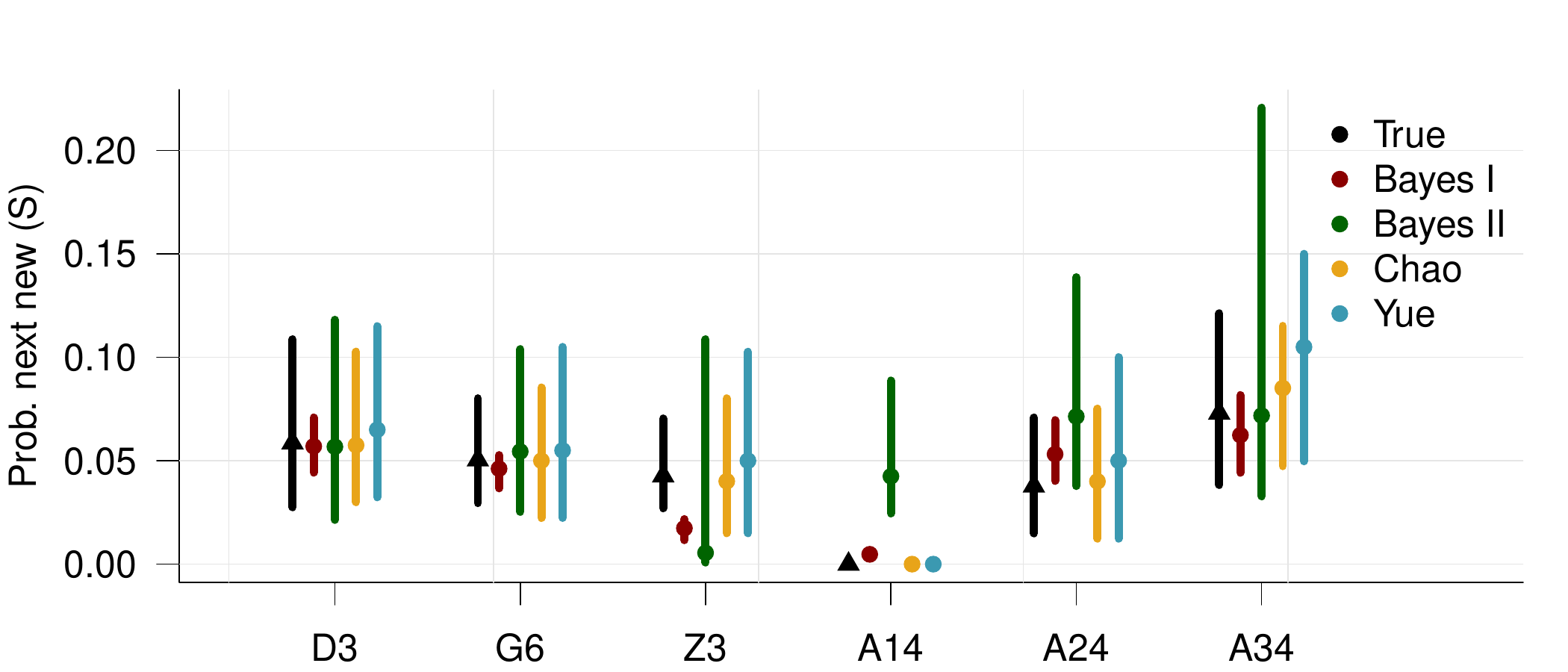}
        \hfill
        \includegraphics[width=0.485\linewidth]{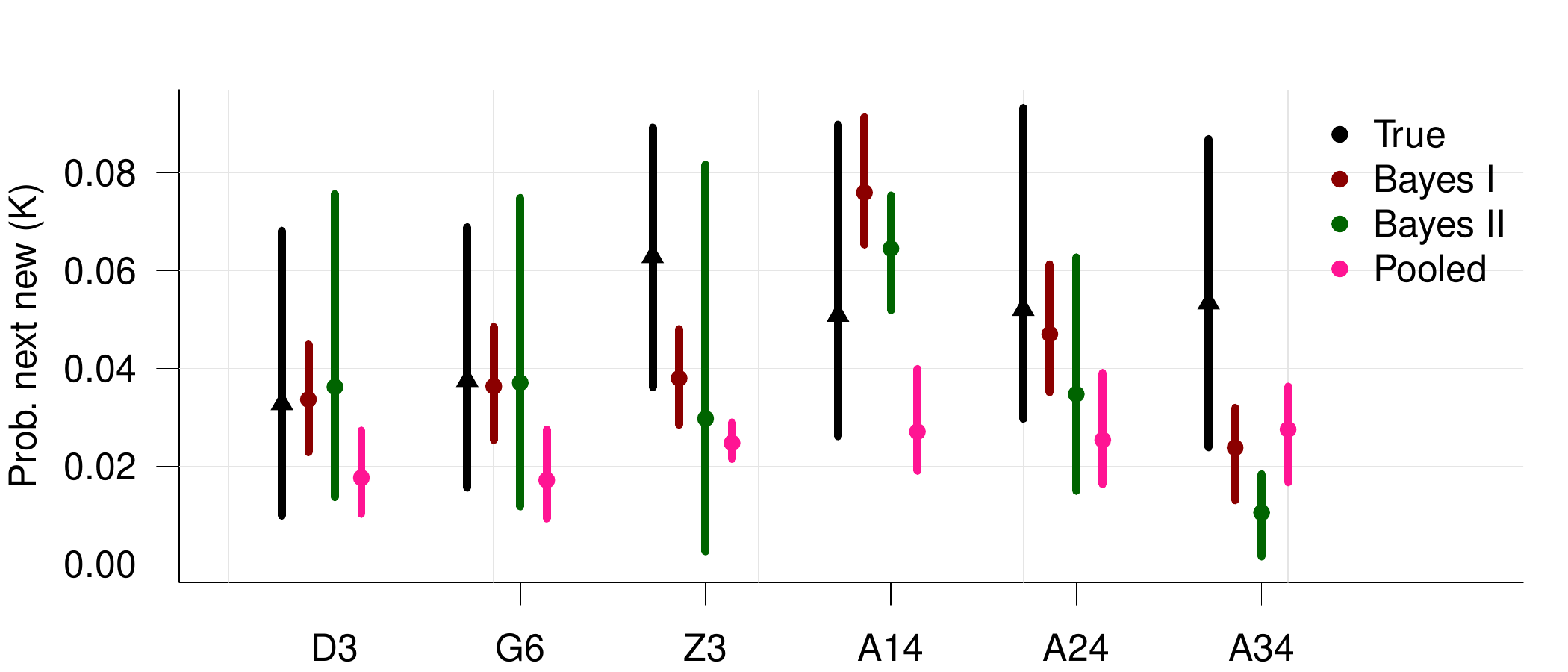}
        \\
        \includegraphics[width=0.485\linewidth]{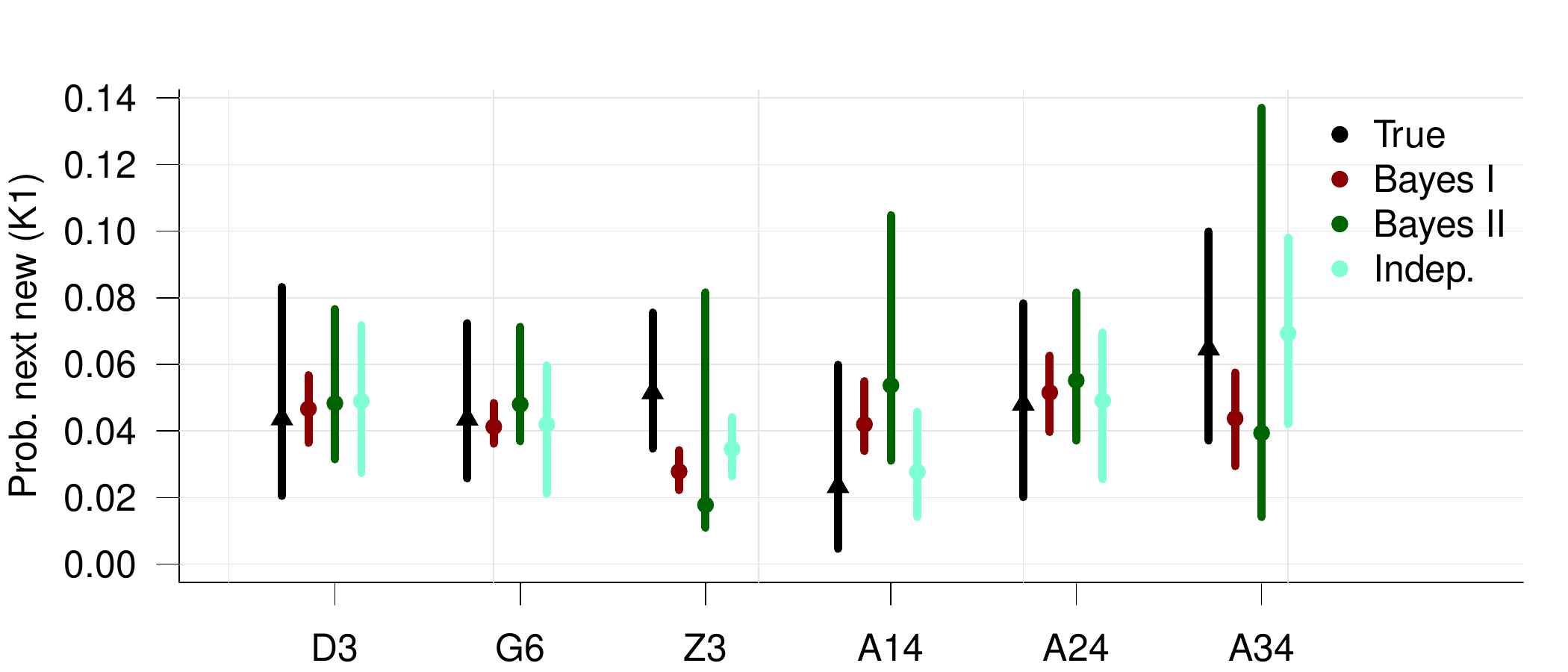}
        \hfill
        \includegraphics[width=0.485\linewidth]{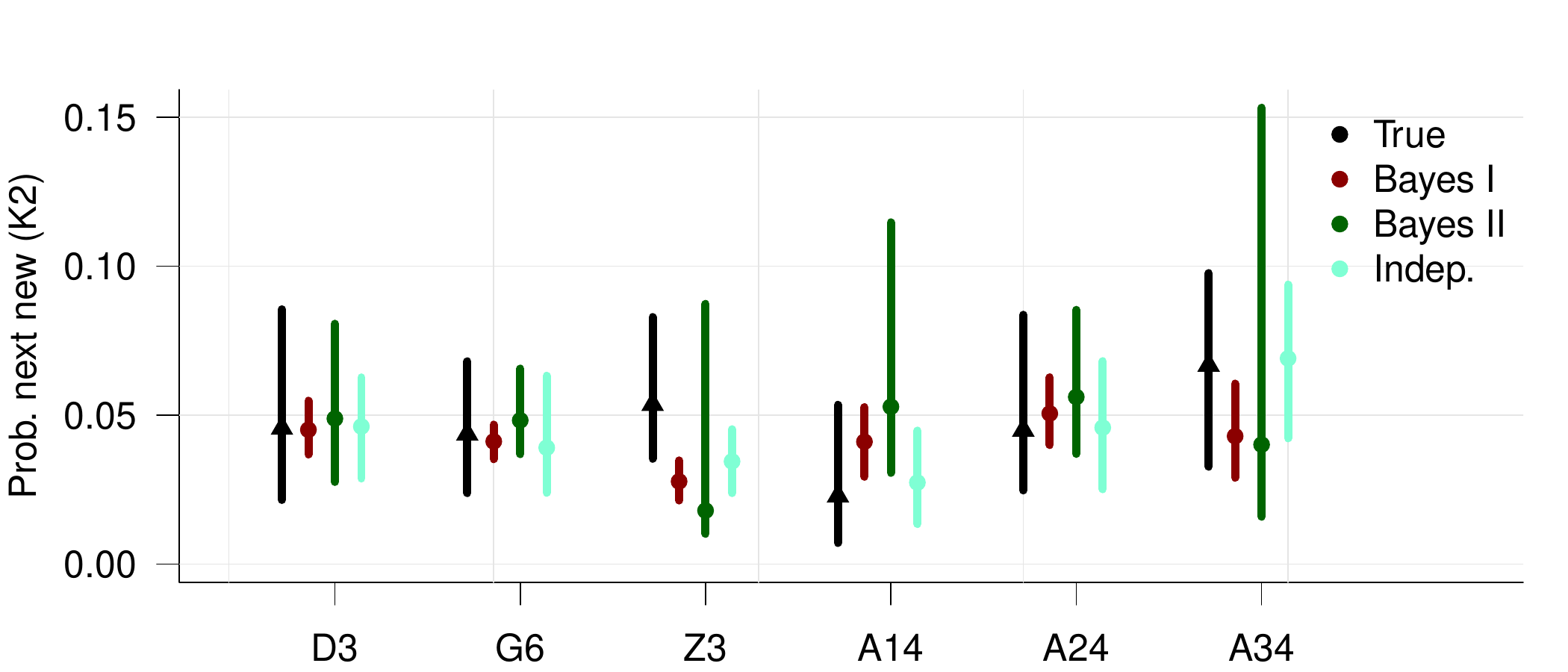}
        \caption{Experiment 2: one-step-ahead prediction probability for new shared species (top-left panel), new global distinct species (top-right panel), and new local distinct species (bottom-left panel) and (bottom-right panel) across selected scenarios.}
        \label{fig:SS_Pr1step_1}
    \end{figure}

    Figure \ref{fig:SS_Pr1step_1} reports the one-step-ahead discovery probabilities, where the black line represents the oracle benchmark; therefore, estimates closer to the black curve indicate better performance. Overall, \textit{Bayes I} exhibits very low variability across replications, whereas \textit{Bayes II} is noticeably more variable, reflecting its different parameter-fitting strategy. In terms of accuracy, both Bayesian estimators systematically improve over the \textit{Pooled} baseline and deliver results that are generally comparable to \textit{Independent}. Importantly, this experiment also includes competitors for shared-species discovery probabilities: both \textit{Bayes I} and \textit{Bayes II} are broadly in line with the state-of-the-art frequentist estimators. Finally, as in the experiment in Section \ref{subsection:SS_exp1}, we note that our method is the only model-based approach that targets all four quantities simultaneously; moreover, unlike the frequentist competitors, it is not restricted to one-step-ahead prediction and naturally extends to $m$-step-ahead coverage probabilities, as illustrated in Section \ref{subsection:SS_exp_mappette}.


    \subsection{Further experiments}
    \label{subsection:further_experiments}
    Additional experiments are reported in the Supplementary material to further illustrate two key aspects of our approach. 
    First, in Section \ref{subsection:SS_exp_mappette}, we investigate $m$-steps-ahead discovery probabilities, showing how the probability of observing at least one new local, global, or shared species varies with both the current sample size and the planned future sampling effort; this provides a practical tool for designing stopping rules and quantifying diminishing returns. 
    Second, in Section \ref{subsection:SS_exp4}, we compare our predictive estimator of shared species with the shared-richness estimator of \cite{Chao2000}, clarifying differences in assumptions and interpretations between the prediction of future discoveries and the estimation of unobserved species richness.

	\section{Analysis of ants data}
	\label{section:application}
	We illustrate the methodology outlined in the previous sections to analyse a real-world dataset coming from a case study conducted by \cite{Zara21} to evaluate 
    the impact of urbanisation on biodiversity. 
    Although urbanisation is often cited as one of the main drivers of species extinction due to factors such as pollution, changes in land use, and the introduction and spread of alien species, the authors argue that urban green spaces can serve as important refuges 
    supporting high levels of species diversity. 
To investigate this topic, the researchers selected several sites in Trieste, a city in north-eastern Italy, representing different levels of urbanisation.    
    For our analysis, 
	we focus on two of these areas because of their ecological distinction: Bosco Bovedo (BB) and Orto Lapidario (OL).
	The former is a semi-natural urban forest located just outside residential areas, acting as a transition zone between urban and natural environments. The latter is a city park within a local museum complex.  The two selected sites are sufficiently far from each other (about 6 kilometres) to rule out migratory contamination between the areas. 
    Hence, we use the methodology introduced in Section \ref{section:VecFDP} to model the two areas using separate distributions, thereby accounting for possible differences in the	underlying distributions due to their spatial locations.

	Ground-dwelling arthropods were sampled using pitfall traps. 
	Formicidae were identified at the species level on the basis of their morphology, and their abundance was recorded.
    Overall, a total number of $2,971$ and $3,489$ ants were collected in the first area (BB) and in the second area (OL), respectively. However, $2037$ out of $2971$ (around $68\%$) observations in the first group belong to the same species (crematogaster schmidti), which we excluded from the analysis. Similarly, in the second area (OL), $1229$ out of $3489$ (around $35\%$) also belong to the same species (pheidole pallidula), which we also excluded. This procedure is common in ecological studies, as highly abundant species are known to contribute little to understanding species diversity. Hence, the dataset analysed consists of $n_1 = 934$ and $n_2 = 2235$ observations. The observed number of local distinct species is $r_1 = 17$ and $r_2 = 23$, while the global number of distinct species and the number of shared species are $r = 30$ and $t = 10$, respectively. The observed species proportions, sorted in decreasing order, are displayed in Figure \ref{fig:Ants_proportions_empirical}, while the species accumulation curves are reported in Figure \ref{fig:Ants_AccCrvs}.


    To analyse the ants dataset, we employ the same strategy introduced in the simulation study. 
    First, we evaluate the prediction performance in terms of the number of new shared and distinct species in an additional, unobserved sample. Secondly, we assess the $m$-steps-ahead prediction performance, starting from $m=1$. For both cases, the competing approaches are those described in Section \ref{subsection:data_generation}.

    Since this is a real dataset and the true future discoveries are not observable, we adopt a training and test split strategy to empirically assess predictive performance.
    Specifically, we partition the observed data into a training set and a test set: the training set is used to estimate the model parameters, which are then plugged into the proposed estimators to predict the number of new species in the test set.
    We repeat the analysis for several training set proportions, namely $\{10\%,\,30\%,\,50\%,\,70\%,\,90\%\}$.  For brevity, we report here only the results corresponding to a $50\%$ training set, with the remaining $50\%$ used as test data. The remaining proportions are deferred to Section \ref{app:application}.
    The experiments are repeated over $100$ independent splits obtained by sampling without replacement from the full dataset. In particular, each split is constructed so as to preserve the original imbalance between the two areas.
     
  Following Section \ref{subsection:SS_exp1}, the predictive performance on the test set is evaluated using RMSE, reported in the left panel of Figure \ref{fig:Ants_application}. The benefits of our approach are particularly evident for shared species, for which no natural competitors are available. Improvements are also observed for the prediction of new local distinct species in the first area, which is the one with the smaller sample size. Here, \textit{Bayes I} achieves a lower RMSE than the \textit{Independent} estimator,
  likely because it borrows strength information from the second area through the dependence structure induced by the model.
  This interpretation is supported by Figure \ref{fig:Ants_application_Pred_rest}, which reports the additional cases where the training set is smaller ($10\%$ and $30\%$). The proposed estimators, especially \textit{Bayes I}, achieve lower prediction error than the competitor models.
  
    Regarding the one-step-ahead prediction, the right panel of Figure \ref{fig:Ants_application} reports the estimated discovery probabilities under \textit{Bayes I}, \textit{Bayes II}, and the competing methods for a total sample size $n=350$. Additional results for other sample sizes, $n\in\{50,\,250,\,450,\,600\}$, are reported in Figure \ref{fig:Ants_application_Pr1step_rest}. We did not consider larger sample sizes since the probabilities are too small to be meaningful. 
    Regarding the shared species discovery probability, the lower tails of the intervals for both \textit{Yue} and \textit{Chao} include zero; therefore, they do not provide any guarantee of discovering a new shared species, despite the fact that the observed data suggest that not all shared species have been detected yet. In contrast, our model 
    yields a discovery probability for new shared species that is strictly greater than zero. Furthermore, the probability of discovering a new global distinct species is smaller under the \textit{Pooled} estimator than under our proposed estimators.
    Overall, the probabilities involved are very close to zero, in agreement with the accumulation curves in Figure \ref{fig:Ants_AccCrvs}, which appear close to saturation. This suggests that discovering something new in a single additional draw is unlikely.

      For this reason, we strengthen our analysis by considering the $m$-steps-ahead discovery probability, with $m>1$, of observing at least one new shared or distinct species. To illustrate this, we propose a two-dimensional visualisation of discovery probabilities across both the current sampling effort and the planned future effort. Specifically, for each quantity of interest, we represent the probability of making at least one new discovery as a function of: (i) the current sample size (horizontal axis), considering equally spaced sizes from $n=158$ to $n=2218$, and (ii) the size of the additional sample (vertical axis), from $m=1$ to $m=1000$ in steps of $10$. 
    The resulting heatmaps provide an immediate summary of the expected gain from further sampling and can be used in practice to decide whether it is worth investing additional resources or whether the experiment is already sufficiently exhaustive. 
    For the sake of space,  we focus only on the \textit{Bayes I} implementation; see Figure \ref{fig:Ants_application_mappette}. The figure confirms that the one-step-ahead analysis is not informative for any practical stopping rule since even at $n=158$ the discovery probability is already close to zero. 
    As the future sample size increases, the discovery probability grows and eventually approaches one, thereby providing a concrete, interpretable notion of the sampling effort required to achieve a desired chance of observing something new. Moreover, for fixed future effort, discovery probabilities tend to decrease as the sample size $n$ increases, reflecting diminishing returns: once many observations have already been collected, substantially larger additional samples are needed to reach the same probability of discovering new distinct or shared species. Finally, the discovery probabilities for new local distinct species in the second area decrease much more quickly as the observed sample size increases than in the first area, reflecting the marked imbalance between the two sites.

    \begin{figure}[ht!]
    	\centering
    	\includegraphics[width=0.49\linewidth]{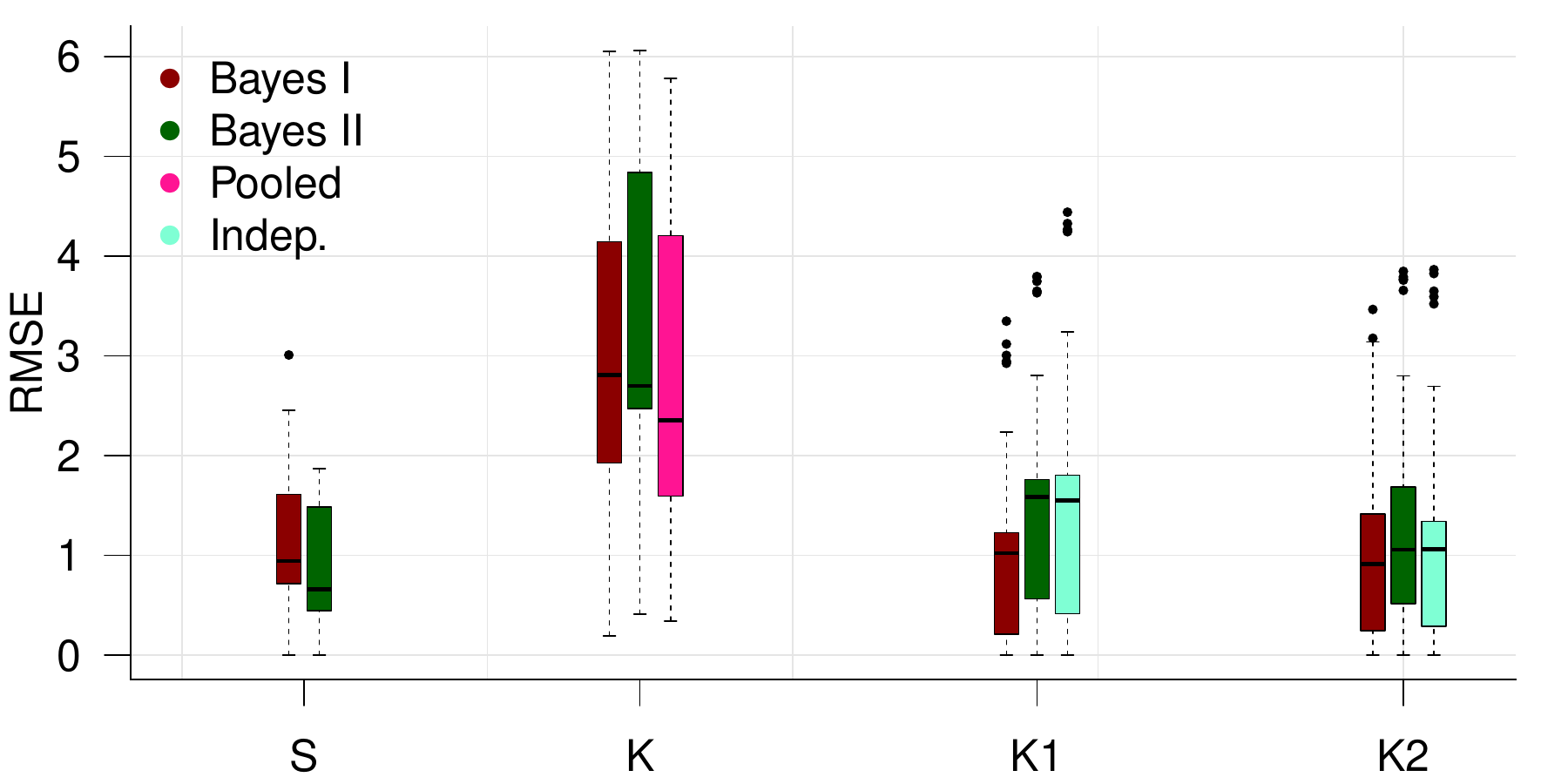}
    	\hfill
    	\includegraphics[width=0.49\linewidth]{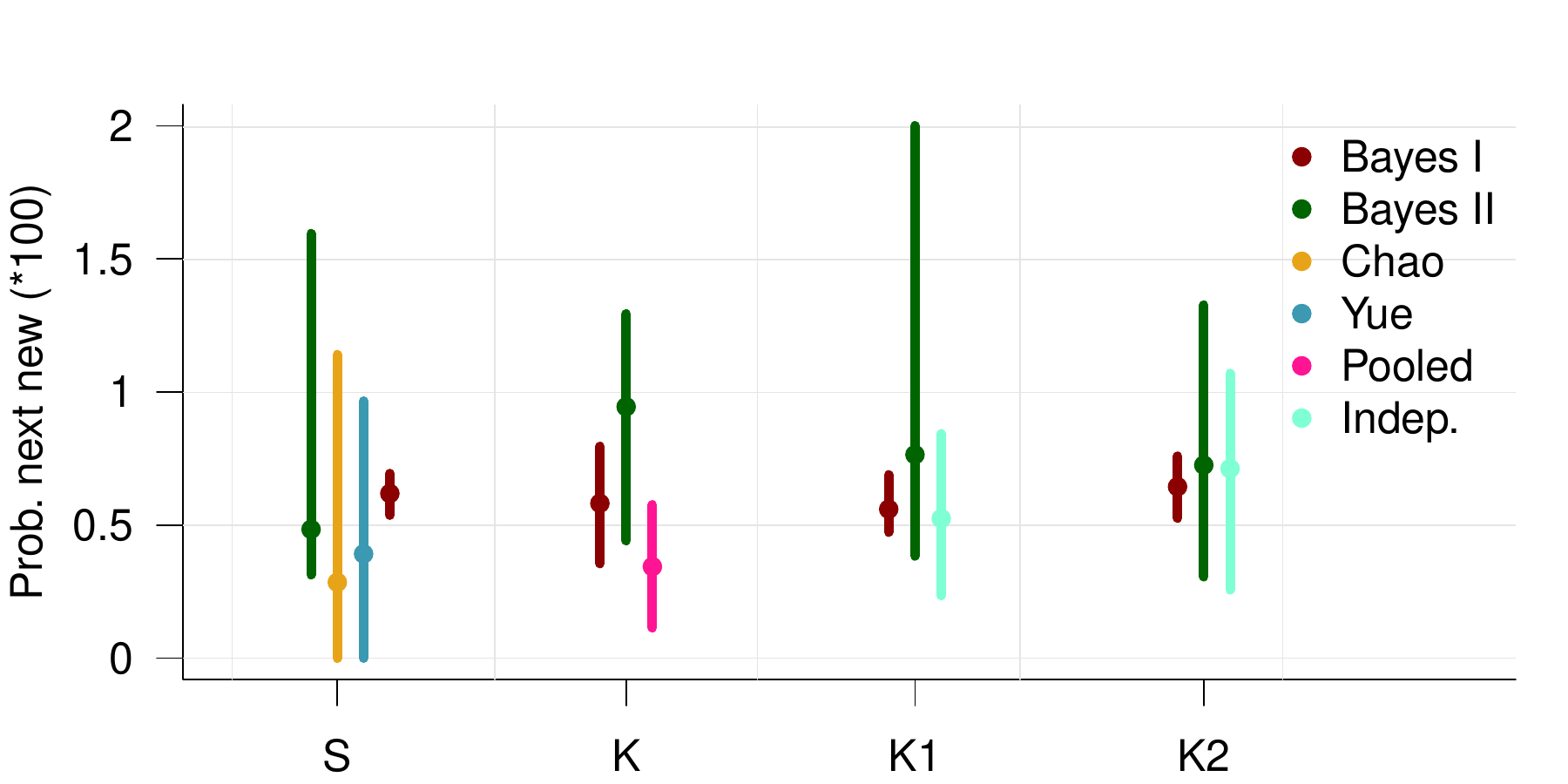}
    	\caption{RMSE of out-of-sample predictions (left panel) and one-step-ahead discovery probability (right panel) for shared (S), global distinct (K), and local distinct (K1, K2) species. Probabilities on the rightmost plot have been multiplied by $100$.}
    	\label{fig:Ants_application}
    \end{figure}
    
    \begin{figure}[ht!]
        \centering
        \includegraphics[width=0.47\linewidth]{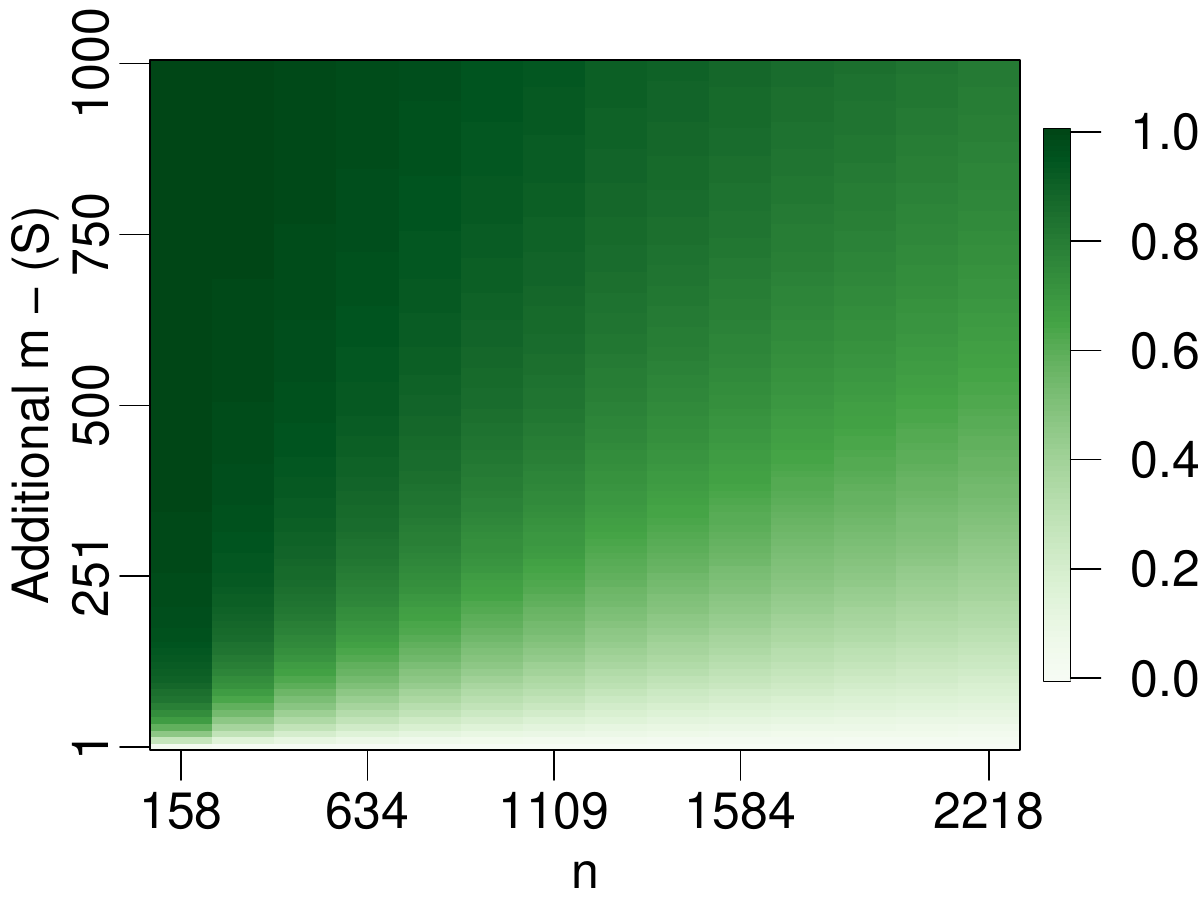}
        \hfill
        \includegraphics[width=0.47\linewidth]{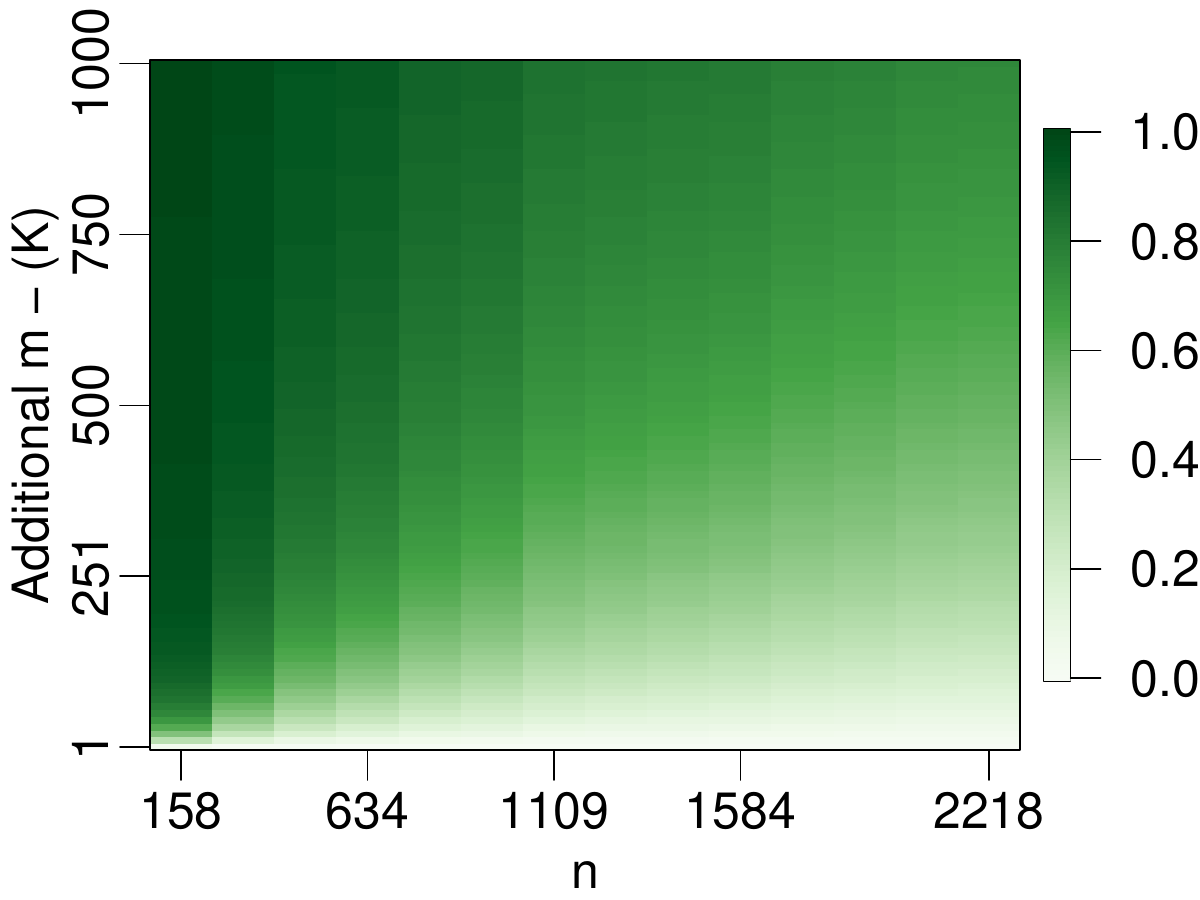}
        \\
        \includegraphics[width=0.47\linewidth]{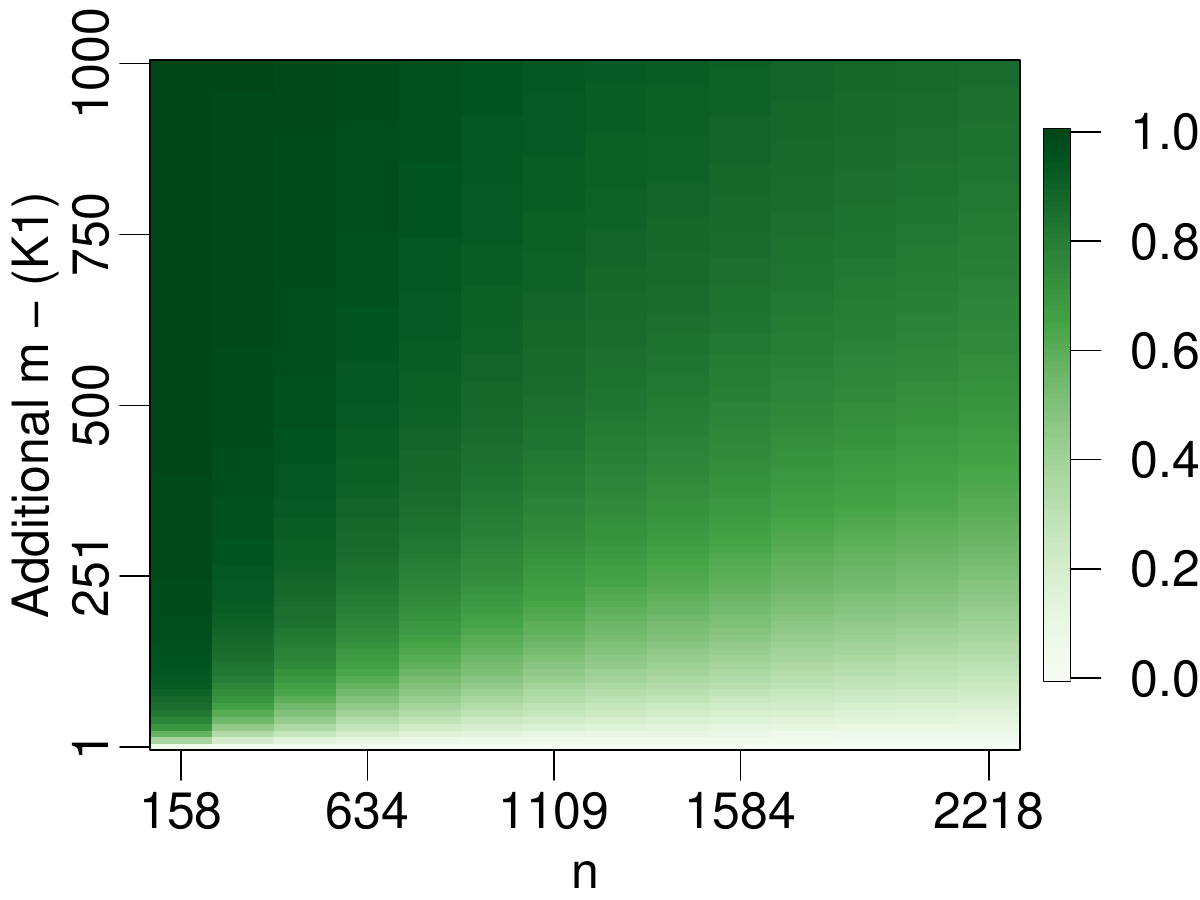}
        \includegraphics[width=0.47\linewidth]{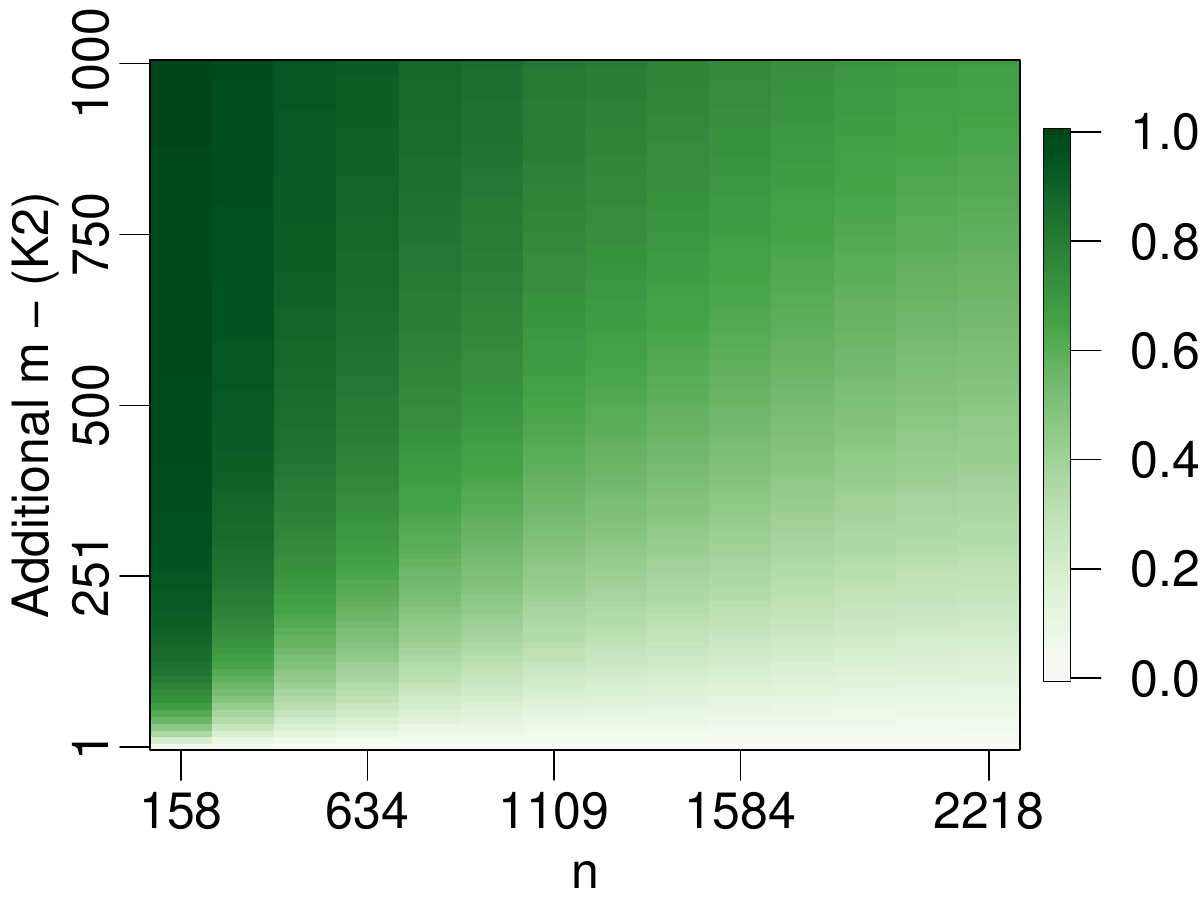}
        \caption{$m$-steps-ahead prediction probabilities for new shared species (top-left), global distinct species (top-right), and local distinct species (bottom-left) and (bottom-right). The x-axes display the total number of observations used in the analysis.}
        \label{fig:Ants_application_mappette}
    \end{figure}
    
    \section{Discussion}
    \label{section:discussion}
    The increased mathematical complexity has led to limited exploration of the species sampling problem across multiple areas in the literature, compared to the single-area case. In this work, we have introduced for the first time a model-based approach within the Bayesian nonparametric setting that relies entirely on closed-form expressions and exact calculations. Specifically, we focused on two primary objectives: predicting the number of distinct and shared species in an additional, unobserved sample of size $m \geq 1$, and 
    estimating the discovery probability of new shared species at $m$-steps ahead. 
    
    We argue that our work provides a foundation for significant future developments. For instance,
    the exact formulas presented in this work require the computation of the generalised factorial coefficients, whose evaluation scales quadratically with the sample size. We hope that a more in-depth investigation into the asymptotic behaviour of our model may pave the way for suitable approximations that make our solutions more scalable, similar to how \cite{favaro2009} improved upon \cite{Lij(07)} in the exchangeable case.
    In our experience, computing the coefficients $V$ does not constitute an additional computational bottleneck, as they converge rapidly. However, the infinite summation complicates the study of additional properties of our model, such as the expected values of the statistics of interest. 
    It remains an open question whether there exists a distribution $q_M$ such that Equation \eqref{eqn:Vprior} admits an explicit solution, similarly to Gnedin's model \citep{gnedin2010} in the single-group case.
    This is not the only open question about these $V$ coefficients. 
    Indeed, in the paper, we began with the model \eqref{eqn:partial_ex}, where the $P_j$'s are specified as in Equation \eqref{eqn:Pj_def}, to derive the
    $V$ coefficients, the pEPPF and a system of predictive rules. An alternative approach, similar to \cite{GnedinPitman2006}, would focus on the definition of a suitable set of $V$ coefficients that satisfy Equation \eqref{eqn:Vprior_recurrence} to obtain the  pEPPF and a system of predictive rules. This strategy would lead to the definition of Gibbs-type priors in a dependent framework, with the model analysed here being a specific example. We wonder if it is possible to identify other tractable examples within this class.

    A possible application for our framework is to the multi-armed setting, where the goal is to select the area to sample from so in order to maximise the probability of discovering a new global distinct species. A Bayesian nonparametric approach to this problem was first developed in \cite{Battiston2018MultiArmed} under the Hierarchical Pitman-Yor process and subsequently extended to the class of multivariate species sampling models in \cite{franzolini25}, which includes the $\operatorname{Vec\mbox{-}FDP}$ prior. The resulting decision rule can be expressed in terms of a functional of the $\operatorname{Vec-FDP}$ posterior, whose explicit form is given in \cite{colombi2023mixture}, including the case of more than two areas ($d>2$), which is typically the most relevant in these applications.
    
An open challenge is how to generalise our theoretical results to more than two areas. Notably, the main difficulty is intrinsic to the combinatorial structure of the multiple areas species problem. 
For $d=2$, the description of the observed sample can be expressed in terms of the local numbers of distinct species and the global number of distinct species, from which the number of shared species follows deterministically (see, e.g., Equation \eqref{eqn:linear_rel_post}). However, already for $d=3$, this relationship no longer holds. 
Indeed, a joint description of the set of quantities of interest requires additional information beyond the local and global counts, namely the statistics for each pair of areas. 
As $d$ increases, one must account not only for all pairwise statistics but also for triple-wise statistics, and so on, with the number of required quantities growing exponentially. This rapidly makes the corresponding distributional theory non-scalable in $d$. A direction for future work is to investigate alternative strategies --beyond the combinatorial marginalization arguments used here--  that directly target the global number of distinct and shared species, without requiring integration over the full collection of overlap statistics.


\noindent{\textbf{Acknowledgments:}}\\
This research was largely conducted while Alessandro Colombi was a Ph.D. student at University of Milano-Bicocca, Italy. We would like to thank Giovanni Bacaro (University of Trieste) for supporting with the ants data and Mario Beraha and Tommaso Rigon (University of Milano-Bicocca) for the helpful discussions. We thank the Editorial Board and the two Referees for their valuable comments and suggestions.
The first and the third authors were supported by the European Union -- Next Generation EU funds, component M4C2, investment 1.1., PRIN-PNRR 2022 (P2022H5WZ9). The second author was partially supported by the European Union -- Next Generation EU funds, component M4C2, investment 1.1., PRIN-PNRR 2022 (P2022AW4LX).

\newpage
\setcounter{page}{1} 
\setcounter{equation}{0}
\setcounter{section}{0} 
\setcounter{table}{0}
\setcounter{figure}{0}
\setcounter{prop}{0}
\renewcommand\thesection{S\arabic{section}}
\renewcommand\thetable{S\arabic{table}}
\renewcommand\thefigure{S\arabic{figure}}
\renewcommand\theequation{S\arabic{equation}}
\renewcommand\theprop{S\arabic{prop}}

\begin{center}
   \LARGE Supplementary materials for:\\
    "Bayesian discovery of species in multiple areas"
\end{center}

\section{Quantities of interest}
\label{app:quantities_details}

\subsection{Prior quantities}
\label{app:prior_quantities}
In this section, we present and detail the prior quantities of interest by means of an example, reported in Figure \ref{fig:esempio_animali_prior}, which shows the observed samples in two different areas. Moreover, Table \ref{tab:prior_quantities_def} summarises the notation and the meaning of each random variable.

\begin{figure}[ht!]
	\centering 
	\includegraphics[width=0.8\linewidth]{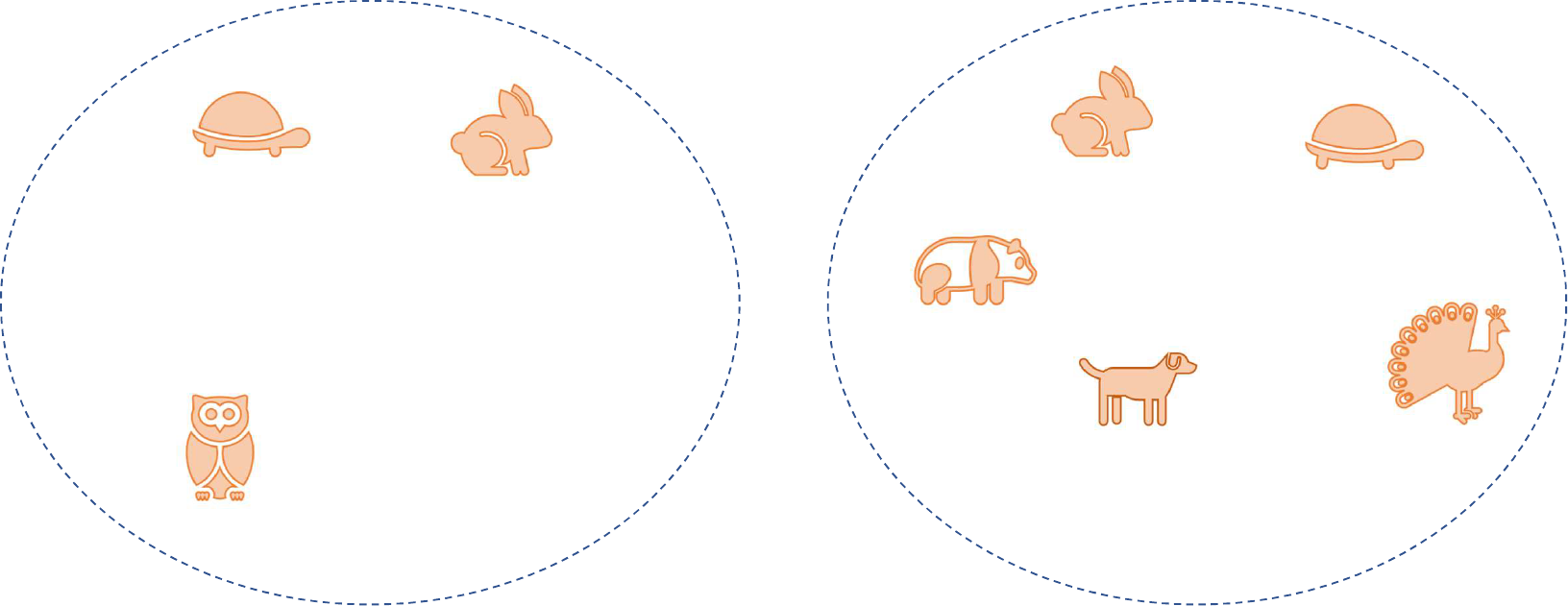}
	\caption{ Observed species in the observed sample. Each dotted circle delimits an area.}
	\label{fig:esempio_animali_prior}
\end{figure}

The first area is composed of $K_{1,n_1}=3$ distinct species. Following the notation of Section \ref{section:VecFDP}, this is a local quantity because it does not require any knowledge of the second area to be computed. Similarly, the local number of distinct species in the second area, $K_{2,n_2}$, equals five.
In general, we denote the observed values of $K_{j,n_j}$ as $r_j$, for $j=1,2$.
Then, moving to global quantities, i.e., those that require both areas to be computed, let the global number of distinct species be $\mathcal K_{n_1,n_2}$. This can be computed by pooling all species and discarding those that are repeated twice. In our running example, $\mathcal K_{n_1,n_2} = 6$ and, in general, this is denoted as $r$. We highlight that $r\geq r_j$, $j=1,2$, because each local distinct species is also a global distinct species, but we have $r\leq r_1+r_2$, because some species may appear in both areas. These are called shared species; the associated random variable is denoted as $\mathcal S_{n_1,n_2}$, and its realisation is $t$. In Figure \ref{fig:esempio_animali_prior}, $t=2$ (the rabbit and the turtle).

We also introduce two additional quantities, $K^*_{1,n}$ and $K^*_{2,n}$, which are useful in our proofs and for the understanding of our framework. Consider $j=1$, $K^*_{1,n}$ represents the number of species observed in the first group but not in the second one. For example, the owl is only seen in area 1, hence $K^*_{1,n} = 1$. On the other hand, $K^*_{2,n} = 3$ (the panda, the turkey, and the dog). In general, we let $K^*_{j,n} = r^*_j$, $j=1,2$. Specifically, even though $r^*_j$ refers to the number of species in a single area, this is also a global quantity, as it is not well defined if a second area is not observed. 

Summing up, we introduced six random variables which are dependent. In the following, we report their main relationships
\begin{equation}
	\label{eqn:prior_system_eq}
	\begin{split}
		t = r_1 + r_2 - r,
		\quad
		r = t + r_1^* + r_2^*, 
		\quad
		r_j = t + r_j^*, \quad j=1,2. 
	\end{split}
\end{equation}
Among these four linear equations, only three of them are linearly independent. As a consequence, the number of linearly independent random variables is three. We only need three out of six quantities to properly characterise the observed sample; the remaining three are deduced using the system of Equations \eqref{eqn:prior_system_eq}.

\begin{table}[ht]
	\centering
	\caption{In-sample statistics}	
	\label{tab:prior_quantities_def}
	\renewcommand{\arraystretch}{1.2}
	\begin{tabular}{ccl}
		\hline
		\textbf{Random variable} & \textbf{Realisations} & \textbf{Description} \\ 
		\hline
		$\mathcal K_{n_1,n_2}$ & $r$ & $\#$ of global distinct species \\
		$\mathcal S_{n_1,n_2}$ & $t$ & $\#$ of shared species\\
		$K_{j,n_j}$ & $r_j$ & $\#$ of local distinct species in group $j$\\
		$K^*_{j,n}$ & $r^*_j$ & $\#$ of distinct species observed in group $j$ \\
		\empty & \empty &  but which are missing group $j^\prime$
	\end{tabular}
\end{table}

\subsection{Posterior quantities}
\label{app:post_quantities}
Similarly to the previous section, we now present and detail the posterior quantities of interest by means of an updated example, reported in Figure \ref{fig:esempio_animali_post}. In red, we show the same observed species as in Figure \ref{fig:esempio_animali_prior} while we add the additional species in green. These are only present in a future sample, and therefore, they are still unobserved. We say that Figure \ref{fig:esempio_animali_post} represents the enlarged sample, composed of both the observed and the future ones.
Finally, Table \ref{tab:post_quantities_def} summarises the notation and the meaning of each random variable.

\begin{figure}[ht!]
	\centering 
	\includegraphics[width=0.8\linewidth]{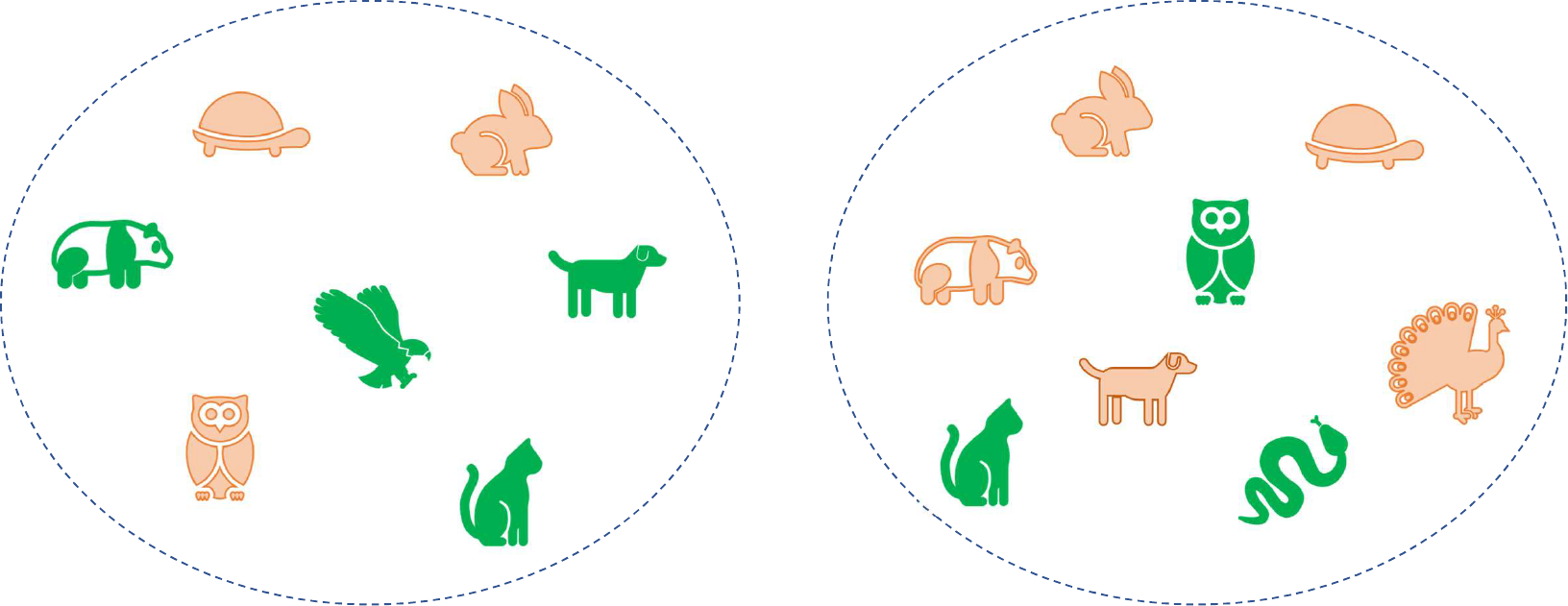}
	\caption{ Observed and future species in the enlarged sample, green species represent those belonging to the future, additional, sample, hence they are unobserved.  }
	\label{fig:esempio_animali_post}
\end{figure}

Once again, we start from the local number of new distinct species, $K^{(n_j)}_{j,m_j}$, for $j=1,2$, and generally denoted as $k_j$. Let $j=1$, this is defined as the number of distinct species in the enlarged sample minus the ones that were already observed in the observed sample. In our running example, this is simply computed as the number of green species in area 1, that is, four, $k_1 = 4$. This does not imply that a future sample can not contain turtles and rabbits, but these would not count as new species because they were already present in the previous sample. Similarly, $k_2 = 3$.
Let us move to the global quantities, starting from the global number of distinct species $\mathcal K^{(n_1,n_2)}_{m_1,m_2}$, whose realisation is denoted as $k$. This is also defined as the number of global species in the enlarged sample minus the same quantity in the observed sample, namely, $\mathcal K^{(n_1,n_2)}_{m_1,m_2} = \mathcal K_{n_1+m_1,n_2+m_2} - \mathcal K_{n_1,n_2}$. Some care is required when computing this quantity from Figure \ref{fig:esempio_animali_post}. Indeed, in Figure \ref{fig:esempio_animali_prior} it was enough to pool the red species and discard the repeated ones. Hence, one may be tempted to pool the green species and discard the repeated ones, but this would be wrong. Indeed, note that the owl and the dog have already been observed in area 1 and 2, respectively, and must not count as new global species. Hence, $k=3$ (the eagle, the cat, and the snake).
Although trivial, this example shows something important, that is, we can not compute $k$ from $k_1$ and $k_2$ only looking at the future sample (green species), but we must also take into account the past sample (red species).
Let us also have a look at the other main global quantity, that is the number of new shared species, $\mathcal S^{(n_1,n_2)}_{m_1,m_2}$, also defined as $\mathcal S^{(n_1,n_2)}_{m_1,m_2} = \mathcal S_{n_1+m_1,n_2+m_2} - \mathcal S_{n_1,n_2}$ and denoted as $s$. Once again, if we repeat what we did in Section \ref{app:prior_quantities}, i.e., pooling the red species and counting the number of the repeated ones, with the green species in Section \ref{fig:esempio_animali_post}, we end up committing a mistake. Indeed, the only repeated green species is the cat, while we claim that $s=4$ in our running example. Why is it so? As before, the correct way to count is to follow the definition and not to miss the mixed cases, i.e., those species that are first observed only in one group and then also in the future sample of the other one. 
In this case, the owl, the panda and the dog are shared species in the enlarged sample $\mathcal S_{n_1+m_1,n_2+m_2}$, as well as the previously mentioned cat and the turtle, but the latter has already been counted as a shared species in the observed sample $\mathcal S_{n_1,n_2}$, hence $s$ equals four.
Let us then characterise and name those quantities we implicitly computed to derive $k$ and $s$. Let $S^*_{m}$ be the number of those species that do not belong to the $r$ observed global species and that are shared among the two areas and the $s^*$ be its realisation. In our example, $s^*=1$, that is the cat, that is the only green species appearing in both areas. Moreover, let $S_{j^\prime,j}$ be the number of species that only appears in area $j^\prime$
for what concerns the observed sample but that are then also seen in area $j$ once the future sample is considered. Their realisations are denoted as $s_{j^\prime,j}$, for $j^\prime,j=1,2$ and $j^\prime\neq j$. For example, if $j=1$, then $j^\prime = 2$ and $s_{2,1} = 2$, i.e., the panda and the dog.
The computation of $k$ and $s$ can be summarised as
$k = k_1 + k_2 - s^* - s_{1,2} - s_{2,1}$ and $s = s^* + s_{1,2} + s_{2,1}$. Hence, $k = k_1 + k_2 - s$.

Finally, we follow Section \ref{app:prior_quantities} and introduce two additional auxiliary quantities, $K^{*(n)}_{j,m}$ for $j=1,2$, denoted as $k^*_j$. Let $j=1$, extending the prior quantity $K^*_{1,n}$, this counts the number of new distinct species that were not part of the observed $r$ global species and that are present in the first area. In our example, $k^*_1 = 1$ (the eagle) and $k^*_2 = 1$ (the snake). As for $K^*_{j,n}$, these are also global quantities that are not defined in the case of a single area.

One further consideration is that we can divide the global posterior random variables into two categories:
(i) those involving some of the $r$ observed species ($S_{1,2}$ and $S_{2,1}$) and (ii) those considering only new species, never observed in area 1 nor in area 2 ($S^*_{m}$, $K^{*(n)}_{1,m}$, and $K^{*(n)}_{2,m}$).
What about the two main quantities of interest $\mathcal K^{(n_1,n_2)}_{m_1,m_2}$ and $\mathcal S^{(n_1,n_2)}_{m_1,m_2}$?
As for the number of new shared species, this is directly related to $S_{1,2}$ and $S_{2,1}$, hence it clearly belongs to (i). On the other hand, the number of new distinct species, by definition, must belong to (ii). Indeed, for the sake of clarity, this has been derived as a function of $K^{(n_j)}_{j,m_j}$ and $S_{j^\prime,j}$, but it can also be computed as $k = s^* + k^*_1 + k^*_2$, which are all quantities belonging to (ii). 
This difference has a major impact when considering the computation of joint and marginal distributions of the posterior quantities of interest.

The main linear relationships among the nine introduced posterior random variables are the following
\begin{equation}
	\label{eqn:post_system_eq}
	\begin{split}
		&k = s^* + k_1^* + k_2^*,
		\quad
		s = s^* + s_{1,2} + s_{2,1},
		\quad
		s = k_1 + k_2 - k, \\
		& k_j = s^* + k_j^* + s_{j^\prime,j},
		\quad j^\prime,j=1,2;j^\prime\neq j
	\end{split}
\end{equation}
Only four of these equations are linearly independent, which means that the posterior set of random variables needs five linearly independent quantities to be fully characterised. 

In Section \ref{app:prior_quantities}, we pointed out that $r\leq r_1+r_2$ and $r\geq r_j$, for $j=1,2$. In this posterior case, the analogous condition $k\leq k_1+k_2$ still holds, but in the main manuscript, we notice that $k\geq k_j$ does not hold any longer. For example, in Section \ref{fig:esempio_animali_post} we have $k=3$, which is smaller than $k_1=4$. Why is it so? Intuitively, it is possible that $k_j$ is growing because of the discovery of many species that were only observed in area $j^\prime$, hence the presence of these in area $j$ is not increasing $k$.
More precisely, the set of equations in \eqref{eqn:prior_system_eq}, it can be shown that the condition for $k_j > k$ to happen is $s_{j^\prime,j}>k^*_{j^\prime}$.

\begin{table}[ht]
	\centering
	\caption{Ouf-of-sample statistics.}
	\label{tab:post_quantities_def}
	\renewcommand{\arraystretch}{1.2}
	\begin{tabular}{ccl}
		\hline
		\textbf{Random Variable} & \textbf{Realization} & \textbf{Description} \\ 
		\hline
		$\mathcal K^{(n_1,n_2)}_{m_1,m_2}$ & $k$ &  $\#$ of new global distinct species \\
		$\mathcal S^{(n_1,n_2)}_{m_1,m_2}$ & $s$ &  $\#$ of new shared species \\
		$K^{(n_j)}_{j,m_j}$ & $k_j$ &  $\#$ of new local distinct species in group $j$\\
		$K^{*(n)}_{j,m}$ & $k^*_j$ &  $\#$ of new distinct species in group $j$\\
		\empty & \empty &   but missing in group $j^\prime$ \\
		$S^*_{m}$ & $s^*$ &  $\#$ of new shared species among \\
		\empty & \empty &   the $k$ new distinct species \\
		$S_{j^\prime,j}$ & $s_{j^\prime,j}$ & $\#$ of species which were first \textit{only} observed in  \\
		\empty & \empty & group $j\prime$ and that are then observed in group $j$ \\
	\end{tabular}
\end{table}

\section{Review of generalised factorial coefficients}
\label{app:Cnumbers}
The results presented in Section \ref{section:prior} and Section \ref{section:posterior} rely on both central and non-central generalised factorial coefficients. In this section, we provide some background about these combinatorial objects and report the most relevant formulae that are extensively used in subsequent sections. We refer to \cite[Ch. 8]{chara2002} for a detailed discussion on this topic.

For any positive integers $n$ and $k$, with $k\leq n$, the generalised factorial coefficient $C(n,k; \gamma)$ is the coefficient of the $(k)$th order falling factorial of $t$ in the expansion of the $(n)$th order generalised factorial of $t$
with scale parameter $\gamma$, namely
\begin{equation}
	\label{eqn:Ccentral_def}
	(\gamma t)_{n\downarrow} = \sum_{k=0}^n C(n,k; \gamma) (t)_{k\downarrow};
\end{equation}
Sometimes, we refer to $C(n,k; \gamma)$ as the central generalised factorial coefficient to distinguish it from its non-central generalisation, which is
\begin{equation}
	\label{eqn:Cnoncentral_def}
	(\gamma t-\rho)_{n\downarrow} =  
	\sum_{k=0}^n C(n,k; \gamma, \rho) (t)_{k\downarrow};
\end{equation}
In particular, we have that $C(n,k; \gamma,0) = C(n,k; \gamma)$.
We wish to highlight the use of the falling factorial in Equations \eqref{eqn:Ccentral_def} and \eqref{eqn:Cnoncentral_def}.
If one were to replace it with the rising factorial, this would lead to a different definition of the generalised factorial coefficient, as used, for instance, in \cite{Lij(07)}. We denote this alternative form as $\Ccr (n,k; \gamma, \rho)$. The two definitions are connected by the identity
$\Ccr(n,k;\gamma,\rho) = (-1)^{n-k}C(n,k;\gamma,\rho)$.

An important formula that relates the central and the non-central generalised factorial coefficients is the following one
\begin{equation}
	\label{eqn:CnC_general}
	C(n,k; \gamma, \rho)
	= 
	\sum_{j=k}^n \binom{n}{j} (\rho)_{(n-j)\downarrow}  C(j,k; \gamma).
\end{equation}
Moreover, let $(x)_n$ denote the $(n)$th order rising factorial of $x$ and recall that $(x)_{n\downarrow} = (-1)^n(-x)_n$. From Equation \eqref{eqn:CnC_general} we also derive the following generalisation of Equation \eqref{eqn:CnC_general} involving the absolute values of the generalised factorial coefficients,
\begin{equation}
	\label{eqn:CnC}
	\lvert C(n,k; \gamma, \rho) \rvert
	= 
	\sum_{j=k}^n \binom{n}{j} (-\rho)_{n-j} \lvert C(j,k; \gamma) \rvert;
\end{equation}
where in Equation \eqref{eqn:CnC} we also exploited the fact that $\lvert C(n,k; \gamma, \rho) \rvert = (-1)^n C(n,k; \gamma, \rho) $. 
Finally, for $\gamma>0$, we also remind the following formula 
\begin{equation}
	\label{eqn:gen_formula}
	\lvert C(n,k; -\gamma) \rvert = 
	\frac{1}{k!} \sum_{(\star)}  \binom{n}{n_1, \ldots , n_k} \prod_{l=1}^k  (\gamma)_{r_l} 
\end{equation}
where the sum is taken over the following set
\[
(\star) = \left\{ (n_1, \ldots , n_k) : \; n_l \geq 1 \,, \; n_1+\cdots+n_k = n \right\}.
\]
It will also be useful to remind an important generalisation of Vandermonde's identity:
\begin{equation}
	\label{eqn:vandermonde}
	\sum_{(\star \star)}  \binom{n}{n_1  ,\ldots , n_k} \prod_{l=1}^k (\gamma_l )_{r_l} =
	(\gamma_1 +\ldots +\gamma_k)_n
\end{equation}
where the sum is taken over the following set
\[
(\star \star) = \left\{ (n_1, \ldots , n_k) : \; n_l \geq 0 , \; n_1+\cdots+n_k = n \right\}
\]
and $\gamma_l>0$ for $l=1,\ldots,k$.

\section{Details and proofs of results in Section \ref{section:predictive}}
\label{app:joint_CRFP}
Section \ref{section:predictive} reports the predictive distribution of the $(n_1+1)$th observation when, according to the Chinese restaurant franchise metaphor, the client enters the first restaurant.
Here, we report the general case when a new pair of clients arrives, one for each group, i.e.,  the $(n_1+1)$th and $(n_2+1)$th clients in the first and second restaurants, respectively. The predictive distribution is 
\begin{equation}
	\label{eqn:full_predictive}
	\begin{split}
		&\P\left( X_{1,n_1+1}\in\,A\,, X_{2,n_2+1}\in\,B\,\mid\bmX\right) \\
		& \ =
		\frac{V^r_{n_1+1,n_2+1}}{V^r_{n_1,n_2}} 
		\sum_{l_1=1}^r\sum_{l_2=1}^r
		\left(n_{1,l_1}+\gamma_1\right)
		\left(n_{1,l_2}+\gamma_2\right)
		\delta_{X^{**}_{l_1}}(A)
		\delta_{X^{**}_{l_2}}(B) \\
		&\quad+
		\frac{V^{r+1}_{n_1+1,n_2+1}}{V^r_{n_1,n_2}} 
		\left\{\,
		\sum_{l_1=1}^r
		\left(n_{1,l_1}+\gamma_1\right)
		\delta_{X^{**}_{l_1}}(A)
		\gamma_2
		P_0(B)
		+
		\gamma_1
		P_0(A)
		\sum_{l_2=1}^r
		\left(n_{1,l_2}+\gamma_2\right)
		\delta_{X^{**}_{l_2}}(B) 
		\,\right\} \\
		&\quad +
		\frac{V^{r+1}_{n_1+1,n_2+1}}{V^r_{n_1,n_2}} 
		\gamma_1\gamma_2
		P_0(A\cap B) +
		\frac{V^{r+2}_{n_1+1,n_2+1}}{V^r_{n_1,n_2}} 
		\gamma_1\gamma_2
		P_0(A)P_0(B),
		\end {split}
	\end{equation}
	for any measurable sets $A$ and $B$. The one-step-ahead predictive distribution in Equation \eqref{eqn:full_predictive} follows from Equation \eqref{eqn:joint_post_d2} after noticing that $ \lvert C(1, 0; -\gamma_j, -(\gamma_j r_j + n_j))\rvert = \gamma_j r_j + n_j$ and $\lvert C(1, 1; -\gamma_j, -(\gamma_j r_j + n_j))\rvert = \gamma_j$.

\subsection{Proof of Equation \eqref{eqn:peppf}}
\label{app:proof_peppf}
\cite{colombi2023mixture} derived the following equivalent form for the pEPPF,
\begin{equation*}
	\label{eqn:proof_peppf_1}
	\begin{split}
		&\Pi_{r}^{(n)}\left(\bmn_1,\bmn_2\right)  \\
		& \ = 
		\bigintsss_{[0,\infty]\times[0,\infty]}
		\Psi(r,u_1,u_2)
		\prod_{j=1}^2
		\frac{u_j^{n_j-1}}{\Gamma(n_j)\left(1+u_j\right)^{n_j+\gamma_j r}}
		{\rm d}u_1{\rm d}u_2
		\,
		\times 
		\prod_{j=1}^2\prod_{l=1}^r
		\left(\gamma_j\right)_{n_{j,l}},
	\end{split}
\end{equation*}
and, as we are only interested in the case of symmetric Dirichlet distributed random weights, $\Psi(r,u_1,u_2)$ takes the following form
\begin{equation*}
	\label{eqn:proof_peppf_2}
	\begin{split}
		\Psi(r,u_1,u_2) \ = \ 
		\sum_{\mstar=0}^\infty
		\frac{(\mstar+r)!}{\mstar!}
		q_M(\mstar + r) 
		\prod_{j=1}^2\left(1+u_j\right)^{\gamma_j\mstar}.
	\end{split}
\end{equation*}
In summary, we want to exchange the integral and the infinite sum and solve the remaining integrals with respect to $u_1$ and $u_2$. By doing do, we have that
\begin{equation*}
	\label{eqn:proof_peppf_3}
	\begin{split}
		&\Pi_{r}^{(n)}\left(\bmn_1,\bmn_2\right)  \\
		& \ = 
		\prod_{j=1}^2\prod_{l=1}^r
		\left(\gamma_j\right)_{n_{j,l}}
		\,
		\sum_{\mstar=0}^\infty
		\left\{
		\frac{(\mstar+r)!}{\mstar!}
		q_M(\mstar + r) 
		\prod_{j=1}^2
		\int_{0}^\infty
		\frac{1}{\Gamma(n_j)}
		\frac{u_j^{n_j-1}}{\left(1+u_j\right)^{n_j+\gamma_j( r + \mstar)}}
		{\rm d}u_j
		\right\}
		\, .
	\end{split}
\end{equation*}
The integral equals a Beta function \citep[p.258]{abram}, hence we have that
\begin{equation}
	\label{eqn:proof_peppf_4}
	\begin{split}
		&\Pi_{r}^{(n)}\left(\bmn_1,\bmn_2\right)  \\
		& \ = 
		\sum_{\mstar=0}^\infty
		\left\{
		\frac{(\mstar+r)!}{\mstar!}
		q_M(\mstar + r) 
		\prod_{j=1}^2
		\frac{B(n_j,\gamma_j(\mstar + r))}{\Gamma(n_j)}
		\right\}
		\prod_{j=1}^2\prod_{l=1}^r
		\left(\gamma_j\right)_{n_{j,l}}
		\, \\
		& \ =
		\sum_{\mstar=0}^\infty
		\left\{
		\frac{(\mstar+r)!}{\mstar!}
		q_M(\mstar + r) 
		\prod_{j=1}^2
		\frac{1}{\left(\gamma_j(\mstar + r)\right)_{n_{j}}}
		\right\}
		\prod_{j=1}^2\prod_{l=1}^r
		\left(\gamma_j\right)_{n_{j,l}} \, .
	\end{split}
\end{equation}
To complete the proof, it is enough to change variables in the infinite sum and define $V_{n_1,n_2}^r$ as
\begin{equation*}
	\label{eqn:proof_peppf_5}
	\begin{split}
		V_{n_1,n_2}^r = 
		\sum_{\mstar=0}^\infty
		(\mstar)_{r\downarrow}
		q_M(\mstar) 
		\prod_{j=1}^2
		\frac{1}{\left(\gamma_j(\mstar)\right)_{n_{j}}}.
	\end{split}
\end{equation*}
The latter coincides with the definition given in Equation \eqref{eqn:Vprior}.

\section{ Analysis of the $V_{n_1,n_2}^r$ coefficients}
\label{section:Vcoeff}
We discuss here some properties of the coefficients $V_{n_1,n_2}^r$, defined in Equation \eqref{eqn:Vprior}. In the exchangeable case, i.e., when all observations are drawn from the same area, 
\cite{GnedinPitman2006} shows that the sampling model, described by the Exchangeable Partition Probability Function (EPPF), admits the product form
\begin{equation}
	\label{eqn:Gibbs_type}
	\Pi_{r}^{(n)}(n_1,\ldots,n_r) \, = \, V_n^r \prod_{l=1}^{r}(1-\sigma)_{n_l-1} \, ,
\end{equation}
for any $\sigma<1$, $n\geq1$, $r\leq n$ and positive integers $n_1,\ldots,n_r$ that sum up to $n$ if and only if the set of non-negative weights $\{V_r^n \,:\, n\geq 1\,,1\leq r \leq n \}$ satisfies the recurrence relationship $V_{n}^r = V_{n+1}^{r+1} + (n - \sigma r)V_{n+1}^r $.
The sampling models in Equation \eqref{eqn:Gibbs_type} are known as Gibbs-type models, see \cite{deblasi2015} for further discussion. Within this class, it is important to distinguish between two cases, namely when $\sigma \in [0,1)$ and when $\sigma < 0$. The former corresponds to models with a potentially infinite number of species, whereas the latter assumes a finite -though random- number of species. This assumption aligns with the hypothesis of our model, and we therefore confine our attention to this case.
Finally, \cite{millerharrison} reparametrize the model in Equation \eqref{eqn:Gibbs_type} so that the recurrence relationship above takes the form
\begin{equation}
	\label{eqn:Vcoeff_recrel_1d}
	V_{n}^r = \gamma V_{n+1}^{r+1} + (\gamma r + n)V_{n+1}^r \, ,
\end{equation}
and the coefficients $V_{n}^r$ admit the infinite sum representation,
\begin{equation}
	\label{eqn:Vcoeff_MH}
	V_{n}^r \ = \
	\sum_{m=1}^\infty
	\, (m)_{r\downarrow} \, q_M(m) \, \frac{1}{(\gamma m)_{n}}.
\end{equation}
The latter expression is also recovered when assuming Dirichlet distributed weights in the Normalized Independent Finite Point Process model by \cite{argiento2022annals}.

Since Equation \eqref{eqn:Vcoeff_MH} is recovered from Equation \eqref{eqn:Vprior} by setting either $n_1$ or $n_2$ equal to zero, we say that the coefficients $V_{n_1,n_2}^r$ are a multi-group extension of the $V_n^r$ coefficients in the case of finitely many species. Indeed, we show that they share similar properties to $V_n^r$. In particular, $V_{n_1,n_2}^r$ multiplies the general term of the series in \eqref{eqn:Vcoeff_MH} by a factor that rapidly decreases to zero, both with respect to the series index and the number of observations. As a consequence, $V_{n_1,n_2}^r$ converges even faster than $V_n^r$. Hence, the next proposition states that the $V_{n_1,n_2}^r$ coefficients are well defined and, for sufficiently large sample sizes, can be accurately approximated by the $(r)$th term of the series. 
Proposition \ref{proposition:Vcoef_conv_asym} extends \cite{GnedinPitman2006} and \cite{millerharrison} to the multi-group setting. 
\begin{prop}
	\label{proposition:Vcoef_conv_asym}
	The $V_{n_1,n_2}^r$ coefficients introduced in Equation \eqref{eqn:Vprior} are well defined for every choice of probability mass function $q_M$. Moreover, for any integer $r\geq1$ such that $q_M(r)>0$, the following approximation holds:
	\begin{equation}
		\label{eqn:Vapprox}
		\begin{split}
			V_{n_1,n_2}^r \, &= \,
			\frac{r!q_M(r)}{(\gamma_1 r)_{n_1}(\gamma_2 r)_{n_2}} \\
			&\quad \times \, 
			\left\{
			1+n_1^{-\gamma_1}n_2^{-\gamma_2}(r+1)(\gamma_1 r)_{\gamma_1}(\gamma_2 r)_{\gamma_2}\frac{q_M(r+1)}{q_M(r)} + o(n_1^{-\gamma_1}n_2^{-\gamma_2})
			\right\} \, .
		\end{split}
	\end{equation}
\end{prop}
\begin{proof}
    Firstly, we use the change of variables $\mstar = m-r$ and rewrite $V_{n_1,n_2}^r$ as
    \begin{equation}
    	\label{eqn:Vconv_proof_1}
    	V_{n_1,n_2}^r
    	\ = \ 
    	\sum_{\mstar=0}^\infty
    	\frac{(m+r)!}{m!}\prod_{j=1}^2 \frac{1}{(\gamma_j (\mstar + r))_{n_j}} q_M(\mstar+r) \, .
    \end{equation}
    Then, to prove the statement, we consider the series of the asymptotic expansion of the general term. To do so, we use the following Stirling approximation for large values of $\mstar$: 
    \begin{equation}
    	\label{eqn:Vconv_proof_2}
    	\frac{(\mstar+r)!}{\mstar!} \ \sim \  e^{-r} \sqrt{\frac{\mstar+r}{\mstar}}\left(\frac{\mstar+r}{\mstar}\right)^{\mstar}(\mstar + r)^r \, ,
    \end{equation}
    where we use the notation $f(x)\sim g(x)$ as a short form for $f(x) = o_{x_0}(g(x))$ for $x\rightarrow x_0$.
    We also recall the following approximation for the ratio of the Gamma function $\Gamma(a+cm)/\Gamma(b+cm) \sim (cm)^{a-b}$ when $m\rightarrow\infty$. Hence, we have that 
    \begin{equation}
    	\label{eqn:Vconv_proof_3}
    	\frac{\Gamma(\gamma_j \mstar + \gamma_j r)}{\Gamma(\gamma_j \mstar + \gamma_j r + n_j)} \ \sim \ 
    	(\gamma_j \mstar)^{-n_j} \, .
    \end{equation}
    Using Equations \eqref{eqn:Vconv_proof_2} and \eqref{eqn:Vconv_proof_3}, the following is the asymptotic expression of the general term in Equation \eqref{eqn:Vconv_proof_1} for large values of $\mstar$ 
    \begin{equation*}
    	\label{eqn:Vconv_proof_4}
    	\begin{split}
    		&\frac{(\mstar+r)!}{\mstar!}\prod_{j=1}^2 \frac{1}{(\gamma_j (\mstar + r))_{n_j}} q_M(\mstar+r)  \\
    		& \quad  \sim 
    		e^{-r} \sqrt{\frac{\mstar+r}{\mstar}}\left(\frac{\mstar+r}{\mstar}\right)^{\mstar}(\mstar + r)^r
    		(\gamma_1)^{-n_1}(\gamma_2)^{-n_2}(\mstar)^{-n} q_M(\mstar + r)\\
    		& \quad \sim 
    		(\gamma_1)^{-n_1}(\gamma_2)^{-n_2}\left(\frac{\mstar + r}{\mstar}\right)^r (\mstar)^{-n+r}q_M(\mstar + r) \, , \\
    		& \quad \sim 
    		(\gamma_1)^{-n_1}(\gamma_2)^{-n_2} \frac{1}{(\mstar)^{n-r}} q_M(\mstar + r) \, , \\ 
    	\end{split}
    \end{equation*}
    where we defined $n=n_1+n_2$, used $\left(\frac{\mstar+r}{\mstar}\right)^{\mstar} \sim e^r$, and wrote $(\mstar)^{-n}$ as $(\mstar)^{-(n-r)-r}$.
    As a consequence, the convergence of $V_{n_1,n_2}^r$ can be assessed by studying the convergence of the following series,
    \begin{equation*}
    	\label{eqn:Vconv_proof_5}
    	\begin{split}
    		\sum_{\mstar = 0}^\infty
    		(\gamma_1)^{-n_1}(\gamma_2)^{-n_2} \frac{1}{(\mstar)^{n-r}} q_M(\mstar + r)
    		\leq
    		(\gamma_1)^{-n_1}(\gamma_2)^{-n_2} \sum_{\mstar = 0}^\infty q_M(\mstar + r) < \infty \, .
    	\end{split}
    \end{equation*}
    The latter follows since $r\leq n$ implies that $\frac{1}{(\mstar)^{n-r}} \leq 1$. Additionally, being $q_M$ a probability mass function, the final sum is less than or equal to one.

    We continue proving the statement about the asymptotic expansion of $V^r_{n_1,n_2}$.
    To simplify the notation, we write $V_{n_1,n_2}^r$ as 
    \begin{equation*}
    	\label{eqn:Vapprox_proof_0}
    	V_{n_1,n_2}^r = \left\{V_{n_1,n_2}^r\right\}_r + \sum_{m=r+1}^\infty \left\{V_{n_1,n_2}^r\right\}_m \, 
    \end{equation*}
    where $\left\{V_{n_1,n_2}^r\right\}_m$ is the $(m)$th term of the series. Namely,
    \begin{equation*}
    	\label{eqn:Vapprox_proof_1}
    	\left\{V_{n_1,n_2}^r\right\}_m \ = \ 
    	(m)_{r\downarrow}\prod_{j=1}^2 \left((\gamma_j m)_{n_j}\right)^{-1} q_M(m) \, ,
    \end{equation*}
    for each integer $m\geq r$. Since $q_M(r)>0$, we can collect the first term, and we get
    \begin{equation*}
    	\label{eqn:Vapprox_proof_2}
    	V_{n_1,n_2}^r \ = \ \left\{V_{n_1,n_2}^r\right\}_r 
    	\left[ 1 + 
    	\frac{1}{\left\{V_{n_1,n_2}^r\right\}_r}
    	\sum_{m=r+1}^\infty \left\{V_{n_1,n_2}^r\right\}_m 
    	\right] .
    \end{equation*}
    Following the same steps of the proof in \cite{millerharrison}, we show that the second term in the squared brackets goes to zero. Once again, we first isolate the $\left\{V_{n_1,n_2}^r\right\}_{r+1}$ term for the infinite sum.
    \begin{equation}
    	\label{eqn:Vapprox_proof_3}
    	\begin{split}
    		&\frac{1}{\left\{V_{n_1,n_2}^r\right\}_r}
    		\sum_{m=r+1}^\infty \left\{V_{n_1,n_2}^r\right\}_m = 
    		\frac{\left\{V_{n_1,n_2}^r\right\}_{r+1}}{\left\{V_{n_1,n_2}^r\right\}_r} + 
    		\frac{1}{\left\{V_{n_1,n_2}^r\right\}_r}
    		\sum_{m=r+2}^\infty \left\{V_{n_1,n_2}^r\right\}_m \\ 
    		& \qquad = 
    		n_1^{-\gamma_1}n_2^{-\gamma_2}(r+1)(\gamma_1 r)_{\gamma_1}(\gamma_2 r)_{\gamma_2}\frac{q_M(r+1)}{q_M(r)} + 
    		\frac{1}{\left\{V_{n_1,n_2}^r\right\}_r}
    		\sum_{m=r+2}^\infty \left\{V_{n_1,n_2}^r\right\}_m \, .
    	\end{split}
    \end{equation}
    The second term in the final line of Equation \eqref{eqn:Vapprox_proof_3} converges to $0$ as $n_1,n_2\rightarrow0$. This follows from \cite[Proposition S1.1]{millerharrison}. Moreover, it can be written as an infinite polynomial with respect to $(n_1)^{-\gamma_1}(n_2)^{-\gamma_2}$ for some polynomial coefficients depending on $m$. Namely
    \begin{equation}
    	\label{eqn:Vapprox_proof_4}
    	\begin{split}
    		&\frac{1}{\left\{V_{n_1,n_2}^r\right\}_r}
    		\sum_{m=r+2}^\infty \left\{V_{n_1,n_2}^r\right\}_m \, \\
    		& \qquad = \ 
    		\sum_{m=r+2}^\infty C_m\left((n_1)^{-\gamma_1}\right)^{m-r}\left((n_2)^{-\gamma_2}\right)^{m-r} = 
    		o((n_1)^{-\gamma_1}(n_2)^{-\gamma_2}) \, .
    	\end{split}
    \end{equation}
    The statement follows by combining Equations \eqref{eqn:Vapprox_proof_3} and \eqref{eqn:Vapprox_proof_4}.
\end{proof}

Furthermore, we derive the recurrence relationship satisfied by $V_{n_1,n_2}^r$. 
\begin{prop}
	\label{proposition:Vcoeff_recrel}
	Let $V_{n_1,n_2}^r$ be the coefficients defined in Equation \eqref{eqn:Vprior}. 
	Then, the following 1-step recurrence relationship holds 
	\begin{equation}
		\label{eqn:Vcoef_1step_recrel}
			V_{n_1,n_2}^r \, =\, \gamma_2 V_{n_1,n_2+1}^{r+1} + (\gamma_2r + n_2)V_{n_1,n_2+1}^r \,.
		\end{equation}
		or equivalently, 
		\begin{equation}
			\label{eqn:Vcoef_1step_recrel_bis}
			V_{n_1,n_2}^r \, = \, \gamma_1 V_{n_1+1,n_2}^{r+1} + (\gamma_1r + n_1)V_{n_1+1,n_2}^r \,.				
		\end{equation}
		Moreover, the following $2$-steps recurrence relationship holds,
		\begin{equation}
			\label{eqn:Vprior_recurrence}    
			\begin{split}
				V_{n_1,n_2}^r \, = \, & 
				\gamma_1\gamma_2
				\left\{
				r^2V_{n_1+1,n_2+1}^{r} + (2r+1)V_{n_1+1,n_2+1}^{r+1} + V_{n_1+1,n_2+1}^{r+2}
				\right\} \\
				&+ n_1V_{n_1+1,n_2}^{r} + n_2V_{n_1,n_2+1}^{r} - n_1n_2V_{n_1+1,n_2+1}^{r} \,.
			\end{split}
		\end{equation}
	\end{prop}
\begin{proof}
    To prove Equation \eqref{eqn:Vcoef_1step_recrel} we exploit the identity
    \begin{equation*}
    	\label{app:Vcoef_1step_recrel_proof1}
    	(m)_{(r+1)\downarrow} \, = \, \frac{\gamma_2m+n_2}{\gamma_2}(m)_{r\downarrow} - (r + \frac{n_2}{\gamma_2})(m)_{(r)\downarrow}\,.
    \end{equation*}
    Then, we have that
    \begin{equation*}
    	\label{app:Vcoef_1step_recrel_proof2}
    	\begin{split}
    		&V_{n_1,n_2+1}^{r+1} \, = \, 
    		\sum_{m=1}^\infty (m)_{r\downarrow}
    		\left(\frac{1}{\gamma_2}\frac{(\gamma_2m + n_2)}{(\gamma_2m)_{(n_2+1)}}\right)
    		\frac{1}{(\gamma_1m)_{(n_1)}}q_M(m) \\
    		&\quad - \,
    		(r + \frac{n_2}{\gamma_2})
    		\sum_{m=1}^\infty (m)_{r\downarrow}
    		\left(\frac{1}{(\gamma_2m)_{(n_2+1)}}\right)
    		\frac{1}{(\gamma_1m)_{(n_1)}}q_M(m)		
    		\,.
    	\end{split}
    \end{equation*}
    The statement follows after recognising the Pochhammer symbol in the first term of Equation \eqref{eqn:Vcoef_1step_recrel}. Equation \eqref{eqn:Vcoef_1step_recrel_bis} can be proven analogously.  
    
    Firstly, we rewrite $(m)_{(r+2)\downarrow}$ in a convenient way.
    \begin{equation*}
    	\label{eqn:Vrecurr_proof_1}
    	\begin{split}
    		(m)_{(r+2)\downarrow} = (m-r-1)(m-r)(m)_{r\downarrow} = 
    		\left\{m^2-m(2r+1) +r(r+1)\right\}(m)_{r\downarrow} \, .
    	\end{split}
    \end{equation*}
    Then, we exploit the following identity
    \begin{equation*}
    	\label{eqn:Vrecurr_proof_2}
    	\begin{split}
    		m^2 = \frac{(\gamma_1 m + n_1)-n_1}{\gamma_1}\frac{(\gamma_2 m + n_2)-n_2}{\gamma_2} \, ,
    	\end{split}
    \end{equation*}
    to write
    \begin{equation}
    	\label{eqn:Vrecurr_proof_3}
    	\begin{split}
    		(m)_{(r+2)\downarrow} = 
    		\Biggl\{
    		&\frac{(\gamma_1 m + n_1)-n_1}{\gamma_1}\frac{(\gamma_2 m + n_2)-n_2}{\gamma_2} - 
    		\frac{n_1(\gamma_2m+n_2)}{\gamma_1\gamma_2} -
    		\frac{n_2(\gamma_1m+n_1)}{\gamma_1\gamma_2} + \\ 
    		&\quad + \frac{n_1n_2}{\gamma_1\gamma_2} + r(r+1) - m(2r+1)
    		\Biggr\}
    		(m)_{r\downarrow} \, .
    	\end{split}
    \end{equation}
    We exploit Equation \eqref{eqn:Vrecurr_proof_3} to have
    \begin{equation*}
    	\label{eqn:Vrecurr_proof_4}
    	\begin{split}
    		&V_{n_1+1,n_2+1}^{r+2} = 
    		\sum_{m=0}^\infty
    		(m)_{(r+2)\downarrow}\frac{1}{(\gamma_1m)_{n_1+1}(\gamma_2m)_{n_2+1}}q_M(m)\\
    		&\, =
    		\sum_{m=0}^\infty
    		\Biggl\{
    		\frac{(\gamma_1 m + n_1)-n_1}{\gamma_1}\frac{(\gamma_2 m + n_2)-n_2}{\gamma_2} - 
    		\frac{n_1(\gamma_2m+n_2)}{\gamma_1\gamma_2} -
    		\frac{n_2(\gamma_1m+n_1)}{\gamma_1\gamma_2} + \frac{n_1n_2}{\gamma_1\gamma_2} \\
    		& \qquad \qquad
    		+ r(r+1) - m(2r+1)
    		\Biggr\}
    		\frac{(m)_{r\downarrow}}{(\gamma_1m)_{n_1}(\gamma_2m)_{n_2}(\gamma_1m+n_1)(\gamma_1m+n_2)}q_M(m) 
    	\end{split}
    \end{equation*}
    Using the definition of $V$ coefficients in Equation \eqref{eqn:Vprior}, we split the sum and recognise some terms.
    \begin{equation}
    	\label{eqn:Vrecurr_proof_5}
    	\begin{split}
    		V_{n_1+1,n_2+1}^{r+2} \, = \,&
    		\frac{1}{\gamma_1\gamma_2} V_{n_1,n_2}^r -
    		\frac{n_1}{\gamma_1\gamma_2} V_{n_1+1,n_2}^r-
    		\frac{n_2}{\gamma_1\gamma_2} V_{n_1,n_2+1}^r+
    		\left\{\frac{n_1n_2}{\gamma_1\gamma_2} + r(r+1)\right\}V_{n_1+1,n_2+1}^r\\
    		& \quad - 
    		(2r+1)\sum_{m=0}^\infty
    		m(m)_{r\downarrow}\frac{1}{(\gamma_1m)_{n_1+1}(\gamma_2m)_{n_2+1}}q_M(m).
    	\end{split}
    \end{equation}
    Using $m(m)_{r\downarrow} = r(m)_{r\downarrow} + (m)_{(r+1)\downarrow}$ and following the same steps as before, it is easy to show that the sum in the final line of Equation \eqref{eqn:Vrecurr_proof_5} equals $rV_{n_1+1,n_2+1}^r + V_{n_1+1,n_2+1}^{r+1}$. Hence, we have
    \begin{equation*}
    	\label{eqn:Vrecurr_proof_6}
    	\begin{split}
    		&V_{n_1+1,n_2+1}^{r+2} \, = \,
    		\frac{1}{\gamma_1\gamma_2}
    		\left\{
    		V_{n_1,n_2}^r - n_1 V_{n_1+1,n_2}^r - n_2 V_{n_1,n_2+1}^r + 
    		\left(n_1n_2 + \gamma_1\gamma_2r(r+1)\right)
    		\right\} \\
    		& \qquad \qquad \qquad
    		- (2r+1) \left\{rV_{n_1+1,n_2+1}^r + V_{n_1+1,n_2+1}^{r+1}\right\} \, .
    	\end{split}
    \end{equation*}
    The statement follows after some trivial linear algebra.

\end{proof}

\section{Details and proofs of the results in Section \ref{section:prior}}
\label{app:proof_sec_post}
\subsection{Proof of Theorem \ref{thm:jointprior_d2}}
\label{app:proof_joint_prior}
	Firstly, we notice from Section \ref{app:prior_quantities} that $\mathcal K_{n_1,n_2}$, $K_{1,n_1}$ and $K_{2,n_2}$ are linearly independent quantities from which we can also derive the remaining ones, i.e., $\mathcal S_{n_1,n_2}=t$ and $K^*_{j,n}$, for $j=1,2$. From Equations \eqref{eqn:prior_system_eq}, it follows that
	\begin{equation}
		\label{eqn:proof_pr_S_Kstar}
		t \ = \ r_1 + r_2 - r, \quad
		r^*_1 \ = \ r - r_2, \quad 
		r^*_2 \ = \ r - r_1.
	\end{equation}
	Moreover, the following conditions must hold: $1\leq r \leq n_1+n_2$ and $1\leq r_j \leq n_j$, for $j=1,2$. 
	The probability of interest may be evaluated as follows
	\begin{equation*}
		\label{eqn:proof_pr_1}
		\P\left(\mathcal K_{n_1, n_2} = r, K_{1, n_1} = r_1, K_{2, n_2} = r_2 \right) =  \sum_{(\star)}  
		\frac{1}{r!}  \prod_{j=1}^2 
		\binom{n_j}{n_{j,1}, \ldots , n_{j,r}} \cdot \Pi_r^{(n)} (\bmn_1,  \bmn_2) 
	\end{equation*}
	where $n = n_1+n_2$ and the sum $(\star)$ are extended over all the vectors $(n_{1,1}, \ldots , n_{1,r})$ and $(n_{2,1}, \ldots , n_{2,r})$ of non-negative integers satisfying the following constraints
	\begin{equation*}
		\label{eqn:proof_pr_insieme1}
		\begin{split}
			& \sum_{l=1}^r n_{j,l} = n_j \text{ with } n_{j,l} \geq 0,
			\quad 
			l = 1, \ldots , r \, j = 1,2 \, , \\
			& n_{1,l}+n_{2,l} \geq 1 \quad l=1, \ldots , r \, , 
			\ 
			\sum_{l=1}^r \delta_{ \{n_{j,l}>0\} } = r_j, \, j = 1,2
			\, .
		\end{split}
	\end{equation*}
	By exploiting the expression of the pEPPF in Equation \eqref{eqn:peppf}, we get
	
	\begin{equation*} 
		\label{eqn:proof_pr_2}
		\begin{split}
			&\P\left(\mathcal K_{n_1, n_2} = r, K_{1, n_1} = r_1, K_{2, n_2} = r_2 \right)  \\
			& \qquad = 
			V_{n_1,n_2}^r
			\sum_{(\star)} \frac{1}{r!}  
			\prod_{j=1}^2
			\left\{
			\binom{n_j}{n_{j,1}, \ldots , n_{j,r}} \prod_{l=1}^r  (\gamma_j)_{n_{j,l}}
			\right\}.
		\end{split}
	\end{equation*}
	In the following, we aim to solve the sum over the set $(\star)$.
	The main difficulty here is the joint condition $n_{1,j}+n_{2,l}\geq 1$; therefore, we elaborate the sum, trying to decouple it into two sums, only involving the local cardinalities $n_{j,l}$.
	To do so, $t$ out of $r$ species must be shared. Without loss of generality, assume that the first $t$ species are these shared species. Moreover, we fix an ordering for them, noticing that this operation can be done in $\binom{r}{t}$ equivalent ways. Hence, we have
	\begin{equation} 
		\label{eqn:proof_pr_3}
		\begin{split}
			&\sum_{(\star)} \frac{1}{r!}  
			\prod_{j=1}^2
			\left\{
			\binom{n_j}{n_{j,1}, \ldots , n_{j,r}} \prod_{l=1}^r  (\gamma_j)_{n_{j,l}}
			\right\}  \\
			& \qquad = 
			\sum_{(\star\star)} 
			\frac{1}{r!} \binom{r}{t}
			\prod_{j=1}^2
			\left\{
			\binom{n_j}{n_{j,1}, \ldots , n_{j,r}} \prod_{l=1}^r  (\gamma_j)_{n_{j,l}}
			\right\},
		\end{split}
	\end{equation}
	where the set $(\star\star)$ must satisfy the following constraints,
	\begin{equation*}
		\label{eqn:proof_pr_insieme2}
		\begin{split}
			& \sum_{l=1}^r n_{j,l} = n_j \, j=1,2 \, ,  \\ 
			&n_{j,l} \geq 1, \quad l = 1, \ldots , t\, , 
			\quad
			n_{j,l} \geq 0, \quad j = 1,2; \, l = t+1, \ldots , r \, ,\\
			& n_{1,l}+n_{2,l} \geq 1 \quad
			n_{1,l}\cdot n_{2,l} = 0
			\quad l=t+1, \ldots , r \, , \\
			&\sum_{l=1}^r \delta_{ \{n_{j,l}>0\} } = r_j, \quad j = 1,2 \, .
		\end{split}
	\end{equation*}
	Equation \eqref{eqn:proof_pr_3} is properly defined as long as $0\leq t \leq r$, which implies $r \leq r_1+r_2 \leq 2r$.    
	
	When moving from set $(\star)$ to set $(\star\star)$, the joint condition only refers to the final $r-t$ species. We can further reorder such species that are not shared among the two groups. Indeed, we know that $r_1^*$ species have been observed in group 1 only, while $r^*_2$ species are specific to group 2 only. 
	Hence, we assume the $r_1^*$ species are in positions from $t+1$ to $t+r_1^*$, and we fix an ordering in any of the $\binom{r-t}{r^*-1}$ possible ways. 
	The remaining species are $r-t-r_1^*$ which, from Equations \eqref{eqn:proof_pr_S_Kstar}, can be shown to equal $r^*_2$.
	Note that, from Equations \eqref{eqn:proof_pr_insieme2}, $n_{1,l}\geq 1$ implies $n_{2,l} = 0$ for $l=t+1, \ldots , r+r^*_1$ while $n_{2,l} \geq 1$ implies $n_{1,l} = 0$ for $l=t+r^*_1+1, \ldots , r$. This fully resolves the joint condition in set $(\star\star)$. 
	Moreover, for each $j=1,2$, the number of non-zero elements in vectors $\bmn_j$ is $t+r^*_j$, and from Equations \eqref{eqn:proof_pr_S_Kstar}, this equals $r_j$. This guarantees one of the conditions in Equation \eqref{eqn:proof_pr_insieme2}. As a consequence, we can discard zero elements in vectors $\bmn_1$ and $\bmn_2$ and we obtain
	\begin{equation} 
		\label{eqn:proof_pr_4}
		\begin{split}
			&\sum_{(\star)} \frac{1}{r!}  
			\prod_{j=1}^2
			\left\{
			\binom{n_j}{n_{j,1}, \ldots , n_{j,r}} \prod_{l=1}^r  (\gamma_j)_{n_{j,l}}
			\right\}  \\
			& \qquad = 
			\frac{1}{r!} \binom{r}{t} \binom{r-t}{r^*_1}
			\prod_{j=1}^2
			\left\{
			\sum_{(\star\star j)} 
			\binom{n_j}{n_{j,1}, \ldots , n_{j,r_j}}
			\prod_{l=1}^{r_j}  (\gamma_j)_{n_{j,l}}
			\right\},
		\end{split}
	\end{equation}
	where the sum over the sets $(\star\star j)$, for $j=1,2$, is extended over all vectors $(n_{j,1},\ldots,n_{j,r_j})$ such that $ n_{j,l} \geq 1$ and $\sum_{l=1}^{r_j} n_{j,l} = n_j$. 
	Equation \eqref{eqn:proof_pr_4} is properly defined as long as $0\leq r^*_1 \leq r-t$, which implies $r_1 \leq r$ and $r_2 \leq r$. In particular, this also ensures that $r_1+r_2 \leq 2r$. 
	
	Finally, we use Equation \eqref{eqn:gen_formula} to solve the final sums over the sets $(\star\star j)$, and we conclude that
	\begin{equation*} 
		\label{eqn:proof_pr_5}
		\begin{split}
			&\P\left(\mathcal K_{n_1, n_2} = r, K_{1, n_1} = r_1, K_{2, n_2} = r_2 \right)  \\
			& \qquad = 
			V_{n_1,n_2}^r
			\frac{1}{r!} \binom{r}{t} \binom{r-t}{r^*_1}
			\prod_{j=1}^2
			r_j!
			\lvert C(n_j, r_j; -\gamma_j )\rvert.
		\end{split}
	\end{equation*}
	The statement follows by plugging the values of $t$ and $r^*_1$ in terms of $r$, $r_1$ and $r_2$ reported in Equations \eqref{eqn:proof_pr_S_Kstar} and rearranging the 
	factorials and the binomial coefficients.
	
	\subsection{Additional details on the in-sample marginal statistics}
	\label{app:insample_marginal}
	
	\cite{colombi2023mixture} stated and proved the following expression for the marginal distribution of the prior number of distinct species:
	\begin{equation}
		\label{eqn:K12_prior}
		\begin{split}
			\P\left(\, \mathcal K_{n_1,n_2} = r \, \right) 
			=
			V_{n_1,n_2}^r 
			\sum_{z_1=0}^r  
			\sum_{z_2=0}^{r-z_1} 
			\binom{r-z_1}{z_2}
			\frac{(r-z_2)!}{z_1!}
			\prod_{j=1}^2  
			\left \lvert C\left(n_j, r-z_j ; -\gamma_j \right) \right \rvert,
		\end{split}
	\end{equation}
	for $r \in\{1,\ldots, n_1 + n_2\}$. 
	In Section \ref{app:prior_quantities}, we stated that the number of linearly independent variables we need to characterise the observed sample is three. In the main manuscript, we choose the local and global number of distinct species. However, one may not be interested in the local quantities while seeking only the case of global quantities. In particular, the most interesting case is the one involving both the distinct and the shared number of species, namely, $\mathcal S_{n_1,n_2}$ and $\mathcal K_{n_1,n_2}$. Hence, if we complete the set of linearly independent variables by augmenting with respect to $K^*_{1,m} = k^*_1$ and then marginalising it out, we obtain the following joint prior distribution
	
	\begin{equation}
		\label{eqn:SKprior_d2}
		\begin{split}
			& \P\left(\mathcal K_{n_1, n_2} =r, \ \mathcal S_{n_1,n_2} = t \right) \ = \ V_{n_1,n_2}^r  \\
			& \qquad \times 
			\sum_{k^*_1=0}^{r-t} 
			\binom{r-k^*_1}{t}
			\frac{(t+k^*_1)!}{k^*_1!}
			\lvert C(n_1, t+k^*_1; -\gamma_1 )\rvert 
			\lvert C(n_2, r-k^*_1; -\gamma_2 )\rvert \, ,
		\end{split}
	\end{equation}
	for $r \in\{1,\ldots, n_1+n_2\}$ and $t\in\{0,\ldots,\min\{r,n_1,n_2\}$.
	The proof of Equation \eqref{eqn:SKprior_d2} follows the same steps as the one in Section \ref{app:proof_joint_prior}.
	Additionally, the marginal distribution of $\mathcal S_{n_1,n_2}$ can be obtained by marginalising $\mathcal K_{n_1, n_2}$ out of Equation \eqref{eqn:SKprior_d2}. We have that,
	\begin{equation}
		\label{eqn:Sprior}
		\begin{split}
			\P\left( \mathcal S_{n_1,n_2} = t \right) = &
			\sum_{r=1}^{n_1+n_2} \sum_{k^*_1=0}^{r-t} V_{n_1,n_2}^r  
			\binom{r-k^*_1}{t}
			\frac{(t+k^*_1)!}{k^*_1!} \\
			& \qquad  \times 
			\lvert C(n_1, t+k^*_1; -\gamma_1 )\rvert 
			\lvert C(n_2, r-k^*_1; -\gamma_2 )\rvert \, .
		\end{split}
	\end{equation}
	
	Let us now move to consider the local quantities, hence fixing $j\in\{1,2\}$. 
	The marginal distribution for the local number of prior distinct species is 
	\begin{equation}
		\label{eqn:Kj_prior}
		\P\left(\, K_{j,n_j} = r_j \, \right) 
		\ = \ 
		V_{n_j}^{r_j} \, 
		\lvert C\left(n_j,r_j; -\gamma_j \right) \rvert,
	\end{equation}
	for $r_j\in\{1,\ldots, n_j\}$ and where $V_{n_j}^{r_j}$ is obtained from $V_{n_1,n_2}^r$, defined in Equation \eqref{eqn:Vprior}, setting $n_{j^\prime} = 0$, for $j^\prime \neq j$. In particular, this implies $r = r_j$.
	We also notice that $V_{n_j}^{r_j}$ coincides with the definition of the V coefficients in \cite{GnedinPitman2006} and \cite{millerharrison}. The proof of this result is given in \cite{argiento2022annals}.

	\section{Details and proofs of the results in Section \ref{section:posterior}}
	\label{app:proof_sec_prior}
	\subsection{Proof of Equation \eqref{eqn:qMpost}}
	\label{app:proof_qMpost}
	The proof of Equation \eqref{eqn:qMpost} easily follows from Section \ref{app:proof_peppf}. Indeed, we can look at Equation \eqref{eqn:proof_peppf_4} as a posterior marginal distribution with respect to $M$. 
	It follows that,
	\begin{equation}
		\label{eqn:proof_qMpost_1}
		\begin{split}
			& q_M(\mstar\mid\bmX) \propto
			\frac{(\mstar+r)!}{\mstar!}
			q_M(\mstar + r) 
			\prod_{j=1}^2
			\frac{1}{\left(\gamma_j(\mstar + r)\right)_{n_{j}}}
			\prod_{j=1}^2\prod_{l=1}^r
			\left(\gamma_j\right)_{n_{j,l}} \\ 
			& \qquad \propto
			\frac{(\mstar+r)!}{\mstar!}
			q_M(\mstar + r) 
			\prod_{j=1}^2
			\frac{1}{\left(\gamma_j(\mstar + r)\right)_{n_{j}}}.
		\end{split}
	\end{equation}
	Then, the normalising constant for the previous expression is
	\begin{equation}
		\label{eqn:proof_qMpost_2}
		\begin{split}
			\sum_{\mstar = 0}^\infty
			\frac{(\mstar+r)!}{\mstar!}
			q_M(\mstar + r) 
			\prod_{j=1}^2
			\frac{1}{\left(\gamma_j(\mstar + r)\right)_{n_{j}}},
		\end{split}
	\end{equation}
	which is, up to a change of variables, the $V_{n_1,n_2}^r$ coefficient defined in Equation \eqref{eqn:Vprior}.
	
	\subsection{Proof of Equation \eqref{eqn:ExpqMpost}}
	\label{app:proof_ExpqMpost}
	In the following, we compute the expected value of the $M^\star$ whose probability mass function is $q^\star_{M\mid\bmX}$, defined in Equation \eqref{eqn:qMpost}, i.e., 
	\begin{equation*}
		\label{eqn:qMpost_1}
		\begin{split}
			&\E_{q^\star_{M\mid\bmX}}\left(M^\star\right) 
			=
			\frac{1}{V_{n_1,n_2}^{r}}
			\sum_{m^\star = 1}^\infty
			m^\star(m^\star+r)_{r\downarrow}q_M(m^\star + r)
			\prod_{j=1}^d \frac{1}{\left(\gamma_j\left(m^\star+r\right)\right)_{n_j}} \\
			& \quad =
			\frac{1}{V_{n_1,n_2}^{r}}
			\sum_{m^{\star\star} = 0}^\infty
			(m^{\star\star}+1)(m^{\star\star}+r+1)_{r\downarrow}q_M(m^{\star\star} + r + 1)
			\prod_{j=1}^d \frac{1}{\left(\gamma_j\left(m^{\star\star}+r+1\right)\right)_{n_j}}
			\\
			& \quad =
			\frac{1}{V_{n_1,n_2}^{r}}
			\sum_{m^{\star\star} = 0}^\infty
			(m^{\star\star} + r + 1)_{(r+1)\downarrow}
			q_M(m^{\star\star} + r + 1)
			\prod_{j=1}^d \frac{1}{\left(\gamma_j\left(m^{\star\star}+r+1\right)\right)_{n_j}} 
			\, ,
		\end{split}
	\end{equation*}
	where we first applied the change of index $m^{\star\star} = m^\star+1$ and then we used the following identity: 
	$(m^{\star\star}+1)(m^{\star\star}+r+1)_{r\downarrow}=(m^{\star\star}+r+1)_{(r+1)\downarrow}$. Then, we note that 
	$(m^{\star\star} + r + 1)_{(r+1)\downarrow} = 0$ for $m^{\star\star} \leq r$. Hence, we change variables once again, setting $\bar{m} = m^{\star\star} + r + 1$. By doing so, we obtain 
	\begin{equation*}
		\label{eqn:qMpost_2}
		\begin{split}
			\E_{q^\star_{M\mid\bmX}}\left(M^\star\right) 
			\, = \,
			\frac{1}{V_{n_1,n_2}^{r}}
			\sum_{\bar{m} = r+1}^\infty
			(\bar{m})_{(r+1)\downarrow}
			q_M(\bar{m})
			\prod_{j=1}^d \frac{1}{\left(\gamma_j\left(\bar{m}\right)\right)_{n_j}} 
			\, = \,
			\frac{V_{n_1,n_2}^{r+1}}{V_{n_1,n_2}^{r}} \, ,
		\end{split}
	\end{equation*}
	where the final equivalence follows by definition, see Equation \eqref{eqn:Vprior}.
	
	\subsection{Proof of Equation \eqref{eqn:ExpqMpost_approx}}
	\label{app:proof_ExpqMpost_approx}
	The expected value in Equation \eqref{eqn:ExpqMpost} is the ratio of two $V$ coefficients. 
	Using the asymptotic expansion given in Equation \eqref{eqn:Vapprox}, we have that
	\begin{equation}
		\label{eqn:ExpqMpost_approx_proof_1}
		\begin{split}
			&\E\left(M^\star\mid\bmX\right) \ = \ \frac{V_{n_1,n_2}^{r+1}}{V_{n_1,n_2}^r} \\
			& \ \sim 
			(r+1)\frac{(\gamma_1 (r+1) )_{n_1}(\gamma_2 (r+1) )_{n_2}}{(\gamma_1 r)_{n_1}(\gamma_2 r)_{n_2}}
			\frac{q_M(r+1)}{q_M(r)} \\
			& \qquad \times
			\frac{
				\left\{
				1+n_1^{-\gamma_1}n_2^{-\gamma_2}(r+2)(\gamma_1 (r+1))_{\gamma_1}(\gamma_2 (r+1))_{\gamma_2}\frac{q_M(r+2)}{q_M(r+1)} + o\left(n_1^{-\gamma_1}n_2^{-\gamma_2}\right)
				\right\}
			}{
				\left\{
				1+n_1^{-\gamma_1}n_2^{-\gamma_2}(r+1)(\gamma_1 r)_{\gamma_1}(\gamma_2 r)_{\gamma_2}\frac{q_M(r+1)}{q_M(r)} + o\left(n_1^{-\gamma_1}n_2^{-\gamma_2}\right)
				\right\}
			} \, .
		\end{split}
	\end{equation}
	To further expand the second term, let us define $C_r = (r+1)(\gamma_1 r)_{\gamma_1}(\gamma_2 r)_{\gamma_2} q_M(r+1)/q_M(r)$ and $\nstar_j = (n_j)^{\gamma_j}$ . Hence, we have that
	\begin{equation*}
		\label{eqn:ExpqMpost_approx_proof_2}
		\begin{split}
			&\frac{\left(1 + C_{r+1}(\nstar_1)^{-1}(\nstar_2)^{-1} + o\left((\nstar_1)^{-1}(\nstar_2)^{-1}\right) \right)}
			{\left(1 + C_{r}(\nstar_1)^{-1}(\nstar_2)^{-1} + o\left((\nstar_1)^{-1}(\nstar_2)^{-1}\right) \right)} \\
			& \quad = 
			\left(1 + C_{r+1}(\nstar_1)^{-1}(\nstar_2)^{-1} + o\left((\nstar_1)^{-1}(\nstar_2)^{-1}\right) \right)
			\left(1 - C_{r}(\nstar_1)^{-1}(\nstar_)2^{-1} + o\left((\nstar_1)^{-1}(\nstar_2)^{-1}\right) \right) \\ 
			& \quad = 
			1 + (C_{r+1}-C_{r})(\nstar_1)^{-1}(\nstar_2)^{-1} + o\left((\nstar_1)^{-1}(\nstar_2)^{-1}\right) \, .
		\end{split}
	\end{equation*}
	Plugging the latter expansion into \eqref{eqn:ExpqMpost_approx_proof_1}, we have that,
	\begin{equation}
		\label{eqn:ExpqMpost_approx_proof_3}
		\begin{split}
			&\E\left(M^\star\mid\bmX\right) \ = \ \frac{V_{n_1,n_2}^{r+1}}{V_{n_1,n_2}^r} \\
			& \ \sim 
			(r+1)\frac{q_M(r+1)}{q_M(r)}
			\prod_{j=1}^2(\gamma_jr)_{\gamma_j}\frac{\Gamma(\gamma_jr + n_j)}{\Gamma(\gamma_jr+\gamma_j+ n_j)} \\
			& \qquad \times
			\left\{
			1 + (C_{r+1}-C_{r})(\nstar_1)^{-1}(\nstar_2)^{-1} + o\left((\nstar_1)^{-1}(\nstar_2)^{-1}\right)
			\right\} \\
			& \ = 
			(r+1)\frac{q_M(r+1)}{q_M(r)}(\gamma_1r)_{\gamma_1}(\gamma_2r)_{\gamma_2}
			(\nstar_1)^{-1}(\nstar_2)^{-1}o\left((\nstar_1)^{-1}(\nstar_2)^{-1}\right) \\
			& \qquad \times
			\left\{
			1 + (C_{r+1}-C_{r})(\nstar_1)^{-1}(\nstar_2)^{-1} + o\left((\nstar_1)^{-1}(\nstar_2)^{-1}\right)
			\right\} \\
			& \ =
			(r+1)\frac{q_M(r+1)}{q_M(r)}(\gamma_1r)_{\gamma_1}(\gamma_2r)_{\gamma_2}
			(\nstar_1)^{-1}(\nstar_2)^{-1}
			\left(1 + o\left((\nstar_1)^{-1}(\nstar_2)^{-1}\right) \right) \, .
		\end{split}
	\end{equation}
	\subsection{Proof of Theorem \ref{thm:jointpost_d2}}
	\label{app:proof_thm_post}
	Firstly, we notice from Section \ref{app:post_quantities} that the number of linearly independent quantities we need to fully characterise all the posterior samples is five. It follows that $K^{(n_1,n_2)}_{m_1,m_2}, K^{(n_1)}_{1, m_1}, \mathcal K^{(n_2)}_{2,m_2}$ are not enough, and this would require introducing two more random variables that would be marginalised out. Of course, we must choose them so that the five selected quantities form a system of linearly independent variables. 
	To achieve this, we choose $S^*_m=s^*$ and $K^{*(n)}_{1,m}=k^*_1$. Now, from Equations \eqref{eqn:post_system_eq}, it follows that 
	\begin{equation}
		\label{eqn:proof_post_others}
		k^*_2 \ = \  k - k^*_1 - s^*, \ 
		s \ = \ k_1 + k_2 - k, \ 
		s_{2,1} \ = \ k_1 - k^*_1 - s^*, \ 
		s_{1,2} \ = \ k_2 + k^*_1 - k
	\end{equation}
	As for the support of the variables, it is natural to ask for $0 \leq k_j \leq m_j$ and $j=1,2$. Furthermore, from Equation \eqref{eqn:proof_post_others}, we see that $k\leq k_1 + k_2$ is a more stringent condition with respect to $k\leq m_1+m_2$. Hence, we also ask for $0 \leq k \leq k_1+k_2$.
	
	Let $(\pi_1,\pi_2)$ denote the partition of the additional observations 
	$\{ (X_{n_j+1}, \ldots , X_{n_j+m_j}): \; j=1, 2 \}$ into $r+k$ sets of distinct values, of which $r$ coincide with already observed values in the initial sample and the remaining $k$ are new. 
	We also indicate by $\bmm_j (\pi_j) = (m_{j,1} (\pi_j), \ldots , m_{j,K+r}(\pi_j)) $ the corresponding frequency counts, as $j = 1,2$. 
	Finally, let $\Pcr_{m_1,m_2,r+k}$ be the space of all such possible partitions, so that $(\pi_1,\pi_2) \in \Pcr_{m_1,m_2,r+k}$. Moreover, let $n=n_1+n_2$ and $m = m_1+m_2$.
	The posterior probability of interest can be evaluated as follows,
	\begin{equation} 
		\label{eqn:proof_jointpost_1}
		\begin{split}
			&\P\left(\mathcal K^{(n_1,n_2)}_{m_1,m_2} = k,
			\ K^{(n_1)}_{1, m_1} = k_1,
			\ K^{(n_2)}_{2,m_2} = k_2
			\mid \bmX \right)
			\\ 
			& \qquad = 
			\sum_{\Pcr_{m_1,m_2,r+k}}  
			\frac{ \Pi_{r+k}^{(n+m)} (\bmn_1+\bmm_1 (\pi_1),  \bmn_2+ \bmm_2 (\pi_2)) }
			{ \Pi_r^{(n)} (\bmn_1,  \bmn_2) },
		\end{split}
	\end{equation}
	where the sum is extended for all $(\pi_1,\pi_2) \in \Pcr_{m_1,m_2,r+k}$.
	
	We now elaborate the numerator in Equation \eqref{eqn:proof_jointpost_1}. 
	Writing explicitly all terms in Equations \eqref{eqn:peppf} and \eqref{eqn:Vprior} we have that
	\begin{equation}
		\label{eqn:proof_jointpost_2}
		\begin{split}
			& \Pi_{r+K}^{(n+m)} (\bmn_1+\bmm_1 (\pi_1),  \bmn_2+ \bmm_2 (\pi_2)) \\
			& \ =  
			\sum_{\mstar=r+k}^\infty 
			(\mstar)_{(r+k)\downarrow} 
			q_M (\mstar) 
			\prod_{j=1}^2
			\left\{
			\frac{1}{(\gamma_j (\mstar))_{n_j+m_j}} 
			\prod_{l=1}^{k+r} 
			(\gamma_j)_{n_{j,l}+m_{j,l} (\pi_j)}
			\right\} \\
			& \ = 
			\sum_{\mbar=k}^\infty 
			(\mbar)_{k\downarrow}\frac{(\mbar+r)!}{\mbar!}
			q_M (\mbar + r) 
			\prod_{j=1}^2
			\left\{
			\frac{1}{(\gamma_j (\mbar + r))_{n_j+m_j}} 
			\prod_{l=1}^{k+r} 
			(\gamma_j)_{n_{j,l}+m_{j,l} (\pi_j)}
			\right\} \\
			& \ = 
			\sum_{\mbar=k}^\infty 
			\Biggl\{
			(\mbar)_{k\downarrow}\frac{(\mbar+r)!}{\mbar!}
			q_M (\mbar + r) 
			\prod_{j=1}^2 \frac{1}{\left(\gamma_j(r+\mbar)\right)_{n_{j}}}
			\prod_{j=1}^2\prod_{l=1}^{r} (\gamma_j)_{n_{j,l}} \\
			& \qquad  \times
			\prod_{j=1}^2\prod_{l=1}^{r} \frac{(\gamma_j)_{n_{j,l}+m_{j,l} (\pi_j)}} {(\gamma_j)_{n_{j,l}}}
			\prod_{j=1}^2\prod_{l=r+1}^{r+k} (\gamma_j)_{m_{j,l} (\pi_j)}
			\prod_{j=1}^2\frac{\left(\gamma_j(r+\mbar)\right)_{n_j}}{\left(\gamma_j(r+\mbar)\right)_{n_j + m_j}}
			\Biggr\}
		\end{split}
	\end{equation}
	The second equality in Equation \eqref{eqn:proof_jointpost_2} follows after the change of variables $\mbar = \mstar - r$ and using the identity $1/(\mbar-k)! = (\mbar)_{k\downarrow}/\mbar!$. The third equality is obtained by rearranging terms after multiplying and dividing by $\left(\gamma_j(r+\mbar)\right)_{n_j + m_j}$ as well as by 
$(\gamma_j)_{n_{j,l}}$, for all $j=1,2$ and $l=1,\ldots,r$. Additionally, we also assumed, without loss of generality, that the first $r$ species are the ones that have already been observed. 
	
	Then, plugging Equations \eqref{eqn:peppf} and \eqref{eqn:proof_jointpost_2} into Equation \eqref{eqn:proof_jointpost_1} we get
	\begin{equation}
		\label{eqn:proof_jointpost_3}
		\begin{split}
			&\P\left(\mathcal K^{(n_1,n_2)}_{m_1,m_2} = k, K^{(n_1)}_{1, m_1} = k_1,
			K^{(n_2)}_{2,m_2} = k_2 \mid \bmX \right)
			\\ 
			& \ = 
			\sum_{\Pcr_{m_1,m_2,r+k}}  
			\sum_{\mbar=k}^\infty 
			\Biggl\{
			(\mbar)_{k\downarrow}
			\frac{1}{V_{n_1,n_2}^r}
			\frac{(\mbar+r)!}{\mbar!}
			q_M (\mbar + r) 
			\prod_{j=1}^2 \frac{1}{\left(\gamma_j(r+\mbar)\right)_{n_{j}}}
			\\
			& \qquad \qquad \times
			\prod_{j=1}^2\prod_{l=1}^{r} \frac{(\gamma_j)_{n_{j,l}+m_{j,l} (\pi_j)}} {(\gamma_j)_{n_{j,l}}}
			\prod_{j=1}^2\prod_{l=r+1}^{r+k} (\gamma_j)_{m_{j,l} (\pi_j)}
			\prod_{j=1}^2\frac{\left(\gamma_j(r+\mbar)\right)_{n_j}}{\left(\gamma_j(r+\mbar)\right)_{n_j + m_j}}
			\Biggr\} \\
			& \ = 
			\sum_{\Pcr_{m_1,m_2,r+k}} \  
			\sum_{\mbar=k}^\infty
			\left\{
			(\mbar)_{k\downarrow} q_M(\mbar\mid\bmX)          \prod_{j=1}^2\frac{\left(\gamma_j(r+\mbar)\right)_{n_j}}
			{\left(\gamma_j(r+\mbar)\right)_{n_j + m_j}} 
			\right\} \\
			& \qquad \qquad \times
			\prod_{j=1}^2\prod_{l=1}^{r} \frac{(\gamma_j)_{n_{j,l}+m_{j,l} (\pi_j)}} {(\gamma_j)_{n_{j,l}}}
			\prod_{j=1}^2\prod_{l=r+1}^{r+k} (\gamma_j)_{m_{j,l} (\pi_j)}.
		\end{split}
	\end{equation}
	In Equation \eqref{eqn:proof_jointpost_3}, we recognized the posterior distribution $q_M(\cdot\mid\bmX)$ defined in Equation \eqref{eqn:qMpost}. 
	Finally, we notice that the inner infinite sum does not depend on $(\pi_1,\pi_2)$. In particular, we have 
	\begin{equation*}
		\label{eqn:Vpost_proof_1}
		\begin{split}
			&\sum_{\mstar=k}^\infty
			\left\{
			(\mstar)_{k\downarrow} q_M(\mstar\mid\bmX)          
			\prod_{j=1}^2\frac{\left(\gamma_j(r+\mstar)\right)_{n_j}}
			{\left(\gamma_j(r+\mstar)\right)_{n_j + m_j}} 
			\right\} \\
			&\quad = 
			\frac{1}{V_{n_1,n_2}^r}
			\sum_{\mstar=k}^\infty
			\left\{
			\frac{\mstar!(\mstar + r)!}{(\mstar-k)!\mstar!} q_M(\mstar + r)          
			\prod_{j=1}^2
			\frac{\Gamma\left(\gamma_j(r+\mstar)\right)}{\Gamma\left(\gamma_j(r+\mstar) + n_j\right)}
			\frac{\Gamma\left(\gamma_j(r+\mstar) + n_j\right)}{\Gamma\left(\gamma_j(r+\mstar) + n_j + m_j\right)}
			\right\} \\
			&\quad = 
			\frac{1}{V_{n_1,n_2}^r}
			\sum_{\mstar=k}^\infty
			\left\{
			\frac{(\mstar + r)!}{(\mstar-k)!} q_M(\mstar + r)          
			\prod_{j=1}^2
			\frac{\Gamma\left(\gamma_j(r+\mstar) \right)}{\Gamma\left(\gamma_j(r+\mstar) + n_j + m_j\right)}
			\right\} \\
			&\quad = 
			\frac{1}{V_{n_1,n_2}^r}
			\sum_{m^{\star\star}=k+r}^\infty
			\left\{
			\frac{m^{\star\star}!}{(m^{\star\star}-r-k)!} q_M(m^{\star\star})          
			\prod_{j=1}^2
			\frac{1}{\left(\gamma_jm^{\star\star}\right)_{n_j+m_j}}
			\right\} 
			= \frac{V_{n_1+m_1,n_2+m_2}^{r+k}}{V_{n_1,n_2}^r}
		\end{split}
	\end{equation*}
	where the final equality follows from noticing that $m^{\star\star}!/(m^{\star\star}-k-r)! = (m^{\star\star})_{(k+r)\downarrow}$.
	
	We now focus on solving the sum over the set of partitions $\Pcr_{m_1,m_2,r+k}$. The partitions $(\pi_1,\pi_2)$ only appear through the cardinalities of the sets; hence, the quantities of interest can equivalently be computed as
	\begin{equation*}
		\label{eqn:proof_jointpost_4}
		\begin{split}
			&\P\left(\mathcal K^{(n_1,n_2)}_{m_1,m_2} = k, K^{(n_1)}_{1, m_1} = k_1,
			K^{(n_2)}_{2,m_2} = k_2 \mid \bmX \right) \ = \ 
			\frac{V_{n_1+m_1,n_2+m_2}^{r+k}}{V_{n_1,n_2}^r} \\
			& \qquad \times 
			\frac{1}{k!}
			\sum_{(\Delta)} 
			\binom{m_j}{m_{j,1}, \ldots , m_{j,r+k}} 
			\prod_{j=1}^2\prod_{l=1}^{r} \frac{(\gamma_j)_{n_{j,l}+m_{j,l}}} {(\gamma_j)_{n_{j,l}}}
			\prod_{j=1}^2\prod_{l=r+1}^{r+k} (\gamma_j)_{m_{j,l}},
		\end{split}
	\end{equation*}
	where the sum is extended over the set $(\Delta)$ of non-negative integers $\bmm_1 = \left( m_{1,1},\ldots, m_{1,r+k}\right)$ and $\bmm_2 = \left( m_{2,1},\ldots, m_{2,r+k}\right)$ that satisfy the following constraints,
	\begin{equation}
		\label{eqn:proof_post_insieme1}
		\begin{split}
			&\sum_{l=1}^{r+k} m_{j,l} = m_j \, , 
			\quad m_{j,l}\geq 0 
			\quad j=1,2; \, l = 1,\ldots,r+k \, , \\
			&m_{1,l} + m_{2,l} \geq 1 
			\quad l = r+1,\ldots,r+k \, , \\
			&\sum_{l=1}^{r} \delta_{\left\{m_{j,l} \geq 1, \, n_{j,l}=0\right\}} + \sum_{l=r+1}^{r+k} \delta_{\left\{m_{j,l} \geq 1 \right\}} = k_j \, ,
			\quad j=1,2 \, .
		\end{split}
	\end{equation}
	As mentioned at the beginning of the proof, the knowledge of $k$, $k_1$, and $k_2$ only is not enough to fully characterise $\bmm_1$ and $\bmm_2$ and therefore decouple the joint condition in Equation \eqref{eqn:proof_post_insieme1} as done in Section \ref{app:proof_joint_prior}.
	To do so, we introduce the auxiliary quantities $S^*_m=s^*$ and $K^{*(n)}_{1,m}=k^*_1$. See Section \ref{app:post_quantities} for their interpretation. Since we are not interested in inferring such quantities, we then marginalised them out. As a consequence of this augmentation, all other posterior quantities (namely, $s$, $k^*_2$, $s_{1,2}$, and $s_{2,1}$) can be recovered, and their expressions are reported in Equation \eqref{eqn:proof_post_others}.
	Additionally, $s^*$ and $k^*_1$ fix some ordering of the new species. 
	We say that the new shared species among the new $k$ species are located in the first $s^*$ positions. Their order is fixed in any of the $\binom{k}{s^*}$ equivalent ways. 
	We point out that as we set $s^*$, we also set $s-s^*$ as the number of new shared species among the already observed $r$ species.
	Then, consecutively to the $s^*$ new shared species, we set the following $k^*_1$ species to be those that are found in group 1 only. The order is fixed in any of the $\binom{k-s^*}{k^*_1}$ equivalent ways. In particular, among the $r$ new species, we are left with $k^*_2 = k-s^*-k^*_1$ species that are specific to area 2 only and which are placed, by construction, in the final positions. 
	It follows that the target probability equals
	\begin{equation*}
		\label{eqn:proof_jointpost_5}
		\begin{split}
			&\P\left(\mathcal K^{(n_1,n_2)}_{m_1,m_2} = k, K^{(n_1)}_{1, m_1} = k_1,
			K^{(n_2)}_{2,m_2} = k_2 \mid \bmX \right) = \\
			& \qquad = 
			\frac{V_{n_1+m_1,n_2+m_2}^{r+k}}{V_{n_1,n_2}^r} 
			\frac{1}{k!}
			\sum_{s^*=0}^{k}
			\sum_{k^*_1=0}^{k-s^*}
			\binom{k}{s^*}\binom{k-s^*}{k^*_1} \\
			& \qquad \qquad \times 
			\sum_{(\Delta\Delta)} 
			\binom{m_j}{m_{j,1}, \ldots , m_{j,r+k}} 
			\prod_{j=1}^2\prod_{l=1}^{r} \frac{(\gamma_j)_{n_{j,l}+m_{j,l}}} {(\gamma_j)_{n_{j,l}}}
			\prod_{j=1}^2\prod_{l=r+1}^{r+k} (\gamma_j)_{m_{j,l}}
			,
		\end{split}
	\end{equation*}
	where the summation over $(\Delta\Delta)$ satisfies the following set of constraints:
	\begin{equation}
		\label{eqn:proof_post_insieme2}
		\begin{split}
			&\sum_{l=1}^{r+k} m_{j,l} = m_j \, , \quad j=1,2 \, , \\
			& m_{j,l} \geq 0 \quad j=1,2; \, l = 1,\ldots, r \, , \\
			&\sum_{l=1}^{r} \delta_{
				\left\{ n_{1,l}\cdot n_{2,l} = 0, \,
				n_{1,l}+m_{1,l} \geq 1, \,
				n_{2,l}+m_{2,l} \geq 1, \, 
				\right\}} = s - s^* \, , \\
			& m_{j,l} \geq 1 \quad j=1,2; \, l = r+1,\ldots, r+s^* \, , \\
			& m_{1,l} \geq 1, \, m_{2,l} = 0 \quad j=1,2; \, l = r+s^*+1,\ldots, r+s^*+k^*_1 \, , \\
			& m_{1,l} = 0, \, m_{2,l} \geq 1 \quad j=1,2; \, l = r+s^*+k^*_1+1,\ldots, r+k \, , \\
			&\sum_{l=1}^{r} \delta_{\left\{m_{j,l} \geq 1, \, n_{j,l}=0\right\}} = s_{j^\prime,j} \quad j,j^\prime=1,2 \, , \ j\neq j^\prime \, .
		\end{split}
	\end{equation}
	The final conditions in Equation \eqref{eqn:proof_post_insieme2} rise noticing that the second sum in the corresponding condition in Equation \eqref{eqn:proof_post_insieme1} has been set equal to $s^*+k^*_j$, for $j=1,2$, and therefore the first sum must be equal to $k_j - s^* - k_j^*$, which coincides with $s_{j^\prime,j}$, see Equation \eqref{eqn:proof_post_others}. 
	
	It is still not possible to decouple $(\Delta\Delta)$ into two disjoint sets because of the joint condition regarding the number of new shared species among the already observed $r$ species. See line 3 in Equation \eqref{eqn:proof_post_insieme2}.
	However, $\bmm_1$ and $\bmm_2$ can be further reordered to ensure this.
	Indeed, from Section \ref{app:prior_quantities}, we know that the observed sample $\bmn_1$ and $\bmn_2$ can be arranged so that the first $t$ out of $r$ species are shared, then the following $r^*_1$ species are only present in the first group while the remaining $r^*_2$ are only present in the second group. 
	Let us fix $j$. In order to satisfy the joint condition in Equation \eqref{eqn:proof_post_insieme2}, it is enough to reorder $m_{j,l}$ and $l=1,\ldots,r$ so that the first $s_{j^\prime,j}$ species, which were first only observed in group $j^\prime$ and then observed in group $j$, are the first $s_{j^\prime,j}$ species among the $r^*_{j^\prime}$ such that $n_{j,l} = 0$. By construction, it follows that the remaining $r^*_{j^\prime} - s_{j^\prime,j}$ must be such that $m_{j,l} = 0$. This can be done in $\binom{r^*_2}{s_{j^\prime,j}}$ equivalent ways. 
	We have,
	\begin{equation*}
		\label{eqn:proof_jointpost_6}
		\begin{split}
			&\P\left(\mathcal K^{(n_1,n_2)}_{m_1,m_2} = k, K^{(n_1)}_{1, m_1} = k_1,
			K^{(n_2)}_{2,m_2} = k_2 \mid \bmX \right) = \\
			& \qquad = 
			\frac{V_{n_1+m_1,n_2+m_2}^{r+k}}{V_{n_1,n_2}^r} 
			\frac{1}{k!}
			\sum_{s^*=0}^{k}
			\sum_{k^*_1=0}^{k-s^*}
			\binom{k}{s^*}\binom{k-s^*}{k^*_1}
			\binom{r^*_1}{s_{1,2}}\binom{r^*_2}{s_{2,1}}\\
			& \qquad \qquad \times 
			\sum_{(\Delta\Delta\Delta)} 
			\binom{m_j}{m_{j,1}, \ldots , m_{j,r+k}} 
			\prod_{j=1}^2\prod_{l=1}^{r} \frac{(\gamma_j)_{n_{j,l}+m_{j,l}}} {(\gamma_j)_{n_{j,l}}}
			\prod_{j=1}^2\prod_{l=r+1}^{r+k} (\gamma_j)_{m_{j,l}}
			,
		\end{split}
	\end{equation*}
	where the set $(\Delta\Delta\Delta)$ describes the following set of constraints,
	\begin{equation}
		\label{eqn:proof_post_insieme3}
		\begin{split}
			&\sum_{l=1}^{r+k} m_{j,l} = m_j  \, , \\
			&m_{1,l} \geq 0 \, ,  \quad l = 1,\dots,t+r^*_1 \, , \\
			&m_{1,l} \geq 1 \, , \quad l = t+r^*_1+1,\ldots,t+r^*_1+s_{2,1} \, , \\
			&m_{1,l} = 0 \, , \quad l = t+r^*_1+s_{2,1}+1,\ldots,r \, , \\
			&m_{2,l} \geq 0 \, , \quad l = 1,\dots,t \quad l = r-r^*_2,\ldots,r \, , \\
			&m_{2,l} \geq 1 \, , \quad l = t+1,\ldots,t+s_{1,2} \, , \\
			&m_{2,l} = 0 \, , \quad l = t+s_{1,2}+1,\ldots,t+r^*_1 \, , \\
			& m_{j,l} \geq 1 \, , \quad j=1,2; \, , l = r+1,\ldots, r+s^* \, , \\
			& m_{1,l} \geq 1,  \, m_{2,l} = 0 \, , \quad j=1,2; \, l = r+s^*+1,\ldots, r+s^*+k^*_1 \, , \\
			& m_{1,l} = 0, \, m_{2,l} \geq 1 \, , \quad j=1,2; \, l = r+s^*+k^*_1+1,\ldots, r+k \, .
		\end{split}
	\end{equation}
	We are finally ready to decouple the set $(\Delta\Delta\Delta)$. We discard the elements such that $m_{j,l} = 0$ and note that the number of species such that $m_{j,l}\geq 1$ is $s_{j^\prime,1}+s^*+k^*j$, which equals $k_j$. Among the first $r$ species, the number of those such that $m_{j,l}\geq 0$ is $t+r^*_j$, which equals $r_j$. 
	We have that,
	\begin{equation}
		\label{eqn:proof_jointpost_7}
		\begin{split}
			&\P\left(\mathcal K^{(n_1,n_2)}_{m_1,m_2} = k, K^{(n_1)}_{1, m_1} = k_1,
			K^{(n_2)}_{2,m_2} = k_2 \mid \bmX \right) = \\
			& \qquad = 
			\frac{V_{n_1+m_1,n_2+m_2}^{r+k}}{V_{n_1,n_2}^r} 
			\frac{1}{k!}
			\sum_{s^*=0}^{k}
			\sum_{k^*_1=0}^{k-s^*}
			\binom{k}{s^*}\binom{k-s^*}{k^*_1}
			\binom{r^*_1}{s_{1,2}}\binom{r^*_2}{s_{2,1}}\\
			& \qquad \qquad \times 
			\prod_{j=1}^2 \sum_{(\Delta j)} 
			\binom{m_j}{m_{j,1}, \ldots , m_{j,r_j+k_j}} 
			\prod_{l=1}^{r_j} \left(\gamma_j + n_{j,l} \right)_{m_j}
			\prod_{l=r_j+1}^{r_j+k_j} \left(\gamma_j\right)_{m_{j,l}}
			,
		\end{split}
	\end{equation}
	where the sum is extended over the sets $(\Delta j)$, for $j=1,2$, of non-negative integers $\left( m_{1,1},\ldots, m_{1,r_j+k_j}\right)$ and $\left( m_{2,1},\ldots, m_{2,r_j+k_j}\right)$ which satisfy the following constraints, 
	\begin{equation*}
		\label{eqn:proof_post_insieme4}
		\begin{split}
			&\sum_{l=1}^{r_j+k_j} m_{j,l} = m_j  \, , \\
			&m_{j,l} \geq 0 \, , \quad l = 1,\dots,r_j \, , \\
			&m_{j,l} \geq 1 \, , \quad l = r_j+1,\dots,r_j+k_j \, ,
		\end{split}
	\end{equation*}
	for each $j=1,2$.
	Moreover, in Equation \eqref{eqn:proof_jointpost_7}, we leveraged the fact that $n_{j,l} = 0$ for $l=r_j+1,\ldots,r$ and the identity $(x)_{n+m}/(x)_{n} = (x+n)_m$, that holds for any non-negative integers $n,m$ and any positive real number $x>0$. We are finally left to compute the sums over the sets $(\Delta j)$. 
	To do so, let us fix $j$ and introduce the additional variable $h_j$, with $h_j \in \{k_j,k_j+1,\dots,m_j\}$, such that 
	$$
	\sum_{j = 1}^{r_j} m_{j,l} = m_j - h_j 
	\, , \quad
	\sum_{j = r_j+1}^{r_j+k_j} m_{j,l} = h_j \, .
	$$
It follows that
	\begin{equation}
		\label{eqn:proof_jointpost_8}
		\begin{split}
			&\sum_{(\Delta j)} 
			\binom{m_j}{m_{j,1}, \ldots , m_{j,r_j+k_j}} 
			\prod_{l=1}^{r_j} \left(\gamma_j + n_{j,l} \right)_{m_j}
			\prod_{l=r_j+1}^{r_j+k_j} \left(\gamma_j\right)_{m_{j,l}}
			\\
			& \qquad =
			\sum_{h_j = k_j}^{m_j}
			\sum_{(\Delta j,1)}
			\sum_{(\Delta j,2)}
			\binom{m_j}{h_j}
			\binom{m_j-h_j}{m_{j,1},\dots,m_{j,r_j}}
			\binom{h_j}{m_{j,r_j+1},\dots,m_{j,r_j+k_j}} \\
			& \qquad \qquad \times
			\prod_{l=1}^{r_j} \left(\gamma_j + n_{j,l} \right)_{m_j}
			\prod_{l=r_j+1}^{r_j+k_j} \left(\gamma_j\right)_{m_{j,l}},
		\end{split}
	\end{equation}
	where the sets $(\Delta j,1)$ and $(\Delta j,2)$ are defined as
	\begin{equation}
		\label{eqn:proof_post_insieme5}
		\begin{split}
			&(\Delta j,1) = \left\{
			(m_{j,1},\dots,m_{j,r_j}) \ : \ 
			m_{j,l} \geq 0, \text{ and } \sum_{l=1}^{r_j} m_{j,l} = m_j - h_j
			\right\} \\
			& (\Delta j,2) = \left\{
			(m_{j,r_j+1},\dots,m_{j,r_j+k_j}) \ : \ 
			m_{j,l} \geq 1, \text{ and } \sum_{l=r_j+1}^{r_j+k_j} m_{j,l} = h_j
			\right\}
		\end{split}
	\end{equation}
	
	Then, from Vandermonde's generalised identity, see Equation \eqref{eqn:vandermonde}, it follows that
	\begin{equation}
		\label{eqn:proof_jointpost_9}
		\begin{split}
			\sum_{(\Delta j,1)} \binom{m_j-h_j}{m_{j,1}, \ldots , m_{j,r_j}} 
			\prod_{l=1}^{r_j}
			\left(\gamma_j+n_{j,l}\right)_{m_{j,l} } 
			= \left(\sum_{l=1}^{r_j}\left(\gamma_j + n_{j,l}\right)\right)_{m_j-h_j}
			= \left(\gamma_j r_j + n_j\right)_{m_j-h_j}
		\end{split}
	\end{equation}
	while, thanks to Equation \eqref{eqn:gen_formula}, we have
	\begin{equation}
		\label{eqn:proof_jointpost_10}
		\begin{split}
			\sum_{(\Delta j,2)} 
			\binom{h_j}{m_{j,r_j+1}, \ldots , m_{j,r_j+k_j}} 
			\prod_{l=r_j+1}^{r_j+k_j}
			(\gamma_j)_{m_{j,l} } = 
			k_j! \,  \lvert C(h_j, k_j; -\gamma_j)\rvert .
		\end{split}
	\end{equation}
	Finally, plugging Equations \eqref{eqn:proof_jointpost_9} and \eqref{eqn:proof_jointpost_10} into \eqref{eqn:proof_jointpost_8}, we have 
	\begin{equation}
		\label{eqn:proof_jointpost_11}
		\begin{split}
			&\sum_{(\Delta j)} 
			\binom{m_j}{m_{j,1}, \ldots , m_{j,r_j+k_j}} 
			\prod_{l=1}^{r_j} \left(\gamma_j + n_{j,l} \right)_{m_j}
			\prod_{l=r_j+1}^{r_j+k_j} \left(\gamma_j\right)_{m_{j,l}}
			\\
			& \qquad =
			k_j!
			\sum_{h_j = k_j}^{m_j}
			\binom{m_j}{h_j} \,
			\lvert C(h_j, k_j; -\gamma_j)\rvert
			\left(\gamma_j r_j + n_j\right)_{m_j-h_j} \\
			& \qquad = 
			k_j! \, \lvert C\left(m_j, k_j; -\gamma_j, -\left(\gamma_j r_j + n_j\right)  \right)\rvert,
		\end{split}
	\end{equation}
	where the final equality follows from Equation \eqref{eqn:CnC}. The statement follows after plugging Equation \eqref{eqn:proof_jointpost_11} into \eqref{eqn:proof_jointpost_7} and rewriting all quantities in terms of $k$, $k_1$, and $k_2$.
	
	\subsection{Proof of Proposition \ref{thm:marginal_post_global}}
	\label{app:proof_K12_post}
	The best way to evaluate $\P\left(\mathcal K_{m_1,m_2}^{(n_1,n_2)} =k \mid \bmX\right)$ is not marginalizing $K_{1,m_1}^{(n_1)}$ and $K_{2,m_2}^{(n_2)}$ out of Equation \eqref{eqn:joint_post_d2}. 
	Section \ref{app:post_quantities} explains that the global number of new distinct species, $\mathcal K_{m_1,m_2}^{(n_1,n_2)}$, can be computed using those posterior quantities that do not require information about the frequencies of species in the future sample that regard the previously observed $r$ distinct species. Namely, $m_{j,l}$ for $l=1,\ldots,r$.
	In summary, the proof follows the steps of the one in Section \ref{app:proof_thm_post} but does not require reordering the first $r$ species. Indeed, following the same steps of Section \ref{app:proof_thm_post}, we can show that the quantity of interest can be computed as
	\begin{equation*} 
		\label{eqn:proof_K12post_1}
		\begin{split}
			&\P\left(\mathcal K^{(n_1,n_2)}_{m_1,m_2} = k \mid \bmX \right) =
			\sum_{\Pcr_{m_1,m_2,r+k}}  
			\frac{ \Pi_{r+k}^{(n+m)} (\bmn_1+\bmm_1 (\pi_1),  \bmn_2+ \bmm_2 (\pi_2)) }
			{ \Pi_r^{(n)} (\bmn_1,  \bmn_2) } \\
			& \qquad = 
			\frac{V_{n_1+m_1,n_2+m_2}^{r+k}}{V_{n_1,n_2}^r}
			\frac{1}{k!}
			\sum_{(\star)} 
			\binom{m_j}{m_{j,1}, \ldots , m_{j,r+k}} 
			\prod_{j=1}^2\prod_{l=1}^{r} \frac{(\gamma_j)_{n_{j,l}+m_{j,l}}} {(\gamma_j)_{n_{j,l}}}
			\prod_{j=1}^2\prod_{l=r+1}^{r+k} (\gamma_j)_{m_{j,l}},
		\end{split}
	\end{equation*}
	where the sum is extended over the set $(\star)$ of non-negative integers $\bmm_1 = \left( m_{1,1},\ldots, m_{1,r+k}\right)$ and $\bmm_2 = \left( m_{2,1},\ldots, m_{2,r+k}\right)$ which satisfy the following constraints,
	\begin{equation}
		\label{eqn:proof_Kpost_insieme1}
		\begin{split}
			&\sum_{l=1}^{r+k} m_{j,l} = m_j \, ,
			\quad m_{j,l}\geq 0 \quad
			j=1,2; \, l = 1,\ldots,r+k \, , \\
			&m_{1,l} + m_{2,l} \geq 1
			\, , \quad
			l = r+1,\ldots,r+k \, .
		\end{split}
	\end{equation}
	We now augment the space by introducing the auxiliary random variables $K^{*(n)}_{1,m} = k^*_1$ and $K^{*(n)}_{2,m} = k^*_2$. Due to Equation \eqref{eqn:proof_post_others}, this also implies that the number of new shared species among the new $k$ distinct species is $s^* = k - k^*_1 - k^*_2$. Hence, we order the new $k$ species so that the first $s^*$ are shared, the following $k^*_1$ are present in group one only and the remaining $k^*_2$ are found in group 2 only. This can be done in $\binom{k}{k^*_1}\binom{k-k^*_1}{k^*_2}$ equivalent ways. As for Section \ref{app:proof_thm_post}, we are not interested in $K^{*(n)}_{1,m}$ and $K^{*(n)}_{2,m}$, which can then be marginalized out.
	Hence, we have that
	\begin{equation*} 
		\label{eqn:proof_K12post_2}
		\begin{split}
			&\P\left(\mathcal K^{(n_1,n_2)}_{m_1,m_2} = k \mid \bmX \right)\\
			& \qquad = 
			\frac{V_{n_1+m_1,n_2+m_2}^{r+k}}{V_{n_1,n_2}^r} \ \frac{1}{k!}
			\sum_{(\star\star)}
			\sum_{k^*_1=0}^{k}\sum_{k^*_2=0}^{k-k^*_1}
			\binom{k}{k^*_1}\binom{k-k^*_1}{k^*_2}
			\binom{m_j}{m_{j,1}, \ldots , m_{j,r+k}} \\
			& \qquad \qquad \times
			\prod_{j=1}^2\prod_{l=1}^{r} \frac{(\gamma_j)_{n_{j,l}+m_{j,l}}} {(\gamma_j)_{n_{j,l}}}
			\prod_{j=1}^2\prod_{l=r+1}^{r+k} (\gamma_j)_{m_{j,l}},
		\end{split}
	\end{equation*}
	where the set $(\star\star)$ defines the following constraints,
	\begin{equation}
		\label{eqn:proof_Kpost_insieme2}
		\begin{split}
			&\sum_{l=1}^{r+k} m_{j,l} = m_j 
			\, , \quad
			m_{j,l}\geq 0 \quad j=1,2; \, l = 1,\ldots,r \, , \\
			&m_{j,l} \geq 1 \, , \quad j=1,2; \, l = r+1,\ldots,r+s^* \, , \\
			&m_{1,l} \geq 1 \, ,  \, m_{2,l} = 0 
			\, , \quad l = r+s^*+1,\ldots,r+s^*+k^*_1 \, , \\
			&m_{1,l} = 0, \, m_{2,l} \geq 1 
			\, , \quad l = r+s^*+k^*_1+1,\ldots,r+k \, .
		\end{split}
	\end{equation}
	The set $(\star\star)$ does not involve any coupled condition as we are not interested in shared quantities. Hence, it can be decoupled similarly to the set $(\Delta\Delta\Delta)$ in Section \ref{app:proof_thm_post}. It follows that,
	\begin{equation*} 
		\label{eqn:proof_K12post_3}
		\begin{split}
			&\P\left(\mathcal K^{(n_1,n_2)}_{m_1,m_2} = k \mid \bmX \right) \ = \
			\frac{V_{n_1+m_1,n_2+m_2}^{r+k}}{V_{n_1,n_2}^r} \frac{1}{k!}
			\sum_{k^*_1=0}^{k}\sum_{k^*_2=0}^{k-k^*_1}
			\binom{k}{k^*_1}\binom{k-k^*_1}{k^*_2} \\
			& \qquad \times
			\prod_{j=1}^2
			\sum_{(\star j)}
			\binom{m_j}{m_{j,1}, \ldots , m_{j,r+s^*+k^*_j}} 
			\prod_{l=1}^{r} \left(\gamma_j + n_{j,l}\right)_{m_{j,l}}
			\prod_{l=r+1}^{r+s^*+k^*_j} \left(\gamma_j \right)_{m_{j,l}}
		\end{split}
	\end{equation*}
	where the sum is extended over the sets $(\star j)$, for $j=1,2$, of non-negative integers $\left( m_{1,1},\ldots, m_{1,r+s^*+k^*_1}\right)$ and $\left( m_{2,1},\ldots, m_{2,r + s^*+k^*_2}\right)$ which satisfy the following constraints, 
	\begin{equation*}
		\label{eqn:proof_Kpost_insieme5}
		\begin{split}
			&\sum_{l=1}^{r+s^*+k^*_j} m_{j,l} = m_j  
			\, , \quad
			m_{j,l} \geq 0 \quad l = 1,\dots,r \, , \\
			&m_{j,l} \geq 1 \, ,  
			\quad l = r+1,\dots,r+s^*+k^*_j \, ,
		\end{split}
	\end{equation*}
	for each $j=1,2$.
	Let us now fix $j$. The sum over $(\star j)$ is solved as in Section \ref{app:proof_thm_post}, that is introducing the additional variable $h_j$, with $h_j \in \{s^*+k^*_j,s^*+k^*_j+1,\dots,m_j\}$, such that 
	$$
	\sum_{j = 1}^{r} m_{j,l} = m_j - h_j \, \text{ and }
	\sum_{j = r+1}^{r+s^*+k^*_j} m_{j,l} = h_j \, .
	$$
	Following the same steps of Section \ref{app:proof_thm_post}, it is possible to show that
	\begin{equation*} 
		\label{eqn:proof_K12post_4}
		\begin{split}
			&\P\left(\mathcal K^{(n_1,n_2)}_{m_1,m_2} = k \mid \bmX \right) \ = \ \frac{V_{n_1+m_1,n_2+m_2}^{r+k}}{V_{n_1,n_2}^r}  \\
			& \ \times
			\frac{1}{k!}
			\sum_{k^*_1=0}^{k}\sum_{k^*_2=0}^{k-k^*_1}
			\binom{k}{k^*_1}\binom{k-k^*_1}{k^*_2} 
			\prod_{j=1}^2
			(s^*+k^*_j)!
			\lvert C\left(m_j,s^*+k^*_j; -\gamma_j, - (\gamma_j r + n_j)\right) \rvert.
		\end{split}
	\end{equation*}
	Finally, we perform a change of variables to lighten the notation and so that $s^*$ is not involved in the final expression, even though it can be written in terms of $k$, $k^*_1$ and $k^*_2$. Let us fix $j$ and let $z_j = k - s^* - k^*_j$. This can be interpreted as the number of new species that are not present in group $j$ but that are present in group $j^\prime$. Hence, $z_j = k^*_{j^\prime}$. If follows that $s^*+k^*_j = k - z_j$. Furthermore, the following identity holds
	$$
	\binom{k}{z_2}\binom{k-z_2}{z_1} \ = \ 
	\binom{k}{z_1}\binom{k-z_1}{z_2},
	$$
	for $z_1 \in \{0,\ldots,k\}$ and $z_2 \in \{0,\ldots,k-z_1\}$.
	It is enough to rearrange the binomial and factorial coefficients to conclude the proof.
	
	\subsection{Proof of Equation \eqref{eqn:Kj_post}}
	\label{app:proof_Kj_post}
	Let us fix $j,j^\prime=1,2$ with $j \neq j^\prime$.
	The statement follows from Equation \eqref{eqn:Kpost_d2}, after setting both $n_{j^\prime}$ and $m_{j^\prime}$ to zero. Indeed, $V(m_j,m_{j^\prime};k) = V(m_j;k_j)$ 
	because, from Equation \eqref{eqn:post_system_eq}, $k = k_j$ as $k_{j^\prime}$ and $s$ must be zero since $m_{j^\prime} = 0$ and $(x)_{0} = 1$ for all $x>0$. Similarly, $n_{j^\prime} = 0$ also implies that $r = r_j$.
	Then, we recall that
	$\lvert C\left(0,0;x,y\right) \rvert = 1$ and 
	$\lvert C\left(0,k;x,y\right) \rvert = 0$ for all $x\neq 0$, $y\in\R$ and $k = 1,2,3,\ldots\,$, see \cite{chara2002}.
	Hence, the only non-zero term in Equation \eqref{eqn:Kpost_d2} is when $z_1 = 0$ and $z_2 = k$. Finally, recalling that $r = r_j$ and $k=k_j$, it follows that
	\begin{equation*}
		\label{eqn:proof_Kjpost_1}
		\begin{split}
			\P\left(K^{(n_j)}_{j,m_j} = k_j \mid \bmX\right)   
			\ = \ 
			\frac{V_{n_j+m_j}^{r_j+k_j}}{V_{n_j}^{r_j}}
			\lvert C\left(m_j,k_j; -\gamma_j, - (\gamma_j r_j + n_j)\right) \rvert.
		\end{split}
	\end{equation*}
	
	\section{Additional details on species discovery}
	\label{app:species_discovery}
    Section \ref{subsection:PrShared_post} reports the $m$-steps ahead coverage estimator as well as the one-step ahead discovery probability of new shared species. Here, we report the analogous quantities in the case of new local and global distinct species.
    Starting from the one-step ahead distribution, we have that $\P\left(\mathcal K^{(n_1,n_2)}_{1,1} = k\mid\bmX\right)$ for $k=\{0,1,2\}$ follows from Equation \eqref{eqn:Kpost_d2} and is given by
	\begin{equation}
		\label{eqn:PrK_1step_full}
		\begin{split}
			\P\left(\mathcal K^{(n_1,n_2)}_{1,1} = 0\mid\bmX\right) &\,= \,    
			\frac{V_{n_1+1,n_2+1}^{r}}{V_{n_1,n_2}^r}
            \left\{
			(\gamma_1r+n_1)(\gamma_2r+n_2)
            \right\}\,,\\
			\P\left(\mathcal K^{(n_1,n_2)}_{1,1} = 1\mid\bmX\right) &\,= \,    
			\frac{V_{n_1+1,n_2+1}^{r+1}}{V_{n_1,n_2}^r}
            \bigl\{
            \gamma_1(\gamma_2r_2+n_2) \,+\,
            \gamma_2(\gamma_1r_1+n_1)  \\
            & \qquad \qquad \qquad  \quad +\, \gamma_1\gamma_2 \left(r^*_1 + r^*_2 + 1\right)
            \bigr\}\,,\\
			\P\left(\mathcal K^{(n_1,n_2)}_{1,1} = 2\mid\bmX\right) &\,= \,    
			\frac{V_{n_1+1,n_2+1}^{r+2}}{V_{n_1,n_2}^r}
			\gamma_1\gamma_2 \, .
		\end{split}
	\end{equation}
    
	The full distribution of the number of new local distinct species in area $j$,  
    $\P\left(K^{(n_j)}_{j,1} = k_j\mid\bmX\right)$ for $k_j=\{0,1\}$ for $j=1,2$, follows from the marginal distribution in Equation \eqref{eqn:Kj_post} and it equals
	\begin{equation}
		\label{eqn:PrKj_1step_full}
		\begin{split}
        \P\left(K^{(n_j)}_{j, 1} = 0\mid\bmX\right) &\,= \, \frac{V_{n_j+1}^{r_j}}{V_{n_j}^{r_j}}\,(\gamma_jr_j + n_j)\,,\\
        \P\left(K^{(n_j)}_{j, 1} = 1\mid\bmX\right) &\,= \, \frac{V_{n_j+1}^{r_j+1}}{V_{n_j}^{r_j}}\,\gamma_j\,.
		\end{split}
	\end{equation}
	Finally, the full distribution $\P\left(\mathcal S^{(n_1,n_2)}_{1,1} = s\mid\bmX\right)$ for $s=\{0,1,2\}$ is 
	\begin{equation*}
		\label{eqn:PrSh_1step_full}
		\begin{split}
			\P\left(\mathcal S^{(n_1,n_2)}_{1,1} = 0\mid\bmX\right) &\,= \,    
			\frac{V_{n_1+1,n_2+1}^{r}}{V_{n_1,n_2}^r}
			(\gamma_1r_1+n_1)(\gamma_2r_2+n_2)
			\\
			& \quad +
			\frac{V_{n_1+1,n_2+1}^{r+1}}{V_{n_1,n_2}^r}
			\left\{
			\gamma_1(\gamma_2r_2+n_2) + \gamma_2(\gamma_1r_1+n_1) 
			\right\}
			+
			\frac{V_{n_1+1,n_2+1}^{r+2}}{V_{n_1,n_2}^r}
			\gamma_1\gamma_2 \, , \\
			\P\left(\mathcal S^{(n_1,n_2)}_{1,1} = 1\mid\bmX\right) &\,= \,    
			\frac{V_{n_1+1,n_2+1}^{r}}{V_{n_1,n_2}^r}
			\left\{
			r^*_2\gamma_1(\gamma_2r_2+n_2) + 
			r^*_1\gamma_2(\gamma_1r_1+n_1) 
			\right\}
			+
			\frac{V_{n_1+1,n_2+1}^{r+1}}{V_{n_1,n_2}^r}
			\gamma_1\gamma_2(r^*_1 + r^*_2 + 1) \, , \\ 
			\P\left(\mathcal S^{(n_1,n_2)}_{1,1} = 2\mid\bmX\right) &\,= \,    
			\frac{V_{n_1+1,n_2+1}^{r}}{V_{n_1,n_2}^r}
			\gamma_1\gamma_2r^*_1r^*_2 \, .
		\end{split}
	\end{equation*}

    Turning our attention to the $m$-steps ahead coverage, the following hold
    \begin{equation*}
    \label{eqn:K_Kj_msteps_coverage}
        \begin{split}
		\P\left(\mathcal K^{(n_1,n_2)}_{m_1,m_2} = 0\mid\bmX\right) &= 
		\frac{V_{n_1+m_1,n_2+m_2}^{r}}{V_{n_1,n_2}^{r}}
		\prod_{j=1}^d 
		\lvert C(m_j,0; -\gamma_j,-(\gamma_j r  + n_j) )\rvert\,, \\
        \P\left(K^{(n_j)}_{j,m_j} = 0\mid\bmX\right) &= 
		\frac{V_{n_j+m_j}^{r_j}}{V_{n_j}^{r_j}}
		\lvert C(m_j,0; -\gamma_j,-(\gamma_j r_j  + n_j) )\rvert
        \,,
        \end{split}
    \end{equation*}
    while has been stated in Equation \eqref{eqn:Shzero_m_steps}.
	\section{Proofs of Equations \eqref{eqn:ssj}-\eqref{eqn:sp1p2}}
	\label{app:proof_stima_param}
	Firstly, we note that:
	\begin{equation*}
		\begin{split}
			\E\left(\sum_{m=1}^M w_{j,m}^2\right) =
			\E_{q_M}\left(\E\left(\sum_{m=1}^M w_{j,m}^2\mid M\right)\right) = 
			\E_{q_M}\left( \sum_{m=1}^M
			\E \left( w_{j,m}^2\mid M \right)
			\right) .
		\end{split}
	\end{equation*}
	Given $M$, $\left(w_{j,1},\ldots,w_{j,M}\right) \sim \operatorname{Dir}_M(\gamma_j,\ldots,\gamma_j)$, hence $\E(w_{j,m}\mid M) = 1/M$ and $\operatorname{var}(w_{j,m}\mid M) = (M-1)/(M^2(\gamma_jM+1))$. It follows that
	\begin{equation*}
		\begin{split}
			\E\left(\sum_{m=1}^M w_{j,m}^2\right) =
			\E_{q_M}\left( 
			\sum_{m=1}^M
			\frac{1+\gamma_j}{M(1+\gamma_jM)}
			\right)  = 
			(1+\gamma_j)
			\E_{q_M}\left( 
			\frac{1}{(1+\gamma_jM)}
			\right).
		\end{split}
	\end{equation*}
	
	\begin{equation*}
		\begin{split}
			&\E\left(\sum_{m=1}^M w_{1,m}w_{2,m}\right) =
			\E_{q_M}\left( \sum_{m=1}^M
			\E \left( w_{1,m}w_{2,m}\mid M \right)
			\right) = \\
			& \qquad
			\E_{q_M}\left( \sum_{m=1}^M
			\E \left( w_{1,m}\mid M \right)
			\E \left( w_{2,m}\mid M \right)
			\right) = 
			\E_{q_M}\left(1/M\right)
			.
		\end{split}
	\end{equation*}
	The equality holds since $\left(w_{1,1},\ldots,w_{1,M}\right)$ and $\left(w_{2,1},\ldots,w_{2,M}\right)$ are independent given $M$.
    \section{Bayesian estimators for diversity indices}
	\label{section:posterior_diversity}
    The classical unbiased estimator of Simpson's diversity index, proposed by \citet{Simpson49}, is
	\begin{equation}
		\label{eqn:ssj_freq}
		\hat{\rho}_{j,\text{unb}} = 
		\sum_{l=1}^r\frac{n_{j,l}}{n_j}\frac{\left(n_{j,l}-1\right)}{\left(n_j-1\right)} \, .
	\end{equation}
	Since $\rho_j$ is interpreted as the probability that two randomly and independently chosen individuals belong to the same species, the estimator in Equation \eqref{eqn:ssj_freq} is obtained by dividing the total number of within-species pairs, $\sum_{l=1}^r n_{j,l}(n_{j,l}-1)/2$, by the total number of possible pairs, $n_j(n_j-1)/2$. 
    
	In this section, we present the Bayesian counterpart of $\rho_j$, $j=1,2$, as well as $\rho_{12}$, given an observed sample $\bmX=(\bmX_1,\bmX_2)$ of sizes $n_1$ and $n_2$, with $r$ distinct species and $t$ shared species. Prior to this, we present some technical lemmas that are needed in the derivation of the quantities of interest.

    \subsection{Technical preliminaries}
	\label{app:scambio_limiti}
	In this section, we show a useful inequality that allows us to apply the dominated convergence theorem to exchange the order of the limit and the sum in the following proofs.
	\begin{lemma}
		\label{lem:scambio}
		Let $j,j^\prime\in\{1,2\}$, $j\neq j^\prime$ and let $n_1\geq 1$, $n_2\geq 1$, $0\geq n_{j,l}<\leq n_j$ for $l=1,\ldots,r$ and $r\geq1$. 
		Let $q^\star_{M\mid\bmX}$ be the probability mass function defined in Equation \eqref{eqn:qMpost}. The following inequality holds for every $\gamma_1>0$, $\gamma_2 > 0$, $\mstar\geq 0$: 
		\begin{equation}
			\label{eqn:ssj_scambio_lim}
			\begin{split}
				& \frac{
					\sum_{l=1}^r n_{j,l}(n_{j,l} + 1) + \gamma_j\left( \gamma_j(r+\mstar) + r+\mstar + 2n_j  \right)
				}{
					n_j(n_j+1) + \gamma_j^2(r+\mstar)^2 + \gamma_j(r+\mstar)(2n_j+1)		
				}
				q^\star_{M\mid\bmX}(\mstar) \\
				& \qquad \lesssim
				\frac{(\mstar+r)!}{\mstar}q_M(\mstar+r) \, .
			\end{split}
		\end{equation}
		\begin{equation}
			\label{eqn:sp1p2_scambio_lim}
			\begin{split}
				\frac{
					\sum_{l=1}^t n_{1,l}n_{2,l} + \gamma_1\gamma_2(r+M^\star) + n_1\gamma_2 + n_2\gamma_1
				}{
					(\gamma_1(r+M^\star) + n_1)(\gamma_2(r+M^\star) + n_2)
				}
				q^\star_{M\mid\bmX}(\mstar)  \lesssim
				\frac{(\mstar+r)!}{\mstar}q_M(\mstar+r) \, ,
			\end{split}
		\end{equation}
		where we use notation $\lesssim$ to indicate that the upper bound holds up to some constant.
		
	\end{lemma}
	\begin{proof}
		We prove Equation \eqref{eqn:ssj_scambio_lim} only. The proof of Equation \eqref{eqn:sp1p2_scambio_lim} follows the same steps.
		\begin{equation*}
			\label{eqn:ssj_scambio_limiti_proof1}
			\begin{split}
				& \frac{
					\sum_{l=1}^r n_{j,l}(n_{j,l} + 1) + \gamma_j\left( \gamma_j(r+\mstar) + r+\mstar + 2n_j  \right)
				}{
					n_j(n_j+1) + \gamma_j^2(r+\mstar)^2 + \gamma_j(r+\mstar)(2n_j+1)		
				}
				q^\star_{M\mid\bmX}(\mstar)\\
				& \quad \leq 
				\Bigl\{
				\frac{ r(n_j+1)^2 }
				{ 	n_j(n_j+1) + \gamma_j\left( \gamma_j(r+\mstar)^2 + r+\mstar + 2n_j(r+\mstar)  \right)  } \\
				& \quad \ +
				\frac{ \gamma_j\left( \gamma_j(r+\mstar) + r+\mstar + 2n_j  \right) }
				{ 	n_j(n_j+1) + \gamma_j\left( \gamma_j(r+\mstar)^2 + r+\mstar + 2n_j(r+\mstar)  \right)  }
				\Bigr\} 
				\frac{(\mstar+r)_{r\downarrow}q_M(\mstar+r)}{  V_{n_1,n_2}(\gamma_j(\mstar+r))_{n_j} (\gamma_{j^\prime}(\mstar+r))_{n_{j^\prime} } } \\
				& \quad \leq
				\left\{
				\frac{ r(n_j+1)^2 }
				{ 	n_j(n_j+1)  }
				+
				\frac{ \gamma_j\left( \gamma_j(r+\mstar) + r+\mstar + 2n_j  \right)}{\gamma_j\left( \gamma_j(r+\mstar)^2 + r+\mstar + 2n_j(r+\mstar)  \right) } 
				\right\}  \\
				& \qquad \ \times
				\frac{(\mstar+r)_{r\downarrow}q_M(\mstar+r)}{  V_{n_1,n_2}(\gamma_j(\mstar+r))_{n_j} (\gamma_{j^\prime}(\mstar+r))_{n_{j^\prime} } } \,.
			\end{split}
		\end{equation*}
		In the final line, we only used the positivity of each term in the denominator. 
		Then, we exploit that $V_{n_1,n_2}$ is a sum of positive terms, hence it is larger than its first positive term. Moreover, note that the second term in the parenthesis is smaller than one since $r+\mstar\geq1$.
		\begin{equation*}
			\label{eqn:ssj_scambio_limiti_proof2}
			\begin{split}
				& \frac{
					\sum_{l=1}^r n_{j,l}(n_{j,l} + 1) + \gamma_j\left( \gamma_j(r+\mstar) + r+\mstar + 2n_j  \right)
				}{
					n_j(n_j+1) + \gamma_j^2(r+\mstar)^2 + \gamma_j(r+\mstar)(2n_j+1)		
				}
				q^\star_{M\mid\bmX}(\mstar)\\
				& \quad \leq 
				\left\{
				\frac{ r(n_j+1)^2 }{ 	n_j(n_j+1)  } +1  
				\right\}  
				\frac{(\mstar+r)_{r\downarrow}q_M(\mstar+r)}{r!q_M(r)} 
				\frac{(\gamma_jr)_{n_j}(\gamma_{j^\prime}r)_{n_{j^\prime}}}{ (\gamma_j(\mstar+r))_{n_j} (\gamma_{j^\prime}(\mstar+r))_{n_{j^\prime} } } \, .
			\end{split}
		\end{equation*}
		Then, due to the monotonicity of the Pochhammer symbol, the final term is smaller than or equal to one, which concludes the proof. As a corollary, note that we also proved that 
		\begin{equation}
			\label{eqn:qMpost_scambio_lim}
			q^\star_{M\mid\bmX}(\mstar) \lesssim \frac{(\mstar+r)!}{\mstar}q_M(\mstar+r) \, .
		\end{equation}
	\end{proof}
	
	Moreover, under the hypothesis that $q_M$ is a probability mass function on the positive integers for which there exists some constant $a\in(0,1)$ such that $q_M(m)\leq a^m$ for large values of $m$, we have that
	\begin{equation*}
		\label{eqn:bound_is_int}
		\sum_{\mstar = 0}^\infty \frac{(\mstar+r)!}{\mstar}q_M(\mstar+r) < \infty \, .
	\end{equation*} 
	The result holds since 
	$$
	\sum_{\mstar=0}^{\infty} \frac{(\mstar+r)!}{\mstar!}a^{\mstar} < \infty \, ,
	$$
	for every $a\in(0,1)$, see \cite{argiento2022annals}. 
	
	We conclude this section about some preliminary results, providing additional details about the limiting distribution of $q^\star_{M\mid\bmX}$ when $\lim_{\gamma_j\rightarrow\gamma_0}$, where $\gamma_0 \in [0,\infty]$. In particular, $q^\star_{M\mid\bmX}$ is a well defined probability mass function on $\N$ for each $\gamma_1>0$, $\gamma_2>0$. We define the limiting distribution 
	$q^\star_{M\mid\bmX,\gamma_0}$ as the pointwise limit of each atom. Namely,
	$$
	q^\star_{M\mid\bmX,\gamma_0}(\mstar) = \lim_{\gamma_j\rightarrow \gamma_0} q^\star_{M\mid\bmX}(\mstar) \, ,
	$$
	for each $\mstar \geq 0$. Clearly, $q^\star_{M\mid\bmX,\gamma_0}(\mstar) \geq 0$ for each $\mstar\geq0$. Moreover, thanks to Equation \eqref{eqn:qMpost_scambio_lim}, we apply the dominated convergence theorem to show that
	$$
	\sum_{\mstar = 0}^\infty q^\star_{M\mid\bmX,\gamma_0}(\mstar) = 
	\sum_{\mstar = 0}^\infty \lim_{\gamma_j\rightarrow \gamma_0} q^\star_{M\mid\bmX}(\mstar) =
	\lim_{\gamma_j\rightarrow \gamma_0} \sum_{\mstar = 0}^\infty  q^\star_{M\mid\bmX}(\mstar) = 1 \, .
	$$
	In particular, $q^\star_{M\mid\bmX,\gamma_0}(\mstar) \leq 1$ for each $\mstar\geq0$ since the sum is one and all terms are non-negative. As a consequence, we conclude that $q^\star_{M\mid\bmX,\gamma_0}$ is a well defined probability mass function on $\N$.
	
	\begin{lemma}
		\label{lem:qminf}
		When $\gamma_0$ is $+\infty$, it holds that 
		\begin{equation}
			\label{eqn:qMpost_liminf_def}
			q^\star_{M\mid\bmX,\infty} (m^\star)\propto \frac{(\mstar +r)_{r\downarrow}}{(\mstar +r)^{n_1}(\mstar +r)^{n_2}}q_M(\mstar +r) \, .
		\end{equation}
	\end{lemma}
	\begin{proof}
		We must evaluate the following limit
		\begin{equation*}
			\label{eqn:qMpost_liminf_def_proof1}
			\begin{split}
				& \lim_{\gamma_1,\gamma_2\rightarrow+\infty} q^\star_{M\mid\bmX}(\mstar) \\
				& \quad = 
				\lim_{\gamma_1,\gamma_2\rightarrow+\infty} 
				\frac{1}{V_{n_1,n_2}^r}
				(\mstar+r)_{r\downarrow}
				q_M(\mstar+r)
				\prod_{j=1}^d \frac{1}{(\gamma_j (\mstar+r))_{n_j}} \\
				& \quad = 
				(\mstar+r)_{r\downarrow}
				q_M(\mstar+r)
				\lim_{\gamma_1,\gamma_2\rightarrow+\infty}
				\left\{
				\sum_{m=0}^\infty
				(m)_{r\downarrow}q_M(m)
				\prod_{j=1}^2
				\frac{(\gamma_j (\mstar+r))_{n_j}}{(\gamma_jm)_{n_j}}
				\right\}^{-1} \,.
			\end{split}
		\end{equation*}
		Exploiting the results in Lemma \ref{lem:scambio}, we exchange the order of the limit and the sum. Hence, we have that
		\begin{equation*}
			\label{eqn:qMpost_liminf_def_proof2}
			\begin{split}
				& \lim_{\gamma_1,\gamma_2\rightarrow+\infty} q^\star_{M\mid\bmX}(\mstar) \, = \,
				(\mstar+r)_{r\downarrow} q_M(\mstar+r) \\
				& \qquad \times  
				\left\{
				\sum_{m=0}^\infty
				(m)_{r\downarrow}q_M(m)
				\prod_{j=1}^2
				\lim_{\gamma_j\rightarrow+\infty}
				\frac{\Gamma(\gamma_j (\mstar+r) + n_j)}{\Gamma(\gamma_j m + n_j)}
				\frac{\Gamma(\gamma_jm)}{\Gamma(\gamma_j (\mstar +r))}
				\right\}^{-1} \\
				& \quad = \, 
				(\mstar+r)_{r\downarrow} q_M(\mstar+r)   
				\left\{
				\sum_{m=0}^\infty
				(m)_{r\downarrow}q_M(m)
				\prod_{j=1}^2
				\left( \frac{\mstar+r}{m} \right)^{n_j}
				\right\}^{-1} \, ,
			\end{split}
		\end{equation*}
		where the final equality holds because of the asymptotic of the ratio of gamma functions, $\Gamma(x+a)/\Gamma(x+b)\sim (x)^{a-b}$, for $x\rightarrow+\infty$.
		Moreover, we note that the normalising constant of $q^\star_{M\mid\bmX,\infty}$ is the pEPPF computed in the corresponding limiting case.
	\end{proof}

    \subsection{Posterior quantities and limiting behaviours}
    The posterior expected value of the Simpson diversity index in area $j$ equals 
	\begin{equation}
		\label{eqn:ssj_post}
		\begin{split}
			&\E\left(\rho_j\mid\bmX\right) \\
            & \qquad = \,
			\E_{q^\star_{M\mid\bmX}}
			\left(
			\frac{
				\sum_{l=1}^r n_{j,l}(n_{j,l} + 1) + \gamma_j\left( \gamma_j(r+M^\star) + r+M^\star + 2n_j  \right)
			}{
				n_j(n_j+1) + \gamma_j^2(r+M^\star)^2 + \gamma_j(r+M^\star)(2n_j+1)		
			}
			\right) \, .		
		\end{split}
	\end{equation}
    \begin{proof}
    The proof follows the same steps given in Section \ref{app:proof_stima_param} but exploits the posterior representation $(P_1,P_2)\mid\bmX$. The latter is provided in \cite{colombi2023mixture} and reported here for the sake of completeness.
	
	$(P_1,P_2)\mid\bmX \stackrel{d}{=} (P_1^\star,P^\star_2)$, where each component is defined as
	\begin{equation*}
		\label{eqn:posterior_P1P2}
		P^\star_j = \sum_{l=1}^r w_{j,l}\delta_{\tau^\star_l} \, + \, \sum_{l=r+1}^{r+M^\star} w_{j,l}\delta_{\tau_l} \, ,
	\end{equation*}
	where $n_{j,l}\geq0$ are the observed counts and $\tau^\star_l$ are the corresponding species labels. Then, the random number of the unseen species $M^\star\sim q^\star_{M\mid\bmX}$ where $q^\star_{M\mid\bmX}$ has been defined in Section \ref{subsection:qM_post}. The labels of the unseen species are $\tau_{l}\mid M^\star \iid P_0(\D \tau)$ for $l=r+1,\ldots,r+M^\star$. Finally, the vector of posterior probabilities $\bmw^\star_{j} = (w^\star_{j,1},\ldots,w^\star_{j,r+M^\star})$ follows a Dirichlet distribution,
	\begin{equation*}
		\label{eqn:wj_post}
		(w^\star_{j,1},\ldots,w^\star_{j,r+M^\star}) \mid M^\star \, \sim \, \operatorname{Dir}_{r+M^\star}\left(\gamma_j+n_{j,1},\ldots,\gamma_j+n_{j,r},\gamma_j,\ldots\,\gamma_j\right) \, .
	\end{equation*}
	Furthermore, $P_1\indep P_2 \mid M^\star,\bmX$.
	
	The marginal distributions of $w^\star_{j,l}\mid M^\star$ are crucial in the derivation of the result. Using the aggregation property of the Dirichlet distribution, we have that
	\begin{equation*}
		\label{eqn:wj_post_marginal}
		w^\star_{j,l} \mid M^\star \, \sim \, \operatorname{Beta}\left(a_{j,l}, n_j + \gamma_j(r+M^\star) - a_{j,l} \right)\, ,
	\end{equation*}
	where $a_{j,l} = n_{j,l} + \gamma_j$ for $l=1,\ldots,r$ and $a_{j,l} = \gamma_j$ for $l=r+1,\ldots,r+M^\star$.

	The posterior expected value of the Simpson index is
	\begin{equation*}
		\label{eqn:ssj_post_proof1}
		\begin{split}
			&\E\left(\rho_j\mid \bmX\right) \, = \,
			\E\left(\sum_{l=1}^{r+M^\star} \left(w^\star_{j,l}\right)^2 \mid \bmX\right) \, = \,
			\E\left( \E\left( \sum_{l=1}^{r+M^\star} \left(w^\star_{j,l}\right)^2 \mid M^\star,\bmX \right) \mid  \bmX \right) \\ 
			& \quad = \,
			\E_{q^\star_{M\mid\bmX}}\left(
			\sum_{l=1}^{r}
			\E\left(
			(w^\star_{j,l})^2
			\mid M^\star,\bmX\right) + 
			M^\star
			\E\left(
			(w^\star_{j,r+1})^2			
			\mid M^\star,\bmX\right)
			\right) \, .
		\end{split}
	\end{equation*}
	To conclude the proof, it is enough to recall that the second moment of the marginal distributions $w^\star_{j,l}\mid M^\star$ equals 
	$$
	\E\left((w^\star_{j,r+1})^2 \mid M^\star,\bmX\right) \, = \, \frac{a_{j,l}}{(n_j + \gamma_j(r+M^\star))} \frac{(a_{j,l}+1)}{(n_j + \gamma_j(r+M^\star)+1)} \, .
	$$
    \end{proof}
	Note that the posterior is computed with respect to the multilevel sample $\bmX$ since $M^\star$ depends on both areas. 
	
	
	The limiting behaviour for $\gamma_j$ going to zero or infinity further explains how to interpret such a parameter and sheds light on the properties of the model.
	Firstly, consider the limit for $\gamma_j$ going to zero while $\gamma_{j^\prime}$ is kept fixed, with $j^\prime\neq j$, that is 
	\begin{equation}
		\label{eqn:ssj_post_lim0}
		\lim_{\gamma_j\rightarrow 0} 
		\E(\rho_j\mid\bmX) = 
		\sum_{l=1}^r\frac{n_{j,l}}{n_j}\frac{\left(n_{j,l}+1\right)}{\left(n_j+1\right)} \, .
	\end{equation}
	\begin{proof}
	    We write the expected value with respect to $q^\star_{M\mid\bmX}$ as an infinite sum.
	      We recall that, although it is not explicit in the notation, $q^\star_{M\mid\bmX}$ depends both on $\gamma_1$ and $\gamma_2$, see Equation \eqref{eqn:qMpost}. 
	   We use the dominated convergence theorem to exchange the limit and the sum, which is valid thanks to the upper bound provided in Section \ref{app:scambio_limiti}.
	
	Then, the limit for $\gamma_j\rightarrow0$:
	\begin{equation*}
		\label{eqn:ssj_post_lim0_proof1}
		\begin{split}
			&\lim_{\gamma_j\rightarrow 0}  \E\left(\rho_j\mid\bmX\right) \\
			& \quad = 
			\lim_{\gamma_j\rightarrow 0} \sum_{\mstar=0}^\infty 
			\frac{
				\sum_{l=1}^r n_{j,l}(n_{j,l} + 1) + \gamma_j\left( \gamma_j(r+\mstar) + r+\mstar + 2n_j  \right)
			}{
				n_j(n_j+1) + \gamma_j^2(r+\mstar)^2 + \gamma_j(r+\mstar)(2n_j+1)		
			}
			q^\star_{M\mid\bmX}(\mstar)\\
			& \quad = 
			\sum_{\mstar=0}^\infty \lim_{\gamma_j\rightarrow 0}
			\frac{
				\sum_{l=1}^r n_{j,l}(n_{j,l} + 1) + \gamma_j\left( \gamma_j(r+\mstar) + r+\mstar + 2n_j  \right)
			}{
				n_j(n_j+1) + \gamma_j^2(r+\mstar)^2 + \gamma_j(r+\mstar)(2n_j+1)		
			}
			q^\star_{M\mid\bmX}(\mstar) \\
			& \quad = 
			\frac{\sum_{l=1}^r n_{j,l}(n_{j,l} + 1) }{n_j(n_j+1)}
			\sum_{\mstar=0}^\infty \lim_{\gamma_j\rightarrow 0} 
			q^\star_{M\mid\bmX}(\mstar)
			= 
			\frac{\sum_{l=1}^r n_{j,l}(n_{j,l} + 1) }{n_j(n_j+1)} \, .
		\end{split}
	\end{equation*}
	The final equality holds since $\sum_{\mstar=0}^\infty \lim_{\gamma_j\rightarrow 0} q^\star_{M\mid\bmX}(\mstar) =  \sum_{\mstar=0}^\infty q^\star_{M\mid\bmX,0}(\mstar) = 1$, see Section \ref{app:scambio_limiti}.
	\end{proof}
	In the model formulation, the limiting case $\gamma_1 \rightarrow 0$ represents a prior belief in which all the mass is concentrated on a single species. In this case, the Simpson diversity index is one, which corresponds to the minimal heterogeneity, i.e., no diversity. 
	A posteriori, this belief is updated only based on the observed sample, which results in an increase in the estimated diversity in any sample with more than one species observed. 
	Finally, we note the similarity between the limit in Equation \eqref{eqn:ssj_post_lim0} and the Simpson estimator in Equation \eqref{eqn:ssj_freq}. 
	The intuition behind the result in Equation \eqref{eqn:ssj_post_lim0} is that the prior effectively contributes to the estimation through an additional observation with unknown group assignment. Thus, the total sample size increases by one, and the number of within-species pairs is augmented to include $n_j$ additional pairs formed between each observation in group $j$ and the prior-induced pseudo-observation.
	Since the Bayesian estimator in Equation \eqref{eqn:ssj_post_lim0} is larger than the estimator in Equation \eqref{eqn:ssj_freq}, small values of $\gamma_j$ shrink 
	the Simpson estimator towards one, coherently with the interpretation of the corresponding prior belief.
	
	We now turn our attention to the limit of $\gamma_j\rightarrow+\infty$, which results in 
	\begin{equation}
		\label{eqn:ssj_post_liminf}
		\lim_{\gamma_j\rightarrow +\infty} 
		\E(\rho_j\mid\bmX) = 
		\E_{q^\star_{M\mid\bmX,\infty}}\left(\frac{1}{r+M^\star}\mid\bmX\right)
		\, ,
	\end{equation}
	where $q^\star_{M\mid\bmX,\infty}$ is the limiting probability distribution of $q^\star_{M\mid\bmX}$ when $\gamma_j\rightarrow+\infty$. 
	Further details are provided in Section \ref{app:scambio_limiti} while the explicit expression is reported in Equation \eqref{eqn:qMpost_liminf_def}. 
	\begin{proof}
    We follow similar steps as in the proof of Equation \eqref{eqn:ssj_post_lim0}:
	\begin{equation*}
		\label{eqn:ssj_post_liminf_proof1}
		\begin{split}
			&\lim_{\gamma_j\rightarrow \infty}  \E\left(\rho_j\mid\bmX\right) \\
			& \quad = 
			\lim_{\gamma_j\rightarrow \infty} \sum_{\mstar=0}^\infty 
			\frac{
				\sum_{l=1}^r n_{j,l}(n_{j,l} + 1) + \gamma_j\left( \gamma_j(r+\mstar) + r+\mstar + 2n_j  \right)
			}{
				n_j(n_j+1) + \gamma_j^2(r+\mstar)^2 + \gamma_j(r+\mstar)(2n_j+1)		
			}
			q^\star_{M\mid\bmX}(\mstar)\\
			& \quad = 
			\sum_{\mstar=0}^\infty \lim_{\gamma_j\rightarrow \infty}
			\frac{
				\sum_{l=1}^r n_{j,l}(n_{j,l} + 1) + \gamma_j\left( \gamma_j(r+\mstar) + r+\mstar + 2n_j  \right)
			}{
				n_j(n_j+1) + \gamma_j^2(r+\mstar)^2 + \gamma_j(r+\mstar)(2n_j+1)		
			}
			q^\star_{M\mid\bmX}(\mstar) \\
			& \quad = 
			\sum_{\mstar=0}^\infty 
			\frac{1}{r+\mstar}
			\lim_{\gamma_j\rightarrow \infty} 
			q^\star_{M\mid\bmX}(\mstar)
			= 
			\E_{q^\star_{M\mid\bmX},\infty}\left(1/(r+M^\star)\right) \, ,
		\end{split}
	\end{equation*}
	where $q^\star_{M\mid\bmX,\infty},\inf$ has been defined in Lemma \ref{lem:qminf}.
	\end{proof}
	We recall that large values of $\gamma_1$ and $\gamma_2$ correspond to setting a uniform prior over all species, i.e., 
	the situation of maximum 
	heterogeneity. A posteriori, after observing a sample containing $r$ distinct species, the expected heterogeneity remains the largest possible one: it is equal to the average of the inverse of $r + M^\star$, where $M^\star$ is the (random) number of unobserved species and is distributed according to the limiting distribution $q^\star_{M\mid\bmX,\infty}$. 
	
	We now focus on the Bayesian posterior estimator of $\rho_{12}$, 
	given by
	\begin{equation}
		\label{eqn:sp1p2_post}
		\E\left(\rho_{12}\mid\bmX\right) = 
		\E_{q^\star_{M\mid\bmX}}
		\left(\frac{
			\sum_{l=1}^t n_{1,l}n_{2,l} + \gamma_1\gamma_2(r+M^\star) + n_1\gamma_2 + n_2\gamma_1
		}{
			(\gamma_1(r+M^\star) + n_1)(\gamma_2(r+M^\star) + n_2) 
		}
		\right)  \, .
	\end{equation}
    \begin{proof}
    Similarly to the derivation of $\E(\rho_j\mid \bmX)$, we exploit the conditional independence of $(P_1,P_2)\mid\bmX$.
    	\begin{equation*}
    		\label{eqn:sp1p2_post_proof1}
    		\begin{split}
    			&\E\left(\rho_{12}\mid \bmX\right) \, = \,
    			\E\left(\sum_{l=1}^{r+M^\star} w^\star_{1,l} w^\star_{2,l} \mid \bmX\right) \, = \,
    			\E\left( \E\left( \sum_{l=1}^{r+M^\star}w^\star_{1,l} w^\star_{2,l} \mid M^\star,\bmX \right) \mid  \bmX \right) \\ 
    			& \quad = \,
    			\E_{q^\star_{M\mid\bmX}}\left(
    			\sum_{l=1}^{r+M^\star}
    			\left\{
    			\E\left(
    			w^\star_{1,l}
    			\mid M^\star,\bmX\right) 
    			\E\left(
    			w^\star_{2,l}
    			\mid M^\star,\bmX\right)  
    			\right\}
    			\right) \\
    			& \quad = \,
    			\E_{q^\star_{M\mid\bmX}}\left(
    			\sum_{l=1}^{r}
    			\frac{(n_{1,l}+\gamma_1)}{(n_1+\gamma_1(r+M^\star))}
    			\frac{(n_{2,l}+\gamma_2)}{(n_2+\gamma_2(r+M^\star))}
    			\, + \,
    			M^\star\frac{\gamma_1\gamma_2}{(n_1+\gamma_1(r+M^\star))(n_2+\gamma_2(r+M^\star))} 
    			\right) \, .
    		\end{split}
    	\end{equation*}
    	The result follows after developing the sum in the numerator of the first term. Indeed,
    	$$
    	\sum_{l=1}^r (\gamma_1 + n_{1,l})(\gamma_2 + n_{2,l}) \, = \, 
    	r\gamma_1\gamma_2 + \gamma_1n_2 + \gamma_1n_1 + \sum_{l=1}^t n_{1,l}n_{2,l} \, .
    	$$
    \end{proof}
	%
	Again, we turn our attention to the limiting case for the $\gamma_j$'s parameters. 
	The limit of $\E\left(\rho_{12}\mid\bmX\right)$ for $\gamma_1,\gamma_2$ both going towards zero equals the plug-in estimator, i.e.,
	\begin{equation}
		\label{eqn:sp1p2_post_lim0}
		\lim_{\gamma_1,\gamma_2\rightarrow 0} 
		\E\left(\rho_{12}\mid\bmX\right) = 
		\sum_{l=1}^t \frac{n_{1,l}}{n_1}\frac{n_{2,l}}{n_2} \, .
	\end{equation}
    \begin{proof}
    We follow the same approach as the proof for Equation \eqref{eqn:ssj_post_lim0}. Firstly,
	consider the limit for $\gamma_1,\gamma_2\rightarrow0$ following the same rate. This is equivalent to assuming there exist some constants $c_1,c_2$ such that $\gamma_j = c_j \gamma$ and letting $\gamma$ go to zero or infinity.
	\begin{equation*}
		\label{eqn:sp1p2_post_lim0_proof1}
		\begin{split}
			&\lim_{\gamma_1,\gamma_2\rightarrow 0}  \E\left(\rho_{12}\mid\bmX\right) \\
			& \quad = 
			\lim_{\gamma_1,\gamma_2\rightarrow 0} \sum_{\mstar=0}^\infty 
			\left(\frac{
				\sum_{l=1}^t n_{1,l}n_{2,l} + \gamma_1\gamma_2(r+M^\star) + n_1\gamma_2 + n_2\gamma_1
			}{
				(\gamma_1(r+M^\star) + n_1)(\gamma_2(r+M^\star) + n_2)
			}
			\right)
			q^\star_{M\mid\bmX}(\mstar)\\
			& \quad = 
			\sum_{\mstar=0}^\infty \lim_{\gamma_1,\gamma_2\rightarrow 0}
			\left(\frac{
				\sum_{l=1}^t n_{1,l}n_{2,l} + \gamma_1\gamma_2(r+M^\star) + n_1\gamma_2 + n_2\gamma_1
			}{
				(\gamma_1(r+M^\star) + n_1)(\gamma_2(r+M^\star) + n_2)
			}
			\right)
			q^\star_{M\mid\bmX}(\mstar) \\
			& \quad = 
			\frac{\sum_{l=1}^t n_{1,l}n_{2,l} }{n_1n_2}
			\sum_{\mstar=0}^\infty \lim_{\gamma_j\rightarrow 0} 
			q^\star_{M\mid\bmX}(\mstar)
			= 
			\frac{\sum_{l=1}^t n_{1,l}n_{2,l} }{n_1n_2} \, .
		\end{split}
	\end{equation*}
    \end{proof}
    
	Equation \eqref{eqn:sp1p2_post_lim0} supports the interpretation of $\gamma_1$ and $\gamma_2$ as homogeneity parameters that can be employed to impose sparsity within each area when pushed towards zero. Indeed, we recall that values of $\gamma_j$ close to zero represent the prior belief of minimum correlation between $P_1$ and $P_2$, that is $\E\left(1/M\right)$ 
	(see Section \ref{subsection:correlation}).
	Then, the posterior expectation $\E\left(\rho_{12}\mid\bmX\right)$ achieves its minimum, which is equal to the plug-in estimator obtained by replacing the species probabilities with normalised observed counts; see, e.g., \cite{archer2014}. 
	This means that the amount of similarity between the two areas must be at least equal to that observed. In the undersampled regime, the plug-in estimator is negatively biased, while the prior offers a correction for this bias. 
	
	Finally, when we assume a priori that all species are equally present in the population, i.e., the case of $\gamma_1,\gamma_2\rightarrow+\infty$, we obtain
	\begin{equation}
		\label{eqn:sp1p2_post_liminf}
		\lim_{\gamma_1,\gamma_2\rightarrow +\infty} \E(\rho_{12}\mid\bmX) = 
		\E_{q^\star_{M\mid\bmX,\infty}}\left(\frac{1}{r+M^\star}\mid\bmX\right)
		\, .
	\end{equation}
    \begin{proof}
        We follow the same approach as the proof for Equation \eqref{eqn:ssj_post_liminf}, letting $\gamma_1$ and $\gamma_2$ go to infinity at the same rate.
    	\begin{equation*}
    		\label{eqn:sp1p2_post_liminf_proof1}
    		\begin{split}
    			&\lim_{\gamma_1,\gamma_2\rightarrow \infty}  \E\left(\rho_{12}\mid\bmX\right) \\
    			& \quad = 
    			\lim_{\gamma_j\rightarrow \infty} \sum_{\mstar=0}^\infty 
    			\left(\frac{
    				\sum_{l=1}^t n_{1,l}n_{2,l} + \gamma_1\gamma_2(r+M^\star) + n_1\gamma_2 + n_2\gamma_1
    			}{
    				(\gamma_1(r+M^\star) + n_1)(\gamma_2(r+M^\star) + n_2)
    			}
    			\right)
    			q^\star_{M\mid\bmX}(\mstar)\\
    			& \quad = 
    			\sum_{\mstar=0}^\infty \lim_{\gamma_1,\gamma_2\rightarrow \infty}
    			\left(\frac{
    				\sum_{l=1}^t n_{1,l}n_{2,l} + \gamma_1\gamma_2(r+M^\star) + n_1\gamma_2 + n_2\gamma_1
    			}{
    				(\gamma_1(r+M^\star) + n_1)(\gamma_2(r+M^\star) + n_2)
    			}
    			\right)
    			q^\star_{M\mid\bmX}(\mstar) \\
    			& \quad = 
    			\sum_{\mstar=0}^\infty 
    			\frac{1}{r+\mstar}
    			\lim_{\gamma_1,\gamma_2\rightarrow \infty} 
    			q^\star_{M\mid\bmX}(\mstar)
    			= 
    			\E_{q^\star_{M\mid\bmX},\infty}\left(1/(r+M^\star)\right) \, .
    		\end{split}
    	\end{equation*}
    	The exchange of limit and series is valid, as shown in Section \ref{app:scambio_limiti}.
    \end{proof}
	Specifically, Equation \eqref{eqn:sp1p2_post_liminf} coincides with the corresponding limit of the Simpson index in Equation \eqref{eqn:ssj_post_liminf}. This implies that the estimated Morisita index reaches its maximum value of one, which is, once again, consistent with our interpretation of the $\gamma_j$'s parameters. In fact, pushing the $\gamma_j$'s towards infinity reflects a strong prior belief that the two populations are perfectly correlated, i.e.,  that they are identical.

    \section{Additional details on the simulation study}
	\label{app:SS_additional}
    Tables \ref{tab:SS_D}-\ref{tab:SS_A} report the names of the settings we considered during our simulation study. For selected cases $(D_3,G_6,Z_3,A_{1,4},A_{2,4},A_{3,4})$, Figure \ref{fig:SS_Morosita} displays the values of the Morisita index defined in Section \ref{subsection:indici_intro} and computed for each replicated dataset in some selected simulation settings. Results for all remaining settings are available in the repository \url{https://github.com/alessandrocolombi/HSSM}. Finally, Figures \ref{fig:AccCrv_D3}-\ref{fig:AccCrv_A34} show the accumulation curves of all quantities of interest in the selected settings. The shaded grey areas represent the $95\%$ bootstrap envelopes, computed over $100$ replicated datasets.
    \begin{table}[h!]
    \centering
    \begin{minipage}{0.32\textwidth}
    \centering
    \begin{tabular}{|c|cc|}
        \hline
        & $0.1$ & $0.5$ \\ \hline
        \multicolumn{1}{|c|}{$0.1$} & $D_1$ & $D_2$ \\
        \multicolumn{1}{|c|}{$0.5$} & $-$ & $D_3$ \\
        \hline
    \end{tabular}
    \caption{Dirichlet weights.}
    \label{tab:SS_D}
    \end{minipage}
    \hfill
    \begin{minipage}{0.32\textwidth}
    \centering
    \begin{tabular}{|c|ccc|}
        \hline
        & $0.80$ & $0.85$ & $0.90$ \\ \hline
        $0.80$ & $G_1$ & $G_2$ & $G_3$ \\
        $0.85$ & $-$ & $G_4$ & $G_5$ \\
        $0.90$ & $-$ & $-$ & $G_6$ \\ \hline
    \end{tabular}
    \caption{Geometric weights.}
    \label{tab:SS_G}
    \end{minipage}
    \hfill
    \begin{minipage}{0.32\textwidth}
    \centering
    \begin{tabular}{|c|cc|}
        \hline
        & $1.3$ & $2$ \\ \hline
        \multicolumn{1}{|c|}{$1.3$} & $Z_1$ & $Z_2$ \\
        \multicolumn{1}{|c|}{$2$} & $-$ & $Z_3$ \\
        \hline
    \end{tabular}
    \caption{Zipf's weights.}
    \label{tab:SS_Z}
    \end{minipage}
    \end{table}
    \begin{table}[h!]
        \begin{tabular}{|c|cccc|}
        \hline
          & $(0.1,0.1)$ & $(0.1,0.5)$ & $(0.5,0.1)$ & $(0.5,0.5)$ \\ \hline
            $0$   & $A_{1,1}$   & $A_{1,2}$   & $A_{1,3}$   & $A_{1,4}$   \\
            $0.5$ & $A_{2,1}$   & $A_{2,2}$   & $A_{2,3}$   & $A_{2,4}$   \\
            $1$   & $A_{3,1}$   & $A_{3,2}$   & $A_{3,3}$   & $A_{3,4}$   \\ \hline
    \end{tabular}
    \caption{All configurations in the Additive case. On the rows, the values of the mass $c$ of the common component. The columns report all possible pairs $(\delta_0,\delta)$ of the Dirichlet distributions.}
    \label{tab:SS_A}
    \end{table}
    \begin{figure}[ht!]
        \centering
        \includegraphics[width=0.7\linewidth]{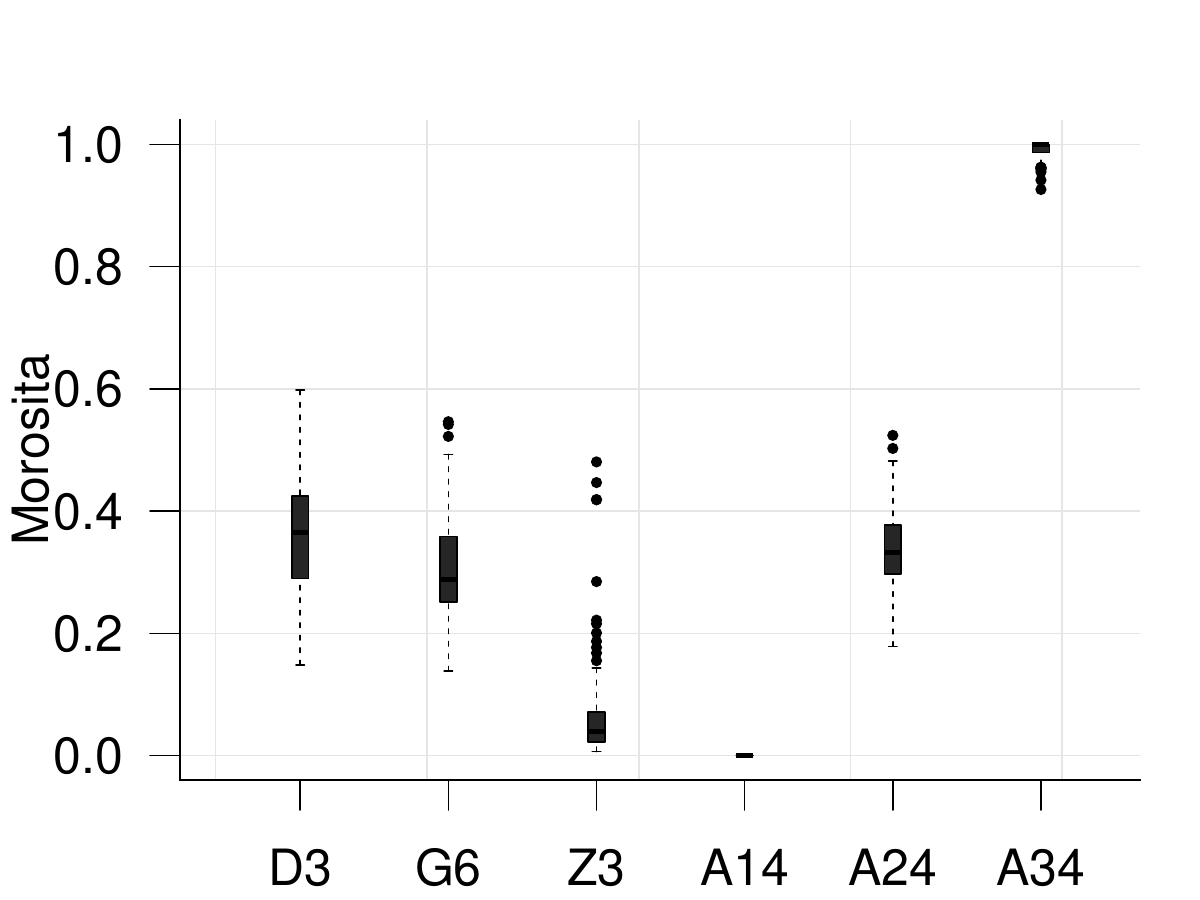}
        \caption{Morisita index for selected settings. }
        \label{fig:SS_Morosita}
    \end{figure}
    \begin{figure}[ht!]
        \centering
        \includegraphics[width=1\linewidth]{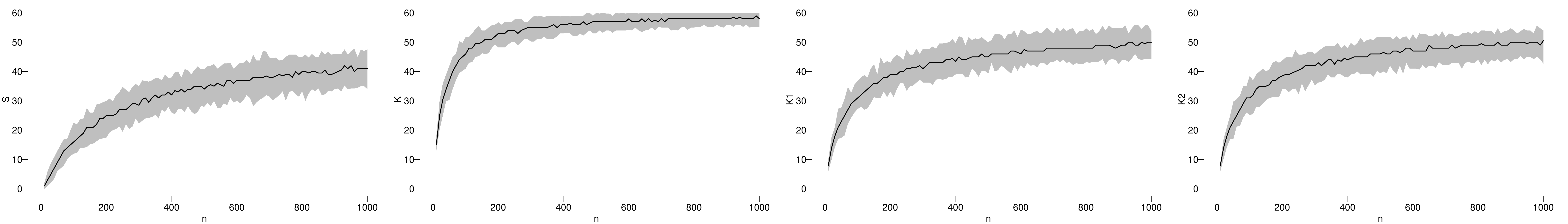}
        \caption{Accumulation curves for setting $D_3$. }
        \label{fig:AccCrv_D3}
    \end{figure}
    \begin{figure}[ht!]
        \centering
        \includegraphics[width=1\linewidth]{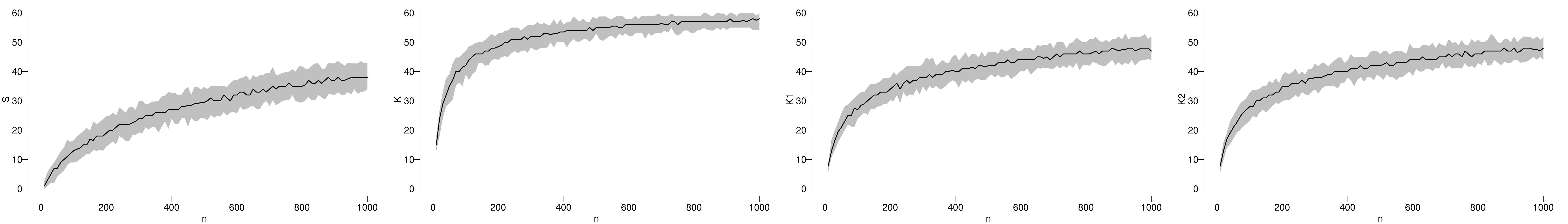}
        \caption{Accumulation curves for setting $G_6$. }
        \label{fig:AccCrv_G6}
    \end{figure}
    \begin{figure}[ht!]
        \centering
        \includegraphics[width=1\linewidth]{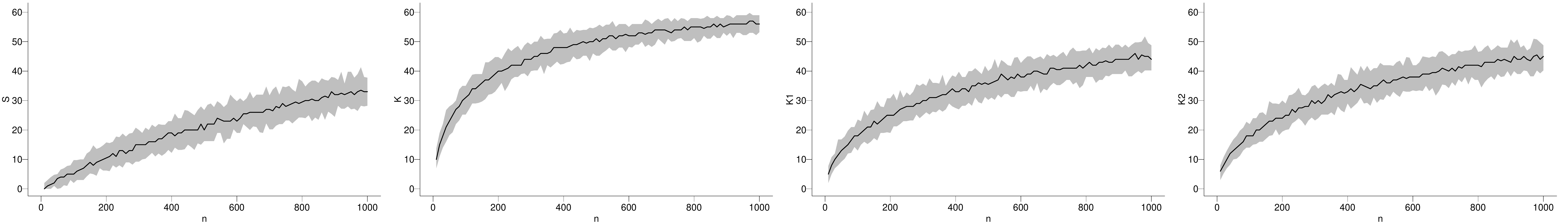}
        \caption{Accumulation curves for setting $Z_3$. }
        \label{fig:AccCrv_Z3}
    \end{figure}
    \begin{figure}[ht!]
        \centering
        \includegraphics[width=1\linewidth]{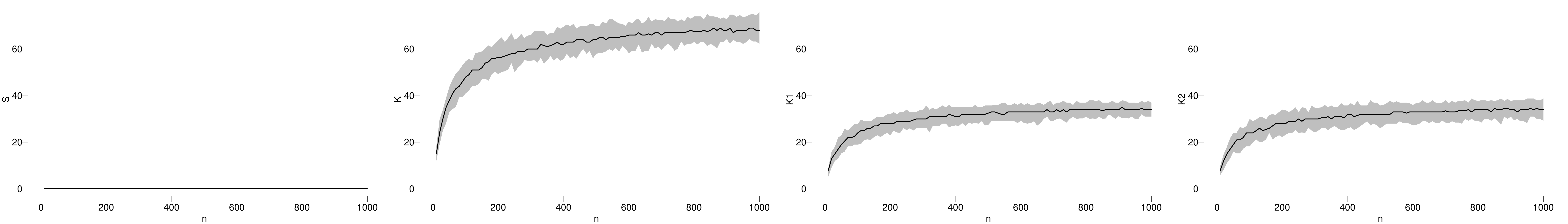}
        \caption{Accumulation curves for setting $A_{1,4}$. }
        \label{fig:AccCrv_A14}
    \end{figure}
    \begin{figure}[ht!]
        \centering
        \includegraphics[width=1\linewidth]{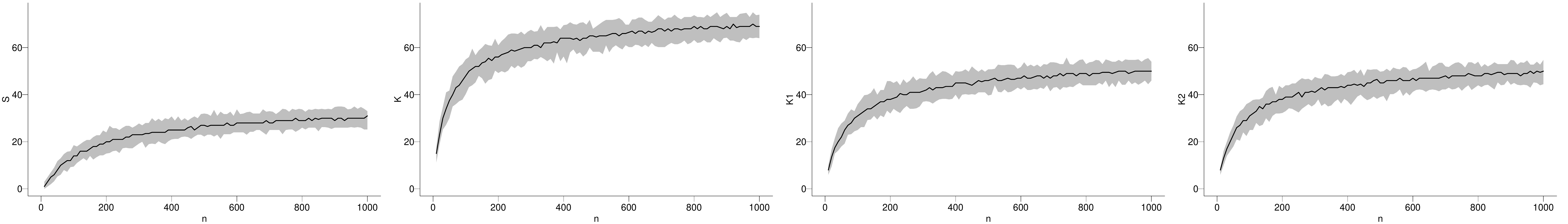}
        \caption{Accumulation curves for setting $A_{2,4}$. }
        \label{fig:AccCrv_A24}
    \end{figure}
    \begin{figure}[ht!]
        \centering
        \includegraphics[width=1\linewidth]{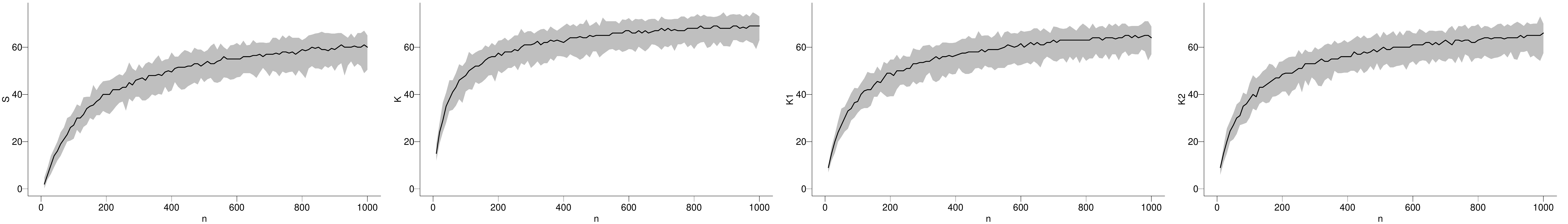}
        \caption{Accumulation curves for setting $A_{3,4}$. }
        \label{fig:AccCrv_A34}
    \end{figure}
    
    \subsection{Additional details on Experiment 1}
	\label{app:SS_Exp1_additional}
    Table \ref{tab:SS1_paramest} show the estimated parameters in the six selected cases reported in Section \ref{subsection:SS_exp1} for $n_1=n_2=400$. The estimated quantities are averaged over the $100$ replicated dataset. The corresponding standard deviation is reported as well. We recall that the true total species numbers are $M_{\mathrm{tot}}=60$ for $(D_3,G_6,Z_3)$ and $M_{\mathrm{tot}}=80$ for $(A_{1,4},A_{2,4},A_{3,4})$. In setting $A_{1,4}$, we set $\Lambda=100$ by default since, by construction, the Morisita index is exactly equal to zero, which does not allow us to solve Equation \eqref{eqn:sp1p2}. 
    \begin{table}[ht]
        \centering
        \begin{tabular}{llcccc}
        \toprule
         & & $K_n+M^\star$ & $\Lambda$ & $\gamma_1$ & $\gamma_2$ \\
        \midrule
        \multirow{2}{*}{$D_3$} 
          & \textit{Bayes I}  & 60.43 (2.48)    & 59.43 (2.48)    & 0.515 (0.113) & 0.516 (0.132) \\
          & \textit{Bayes II} & 62.27 (7.29)    & 62.93 (17.25)   & 0.594 (0.310) & 0.589 (0.344) \\
        \addlinespace
        \multirow{2}{*}{$G_6$} 
          & \textit{Bayes I}  & 60.86 (3.72)    & 59.86 (3.72)    & 0.355 (0.051) & 0.356 (0.059) \\
          & \textit{Bayes II} & 62.11 (9.52)    & 66.07 (19.48)   & 0.442 (0.191) & 0.450 (0.225) \\
        \addlinespace
        \multirow{2}{*}{$Z_3$} 
          & \textit{Bayes I}  & 71.71 (9.64)    & 70.71 (9.64)    & 0.160 (0.034) & 0.154 (0.033) \\
          & \textit{Bayes II} & 172.97 (123.71) & 156.88 (128.68) & 0.073 (0.119) & 0.085 (0.143) \\
        \addlinespace
        \multirow{2}{*}{$A_{1,4}$} 
          & \textit{Bayes I}  & 504.99 (3.57)   & 518.86 (4.24)   & 0.0157 (0.0016) & 0.0155 (0.0017) \\
          & \textit{Bayes II} & 98.94 (5.99)    & 100.00 (0.00)   & 0.158 (0.041)   & 0.153 (0.050)   \\
        \addlinespace
        \multirow{2}{*}{$A_{2,4}$} 
          & \textit{Bayes I}  & 74.97 (5.49)    & 73.97 (5.49)    & 0.304 (0.046) & 0.320 (0.056) \\
          & \textit{Bayes II} & 68.30 (5.65)    & 60.25 (12.29)   & 0.487 (0.161) & 0.510 (0.165) \\
        \addlinespace
        \multirow{2}{*}{$A_{3,4}$} 
          & \textit{Bayes I}  & 64.43 (3.99)    & 63.34 (3.95)    & 0.889 (0.183) & 0.879 (0.203) \\
          & \textit{Bayes II} & 63.03 (3.89)    & 28.04 (4.22)    & 3.53 (6.61)   & 9.74 (42.85)   \\
        \bottomrule
        \end{tabular}
        \caption{Estimated values (standard deviations) for $(K_n+M^\star,\Lambda,\gamma_1,\gamma_2)$ obtained with \textit{Bayes I} (first subrow) and \textit{Bayes II} (second subrow). True total species numbers are $M_{\mathrm{tot}}=60$ for $(D_3,G_6,Z_3)$ and $M_{\mathrm{tot}}=80$ for $(A_{1,4},A_{2,4},A_{3,4})$.}
        \label{tab:SS1_paramest}
    \end{table}
    Figure \ref{fig:SS1_ExecTime} displays the running times as the sample size increases. In particular, we distinguish between the execution time needed to estimate the parameters of interest (left panel) and the time needed to perform out-of-sample prediction (right panel). The latter does not significantly change as the sample size increases since the test set is maintained constant as the prediction horizon is $m_1=m_2=200$. The execution time needed to estimate the parameters via \textit{Bayes I} increases with $n$ as it requires the evaluation of the marginal likelihood (i.e., the pEPPF in Equation \eqref{eqn:peppf}), while it is constant for \textit{Bayes II} since we only need to solve a system of two decoupled equations, as explained in Section \ref{subsection:param_estimation}.
    \begin{figure}[ht!]
        \centering
        \includegraphics[width=0.49\linewidth]{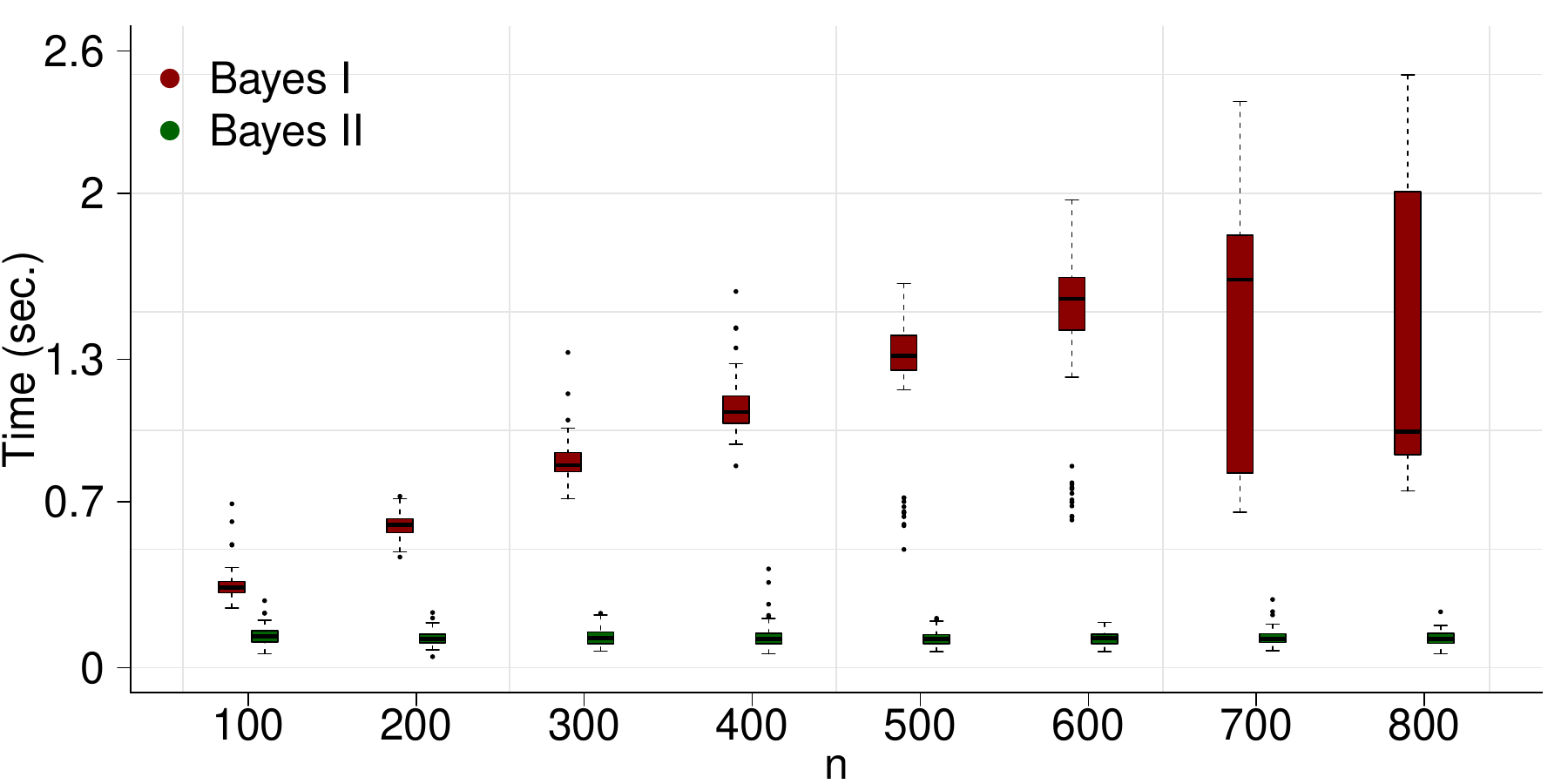}
        \hfill
        \includegraphics[width=0.49\linewidth]{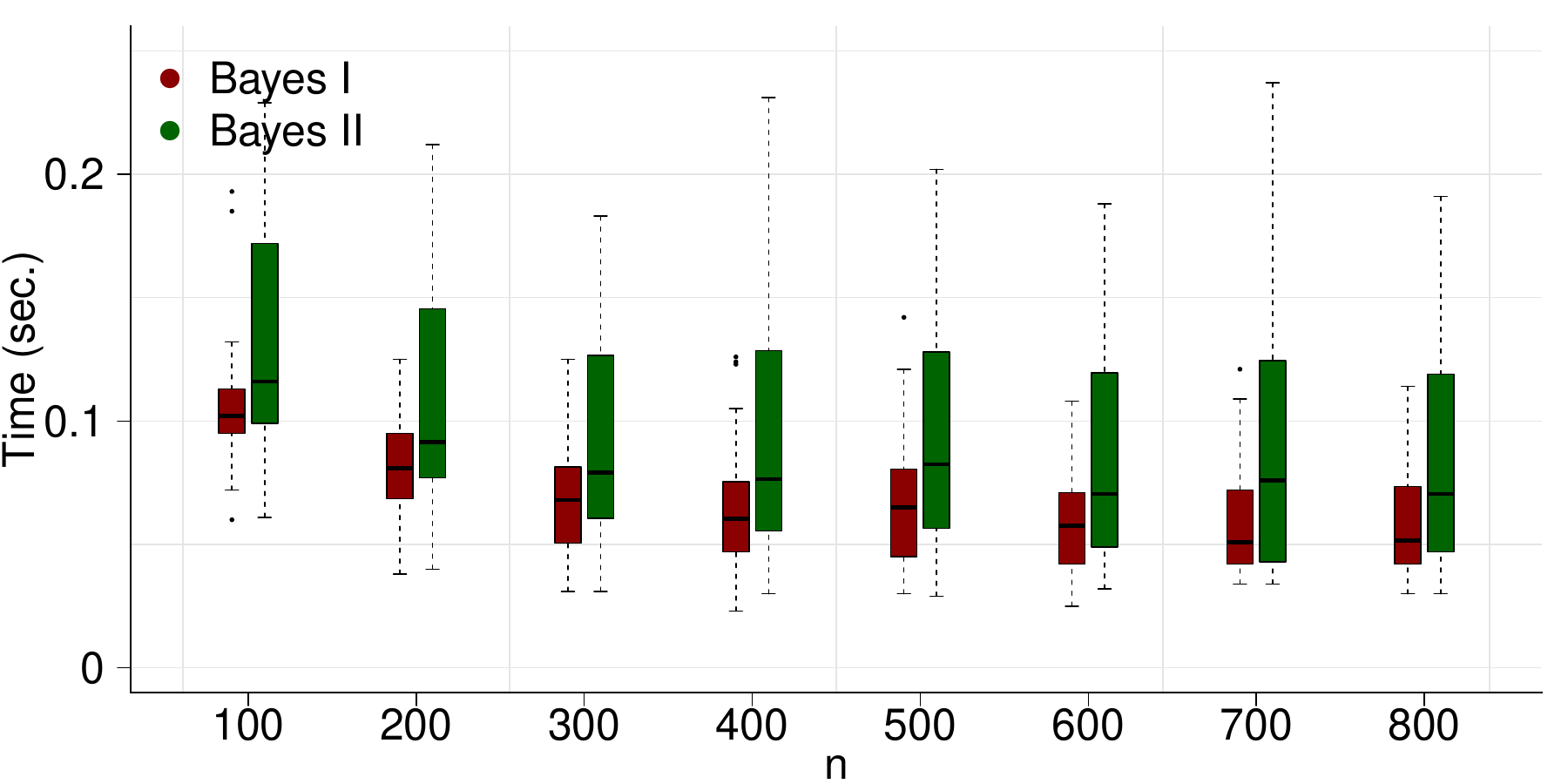}
        \caption{Experiment 1, execution times to fit the model (left panel) and to perform the out-of-sample prediction (right panel). Time is reported in seconds.}
        \label{fig:SS1_ExecTime}
    \end{figure}
    Finally, Figure \ref{fig:Exp1_n1small} shows the results for the unbalanced case $n_1 = 60$ and $n_2 = 300$. 
    \begin{figure}[ht!]
        \centering
        \includegraphics[width=0.485\linewidth]{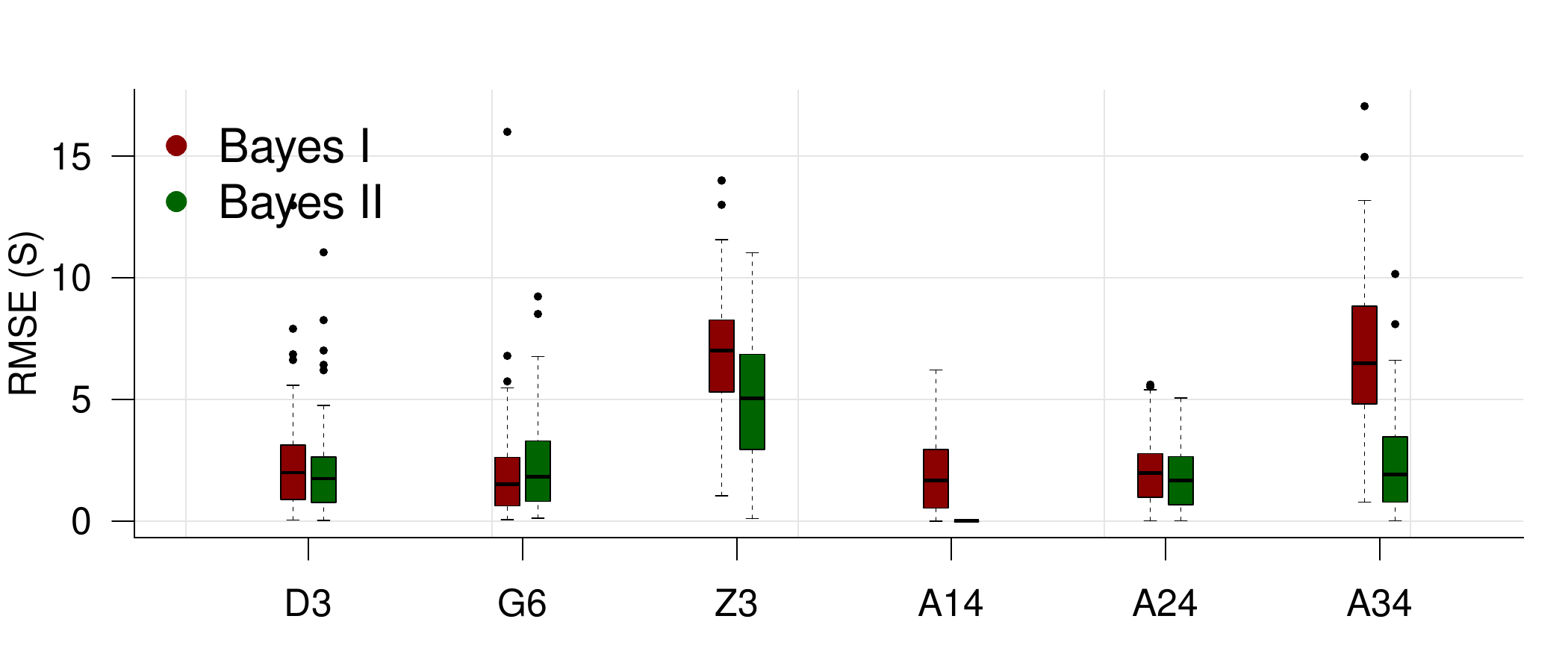}
        \hfill
        \includegraphics[width=0.485\linewidth]{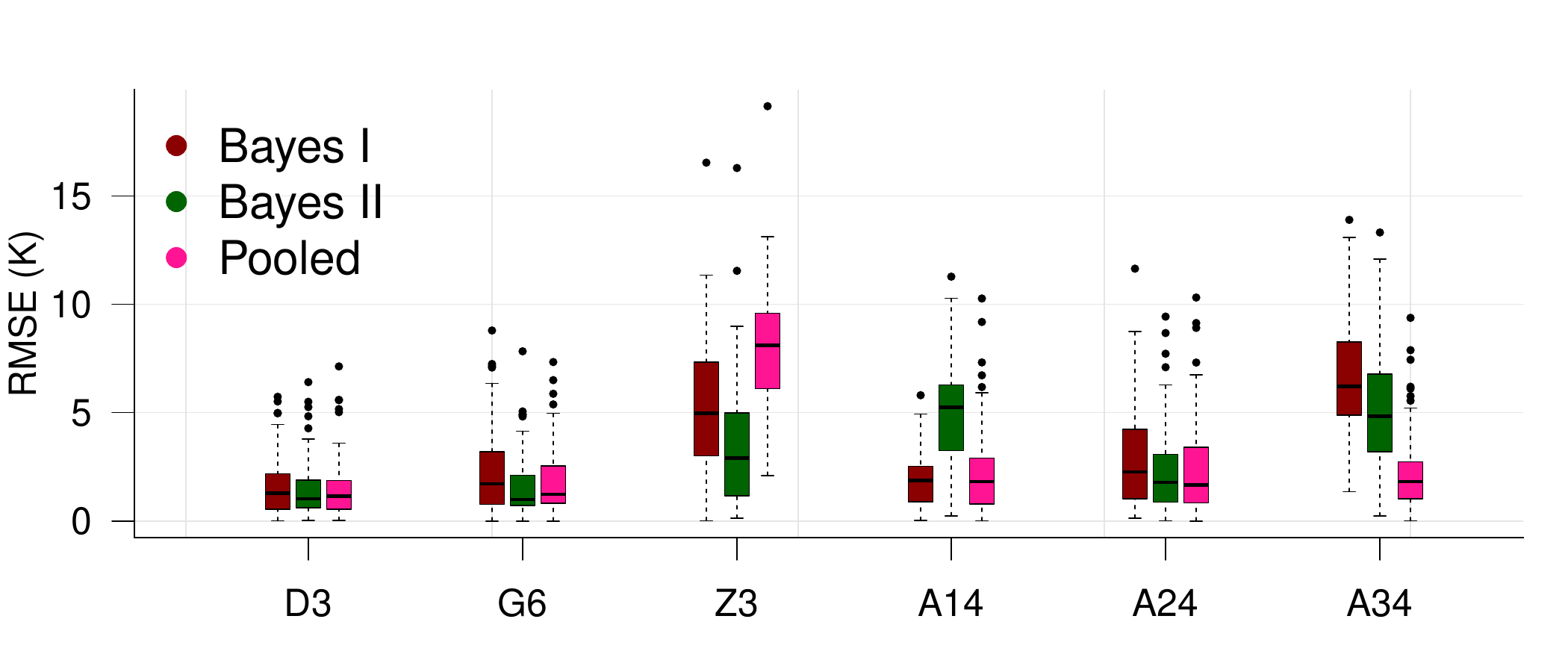}
        \\
        \includegraphics[width=0.485\linewidth]{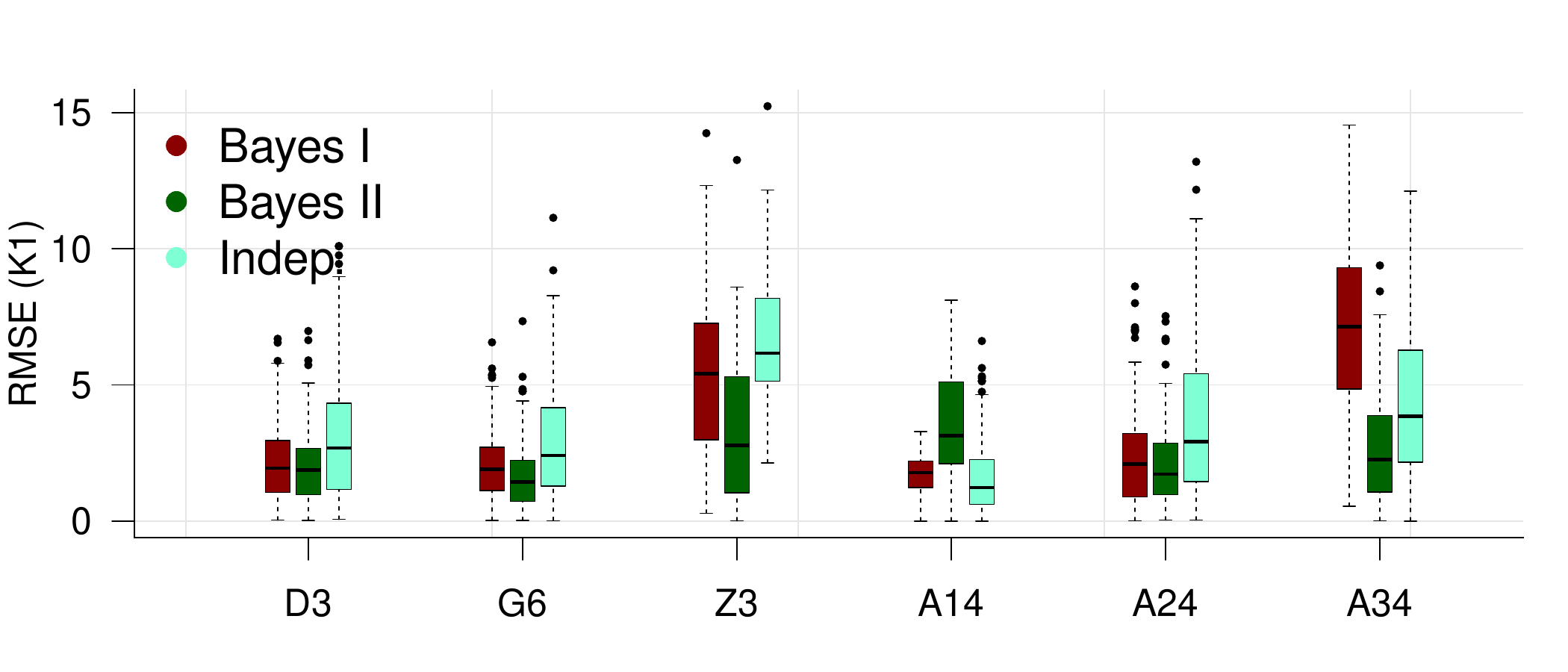}
        \hfill
        \includegraphics[width=0.485\linewidth]{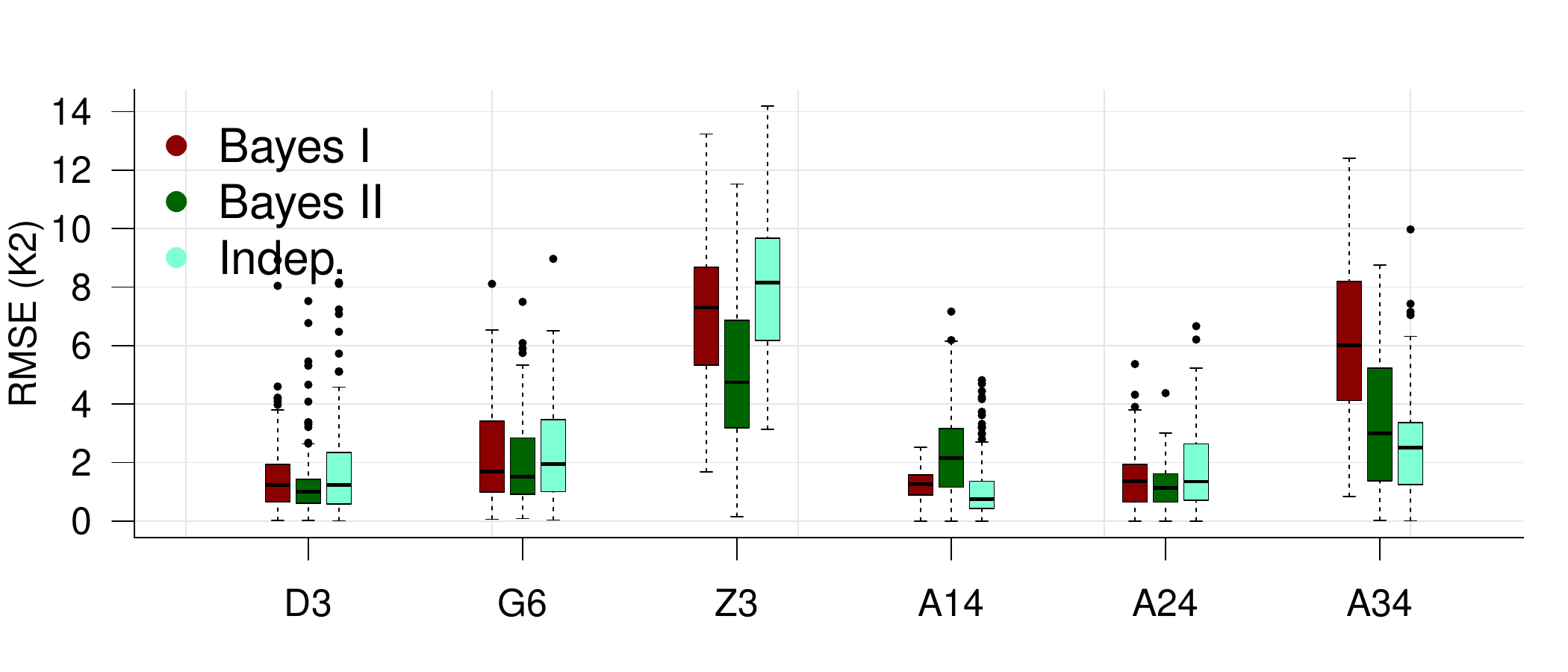}
        \hfill
        \caption{Experiment 1, $n_1\ll n_2$: RMSE of out-of-sample predictions for new shared species (top-left panel), new global distinct species (top-right panel), and new local distinct species (bottom-left panel) and (bottom-right panel) across selected scenarios.}
        \label{fig:Exp1_n1small}
    \end{figure}
    Comparing the bottom-left panel in Figure \ref{fig:Exp1_n1small} with its balanced counterpart in Figure \ref{fig:SS_Pred_1} highlights the advantage of a joint modelling perspective when one area is poorly sampled. In particular, for the prediction of new local distinct species in area 1 (the smaller sample), the \textit{Independent} approach exhibits a markedly larger RMSE in the unbalanced setting, with wider boxplots reflecting increased uncertainty due to limited data. In contrast, our joint model mitigates this effect by borrowing strength from the second area, resulting in more stable predictions.
    \subsection{Additional details on Experiment 2}
	\label{app:SS_Exp2_additional}
    Equation \eqref{eqn:Pr1step_true} reports the oracle estimators for the one-step ahead probabilities of discovery, that is
    \begin{equation}
	\label{eqn:Pr1step_true}
		\begin{split}
            &\P_{\text{true}}\left(K^{(n_j)}_{j,1} > 0\mid\bmX\right)  \\
			&\quad = 
            \sum_{m=1}^{M_{\text{true}}} 
			\
			p_{j,m}\indic(n_{j,m} = 0)\,, \quad j=1,2\,,\\
			&\P_{\text{true}}\left(\mathcal K^{(n_1,n_2)}_{1,1} > 0\mid\bmX\right) \\
            &\quad = 
			\sum_{m=1}^{M_{\text{true}}} 
			\
			\left(p_{1,m}+p_{2,m}-p_{1,m}p_{2,m}\right)\indic(n_{1,m} = 0, \, n_{2,m} = 0) 
			\,,\\
			&\P_{\text{true}}\left(\mathcal S^{(n_1,n_2)}_{1,1} > 0\mid\bmX\right)  \\
            &\quad = 
			\sum_{m=1}^{M_{\text{true}}}  
			\
			p_{1,m}p_{2,m}\indic(n_{1,m} = 0, \, n_{2,m} = 0) \\
			&\qquad\qquad + \
			p_{1,m}\indic(n_{1,m} = 0, \, n_{2,m} > 0) 
			   + 
			p_{2,m}\indic(n_{1,m} > 0, \, n_{2,m} = 0), 
			\,
		\end{split}
	\end{equation}
	where $n_{j,m}$ is the observed absolute frequency of the $(m)$th species in the $(j)$th group. 
    Finally, Figure \ref{fig:Exp2_n1small} shows the results for the unbalanced case $n_1 = 60$ and $n_2 = 300$. 
    \begin{figure}[ht!]
        \centering
        \includegraphics[width=0.485\linewidth]{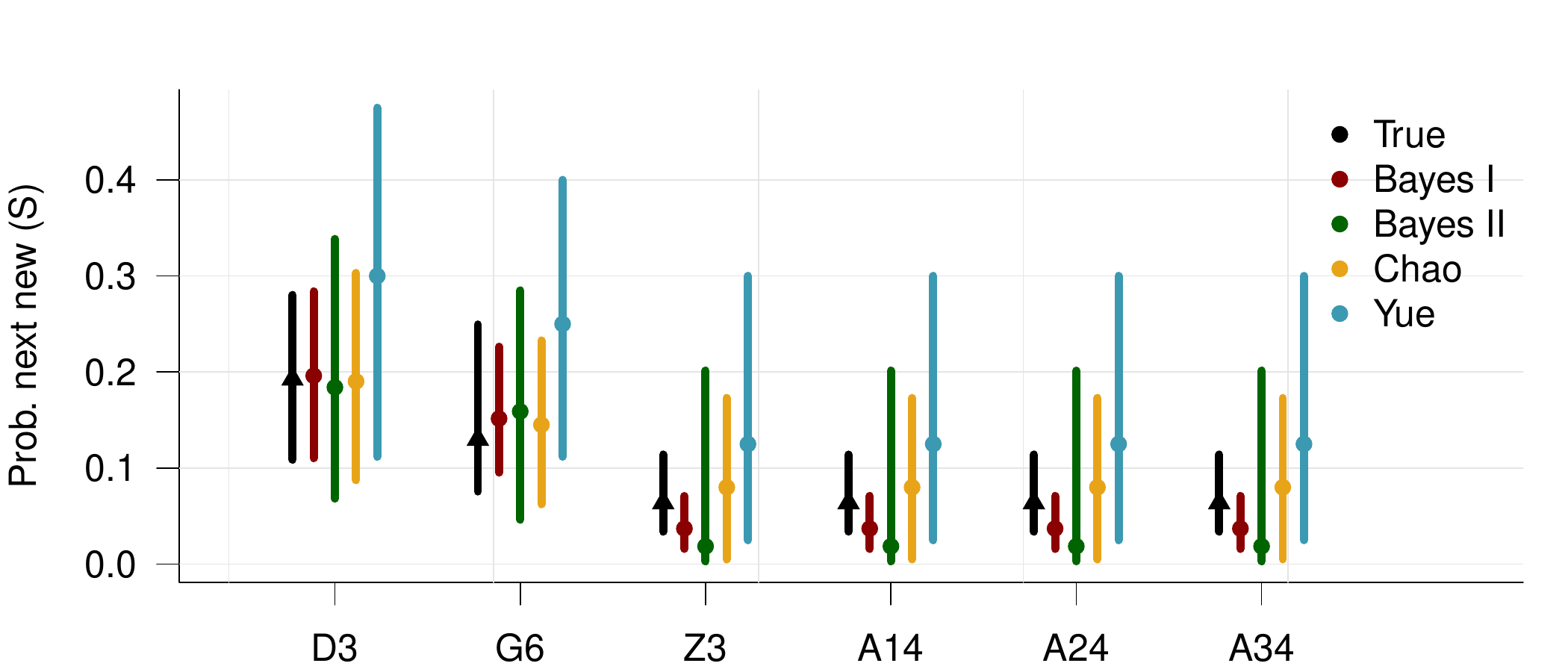}
        \hfill
        \includegraphics[width=0.485\linewidth]{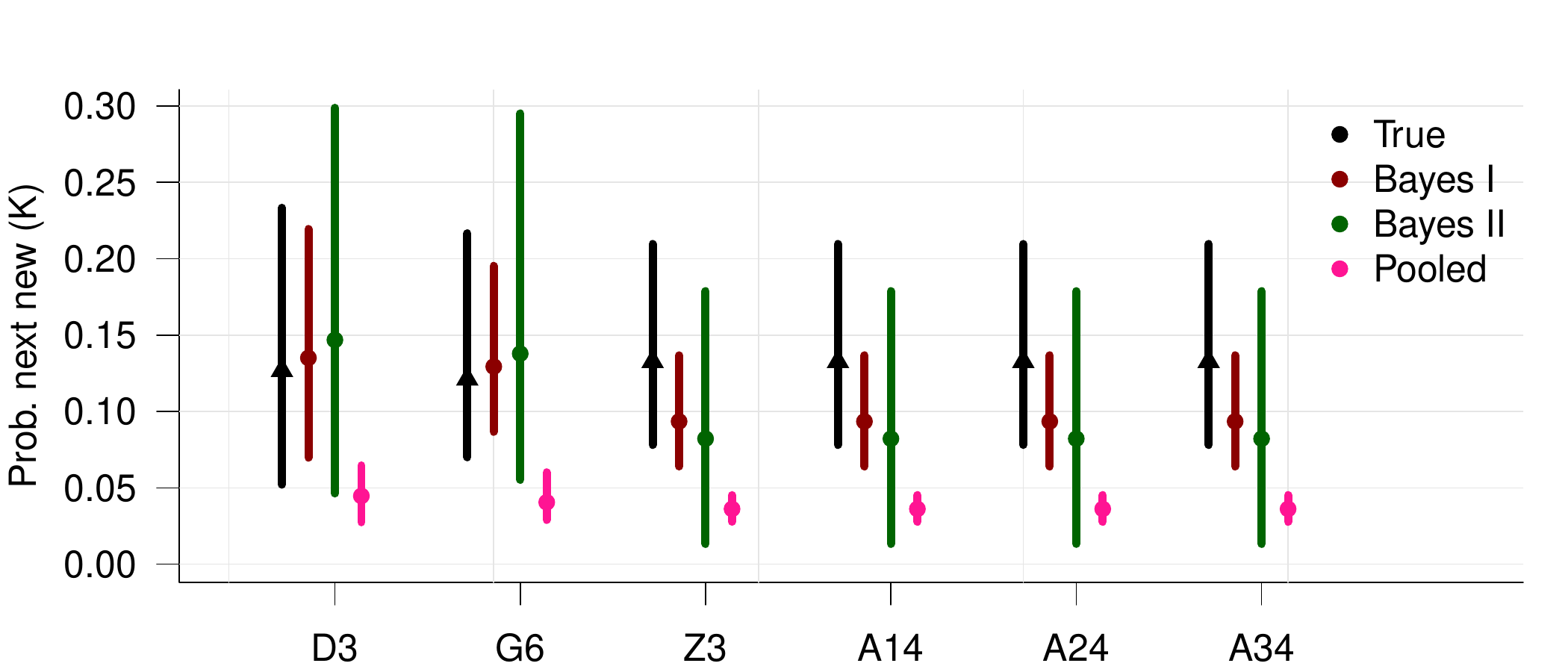}
        \\
        \includegraphics[width=0.485\linewidth]{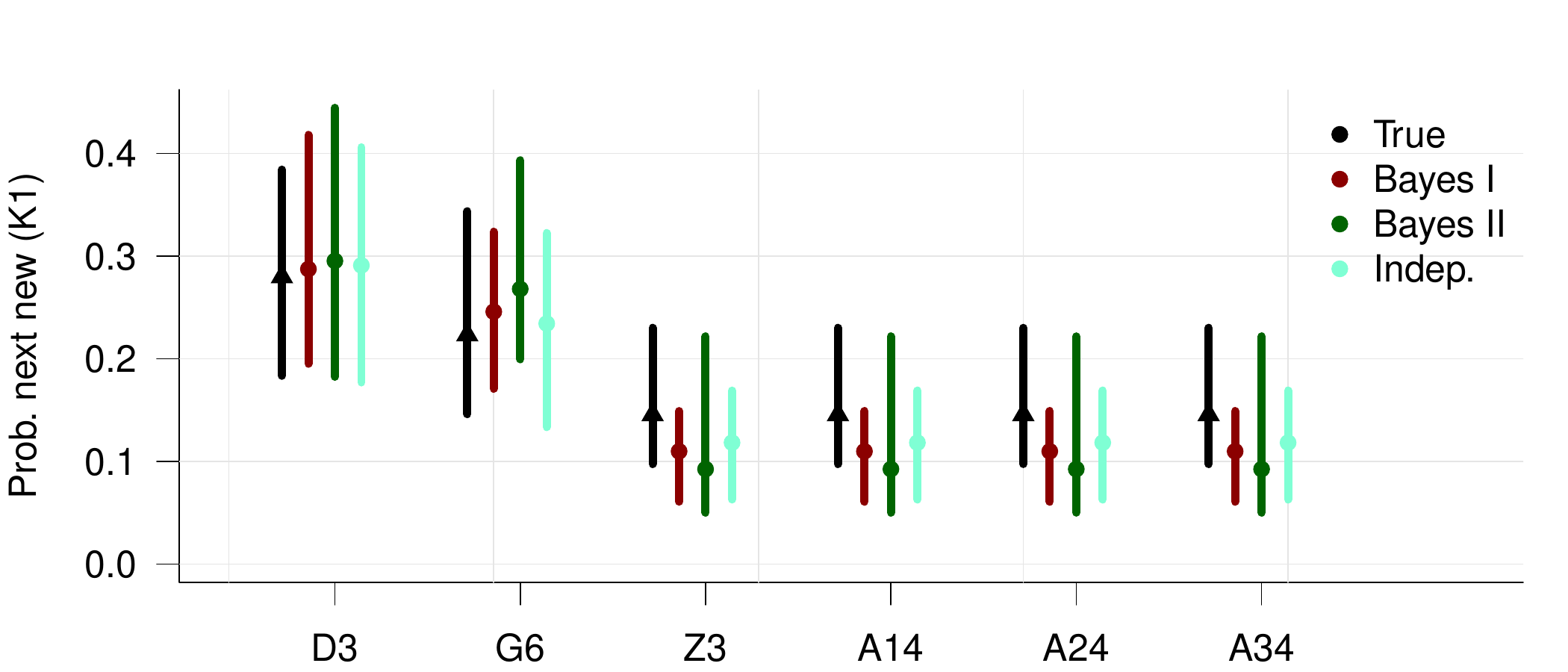}
        \hfill
        \includegraphics[width=0.485\linewidth]{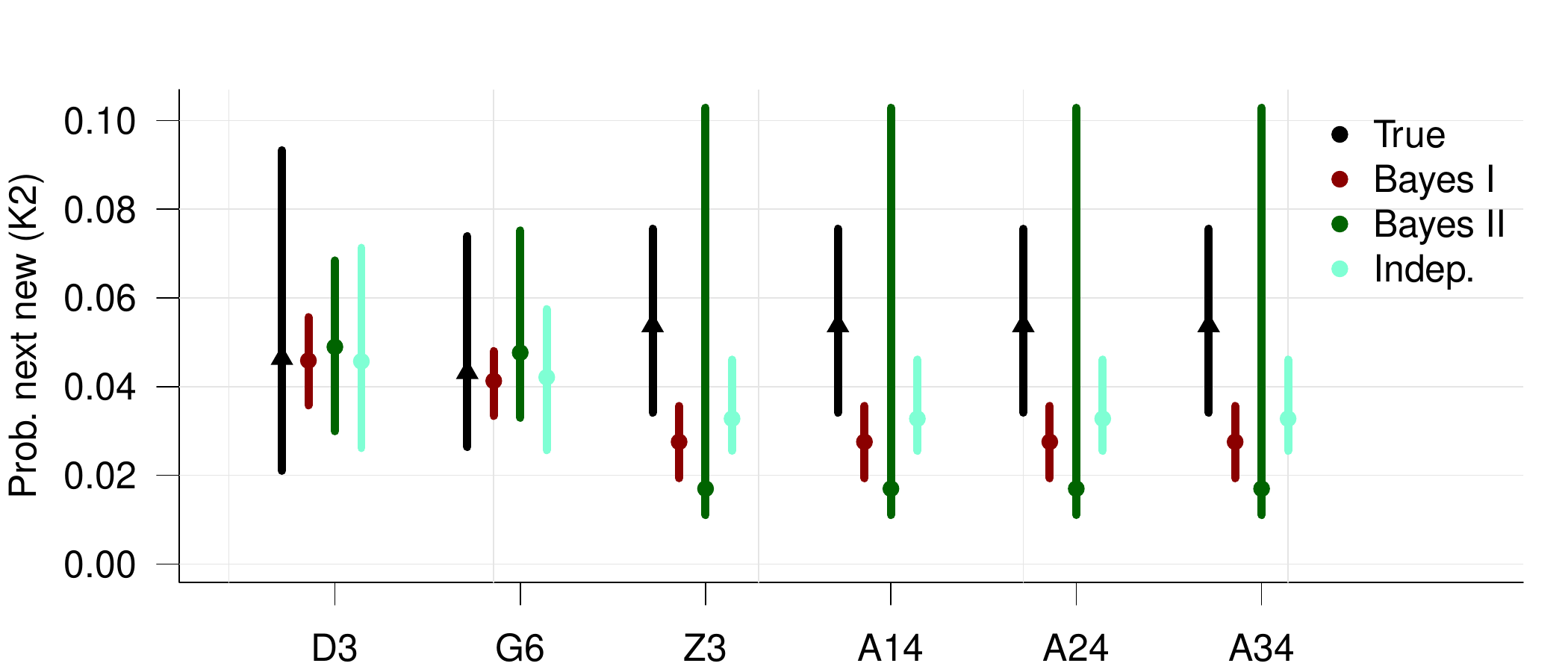}
        \hfill
        \caption{Experiment 2, $n_1\ll n_2$. One-step-ahead prediction probability for new shared species (top-left panel), new global distinct species (top-right panel), and new local distinct species (bottom-left panel) and (bottom-right panel) across selected scenarios.}
        \label{fig:Exp2_n1small}
    \end{figure}
	\subsection{Experiment 3}
    \label{subsection:SS_exp_mappette}
    This experiment introduces a practical tool that is naturally available under our proposed model but is largely unavailable in the existing literature, especially in the frequentist framework. In Section \ref{section:posterior} we derived both $1$-step-ahead and $m$-steps-ahead predictive distributions for the number of new local, global, and shared species. In Section \ref{subsection:SS_exp2} we focused on the one-step-ahead case because it is the only setting where direct comparisons with frequentist competitors are possible. The key methodological novelty of our approach, however, is that it enables coherent predictions at an arbitrarily distant horizon, i.e., for any future sample sizes $(m_1,m_2)$.
    
    To illustrate this capability, we propose a two-dimensional visualization of discovery probabilities across both the current sampling effort and the planned future effort. Specifically, for each quantity of interest, we represent the probability of making at least one new discovery as a function of: (i) the current sample size $n$ (horizontal axis), and (ii) the size of the additional sample (vertical axis). The resulting heatmaps provide an immediate summary of the expected gain from further sampling, and can be used in practice to decide whether it is worth investing additional resources or whether the experiment is already sufficiently exhaustive. Importantly, this assessment is obtained under a unified model that jointly describes all relevant quantities and supports $m$-step-ahead predictions.

    The results for scenario $G_6$,  $D_3$ and $Z_3$ are reported in Figures \ref{fig:SS_Pr_m_step_Geom}, 
\ref{fig:SS_Pr_m_step_Dir} and \ref{fig:SS_Pr_m_step_Zipfs}, respectively.
    In each figure, the estimates (top rows) are compared with the corresponding ground-truth values computed under the data-generating mechanism (bottom rows). For the sake of space, in these plots, we focus only on the \textit{Bayes I} implementation. Across settings and across target quantities, a common qualitative pattern emerges: one-step-ahead discovery probabilities are often very small and, therefore, may be of limited practical use for planning purposes. As the future sample size increases, the discovery probability grows and eventually approaches one, thereby providing a concrete, interpretable notion of the sampling effort required to achieve a desired chance of observing something new. Moreover, for fixed future effort, discovery probabilities tend to decrease as the currently available sample size $n$ increases, reflecting diminishing returns: once many observations have already been collected, substantially larger additional samples are needed to reach the same probability of discovering new distinct or shared species.

    \begin{figure}[ht!]
        \centering
        \includegraphics[width=0.24\linewidth]{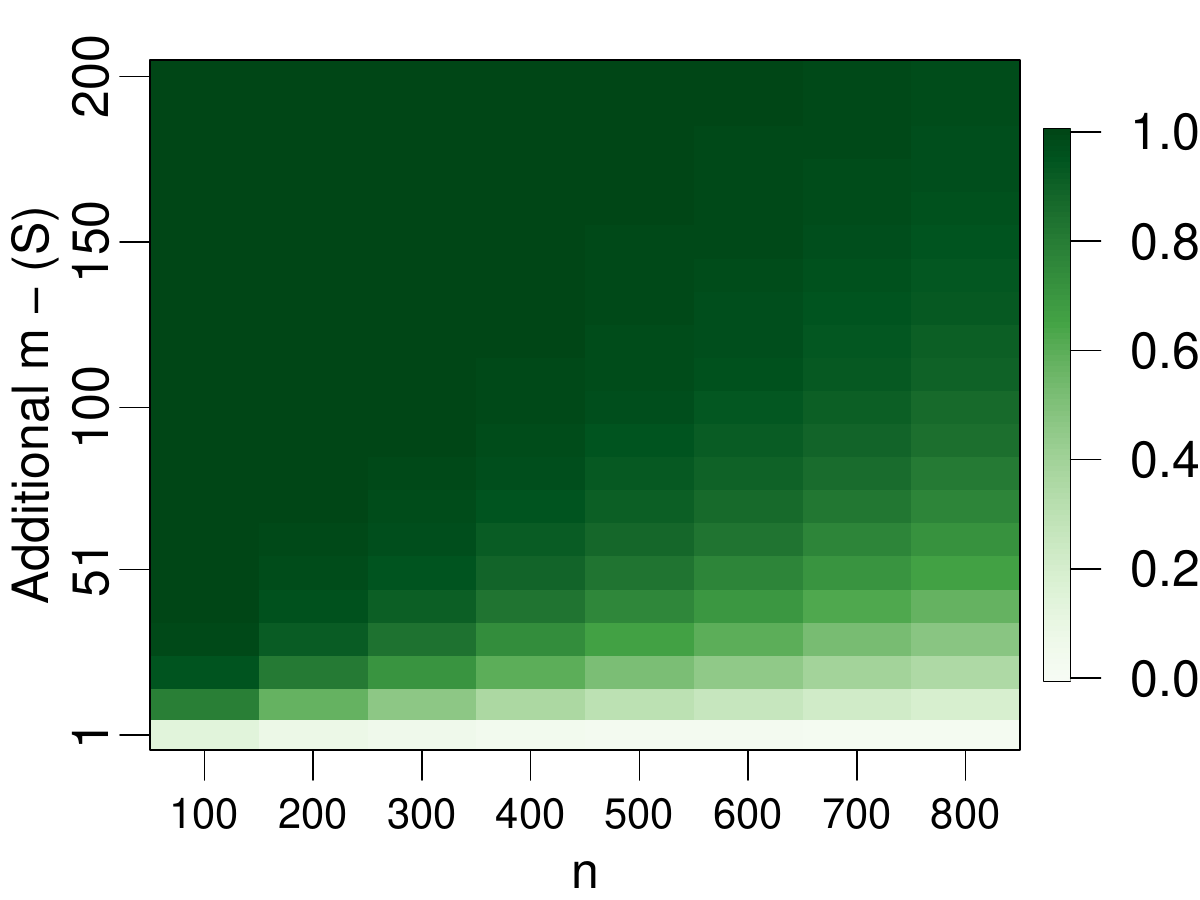}
        \hfill
        \includegraphics[width=0.24\linewidth]{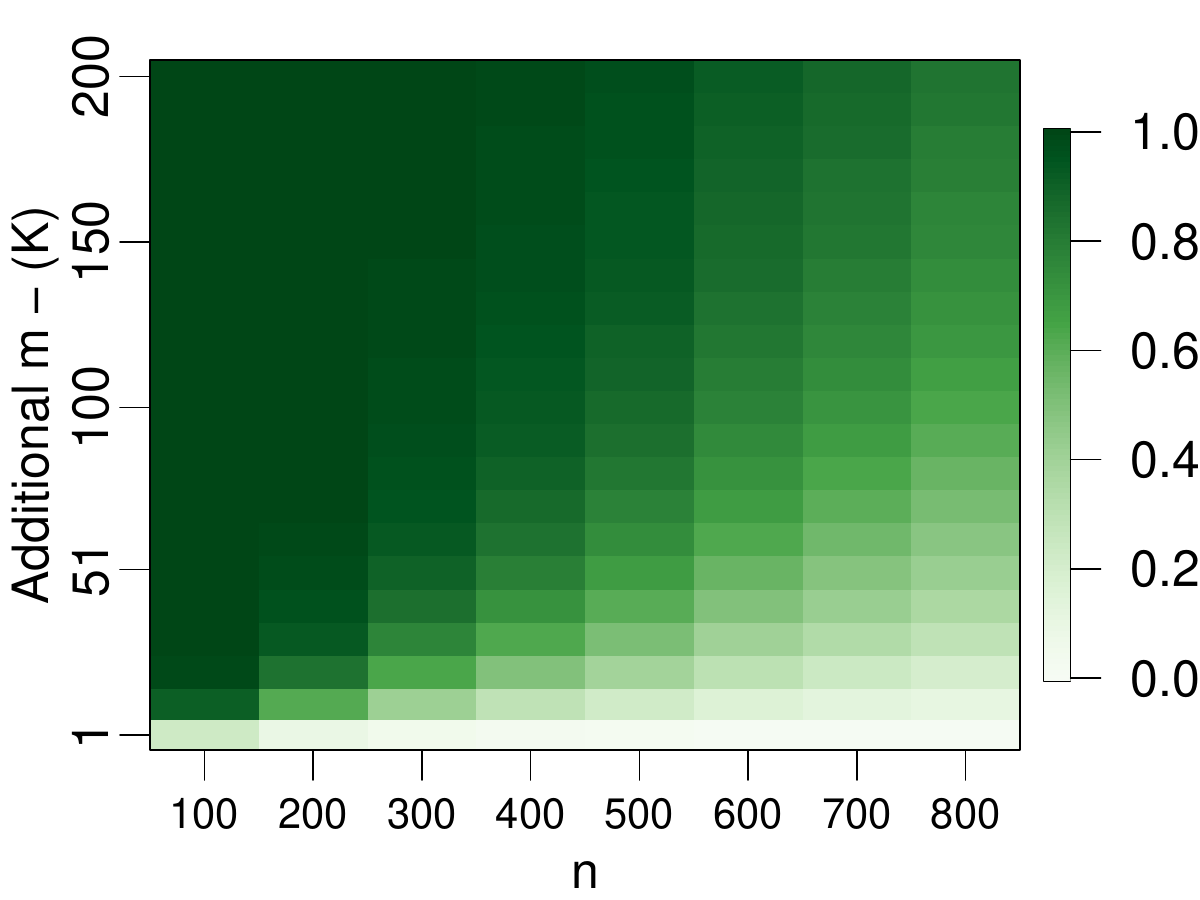}
        \hfill
        \includegraphics[width=0.24\linewidth]{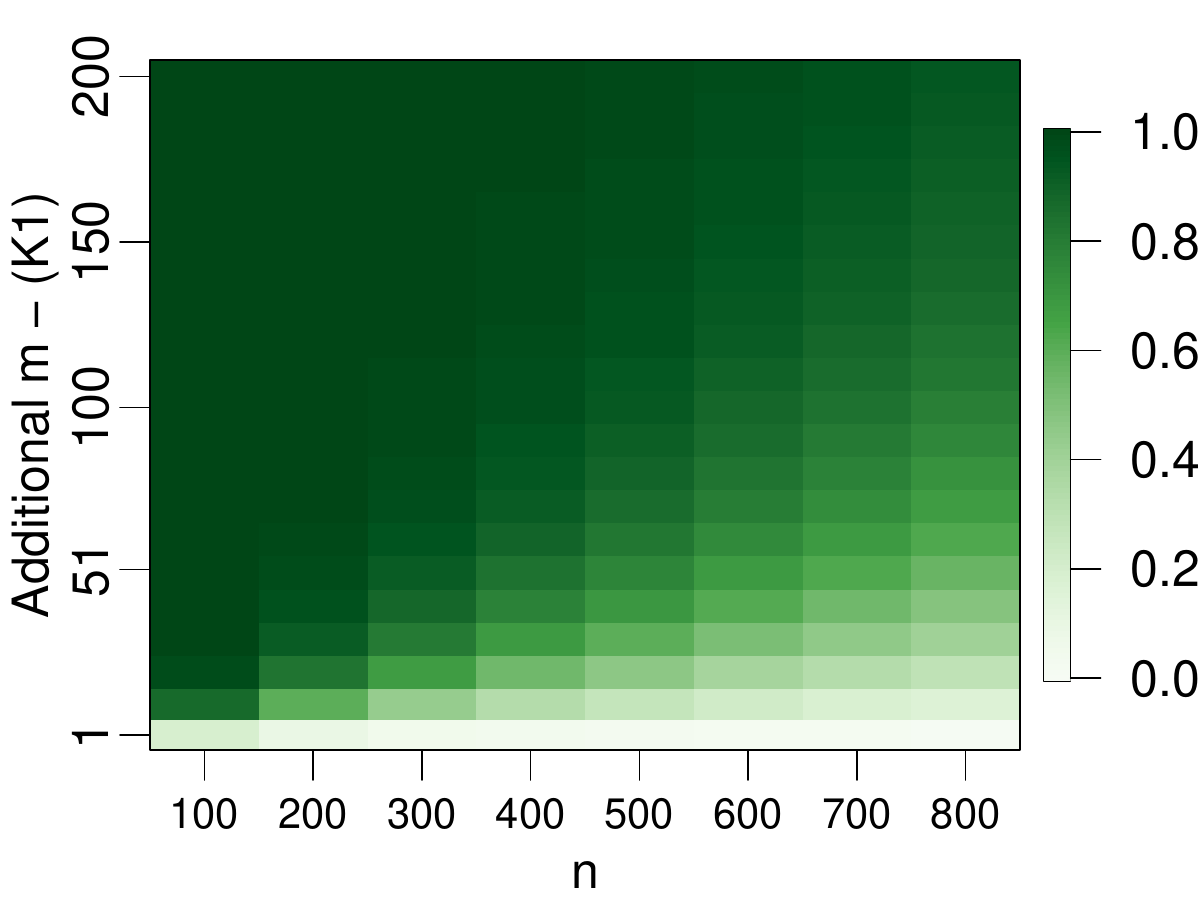}
        \hfill
        \includegraphics[width=0.24\linewidth]{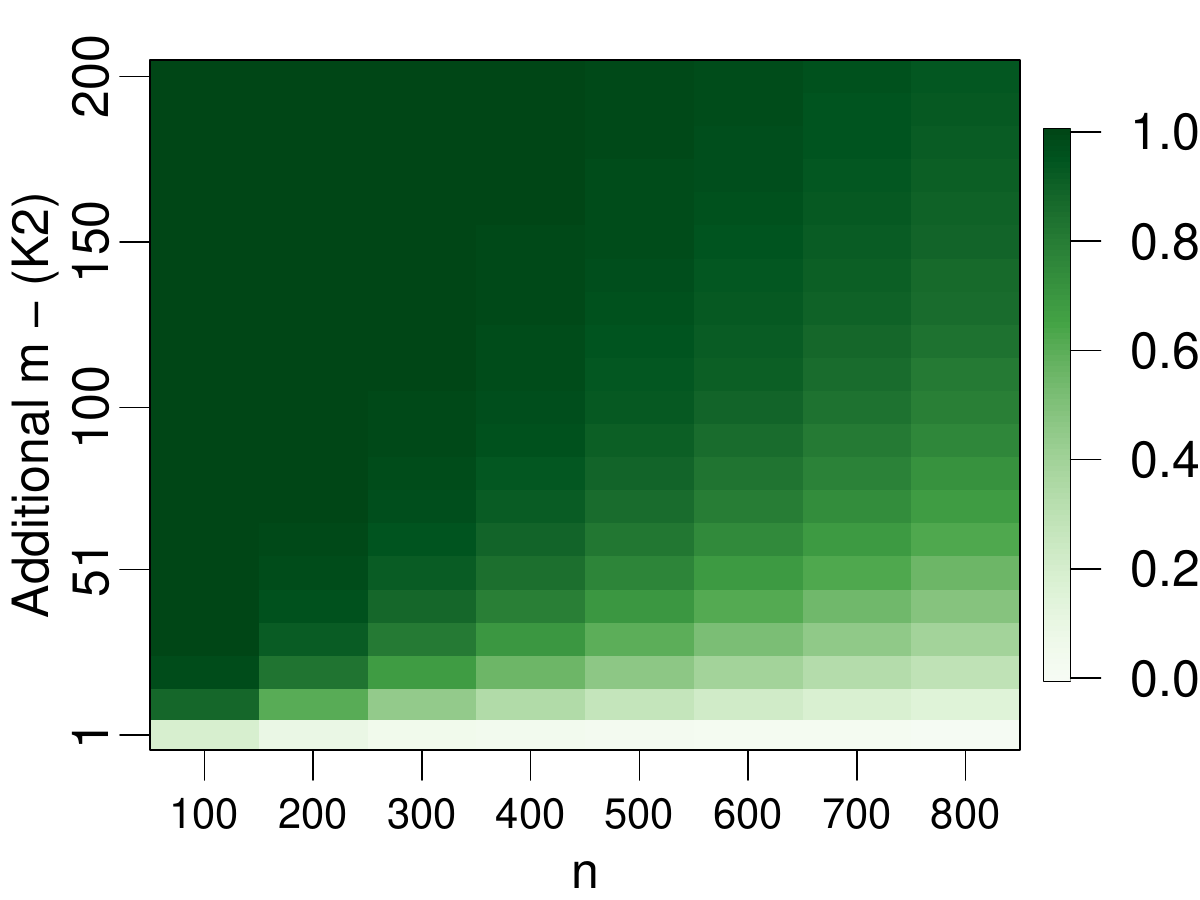}
        \\
        \includegraphics[width=0.24\linewidth]{img/Pr_m_steps_Geom_ShMLE_09_09.pdf}
        \hfill
        \includegraphics[width=0.24\linewidth]{img/Pr_m_steps_Geom_KMLE_09_09.pdf}
        \hfill
        \includegraphics[width=0.24\linewidth]{img/Pr_m_steps_Geom_K1MLE_09_09.pdf}
        \hfill
        \includegraphics[width=0.24\linewidth]{img/Pr_m_steps_Geom_K2MLE_09_09.pdf}
        \caption{Experiment 3 - $G_6$ case: $m$-steps-ahead prediction probabilities for new shared species (first column), new global distinct species (second column), and new local distinct species (third column) and (fourth column) under $G_6$ setting. 
        Top row shows the estimated probabilities via \textit{Bayes I} estimator while the bottom row represents the oracle estimators. }
        \label{fig:SS_Pr_m_step_Geom}
    \end{figure}
    \begin{figure}[ht!]
        \centering
        \includegraphics[width=0.24\linewidth]{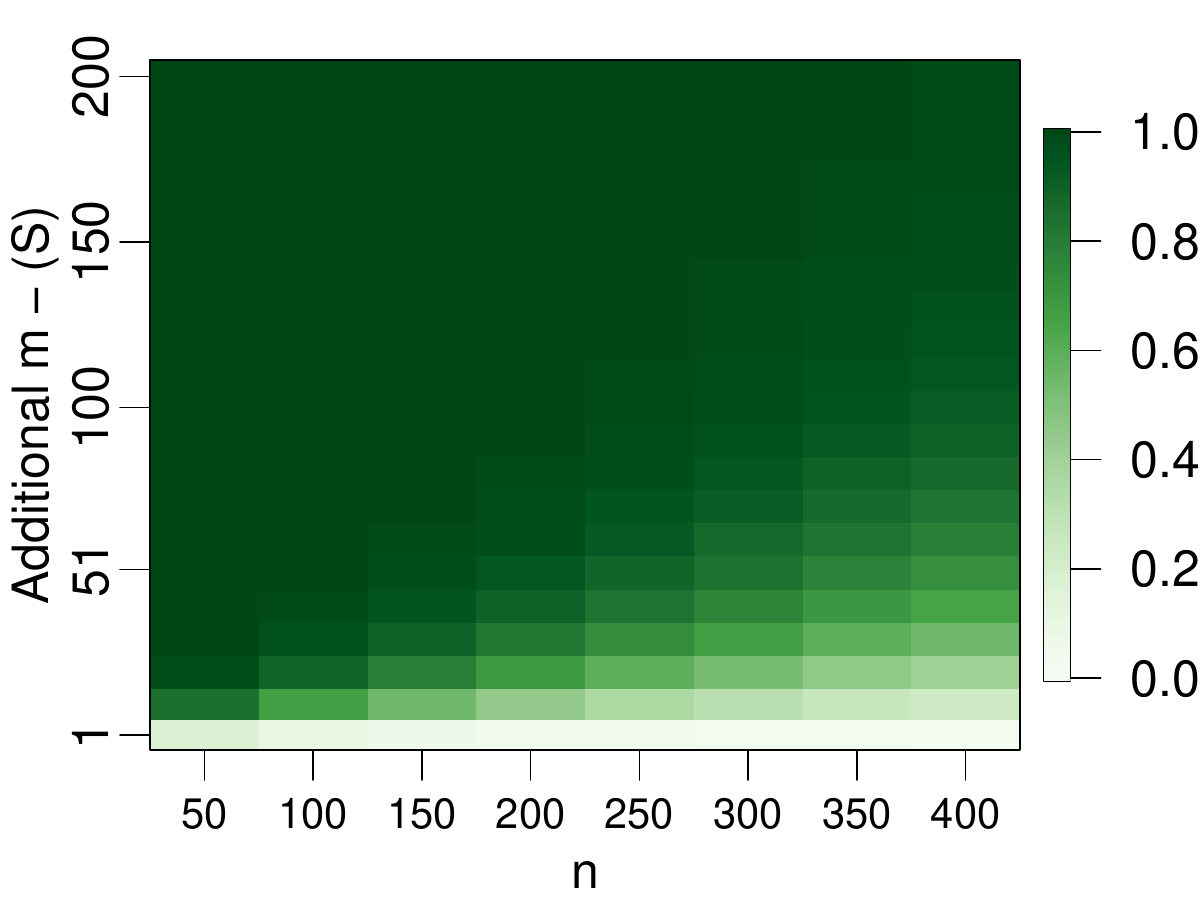}
        \hfill
        \includegraphics[width=0.24\linewidth]{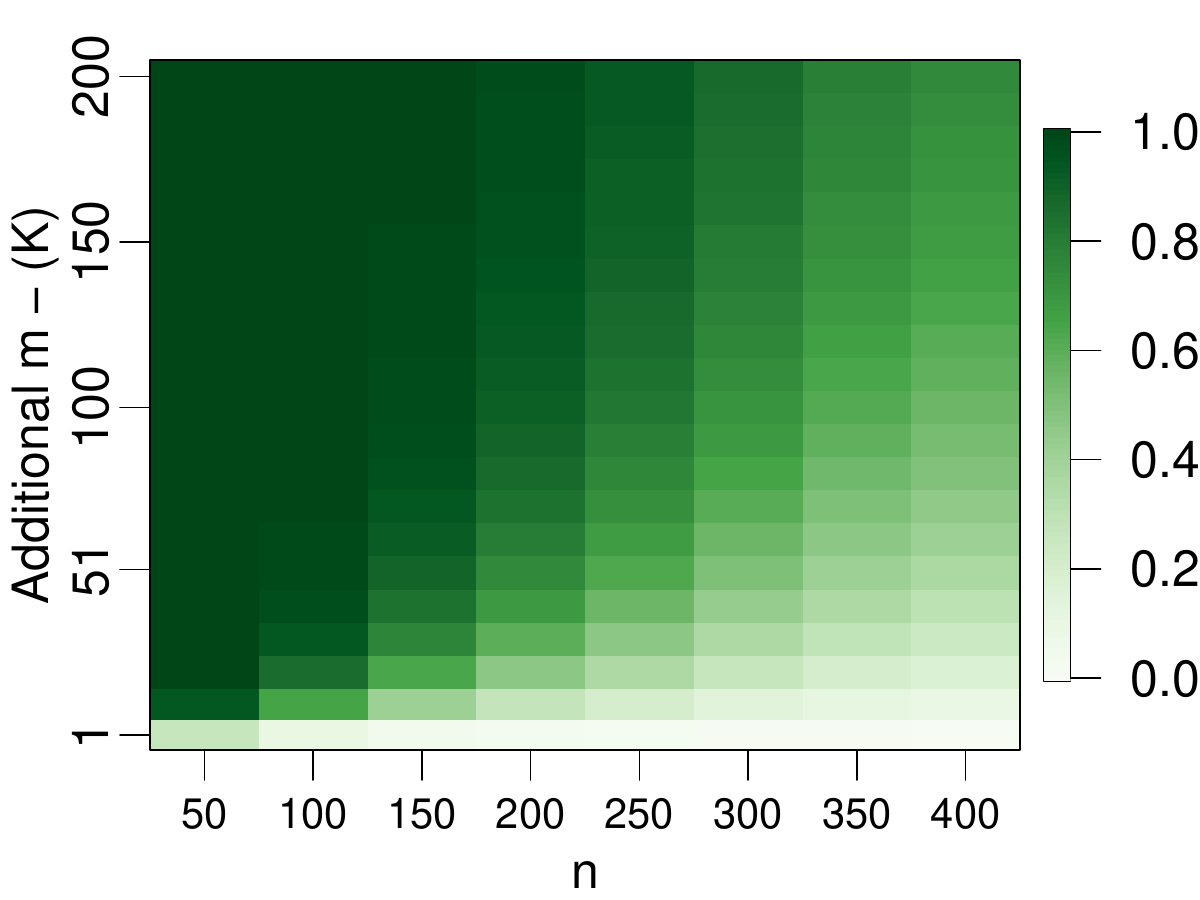}
        \hfill
        \includegraphics[width=0.24\linewidth]{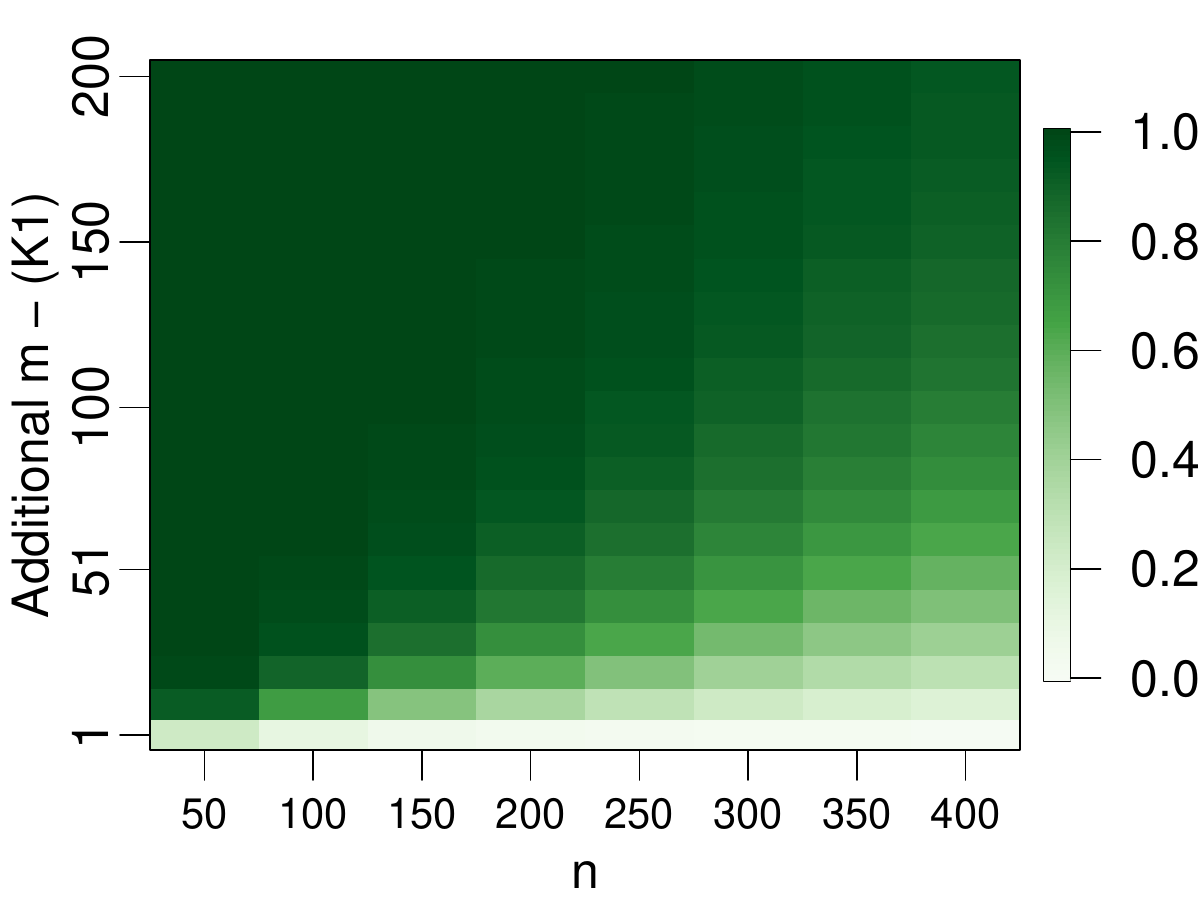}
        \hfill
        \includegraphics[width=0.24\linewidth]{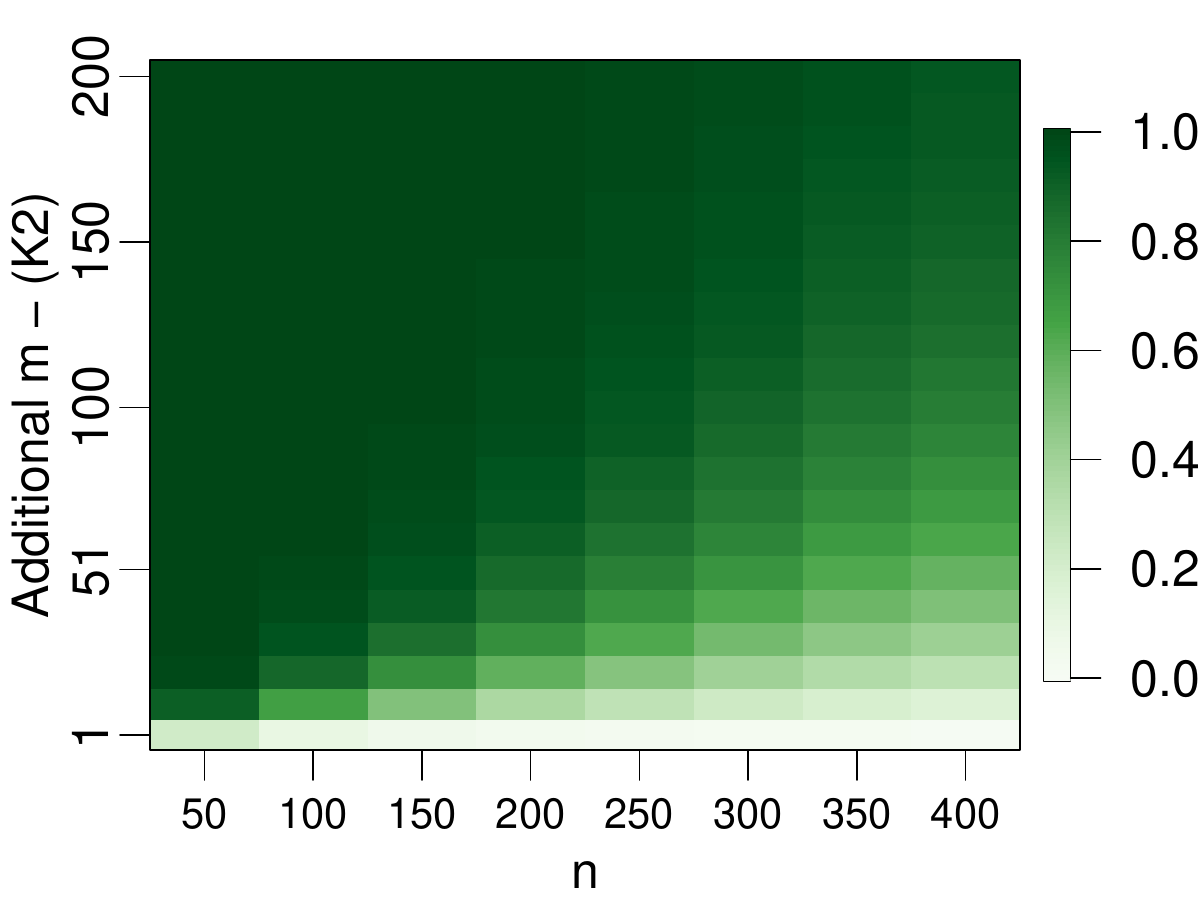}
        \\
        \includegraphics[width=0.24\linewidth]{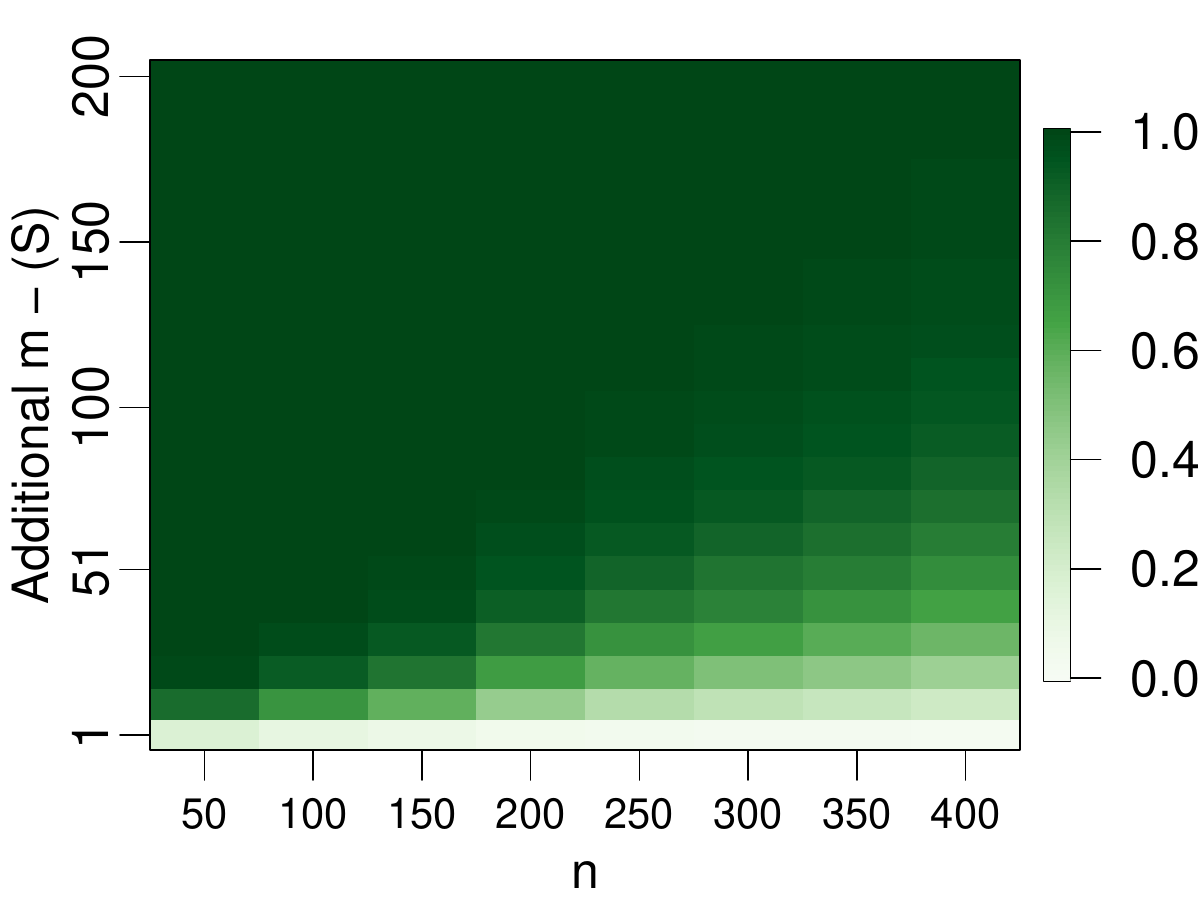}
        \hfill
        \includegraphics[width=0.24\linewidth]{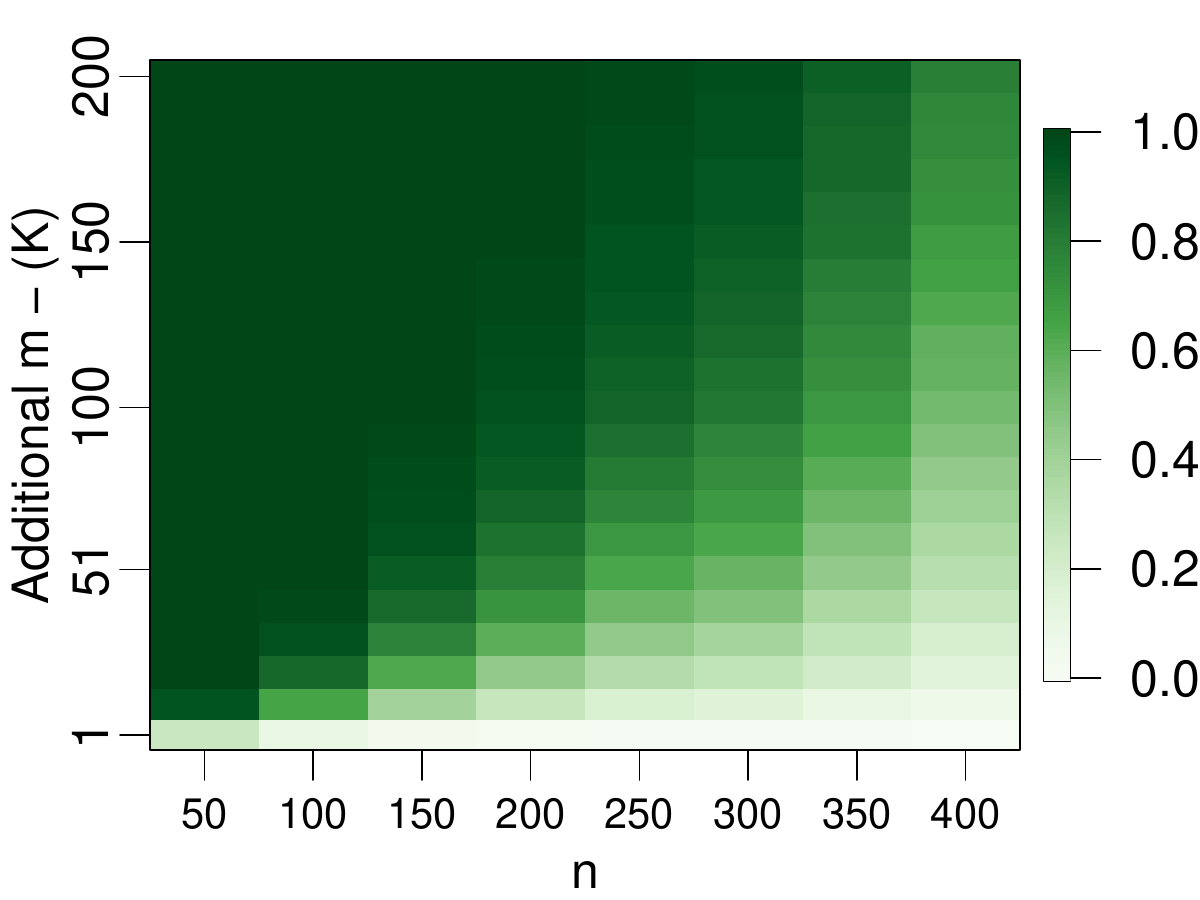}
        \hfill
        \includegraphics[width=0.24\linewidth]{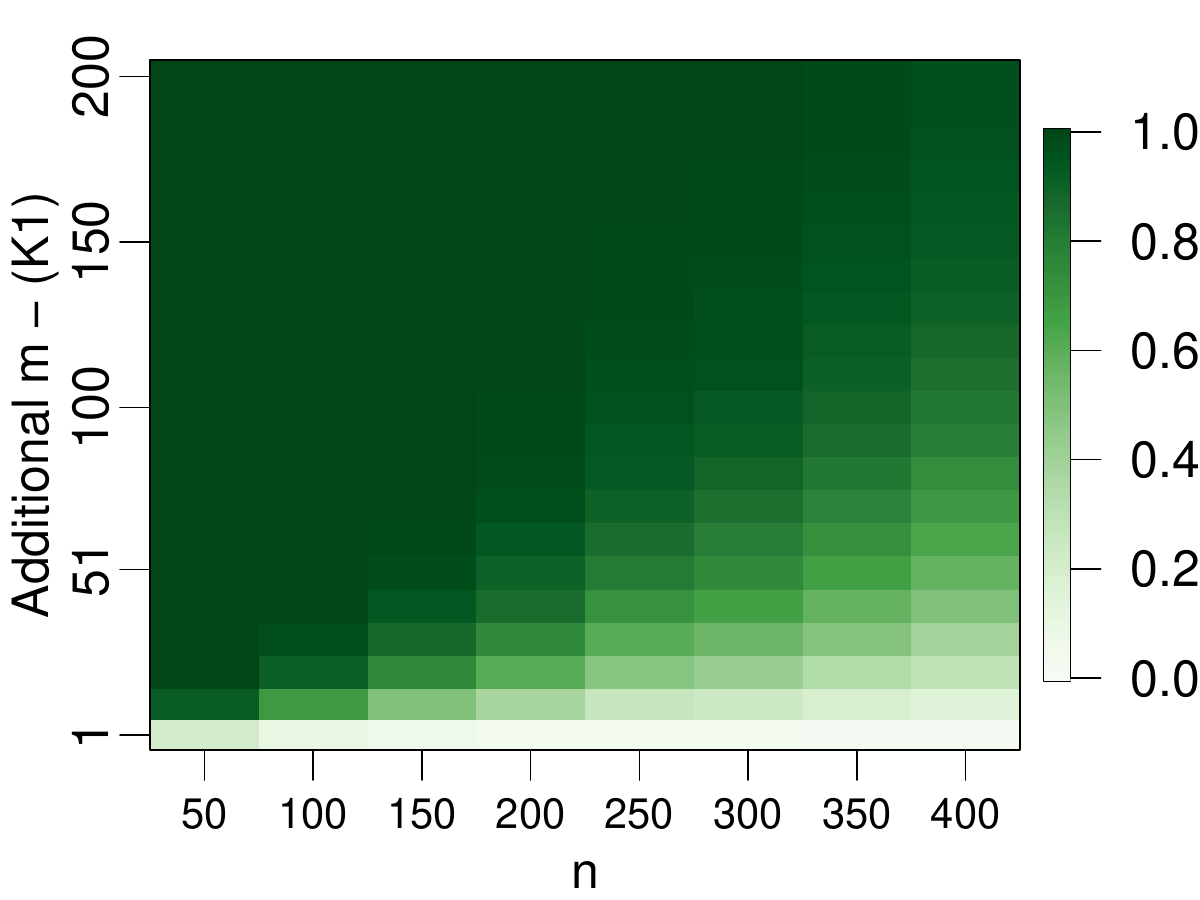}
        \hfill
        \includegraphics[width=0.24\linewidth]{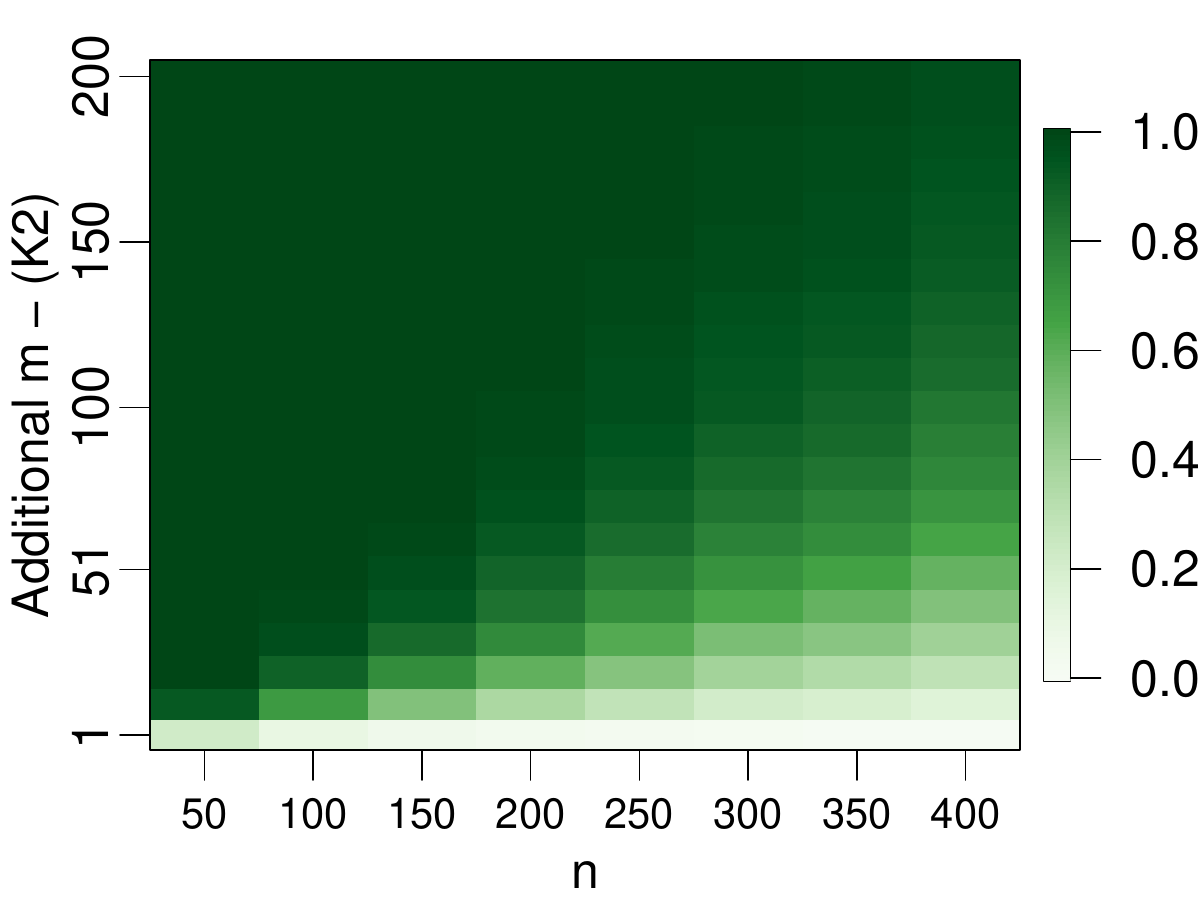}
        \caption{Experiment 3 - $D_3$ case: $m$-steps-ahead prediction probabilities.
        Top row shows the estimated probabilities via \textit{Bayes I} estimator while the bottom row represents the oracle estimators. }
        \label{fig:SS_Pr_m_step_Dir}
    \end{figure}
    \begin{figure}[ht!]
        \centering
        \includegraphics[width=0.24\linewidth]{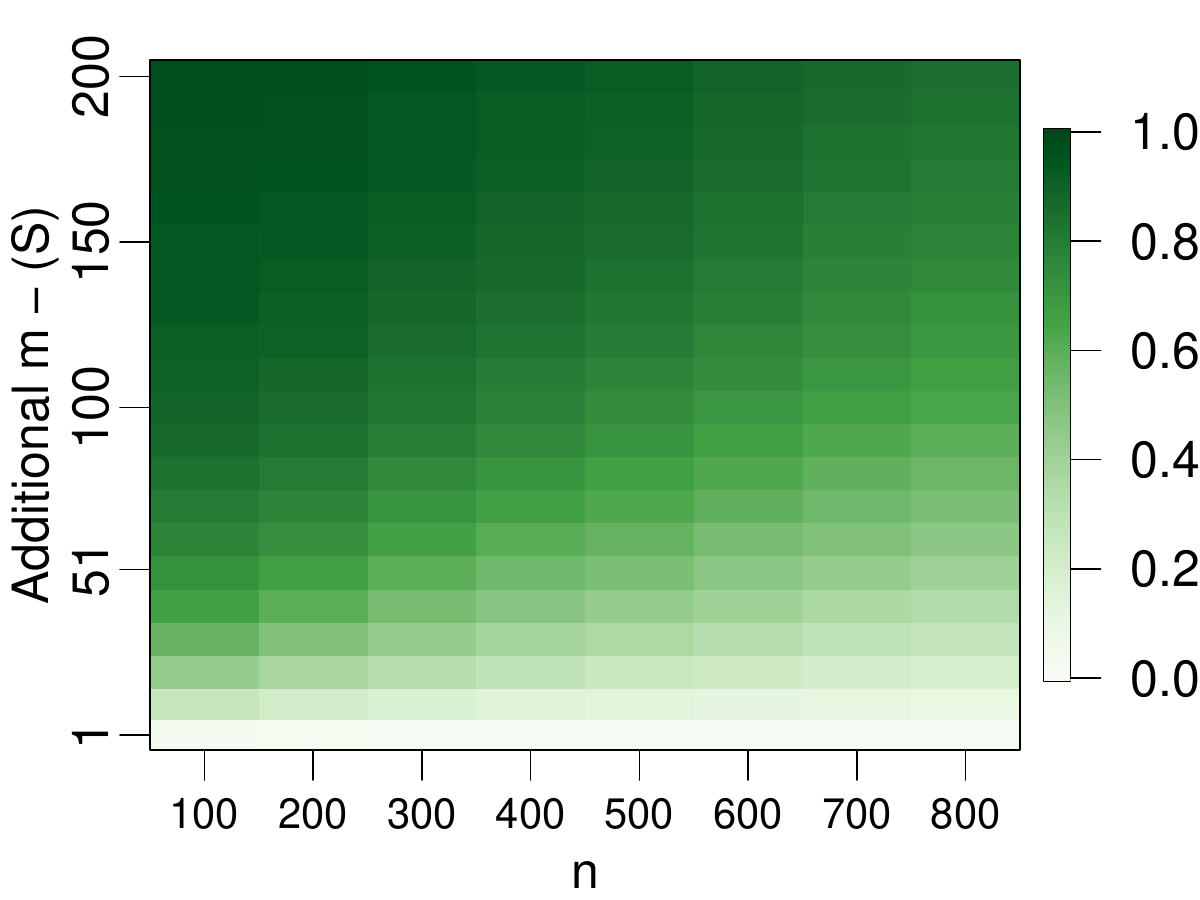}
        \hfill
        \includegraphics[width=0.24\linewidth]{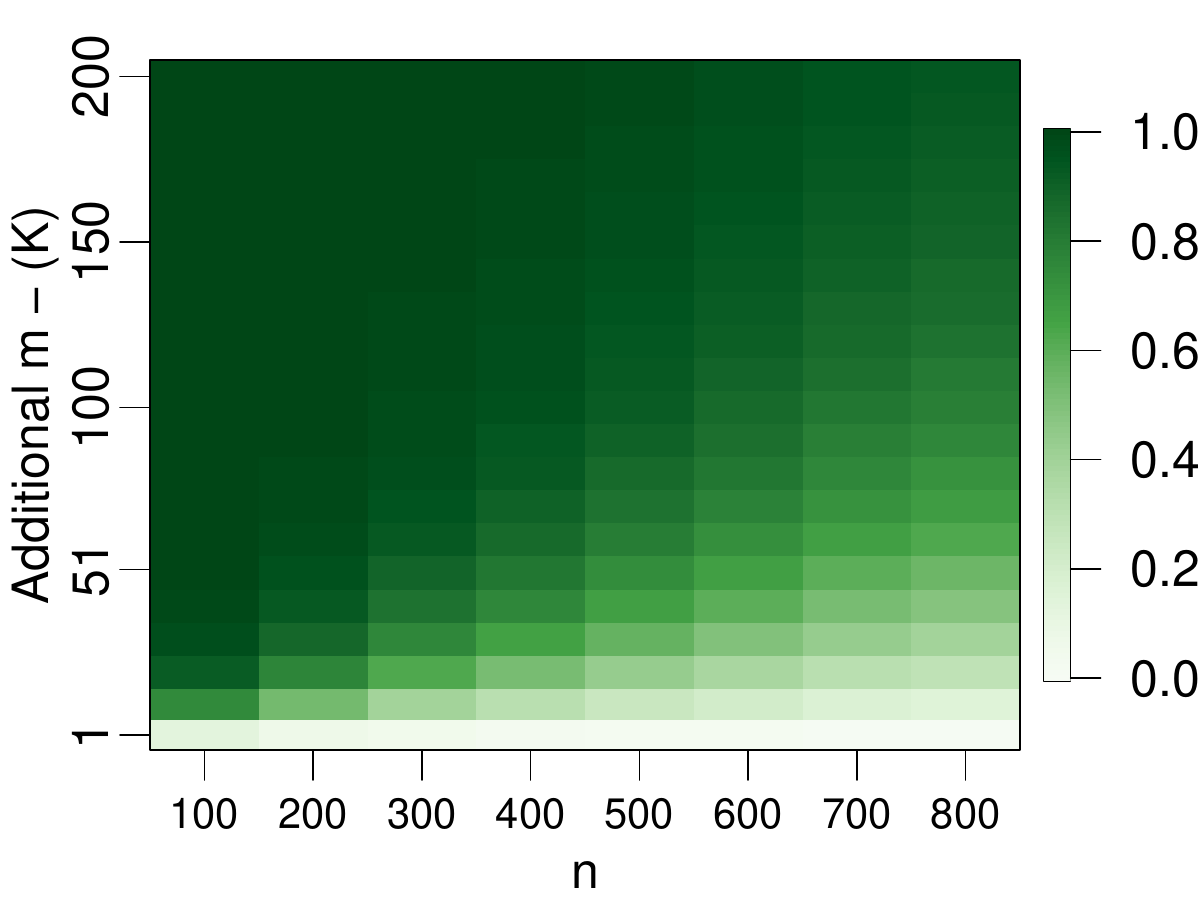}
        \hfill
        \includegraphics[width=0.24\linewidth]{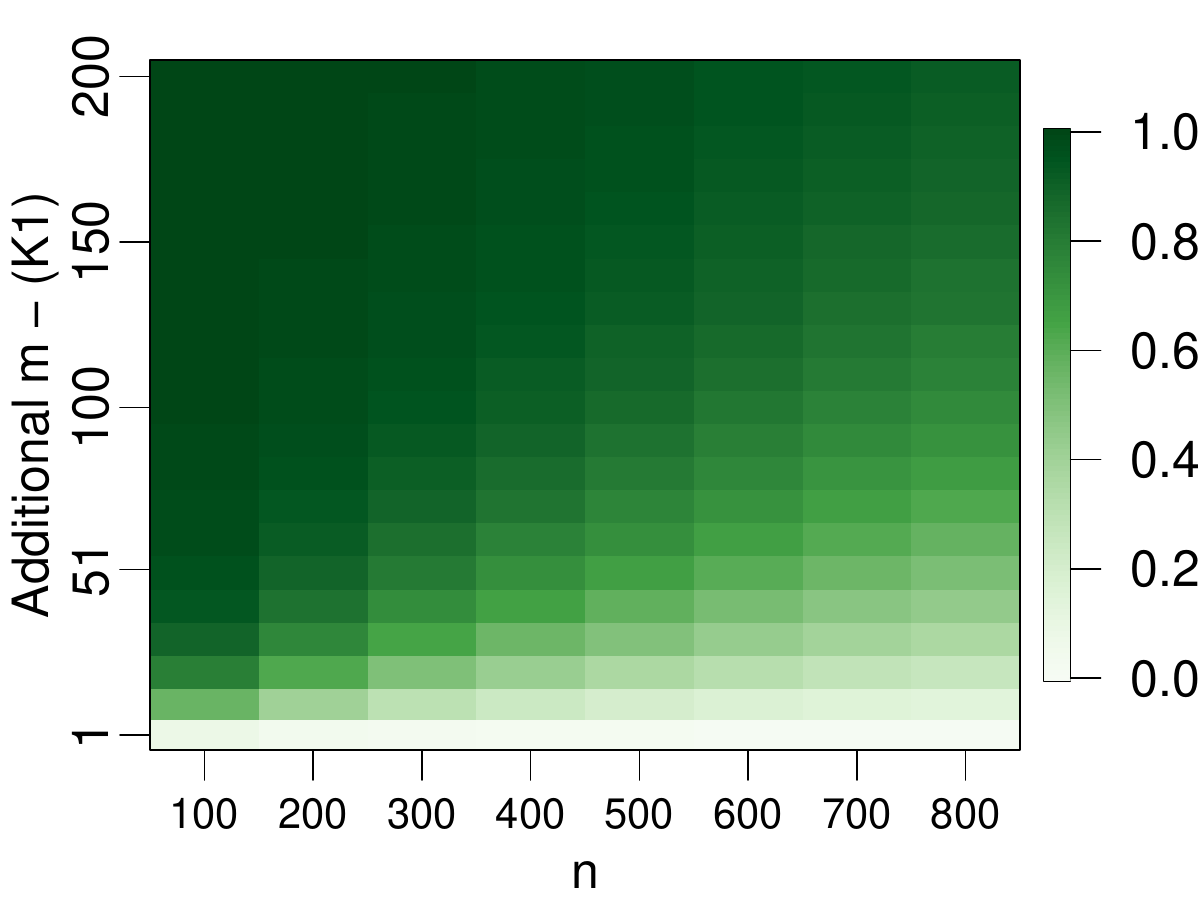}
        \hfill
        \includegraphics[width=0.24\linewidth]{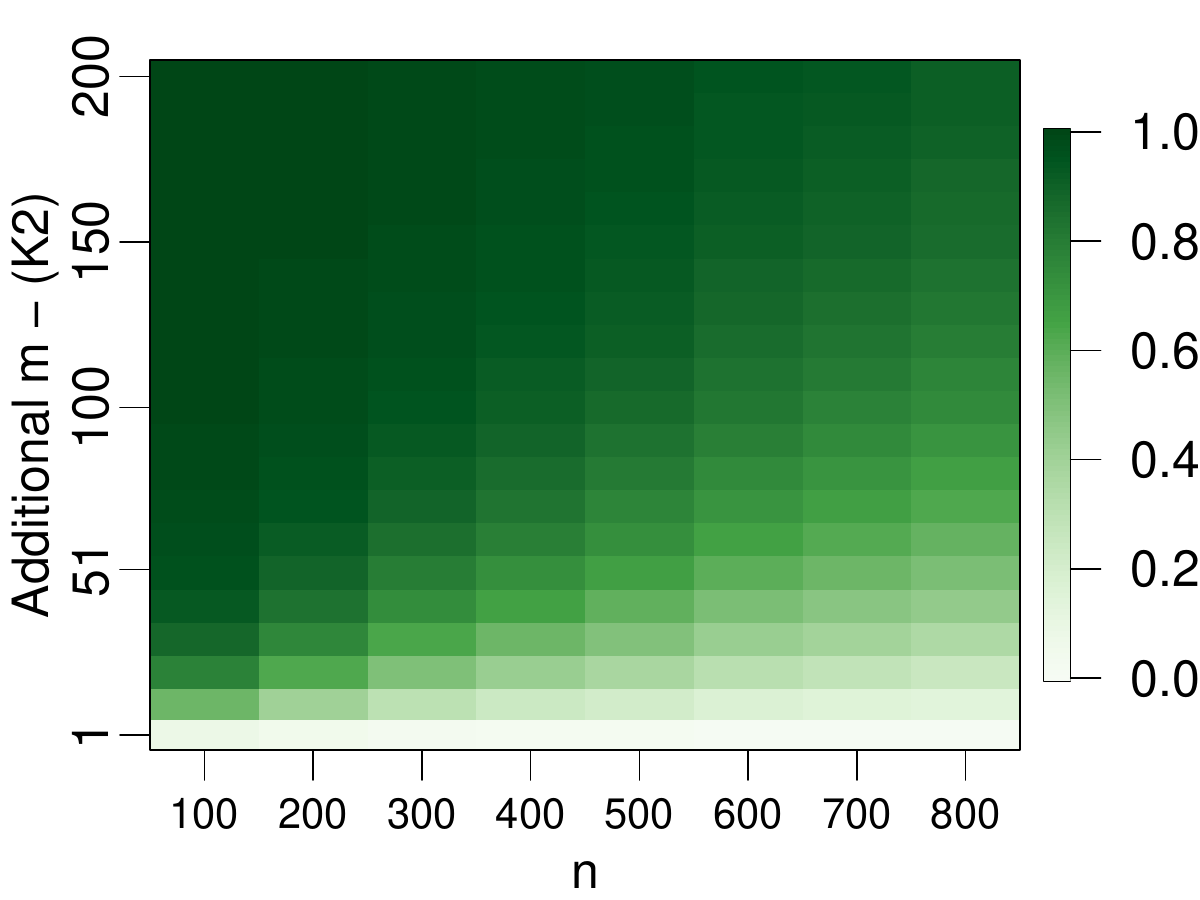}
        \\
        \includegraphics[width=0.24\linewidth]{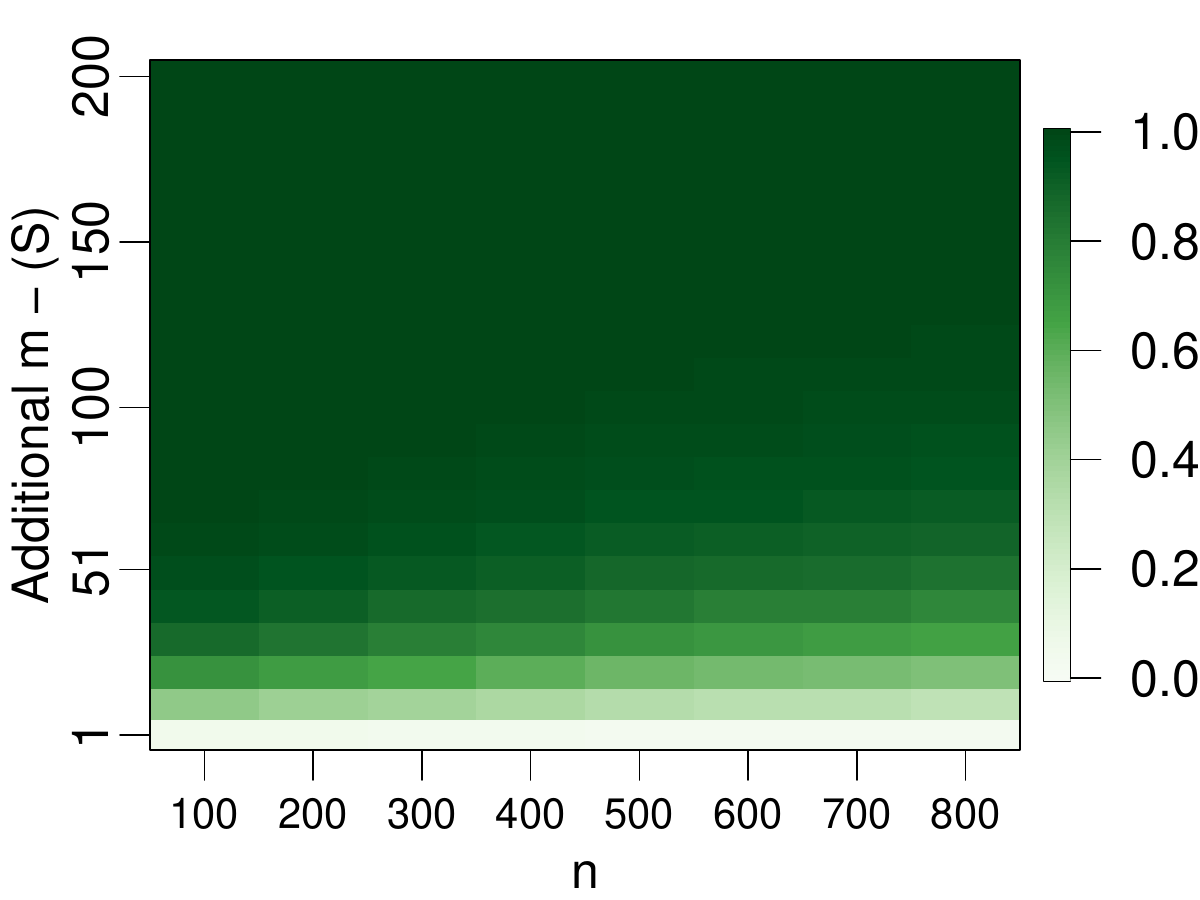}
        \hfill
        \includegraphics[width=0.24\linewidth]{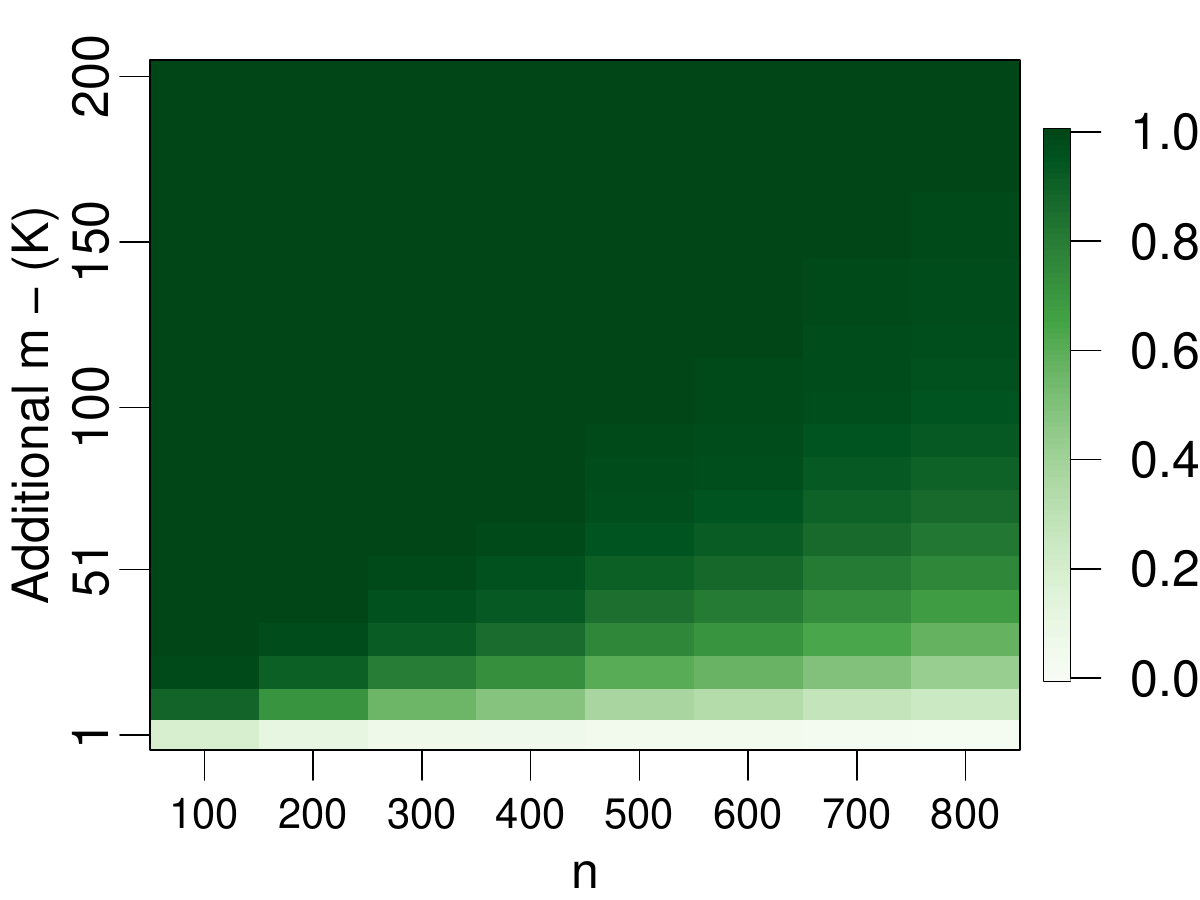}
        \hfill
        \includegraphics[width=0.24\linewidth]{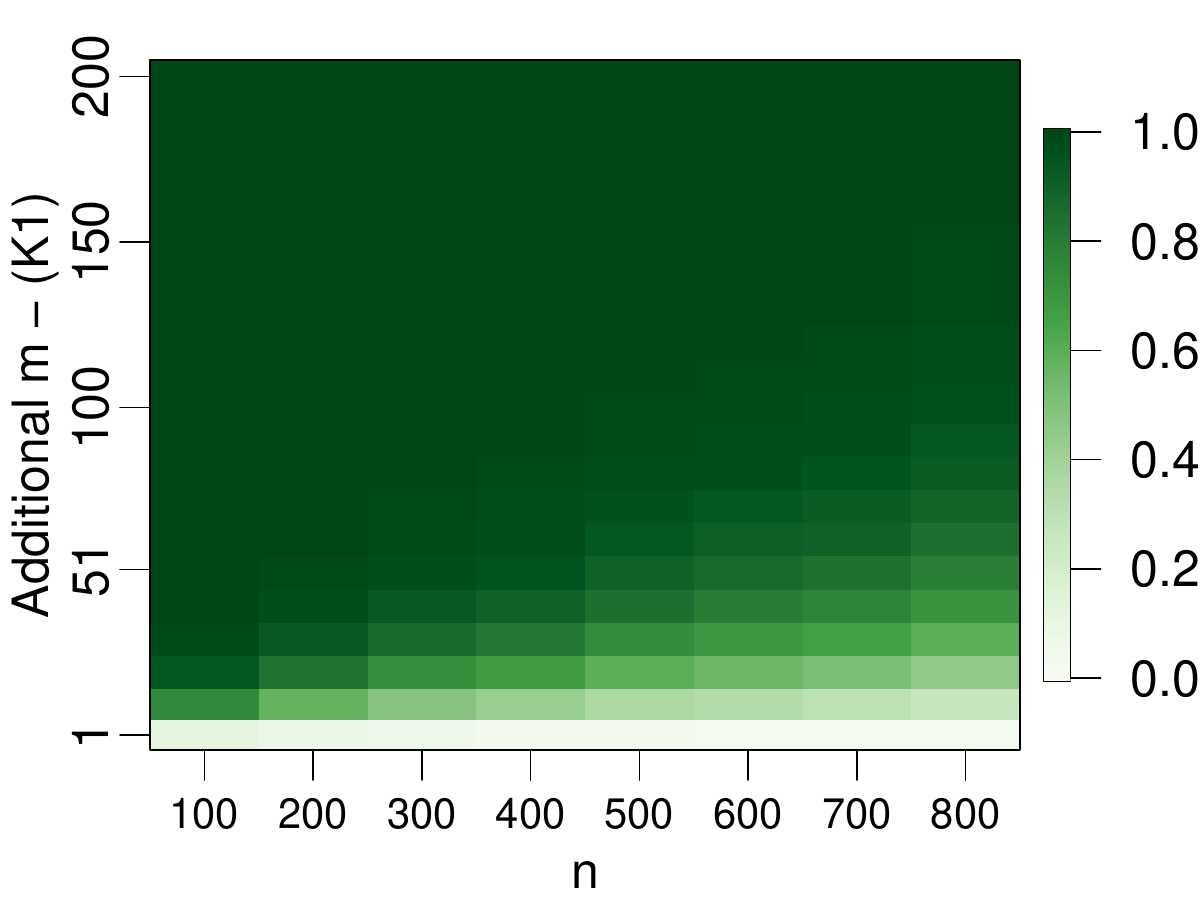}
        \hfill
        \includegraphics[width=0.24\linewidth]{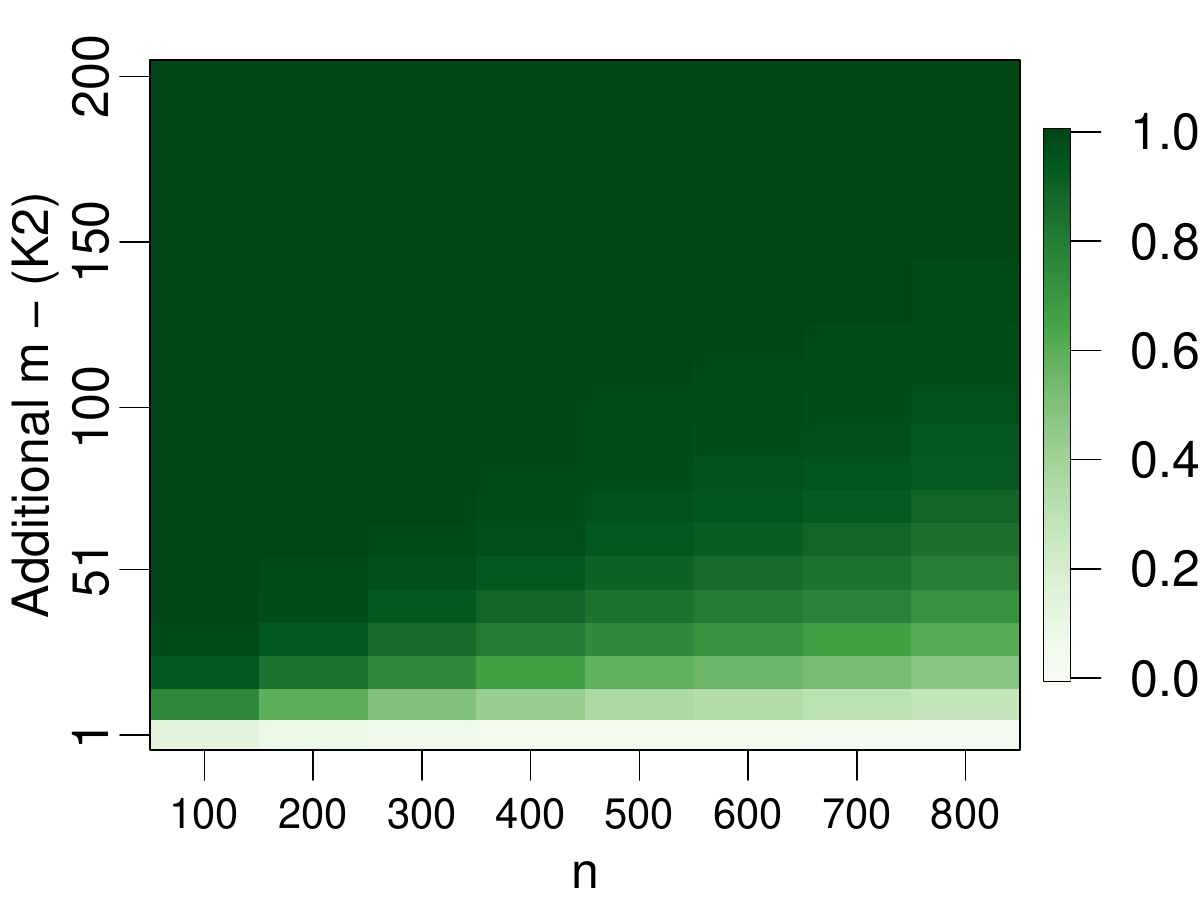}
        \caption{Experiment 3 - $Z_3$ case: $m$-steps-ahead prediction probabilities.
        Top row shows the estimated probabilities via \textit{Bayes I} estimator while the bottom row represents the oracle estimators. }
        \label{fig:SS_Pr_m_step_Zipfs}
    \end{figure}

    \subsection{Experiment 4}
    \label{subsection:SS_exp4}
    This experiment is intended to highlight analogies and differences between our estimator in Equation \eqref{eqn:shared_estimator}, which predicts the number of additional shared species in a future, unobserved test set, and the state-of-the-art estimator for shared species richness introduced by \cite{Chao2000}, as implemented in the \texttt{R} package \texttt{SpadeR} \citep{SpadeR}. This estimator, hereafter denoted by \textit{Chao2000}, targets the total number of shared species across two groups.
    Consequently, the goal of this experiment is not to compare the two methods in terms of performance, but rather to clarify their different modelling assumptions and the interpretation of their outputs. Although subsequent works have extended the \textit{Chao2000} estimator (see, among others, \cite{ChaoLB2009,Chuang2015,chao2017}), we focus on \cite{Chao2000} because of its foundational role in the literature as the first proposal of this kind. Finally, since the purpose is conceptual comparison rather than efficiency benchmarking, and for graphical simplicity, we compare \textit{Chao2000} only with the \textit{Bayes II} estimator; the same conclusions also apply to \textit{Bayes I}.
    
    The experiment proceeds as follows. We generate data from settings $(G_2,G_5,G_6)$, as defined in Section \ref{section:simulation_summary}, with $n_1=n_2=400$. We then consider a grid of training set percentages ranging from $0.1$ to $0.9$. For each percentage, we construct the training set as the corresponding fraction of the full dataset, and use the remainder as a test set. We compute the reference quantity $S_{\text{true}}$, defined as the number of shared species observed in the full dataset (training plus test), which serves as the benchmark for our comparisons.
    The training set is used to estimate the model parameters as described in Section \ref{subsection:param_estimation} and to compute the observed number of shared species, $S_{\text{obs}}$. Given the estimated parameters, we predict the expected number of new shared species in the test set, $S_{\text{est}}$, as explained in Section \ref{section:posterior}. We then compare $S_{\text{true}}$ with the predicted total $S_{\text{obs}} + S_{\text{est}}$. For each percentage value, we replicate the procedure over $100$ independently generated training--test splits.
    
    The results of the experiment, shown in Figure \ref{fig:PredTestSet_exp2}, reveal substantial differences between our model and \textit{Chao2000}. Our model is based on the assumption that the two groups are generated from a process that, under infinite sampling, would produce the same set of species in both groups, albeit with different proportions. In contrast, \cite{Chao2000} adopt a different modelling perspective. For each group $j$, they allow some of the probabilities $w_{j,m}$ in Equation \eqref{eqn:Pj_def} to be exactly zero for some $m$. This implies that, even with infinite sampling, some species present in one group may never appear in the other group. Equivalently, this corresponds to assuming that the total number of species can differ across areas, which ultimately permits inference on the total number of shared species.
    Figure \ref{fig:PredTestSet_exp2} shows that our \textit{Bayes II} estimator converges to the true value $S_{\text{true}}$ as the training set percentage increases, in all the reported cases. On the other hand, \textit{Chao2000} tends to stabilize at slightly higher values than $S_{\text{true}}$, showing a clear discrepancy between the dashed line in Figure \ref{fig:PredTestSet_exp2} and the final prediction of \textit{Chao2000}, i.e., when the training set comprises nearly the entire dataset. This gap corresponds to the estimate of shared species that are assumed to exist in the population but are not observed in the available dataset.
    
    As the results show, \textit{Chao2000} is inferential in nature: it aims to estimate an unknown population-level quantity (shared species richness) and relies on asymptotic assumptions, such as having sufficiently large sample sizes. In contrast, our approach is predictive: it is applicable for any observed sample sizes $n_1$ and $n_2$ and for any future sample sizes $m_1$ and $m_2$, although we expect some bias for large values of $m_1$ and $m_2$ when the modelling assumptions are not matched in practice. In other words, our model answers the question, ``How many species that we have not yet observed will we find in the future?'', whereas \textit{Chao2000} answers the question, ``How many species that we have not yet observed actually exist?'', even though, under its assumptions, some of these species may never be encountered.

	\begin{figure}[ht!]
		\centering 
		\includegraphics[width=0.32\linewidth]{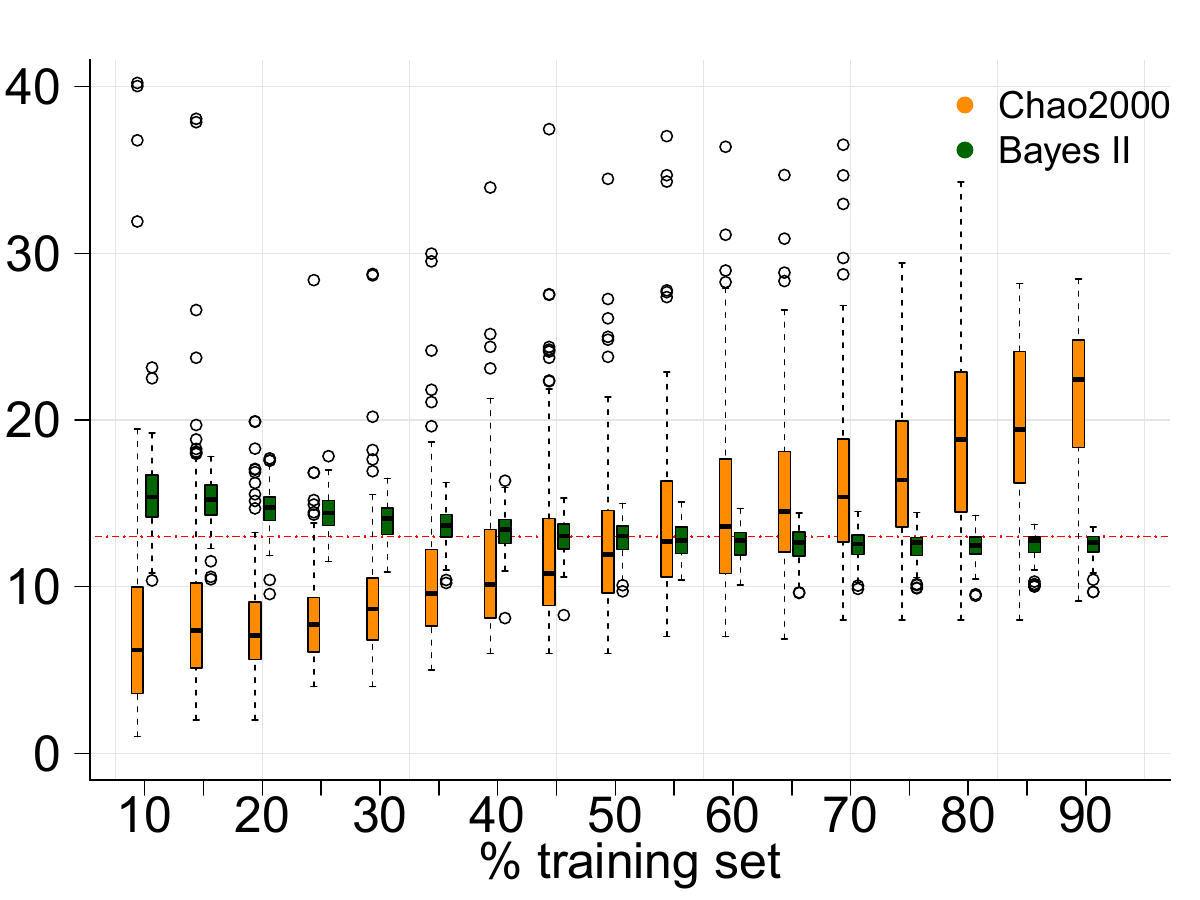}
		\hfill
		\includegraphics[width=0.32\linewidth]{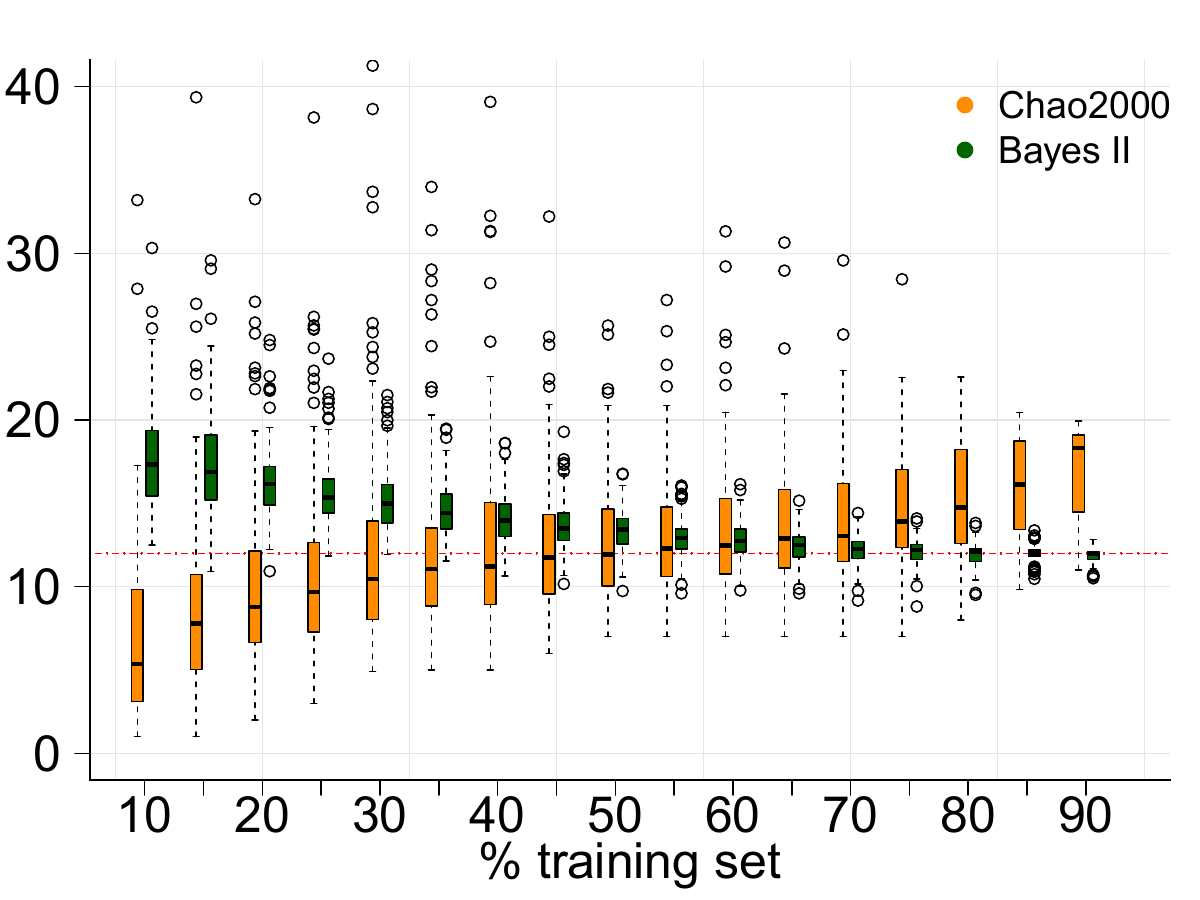}
		\hfill
		\includegraphics[width=0.32\linewidth]{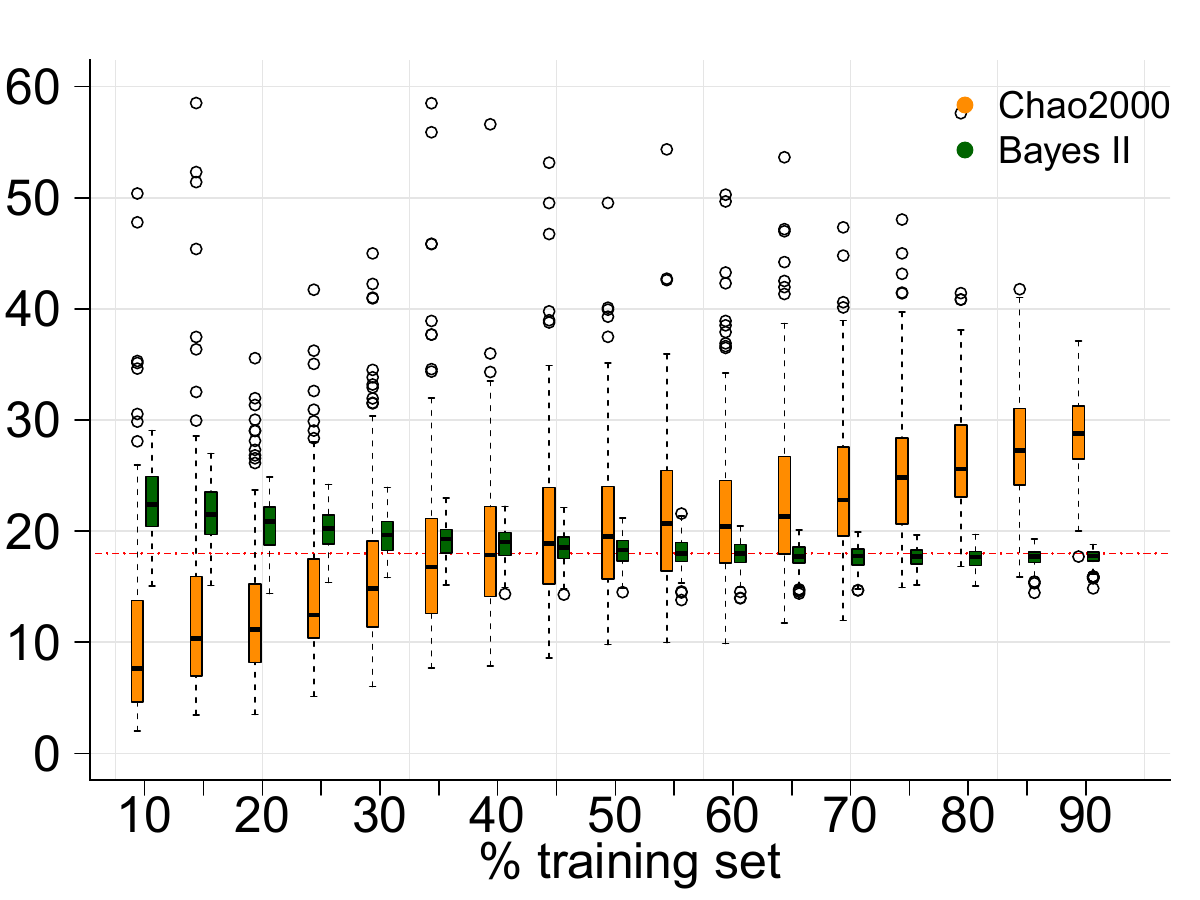}
		\caption{Experiment 4: predicted number of shared species for different training set percentages. The red line represents $S_{\text{true}}$.}
		\label{fig:PredTestSet_exp2}
	\end{figure}

	\section{Additional details on the analysis of ants data}
	\label{app:application}
	Figure \ref{fig:Ants_proportions_empirical} displays the observed species proportions in the ants dataset, sorted in decreasing order while Figure \ref{fig:Ants_AccCrvs} shows the accumulation curves of the data.
	\begin{figure}
		\centering
		\includegraphics[width=1\linewidth]{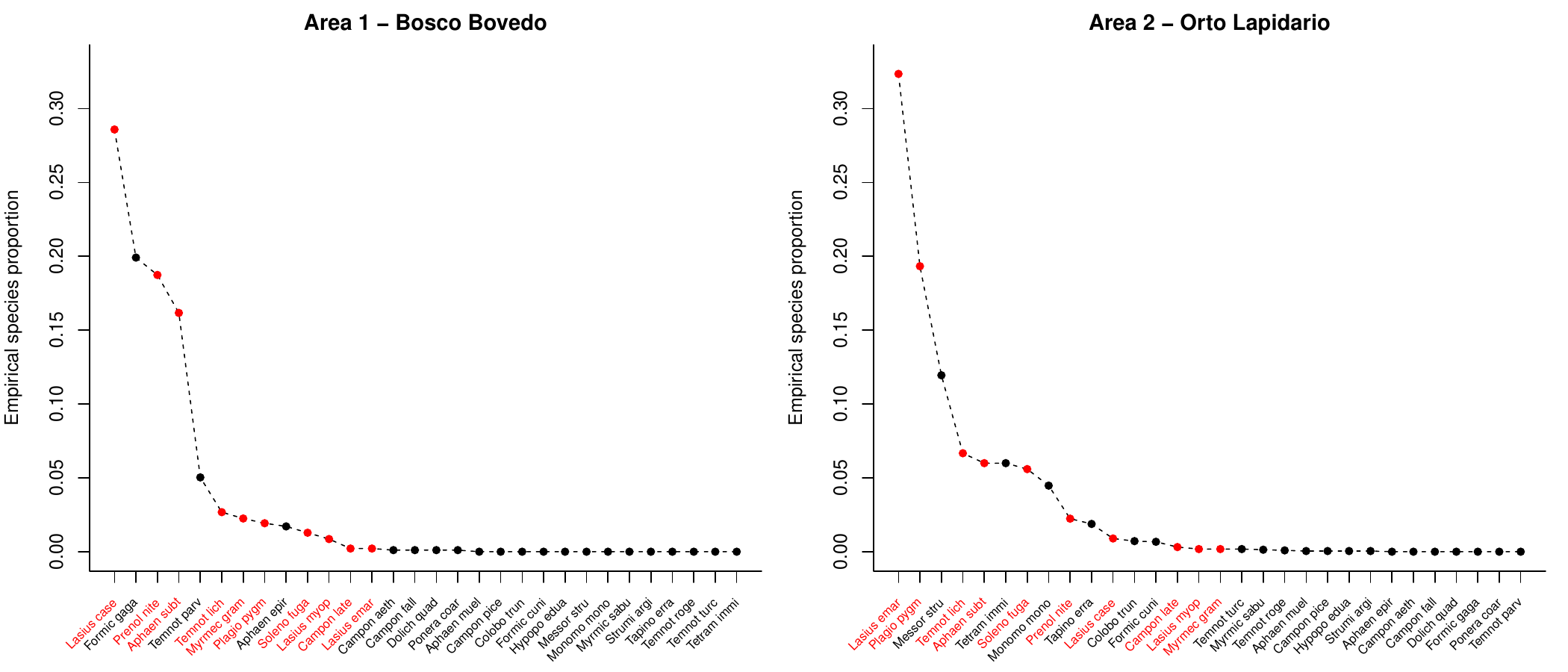}
		\caption{Graphical representations of observed species proportions in the dataset. Species have been sorted, within each group, in decreasing order. Red points and names represent shared species. }
		\label{fig:Ants_proportions_empirical}
	\end{figure}

    \begin{figure}[ht!]
        \centering
        \includegraphics[width=0.48\linewidth]{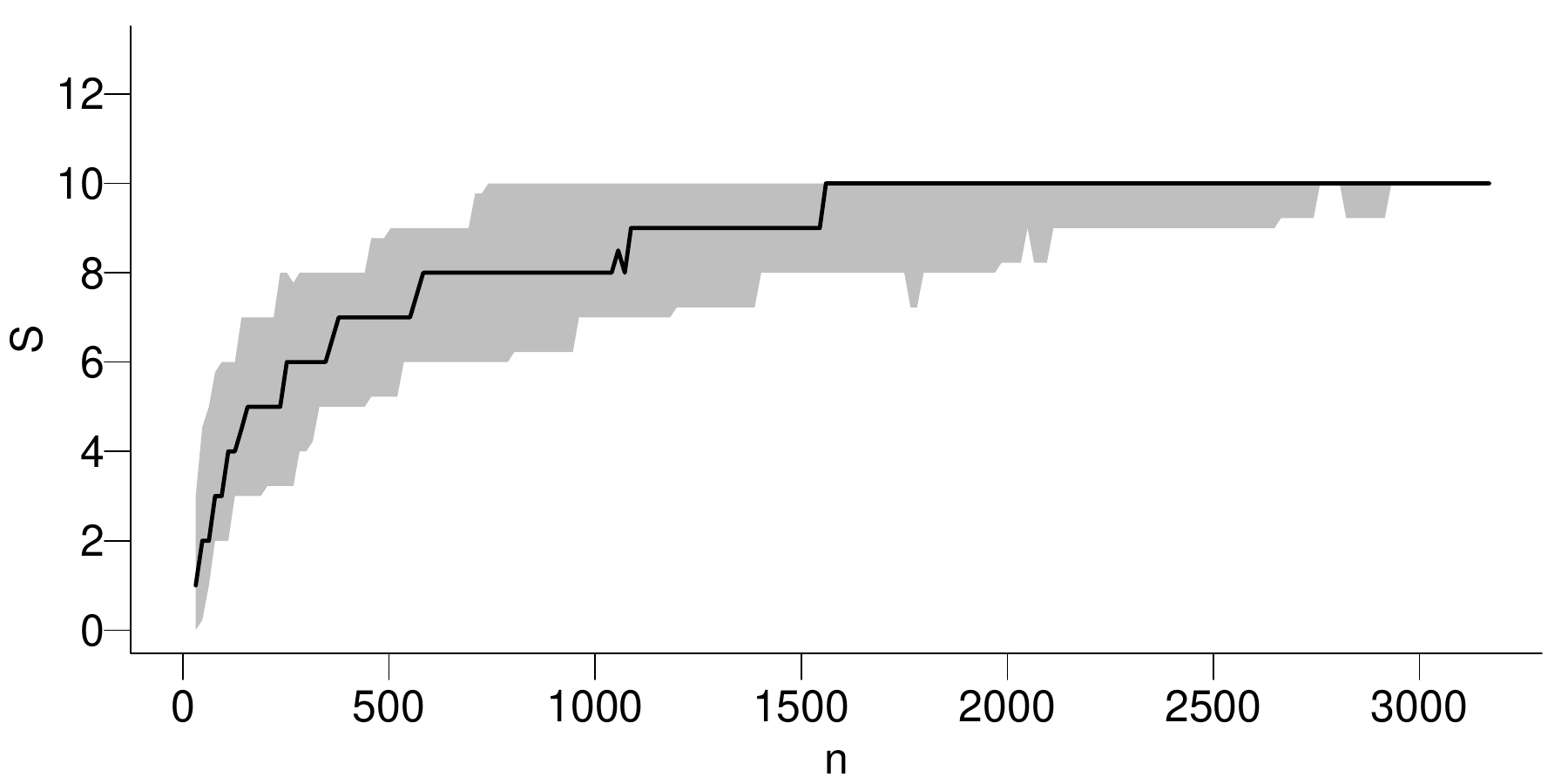}
        \hfill
        \includegraphics[width=0.48\linewidth]{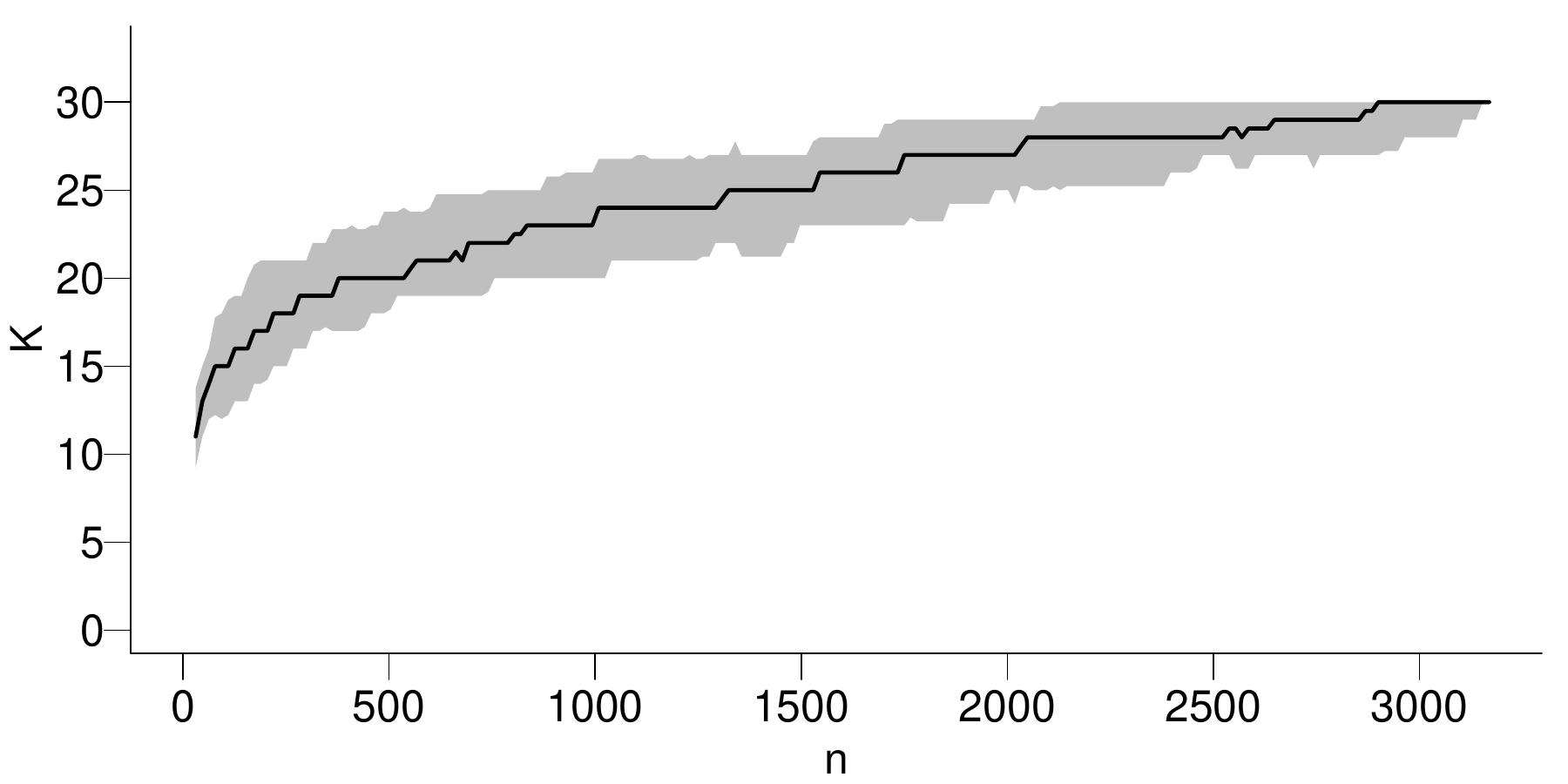}\\
        \includegraphics[width=0.48\linewidth]{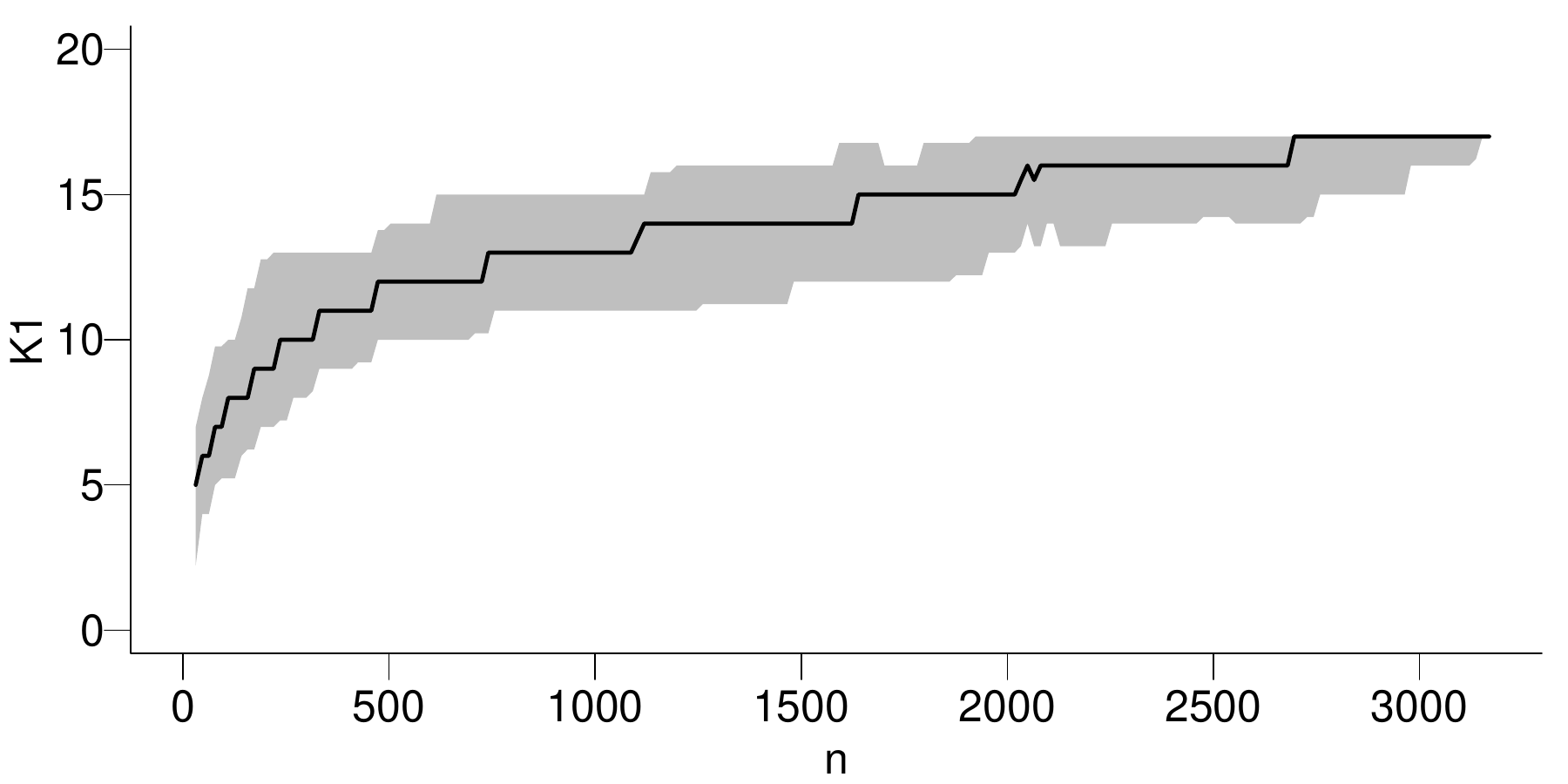}
        \hfill
        \includegraphics[width=0.48\linewidth]{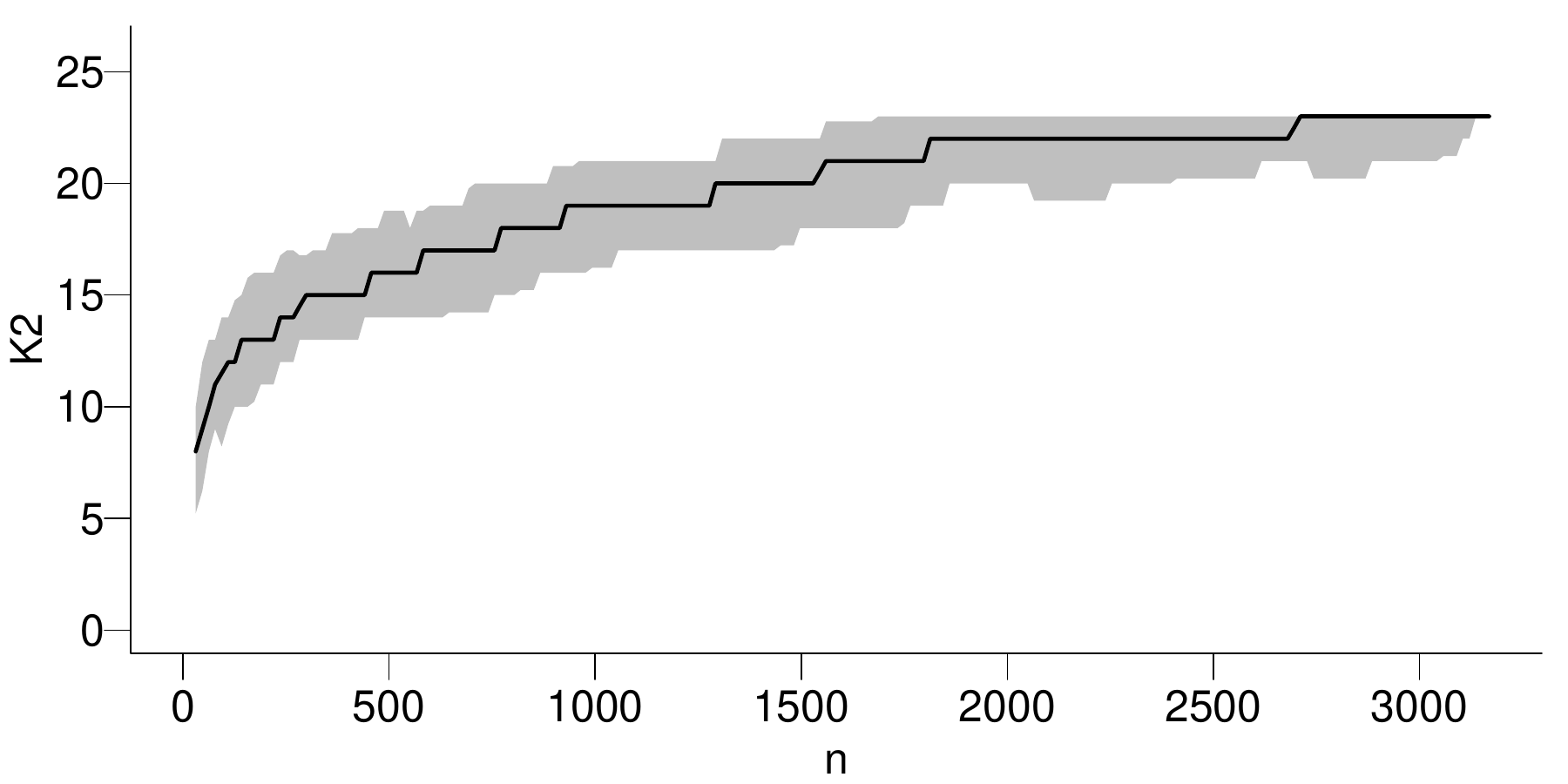}
        \caption{Accumulation curves Ants. }
        \label{fig:Ants_AccCrvs}
    \end{figure}

    \begin{figure}[ht!]
    	\centering
    	\includegraphics[width=0.48\linewidth]{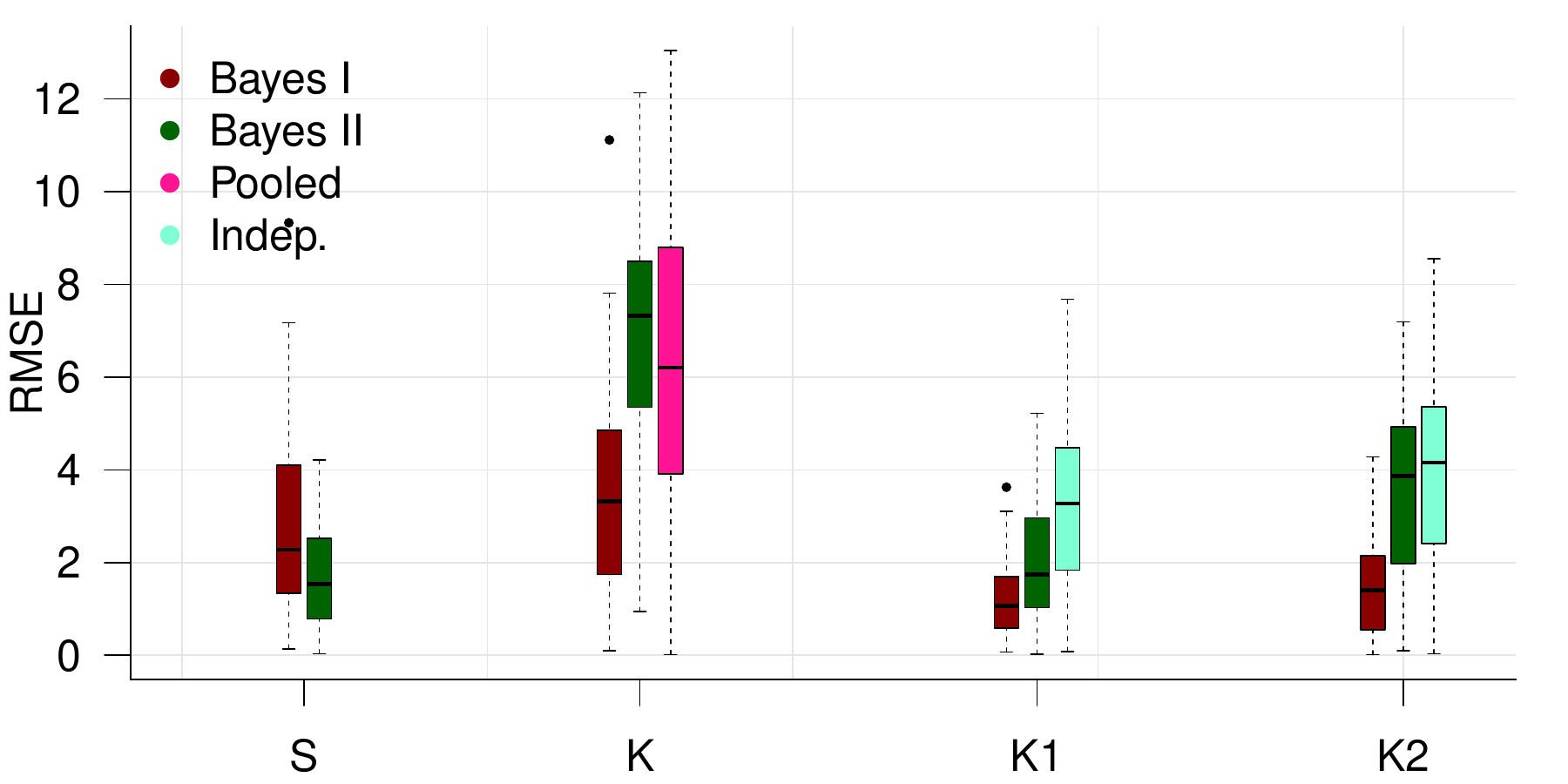}
    	\hfill
    	\includegraphics[width=0.48\linewidth]{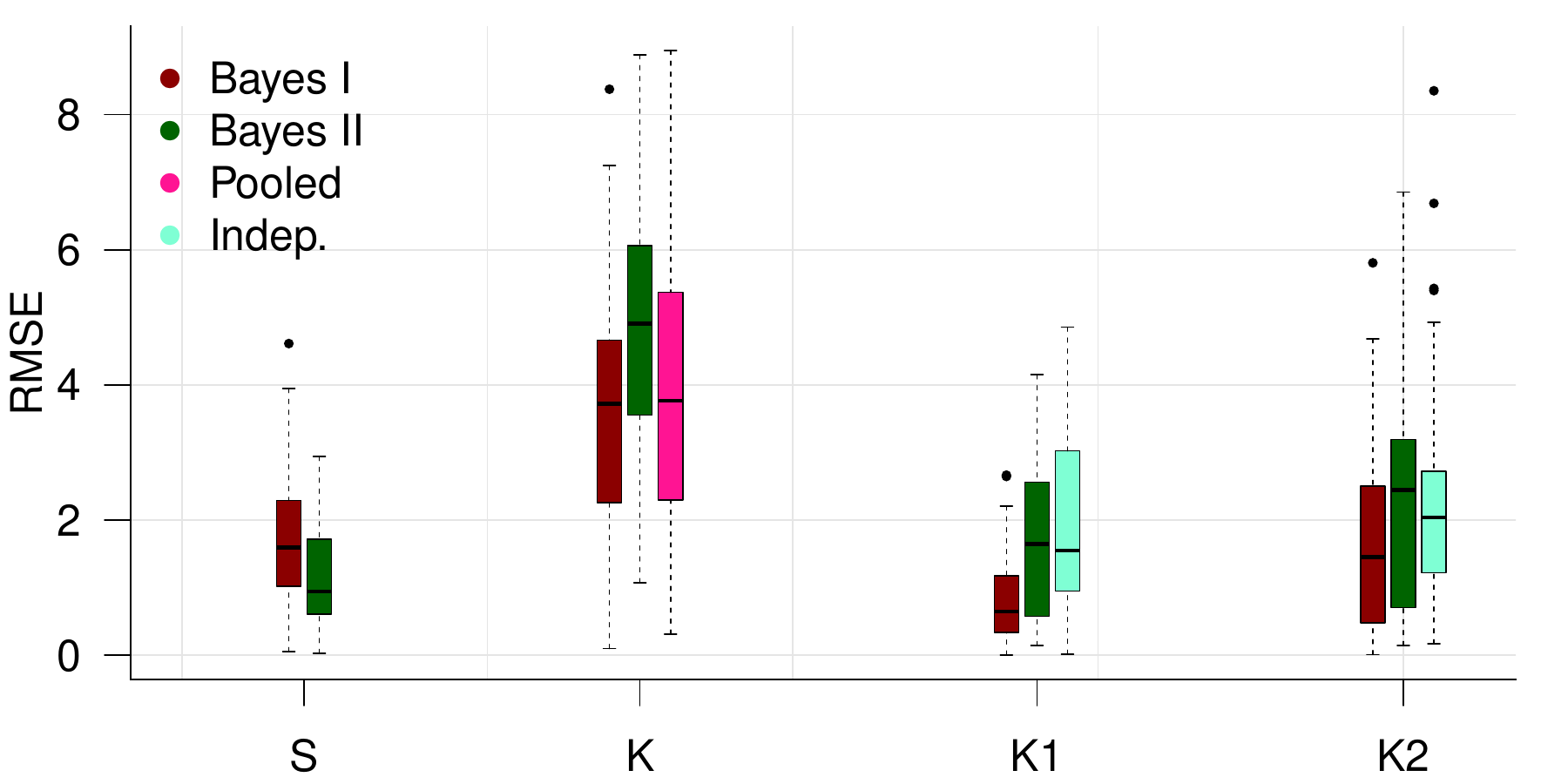}\\
    	\includegraphics[width=0.48\linewidth]{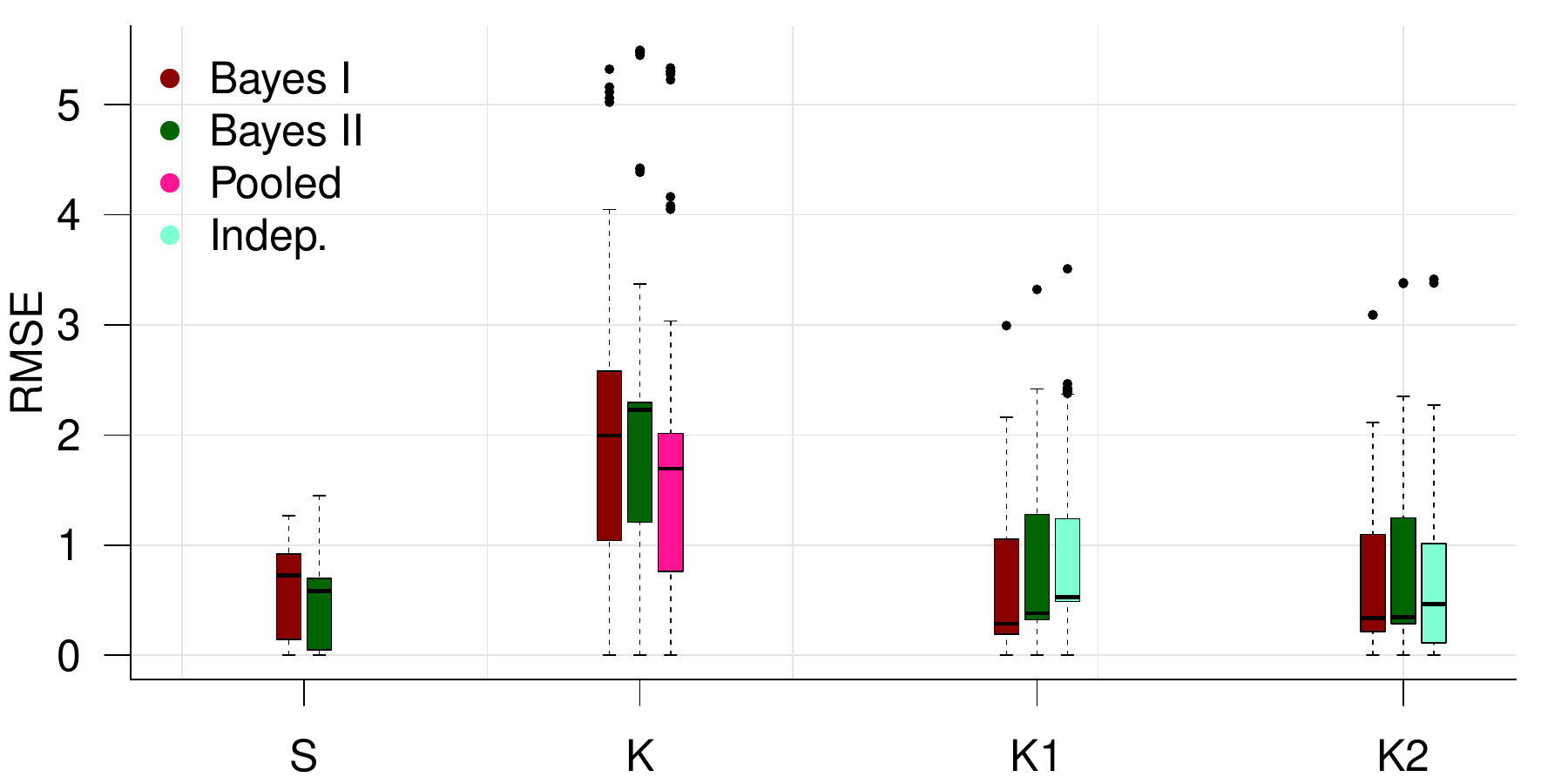}
    	\hfill
    	\includegraphics[width=0.48\linewidth]{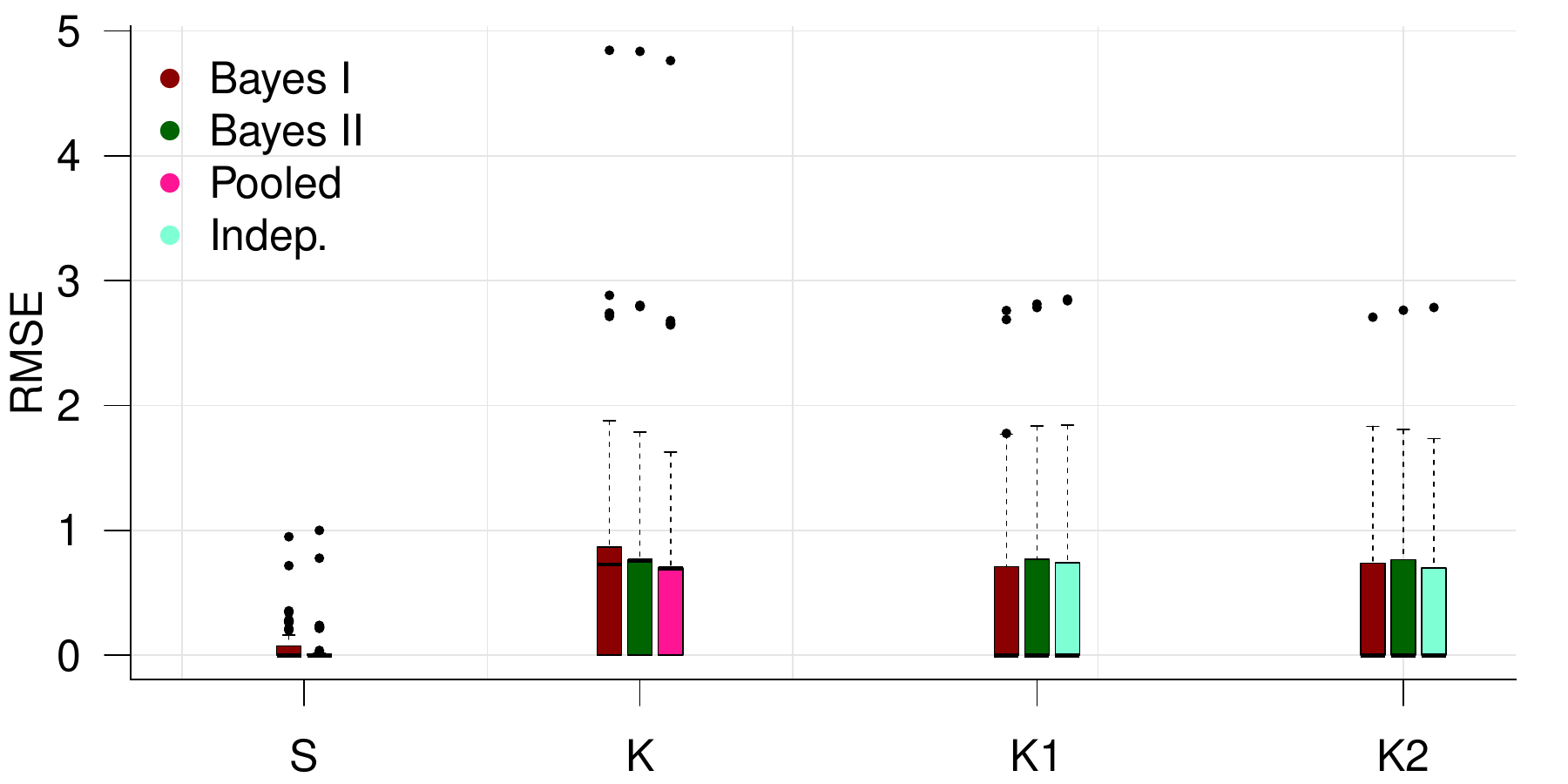}\\
        \caption{RMSE of out-of-sample predictions for new shared species, new global distinct species, and new local distinct species under different training-test splits. Training set proportions are $10\%$ (top-left), $30\%$ (top-right), $70\%$ (bottom-left), and $90\%$ (bottom-right) of the full dataset.}
    	\label{fig:Ants_application_Pred_rest}
    \end{figure}

    \begin{figure}[ht!]
    	\centering
    	\includegraphics[width=0.48\linewidth]{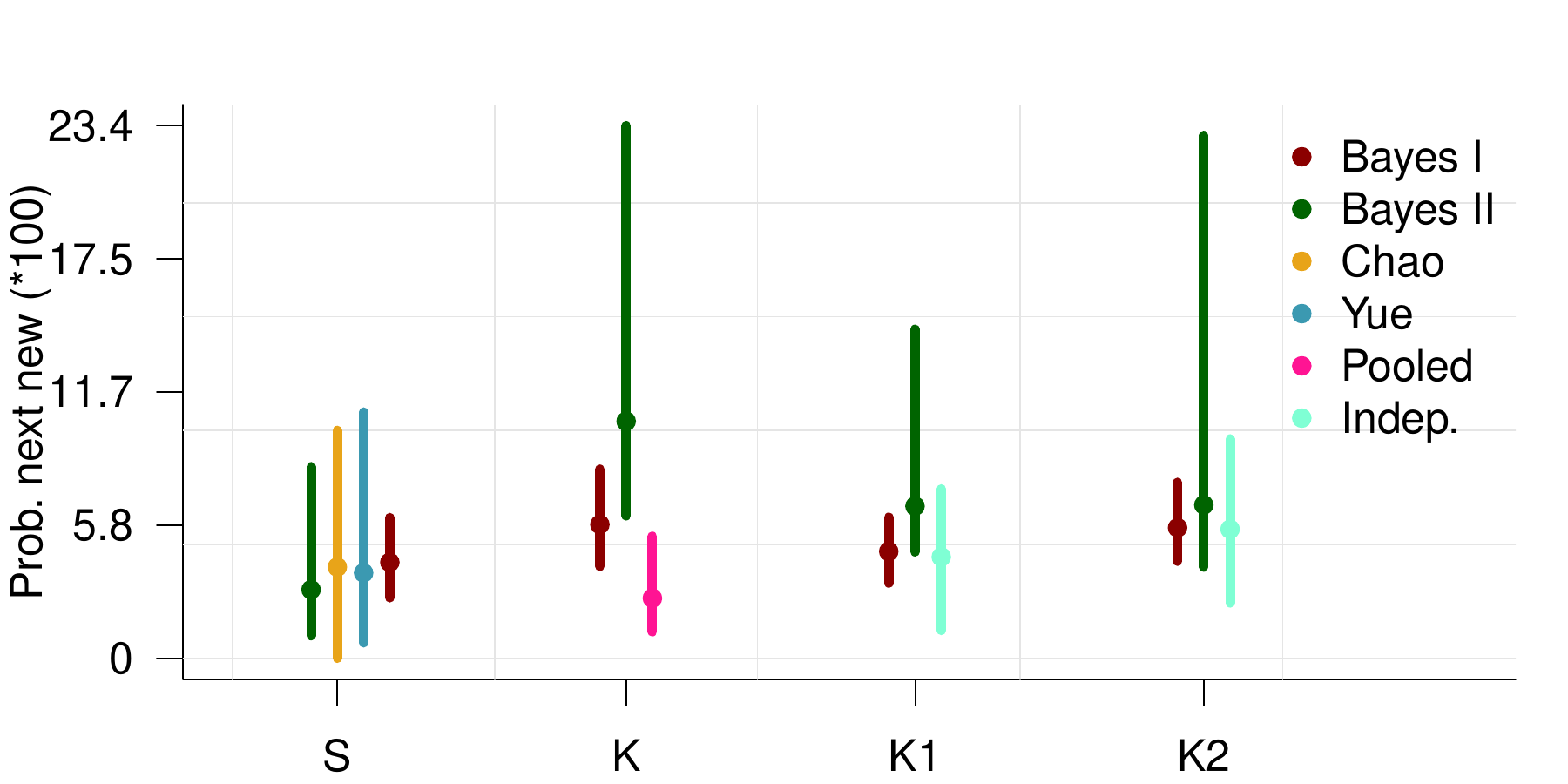}
    	\hfill
    	\includegraphics[width=0.48\linewidth]{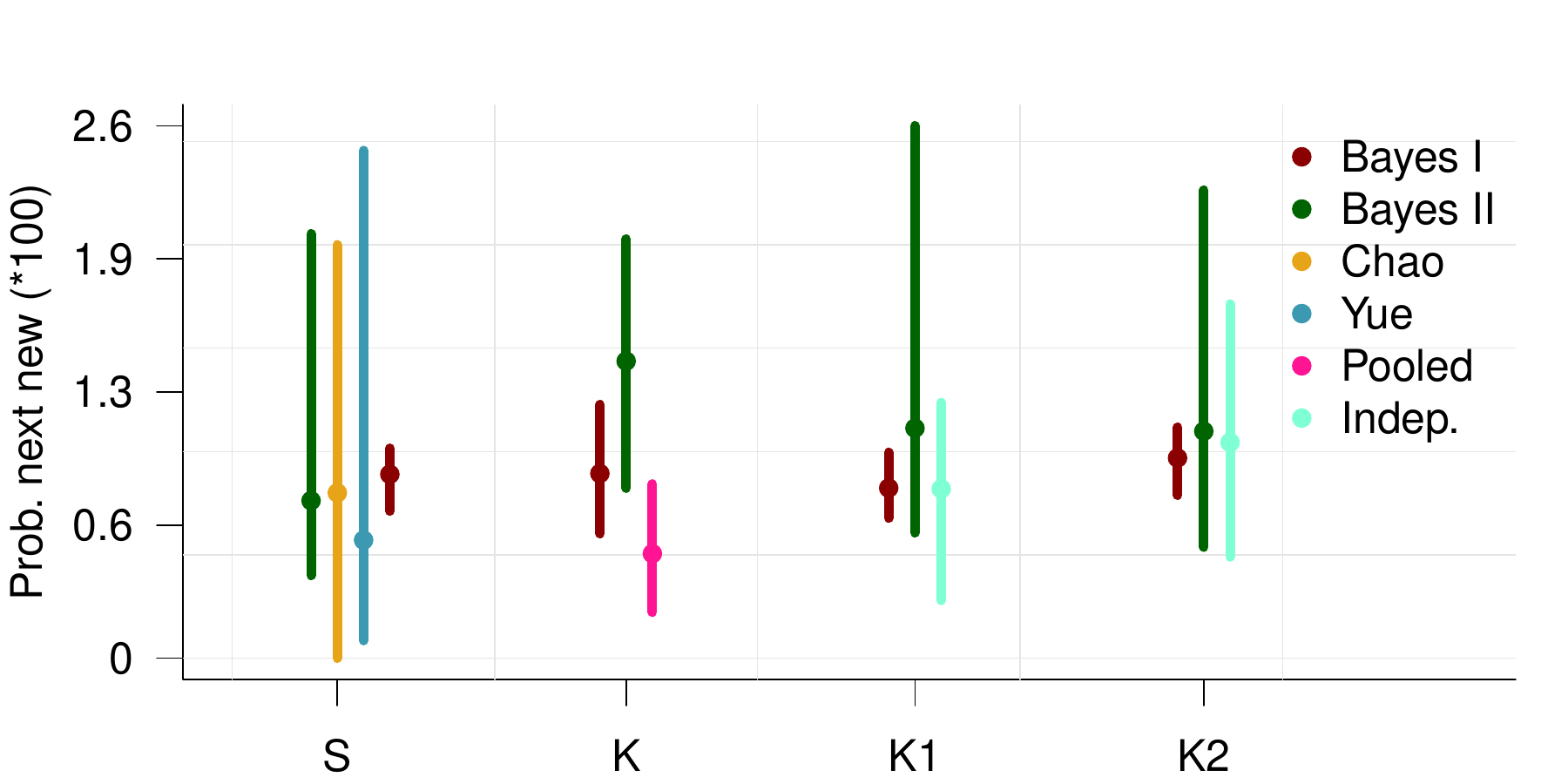}\\
    	\includegraphics[width=0.48\linewidth]{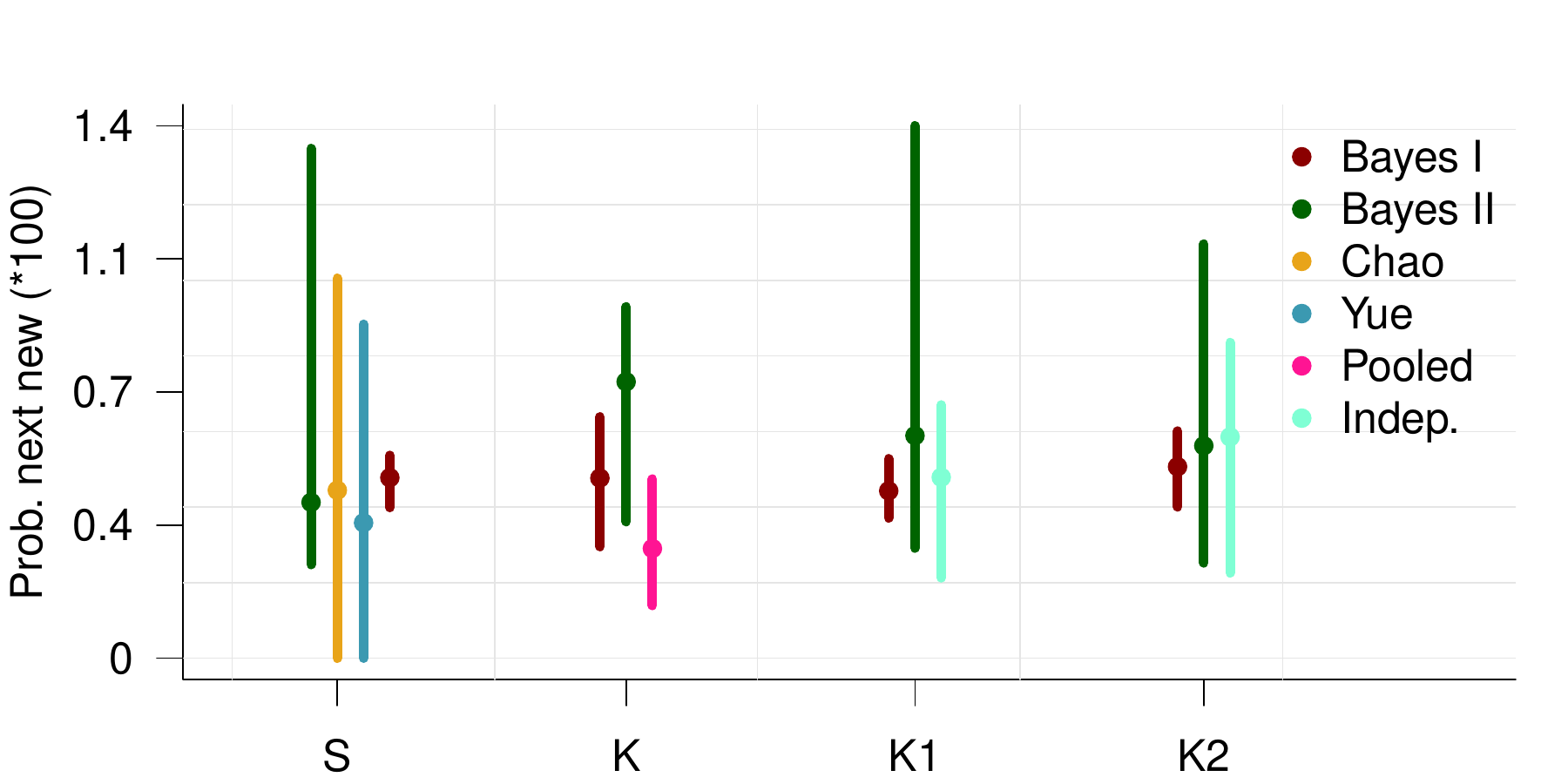}
    	\hfill
    	\includegraphics[width=0.48\linewidth]{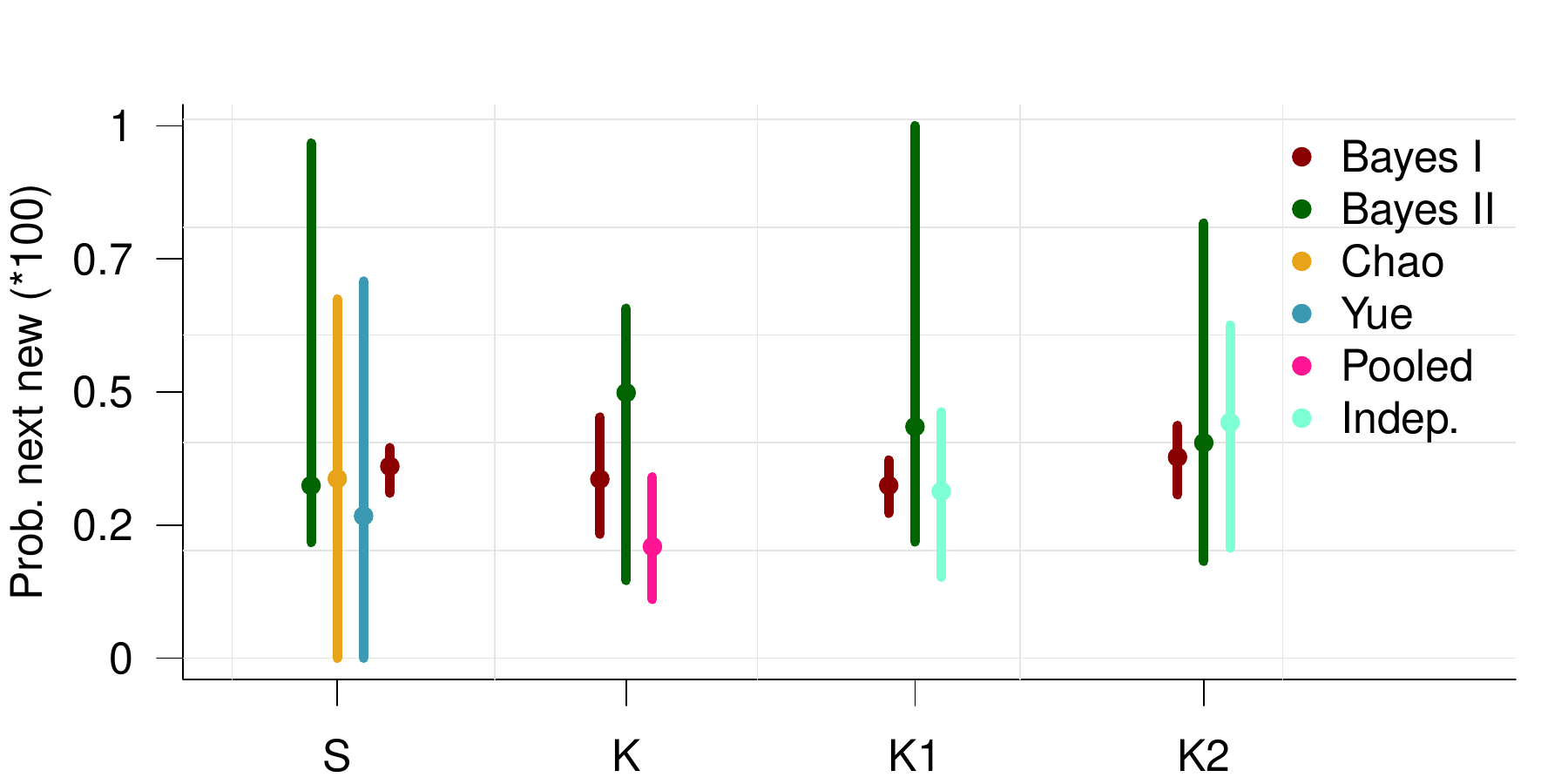}\\
        \caption{One-step-ahead discovery probability for new shared species, new global distinct species, and new local distinct species.
        Each panel represents a different sample size, $n=50$ (top-left), $n=250$ (top-right), $n=450$ (bottom-left) and $n=600$ (bottom-right).
        Probabilities on the rightmost plot have been multiplied by $100$ to improve readability.   }
    	\label{fig:Ants_application_Pr1step_rest}
    \end{figure}

\clearpage
\bibliographystyle{chicago}
\bibliography{references}

\end{document}